\documentclass[
12pt,		
openright,	
twoside,  
a4paper,			
chapter=TITLE,		
english,			
french,				
spanish,			
brazil				
]{USPSC-classe/USPSC}
\usepackage[T1]{fontenc}		
\usepackage[utf8]{inputenc}		
\usepackage{lmodern}		
\usepackage{cancel}

\usepackage{lastpage}			
\usepackage{indentfirst}		
\usepackage{color}				
\usepackage{graphicx}	
\usepackage{gensymb}
\usepackage{float} 				
\usepackage{chemfig,chemmacros} 
\usepackage{tikz}				
\usetikzlibrary{positioning}
\usepackage{pdfpages}
\usepackage{caption}
\usepackage{subcaption}
\usepackage{makeidx}            
\usepackage{hyphenat}          
\usepackage[absolute]{textpos} 
\usepackage{eso-pic}           
\usepackage{makebox}           
\usepackage[hang]{footmisc}
\usepackage{lipsum}
\usepackage{hyperref}

    \usepackage{setspace}
    \usepackage{mathtools}
    \usepackage{mathrsfs}
    \usepackage{amssymb}
    \usepackage{amsthm}
    \usepackage{amsmath}
\newtheorem{theorem}{Theorem}[section]
\newtheorem*{corollary*}{Corollary}
\newtheorem*{cosmic*}{Cosmic censor conjecture}
\newtheorem*{state*}{Stationary state conjecture}
\newenvironment{sproof}{%
  \proof}{\endproof}

\usepackage[superscript]{cite}              
\usepackage[num, abnt-emphasize=bf, abnt-thesis-year=both, abnt-repeated-author-omit=no, abnt-last-names=abnt, abnt-etal-cite, abnt-etal-list=3, abnt-etal-text=it, abnt-and-type=e, abnt-doi=doi, abnt-url-package=none, abnt-verbatim-entry=no]{abntex2cite} 
\usepackage{perpage}
\MakePerPage{footnote}

\renewcommand{\footnotesize}{\small} 

\usepackage{lipsum}				

\usepackage{multicol}	
\usepackage{multirow}	
\usepackage{longtable}	
\usepackage{threeparttablex}    
\usepackage{array}

\usepackage{USPSC-classe/ABNT6023-2018}
\usepackage{setspace} 
\siglaunidade{IFSC}
\programa{DFAe}
\definecolor{blue}{RGB}{41,5,195}

\makeatletter
\hypersetup{
	pdftitle={\@title}, 
	pdfauthor={\@author},
	pdfsubject={\imprimirpreambulo},
	pdfcreator={LaTeX with abnTeX2},
	pdfkeywords={abnt}{latex}{abntex}{USPSC}{trabalho acadêmico}, 
	colorlinks=true,       		
	linkcolor=blue,            	
	citecolor=blue,        		
	filecolor=magenta,      		
	urlcolor=blue,
	bookmarksdepth=4	
}
\makeatother

\makeindex
\begin{document}
\newtheorem{proposition}[theorem]{Proposition}
\numberwithin{equation}{section}
\let\oldchapter\chapter
\renewcommand{\chapter}{
  \renewcommand{\theequation}{\thechapter.\arabic{equation}}
  \oldchapter}
  \let\oldsection\section
\renewcommand{\section}{
  \renewcommand{\theequation}{\thesection.\arabic{equation}}
  \oldsection}
\selectlanguage{english}

\frenchspacing 

\renewcommand{\ABNTEXchapterfontsize}{\fontsize{12}{12}\bfseries}
\renewcommand{\ABNTEXsectionfontsize}{\fontsize{12}{12}\bfseries}
\renewcommand{\ABNTEXsubsectionfontsize}{\fontsize{12}{12}\normalfont}
\renewcommand{\ABNTEXsubsubsectionfontsize}{\fontsize{12}{12}\normalfont}
\renewcommand{\ABNTEXsubsubsubsectionfontsize}{\fontsize{12}{12}\normalfont}

\imprimircapa
\imprimirfolhaderosto*







\begin{agradecimentos}

I am grateful to my advisor, Prof. Dr. Esmerindo de Sousa Bernardes, for his help in the development of this work, which ranges from his availability to supervise it to the countless hours dedicated to discussions. This work would not have been developed without him.

I thank Prof. Dr. Daniel A. T. Vanzella for his lectures on general relativity and several hours dedicated to discussions that were fundamental to this work.

I thank Prof. Dr. Leo Maia for his lectures on statistical mechanics and many important and enthusiastic discussions on physics and mathematics.

I thank Prof. Dr. Fernando Mortari for his lectures on linear algebra and helpful comments on a draft of appendix A.

I thank the librarian Neusa for the revision of this dissertation.

This study was financed by the Coordenação de Aperfeiçoamento de Pessoal de
Nível Superior – Brasil (CAPES) – Finance Code 001.
	
\end{agradecimentos}

\begin{resumo}[Abstract]
 \begin{otherlanguage*}{english}
	\begin{flushleft} 
		\setlength{\absparsep}{0pt} 
 		\SingleSpacing  		\imprimirautorabr~~\textbf{\imprimirtitleabstract}.	\imprimirdata.  \pageref{LastPage}p. 
		\imprimirtipotrabalhoabs~-~\imprimirinstituicao, \imprimirlocal, 	\imprimirdata. 
 	\end{flushleft}
	\OnehalfSpacing 
The aim of this work is to present the black hole information problem and discuss the assumptions and hypotheses necessary for its formulation. As the problem arises in the framework of semiclassical gravity, we first review the necessary notions to describe Lorentzian manifolds equipped with physical properties, as well as the physical concepts of the theory that describes the gravitational interaction as the curvature of spacetime, general relativity. From its classical perspective, we develop the formalism to study the dynamical aspects of black holes in spacetimes obeying suitable causality conditions. Equipped with conjectures that nature censors naked singularities and that black holes reach a stationary configuration after they form, the black hole uniqueness theorems allow us to review several relations for the geometrical quantities associated with them. Following considerations of the other fundamental interactions, which are described by quantum field theory, we review the arguments in the formalism of quantum field theory in curved spacetime that give rise to the effective particle creation effect, its approximately thermal character, and the concept of black hole evaporation. With a precise quantification of information in quantum mechanics and assuming that the condition for physically acceptable states is given by the Hadamard condition, we review the result that entanglement between causally complementary regions is an intrinsic feature of quantum field theory. As a consequence, we discuss how the formation and complete evaporation of black holes leads to information loss. Conscious that such a prediction follows if no deviations from the semiclassical picture occur at the Planck scale, we discuss alternatives to this non-unitary dynamical evolution and formulate the black hole information problem. Lastly, we analyze the assumptions and hypotheses that lead to the problem.

   \vspace{\onelineskip}
 
   \noindent 
   \textbf{Keywords}: General relativity. Quantum field theory in curved spacetime.  Black hole thermodynamics. Black hole information problem.
 \end{otherlanguage*}
\end{resumo}


\setlength{\absparsep}{18pt} 
\begin{resumo}
	\begin{flushleft} 
			\setlength{\absparsep}{0pt} 
			\SingleSpacing 
			\imprimirautorabr~~\textbf{\imprimirtituloresumo}.	\imprimirdata. \pageref{LastPage}p. 
			\imprimirtipotrabalho~-~\imprimirinstituicao, \imprimirlocal, \imprimirdata. 
 	\end{flushleft}
\OnehalfSpacing 		

 O objetivo desse trabalho é apresentar o problema da informação em buracos negros e discutir as suposições e hipóteses necessárias para sua formulação. Como o problema surge com considerações de gravitação semiclássica, nós primeiramente revisamos as noções necessárias para descrever variedades Lorentzianas equipadas com propriedades físicas, assim como os conceitos físicos da teoria que descreve a interação gravitacional como a curvatura do espaço-tempo, relatividade geral. A partir da sua visão clássica, nós desenvolvemos o formalismo para estudar aspectos dinâmicos de buracos negros em espaço-tempos que obedecem a condições de causalidade adequadas. Equipados com conjecturas de que a natureza censura singularidades nuas e que buracos negros atingem configurações estacionárias após sua formação, os teoremas de unicidade de buracos negros nos permitem revisar várias relações para as quantidades geométricas associadas a eles. Seguindo conside\-rações das outras interações fundamentais, descritas por teoria quântica de campos, nós revisamos os argumentos no formalismo de teoria quântica de campos em espaço-tempo curvo que dão origem ao efeito de criação de partículas, seu caráter aproximadamente térmico, e o conceito de evaporação de buracos negros. Com uma quantificação precisa de informação em mecânica quântica e assumindo que a condição para estados fisicamente aceitáveis é dada pela condição de Hadamard, nós revisamos o resultado que emaranhamento entre regiões causalmente complementares é uma característica intrínseca de teoria quântica de campos. Como uma consequência, nós discutimos como a formação e evaporação completa de buracos negros resulta em perda de informação. Cientes de que tal previsão só é válida se não houver desvios das previsões semiclássicas na escala de Planck, nós discutimos alternativas para essa evolução dinâmica não-unitária e formulamos o problema da informação em buracos negros. Por último, nós analisamos as suposições e hipóteses que levam ao problema.

 \textbf{Palavras-chave}: Relatividade geral. Teoria quântica de campos em espaço-tempo curvo. Termodinâmica de buracos negros. Problema da informação em buracos negros.
\end{resumo}

\pdfbookmark[0]{\listfigurename}{lof}
\listoffigures*
\cleardoublepage






\pdfbookmark[0]{\contentsname}{toc}
\tableofcontents*
\cleardoublepage

\textual
\chapter[Introduction]{Introduction}\label{Introduction}

The purpose of this work is to present the black hole information problem, which arises due to predictions from two highly successful theories that describe the fundamental interactions of nature known to date. In essence, black holes are regions of spacetime predicted by the theory that describes the gravitational interaction, and when one considers the description of the electromagnetic and nuclear interactions, one is led to the conclusion that the evolution of these regions can result in information loss. There is, however, a question regarding the physical plausibility of this conclusion, mainly following from the expected limitations of the theories involved, which characterizes the black hole information problem. Building on the basis of each pertinent theory, we review the description of black holes and discuss the assumptions and hypotheses that lead to this question concerning their evolution. It should be noted that this work is self-contained in the sense that knowledge of undergraduate physics courses is sufficient to be able to follow the developments and arguments presented.

Consider, first, the gravitational interaction, which is the weakest of the known fundamental interactions. Its universal character, also known as the \textit{equivalence principle}, was first quantified by Newtonian gravity, which translated to the idea that the gravitational ``charge'' equaled the inertial mass of bodies (a hypothesis currently referred to as the \textit{weak equivalence principle}). Newtonian gravity gave a successful description of why objects fall as well as an explanation of the movement of astronomical bodies, i.e., Kepler's laws. As a matter of fact, these explanations were so accurate that they strengthened the fact that, at least in some limit, gravitational interaction must be reduced to the equations of motion predicted by the universal law of Newtonian gravity. Additionally, this theory gives an interesting result when one considers a body with a high mass concentrated in a small region, so that its escape velocity equals that of the speed of light. From a viewpoint in which light is a massive particle, Newtonian gravity then predicts that this body would be perceived as black by observers infinitely distant from it. Bodies with such an extreme gravitational field became known as \textit{dark stars}.

With the unification of electric and magnetic interactions, the lack of experimental evidence of a medium for the propagation of electromagnetic waves gave rise to the necessity of correcting coordinate transformation laws between inertial frames. Following the interpretation of this correction, the theory of special relativity was then proposed, which carried drastic consequences for the concepts of time and space (see, e.g., \cite{Wald1984} for an objective review) and was quickly found to be a successful explanation for several phenomena. In fact, the limit of the speed of signals postulated by special relativity, given by the speed of light, $c$, was one of the main reasons as to why Newtonian gravity was deemed to not be the best possible explanation of the gravitational interaction, since in such a framework, the interaction is instantaneous. Seeking to unify the concepts of special relativity with the gravitational interaction \cite{Misner1973}, \textit{general relativity} \cite{Einstein1916} was proposed, a \textit{classical} theory describing gravity as the curvature of spacetime generated by an energy distribution. By classical, it is meant that the characterization of the gravitational interaction described by it, given by the spacetime \textit{metric}, assumes definite values in any order of length, not only those much larger than atomic ones (i.e., distances of order $10^{-9}\;\text{m}$). In contrast, a \textit{quantum} theory takes into account, for instance, the intrinsic probabilistic nature of measurements in a system, wave-particle duality, and the discretization of possible values of energy and momentum, characteristics that become relevant at scales of order of atomic ones.

In a very brief summary, the content of general relativity follows from the \textit{Einstein equivalence principle} \cite{Will2018}, which states that the weak equivalence principle holds and that the outcome of any local\footnote{In the sense that inhomogeneities in the gravitational field can be neglected throughout the region where the experiment is carried out.} non-gravitational\footnote{The measurement of the electromagnetic interaction between two distributions of electric charge would constitute such an experiment. In contrast, measurement of the gravitational interaction between such distributions would not.} experiment is independent of its position in space and time (i.e., \textit{local position invariance}), as well as the speed of the frame of reference in free fall (i.e., \textit{local Lorentz invariance}). Stated in this manner, the Einstein equivalence principle can be interpreted as the physical equivalence of gravitational acceleration and inertial acceleration. As a consequence, a theory that obeys the Einstein equivalence principle has to be a \textit{metric theory of gravity} \cite{Will2014}, which means that a symmetric metric is defined on spacetime, freely falling bodies follow ``locally straight'' curves, and that for local freely falling frames of reference, non-gravitational physics must reduce to that predicted by special relativity. The main point of these principles and ideas is that the effects of gravitation must be described by a curved spacetime, in which a free energy distribution that ``bends'' it follows the ``straightest possible'' trajectories.

Although still a classical theory, general relativity provides more accurate results than Newtonian gravity. Indeed, the measurement of light deflection due to gravity \cite{Dyson1920} was the first experiment that consolidated general relativity as a superior theory to describe the gravitational interaction. Nevertheless, it should come as no surprise to the reader that, even to this day, one describes the launch of a rocket or even the entire solar system to an excellent order of approximation, using only Newtonian gravity. This is because in these descriptions one is dealing with a \textit{low curvature} regime (i.e., the curved spacetime is very similar to Minkowski spacetime), small time variations of the metric when compared to spatial ones, and a low speed regime (i.e., $v\ll c$), so that the predictions of general relativity reduce to those of Newtonian gravity \cite{Wald1984}. Still, there are critical philosophical differences concerning \textit{why} the movement of observers is described by the same equations of motion in such a regime, as a consequence of the curved spacetime. Perhaps the best example of the disparity that arises in the movement of observers in the framework of general relativity is the necessity to correct the passage of time for a satellite, not only because of its relative speed to the Earth's surface (as per special relativity), but also because of the difference in intensity of gravity. In this manner, the success of general relativity is undeniable \cite{Misner1973}, being in agreement with every experimental test made to date (see, e.g., \cite{LIGO2016} for one of the most notable and recent agreements, and \cite{Will2018} for an extensive review). Consequently, it is then natural to expect the predictions of general relativity to be of interest to further our knowledge of the gravitational interaction.

One of the most intriguing of these predictions \cite{Oppenheimer1939,Penrose1965, Penrose1969} is of gravitational collapse resulting in regions of spacetime where gravity acts in such a manner that nothing can escape from it \cite{Hawking1972a}, and that can also be related to pathologies in the structure of spacetime \cite{Geroch1968, Hawking1970}. Although these regions, known as \textit{black holes}, may seem similar to the dark stars that can exist in Newtonian gravity, they are fundamentally different. For instance, in the Newtonian theory, the body would be perceived as black only by observers sufficiently distant from it. Hence, it would be possible for observers inside or on the dark star to send signals to some observers outside of it. However, because of the nature of the gravitational interaction in the framework of general relativity and the postulate that an energy distribution can travel, at most, at the speed of light, observers inside the black hole would not be able to communicate with those outside of it, regardless of physical distance.

An extensive investigation of black holes in the purely classical framework of general relativity brought to light several important properties (e.g., \cite{Israel1967, Carter1971,Penrose1971, Hawking1971, Hawking1973, Bardeen1973}). Some of the most interesting of them are those summed up in the black hole \textit{uniqueness theorems} \cite{Heusler1996}, which are a series of results that state that under suitable conditions, any black hole that does not change over time must be completely characterized by three parameters: its mass, angular momentum, and electric charge. Any other physical property, such as the area of its surface or the gravitational acceleration on it as measured by distant observers, are all dependent only on these three parameters. Consequently, the lack of indistinguishable external features for time independent black holes with the same parameters also merits the nomenclature for these results as the \textit{no-hair theorems}.

In light of this, one can then identify the peculiar character of black holes when it comes to the accessibility of information concerning the energy distribution that gave rise to them. For instance, in the case of a dark star in Newtonian gravity, it would be possible for an observer to get closer to the matter distribution that gave rise to it, make measurements, and then share that data with any other observer. However, the same is not possible for a black hole, since if an observer enters it, it will never be able to communicate with observers outside of it. The significance of this line of reasoning is evident when one considers the conjecture that black holes must reduce to a time invariant configuration after they form. In particular, such an assumption is justified by the expectation that physical quantities not associated with conservation laws should be radiated away. Indeed, in the framework of general relativity, details about a gravitational collapse are expected to be radiated away in the form of \textit{gravitational radiation}, in a very similar fashion as an electric charge distribution radiates away its higher order multipole moments. Consequently, if a gravitational collapse produces a black hole that reduces to a time invariant configuration, the black hole uniqueness theorems state that such a black hole should be characterized only by three parameters. In this sense, details about the energy distribution that gave rise to it (e.g., its degrees of freedom or details about elementary particle composition) will forever be concealed from any observer that remains outside of it. Evidently, this does not mean that information is lost, but simply that a class of observers cannot have access to it, in contrast with a dark star in the Newtonian case.

These arguments also lead one to a question concerning the thermodynamic properties of black holes. In particular, since nothing can get out of a black hole, its physical temperature (i.e., the one associated with the emission of radiation in accordance with a black body spectrum form) should be null, which in turn would mean that black holes could act as a way to reduce the entropy of the universe. Considering the fact that thermodynamics is also a highly successful theory to describe macroscopic properties of systems in \textit{equilibrium}, an \textit{ad-hoc} proposal for black hole entropy \cite{Bekenstein1972,Bekenstein1973} surfaced. Although such a proposal was a way to make it so that black holes obey some ``generalized second law of thermodynamics'', there was no microscopic justification for it. This state of affairs developed further when several classical properties of black holes \cite{Bardeen1973} pointed to a mathematical analogy with the \textit{laws of thermodynamics}. However, because a black hole should have a vanishing physical temperature, the proposed correspondence between the classical characterization of black holes and the laws of thermodynamics would not merit a physical status.

Further analysis of black holes then required consideration of the other fundamental interactions, which are currently described by \textit{quantum field theory} (see, e.g., \cite{Weinberg1995, Mandl1984}). In essence, this theory can be understood as the union of quantum mechanics and special relativity, which takes into account not only relativistic effects but also systems with varying particle content. Additionally, quantum field theory possesses the very attractive feature of being a description of fundamental interactions in terms of fields, which, for example, is the fundamental quantity of the electromagnetic interaction. As a matter of fact, it has provided a unified description of the electromagnetic and weak interactions, known as the \textit{electroweak} interaction. So far, this theory has proven to be in excellent agreement with experiments to describe the electroweak and \textit{strong} nuclear interactions, and indeed, it is the basis of the \textit{Standard Model} (see, e.g., \cite{Nagashima2010, Peskin1995}). Perhaps the best example of its success is the extremely accurate agreement with the experimental value of the ratio of the electron's magnetic and angular momentum \cite{Fan2022}. Being highly successful in this sense, quantum field theory also provides an interesting feature of the interpretation of particles as excitations of a field. Very much so, one of its most fascinating predictions is the \textit{Fulling-Davies-Unruh effect} \cite{Fulling1973,Davies1975,Unruh1976}, which states that a uniformly accelerated observer measures a thermal bath of particles, demonstrating the dramatic feature that the particle content of a quantum state is \textit{observer dependent}. Nevertheless, the predictions of quantum field theory are made by studying fields in Minkowski spacetime, i.e., not taking into account gravitation.

The consideration of gravity and the other fundamental interactions simultaneously is currently best made by the framework of \textit{quantum field theory in curved spacetime} \cite{Parker1968, Parker1969, Parker2009, Wald1994, Birrell1984, Fulling1989}. This formalism, also referred to as \textit{semiclassical} gravity, considers the propagation of quantized fields on a curved spacetime background. More precisely, gravitation is still treated classically, being described by a spacetime metric that has definite values in any order of length, while the quantum fields propagate on it. Consequently, semiclassical gravity is expected to provide a good description of quantum effects in gravitation for scales in which general relativity is an adequate description of gravity \cite{Wald1984}. Dimensional arguments \cite{Misner1973} suggest that this will not be the case when the scales of the spacetime structure reach the \textit{Planck length},
\begin{equation}\label{planck} \ell_p=\left(\frac{G\hbar}{c^3}\right)^{1/2}\approx 1.62\;10^{-35}\;\text{m},\end{equation}
where $G$ is the \textit{Newtonian constant of gravitation} and $\hbar$ is the \textit{reduced Planck constant}. For instance, when studying gravitation in a \textit{high curvature} regime, such that the curvature of spacetime is of order $\ell_p^{-2}$, one justifiably expects that the classical description of gravity will no longer be adequate, in the sense that details about the quantum structure of the gravitational interaction should become relevant. It is in this manner that general relativity is expected to be only an approximation on scales much larger than the Planck length, the \textit{Planck time},
\begin{equation}t_p=\left(\frac{G\hbar}{c^5}\right)^{1/2}\approx 5.39\;10^{-44}\;\text{s},\end{equation}
and the \textit{Planck mass},
\begin{equation} m_p=\left(\frac{\hbar c}{G}\right)^{1/2}\approx 2.18\;10^{-8}\;\text{kg}.\end{equation}

Hence, the development of a theory that takes into account a quantum description of gravity should provide a much better picture of phenomena in regimes of extreme gravity, such as black holes. In fact, even though phenomena at the Planck scale may be out of experimental reach for the foreseeable future, the scales of interest for a possible unification of the electroweak and strong interactions suggest that quantum gravitational effects should be of importance. However, the construction of this more fundamental, \textit{quantum gravity} theory, is an open problem. Nonetheless, the predictions of semiclassical gravity have been promising and extensive in many ways, and it may very well be interpreted as an intermediate step in the development of the quantum theory of gravity. More specifically, following the particle interpretation of semiclassical gravity, some of the main results concern the creation of particles due to the expansion of the universe and in spacetimes containing a black hole. The former provided a justification for the anisotropies of the \textit{Cosmic Background radiation} \cite{Parker2009}, while the latter was responsible for giving rise to the concept of black hole \textit{evaporation}, as a consequence of the effective emission of radiation with an approximately thermal character (in the sense that it approximately obeys a black body spectrum). This effective particle creation effect is also referred to as the \textit{Hawking effect} \cite{Hawking1974, Hawking1975}.

When one considers the effective particle creation by a black hole predicted by semiclassical gravity, the almost full picture of formation and evaporation of most black holes can be analyzed with certainty. In essence, suppose one starts with a time independent configuration, which is described by its mass, angular momentum, electric charge, degrees of freedom, baryon and lepton number, and any other physical property. From the perspective of quantum mechanics, complete knowledge of the quantum state of this energy distribution would correspond to a \textit{pure} state. An example of this would be a star at the early stages of its life, composed mainly of hydrogen and helium in which one has all the possible information about the energy distribution. Of course, such a configuration would not be exactly time independent due to nuclear fusion processes, but in regions far away from the star at the very early stages of its life, it would be a good approximation to describe it by a time independent spacetime. Now, as this energy distribution evolves in time, if it is massive enough, the gravitational effects will not be sustained, and a black hole will form as a consequence of gravitational collapse. Even though the black hole formation process is expected to be highly dependent on the details of the energy distribution and how it collapses, the asymmetries of the distribution are expected to be radiated away in the form of gravitational radiation. Following the conjecture that black holes must eventually reach a time independent state, the initial energy distribution will then give rise to a black hole described only by three parameters: its mass, $M$, angular momentum, $L$, and electric charge, $q$. At the same time, the time independent black hole should behave as a gray body\footnote{Gray in the sense that its black body spectrum has a factor due to transmission probabilities.}, effectively emitting an approximately thermal spectrum of radiation. This radiation carries away the mass, angular momentum, and electric charge of the black hole, and after some finite time, the black hole should evaporate completely, leaving only radiation. From the perspective of quantum mechanics, this condition of the energy distribution would correspond to a \textit{mixed} state, which arises in a context in which one does not have all the possible information about a state. Therefore, one would then have a configuration in which information about the details of the energy distribution that gave rise to the black hole is completely lost \cite{Hawking1976, Unruh2017}. Fig. \ref{fig:intro} illustrates the process of black hole formation and complete evaporation, as well as how it can lead to loss of information.

 \begin{figure}[h]
\centering
\includegraphics[scale=1.1]{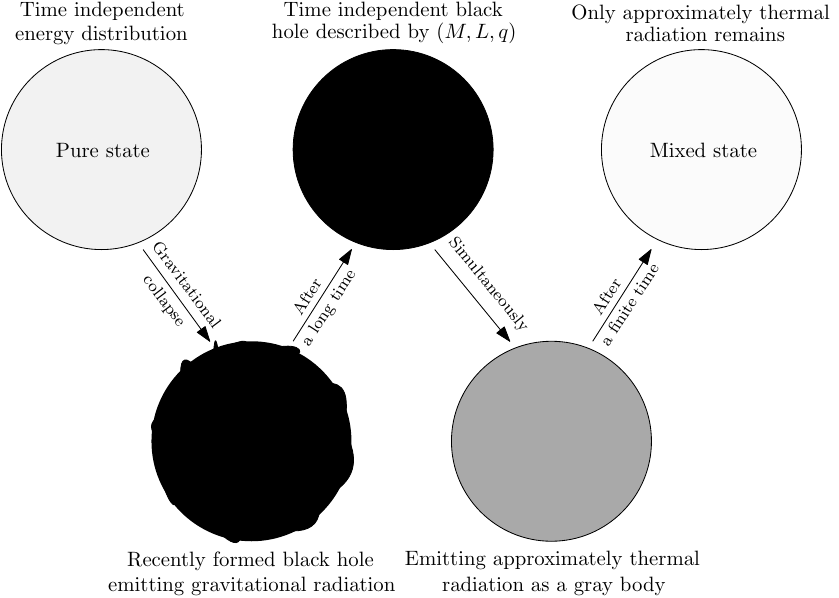} 
\caption{Process of black hole formation and complete evaporation.}
\caption*{Source: By the author.}
\label{fig:intro}
\end{figure}

Evidently, such a prediction of \textit{information loss} follows from the expectation that the evaporation process occurs completely, and lacking a complete theory of quantum gravity, it is still unclear how the evolution of a black hole precisely occurs when it reaches the Planck scale. In summary, the \textit{black hole information problem} can be stated as a question regarding the final state of a black hole in light of its evaporation process. In addition, the classical properties of black holes and the character of the effective particle creation effect point to a possible connection between gravitation, quantum theory and thermodynamics. Arguably, this question and state of affairs surrounding the semiclassical description of black holes are the best clues available for the development of a satisfying theory of quantum gravity. The purpose of this work is to present a review of the classical and semiclassical properties of black holes, a description of the process of evaporation and formulation of the black hole information problem, as well as an analysis of the assumptions and hypotheses that each of the pertinent theories brings to the evaporation process. The organization is as follows.

Chapter \ref{chapter1} concerns the properties of spacetimes and some physical aspects of general relativity necessary for the description and analysis of black holes. We discuss how symmetries can be defined in curved spacetimes and how to evaluate the conserved quantities associated with them. The causal structure of general spacetimes is also analyzed, in which the condition for a spacetime to be considered deterministically and causally ``well behaved'' is presented. Further concepts of interest, such as restrictions on the energy-momentum tensor of a suitable energy distribution, the dynamics of null curves, a definition of a singularity, and a satisfying notion of what it means for a distribution to be isolated, are also discussed in detail.

Chapter \ref{chapter2} presents a definition of the black hole region of a spacetime and the derivation of several properties that follow from chapter \ref{chapter1}. These purely classical results provide relations to the three parameters that describe a time invariant black hole and restrictions on the evolution of general black holes for spacetimes obeying suitable causality conditions. Namely, we will show that these results are rigorous derivations that follow from differential geometry and discuss two conjectures that concern the existence and evolution of black holes.

Chapter \ref{chapter3} provides an introduction to the elements of quantum field theory in curved spacetime that are necessary for the derivation of the effective particle creation effect by black holes. We will see that these results follow from the interpretation of particles as excitations of a field, the conjecture that black holes reduce to time independent configurations, and the classical properties discussed in chapter \ref{chapter2}. We also mention the interpretation of classical properties in light of a thermodynamic perspective, which concerns the possible physical entropy of a black hole and the translation of the effective particle creation process as a consequence of its physical temperature. Additionally, using the formalism of density operators and the assumption that physically acceptable states are those for which some notion of ``energy-momentum expectation value'' can be well defined, we show that entanglement between states inside and outside a black hole is an intrinsic feature of quantum field theory in curved spacetime.

Chapter \ref{chapter4} introduces the black hole information problem. Following the dynamical evolution of a black hole in light of the effective particle creation effect, we show that information will be lost if the evaporation process occurs completely, in the sense that an initial pure state will inevitably evolve into a mixed state. We briefly discuss alternatives to this result and conclude the chapter by discussing the assumptions and hypotheses that lead to the black hole information problem.

Chapter \ref{conclusion} concludes this work with the final remarks and mention of perspectives with regard to further developments on the black hole information problem.

In appendix \ref{A}, we present a precise definition of spacetime and an objective review of the tools necessary to describe it: the geometrical quantities associated with it and the objects that are used to evaluate how they change over events. In particular, readers not familiar with differential geometry and topology are advised to read appendix \ref{A} before chapter \ref{chapter1}. No prior knowledge of general relativity is necessary to start from chapter \ref{chapter1} if one has a solid basis in differential geometry and topology.

In appendix \ref{apb}, we present pertinent results concerning the most general spacetime that describes time independent black holes. These results will mainly be of use in some developments of chapters \ref{chapter2} and \ref{chapter3}, and the reader will be advised to consult it when necessary.

In appendix \ref{information}, we present an objective review of the characterization of information in quantum mechanics. Namely, the formalism of density operators, the von Neumann entropy, and the concept of entanglement. The concepts presented there will be of use for arguments and developments in chs. \ref{chapter3} and \ref{chapter4}, and the reader will be advised to consult it when necessary.

Regarding notation, the signature of the spacetime metric is adopted to be given by $-+\ldots+$. The notation used for tensors is known as the \textit{abstract index notation} \cite{Wald1984}. In particular, Greek letters represent abstract indices, while Latin letters represent concrete indices (see appendix \ref{A2} for details on tensors and this notation). Other notation is introduced as needed. Lastly, most of the calculations will be developed in SI units, in which the pertinent constants will be inserted with the physical parameters. For example, we will refer to the mass, $M$, of a region of spacetime by its \textit{Schwarzschild radius}, 
\begin{equation}\label{schr}
    r_s=\frac{2GM}{c^2}\approx 2.94\;10^{3} \left(\frac{M}{M_\odot}\right)\; \text{m},
\end{equation}
where $M_\odot$ is the Sun's mass. However, with a few exceptions of lengthy developments, \textit{geometrical} (i.e., $G=c=1$) or \textit{natural} (i.e., $G=c=\hbar=k_B=1$, where $k_B$ is Boltzmann's constant) units will be adopted.

 \chapter{Spacetime and general relativity}\label{chapter1}

General relativity is a theory that describes gravitation as the curvature of space and time due to the presence of an energy distribution. It postulates that the universe is a four-dimensional spacetime (see appendix \ref{A2} for the definition of a spacetime) whose Lorentzian metric, $g_{\mu\nu}$, is related to the energy-momentum tensor, $T_{\mu\nu}$, by \textit{Einstein's equation} \cite{Wald1984,Hanoch2023},
\begin{equation}\label{eq1}
    R_{\mu\nu}-\frac{1}{2}Rg_{\mu\nu}=\frac{8\pi G}{c^4}T_{\mu\nu},
\end{equation}
where $R_{\mu\nu}$ is the Ricci tensor and $R$ is the Ricci scalar (see appendix \ref{A3}). Of course, in order for the theory to be consistent with special relativity, it also postulates that the speed of light is a universal constant (which follows from a Lorentzian metric) and that the laws of physics are the same in all inertial frames of reference (i.e., the principle of relativity). Following the Einstein principle of equivalence, the effect of gravity is not to accelerate test bodies, but rather, to shape the paths that they follow on the curved spacetime. Indeed, the physical content of Einstein's equation can perhaps be best summed in Wheeler's words \cite{Wheeler2000}:

\noindent``\textit{Spacetime tells matter how to move; matter tells
spacetime how to curve.}''

Evidently, solving eq. \ref{eq1} in general is not a trivial task, as it reduces to a system of non-linear coupled differential equations. Nonetheless, much can be said about the structure of a physical spacetime (i.e., one whose metric obeys Einstein's equation) given the form of eq. \ref{eq1} and by analyzing the physical content of a Lorentzian metric. In other words, in this chapter we will not be interested in particular solutions of Einstein's equation, but rather, in how its form can lead to results concerning physics under the lens of general relativity. Hence, most of the results that we will review in this chapter are highly geometrical, in the sense that they are intrinsic properties of spacetimes. Regardless, one may dive deeper into their meaning by considering Einstein's equation and physical concepts, such as restrictions on the energy-momentum tensor and how gravitational interaction is expected to behave far from sources.

Although the developments that follow are detailed enough for a complete description of black holes, some aspects of general relativity which are also minimally relevant for our discussion will not be given the same level of attention and detail. For instance, in light of Einstein's equation, a perturbed Minkowski spacetime metric leads one to the conclusion that undulations of curvature can propagate in spacetime. A detailed analysis of this gravitational radiation, known as \textit{gravitational waves}, can be found in, e.g., \cite{Wald1984}. Moreover, a treatment of \textit{cosmological models}, i.e., solutions of Einstein's equation which are in agreement with large-scale experimental data (such as homogeneous and isotropic distribution of galaxies), can be found in, e.g., \cite{Wald1984}.

This chapter is organized as follows. We first use the notion of diffeomorphisms to give a precise definition of the symmetries of a spacetime, which in turn gives us a prescription for conserved quantities. Second, we discuss the causal structure of arbitrary spacetimes, while analyzing constraints and global properties on those deemed physically reliable. We then turn our attention to the possible restrictions on the energy-momentum tensor, considering how one expects physical, classical matter to behave. Next, we discuss the dynamics of null geodesics, and how one can use geodesics in general to give a satisfying notion of ``singularities'' in the structure of a spacetime. We conclude this chapter by presenting the concept of asymptotic flatness, which gives a precise definition of what it means for a distribution of energy to be isolated in the framework of general relativity. These developments rely heavily on several notions of differential geometry and topology, and we refer the reader to appendix \ref{A} for an objective review of them. 

\section{Symmetries}\label{symmetry}

A symmetry in a physical system is defined by a continuous or discrete coordinate transformation that preserves the system. They are of interest because they not only facilitate the analysis of systems, but also give rise to conserved quantities. In the context of general relativity, a symmetry can be associated with a coordinate transformation for which the structure of the spacetime, i.e., the metric, is invariant. In this manner, a precise definition of a symmetry can be given by considering that diffeomorphisms can be interpreted as coordinate transformations.

More specifically, let $(M,g_{\mu\nu})$ be a spacetime and $\psi_s:\mathbb{R}\times M\to M$ be a one-parameter group of diffeomorphisms (see appendix \ref{derivativeoperators}). For a fixed $s$, if the action of $\psi_s$ leaves $g_{\mu\nu}$ unchanged, i.e., $\psi^*_s g_{\mu\nu}=g_{\mu\nu}$ for all $a\in M$, then $\psi_s$ is said to be a symmetry transformation for $g_{\mu\nu}$. In particular, each map is then known as an \textit{isometry}, and $\psi_s$ is said to be a one-parameter group of isometries. Thus, the necessary and sufficient condition for a vector, $\chi^{\mu}$, to be the generator of a one-parameter group of isometries is that the Lie derivative of the metric with respect to it vanishes. From eq. \ref{Liemetric}, this condition reads
\begin{equation}\label{eq2}
    \mathcal{L}_{\chi}g_{\mu\nu}=\chi^{\alpha}\nabla_{\alpha}g_{\mu\nu}+g_{\nu\alpha}\nabla_{\mu}\chi^{\alpha}+g_{\mu\alpha}\nabla_{\nu}\chi^{\alpha}=0,
\end{equation}
but since $\nabla_{\mu}$ is the Levi-Civita connection, eq. \ref{eq2} reduces to 
\begin{equation}\label{eq3}
    \nabla_{(\mu}\chi_{\nu)}=0,
\end{equation}
where the parentheses enclosing the indices are representative of a symmetric permutation of them (see eq. \ref{symmetric}). Eq. \ref{eq3} is known as \textit{Killing's equation}, and any vector that satisfies it is called a \textit{Killing vector}. Consequently, if $\chi^{\mu}$ is the generator of the one-parameter group of isometries, $\psi_s$, then the orbit of $\psi_s$ “follows” the path for which the metric is invariant in a spacetime, and $\chi^{\mu}$ “points” in that direction. In this sense, Killing vectors precisely capture the notion of a \textit{continuous} symmetry. On the other hand, \textit{discrete} symmetries are not generated by Killing vectors. Nonetheless, these discrete coordinate transformations (e.g., reflections) are relevant to a spacetime, and should be regarded as elements of its full group of isometries (see \S~\ref{schsec} for an example).  

An immediate result of Killing's equation is the conservation of a quantity along geodesics. Let $\chi^{\mu}$ be a Killing vector and $\gamma$ be a geodesic parametrized by an affine parameter, $t$, with tangent vector $s^{\mu}$. Then

\begin{equation}\label{eq4}
    \frac{d(s^{\mu}\chi_{\mu})}{dt}=s^{\nu}\nabla_{\nu}(s^{\mu}\chi_{\mu}) =s^{\nu}s^{\mu}\nabla_{\nu}\chi_{\mu}+\chi_{\mu}s^{\nu}\nabla_{\nu}s^{\mu}=0,
\end{equation}
as the first term vanishes by being the contraction of a symmetric and an antisymmetric tensor, while the second term vanishes as a consequence of the geodesic equation (see eq. \ref{aeq8}). Eq. \ref{eq4} can be interpreted as stating that $s^{\mu}\chi_{\mu}$ is constant along $\gamma$, i.e., as $t$ varies. The physical significance of such a conserved quantity can be analyzed straightforwardly when one is given the explicit form of the Killing vector, as will be discussed in detail when we study a specific metric in ch. \ref{chapter2}. Although this conserved quantity along geodesics is an interesting result, this is not the conserved quantity associated with the conservation law that arises due to the continuous symmetry generated by $\chi^{\mu}$, as per Noether's theorem \cite{Greiner1996}.

In order to derive this law, one needs to consider the failure of multiple applications of the Levi-Civita connection on $\chi^{\mu}$ to commute \cite{Dowker}. From eq. \ref{Riedual}, this failure is given by
\begin{equation}\label{eq39}
    (\nabla_{\mu}\nabla_{\nu}-\nabla_{\nu}\nabla_{\mu})\chi_{\alpha}=R_{\mu\nu\alpha}\mathstrut^{\beta}\chi_{\beta},
\end{equation}
which due to Killing's equation, eq. \ref{eq3}, can be rewritten as
\begin{equation}\label{19}
    \nabla_{\mu}\nabla_{\nu}\chi_{\alpha}+\nabla_{\nu}\nabla_{\alpha}\chi_{\mu}=R_{\mu\nu\alpha}\mathstrut^{\beta}\chi_{\beta}.
\end{equation}
Relabeling the indices in eq. \ref{19} allows one to write
\begin{equation}\label{22}
    \nabla_{\nu}\nabla_{\alpha}\chi_{\mu}+\nabla_{\alpha}\nabla_{\mu}\chi_{\nu}=R_{\nu\alpha\mu}\mathstrut^{\beta}\chi_{\beta},
\end{equation}
\begin{equation}\label{35}
    \nabla_{\alpha}\nabla_{\mu}\chi_{\nu}+\nabla_{\mu}\nabla_{\nu}\chi_{\alpha}=R_{\alpha\mu\nu}\mathstrut^{\beta}\chi_{\beta},
\end{equation}
which by summing eq. \ref{19} with eq. \ref{22} and subtracting eq. \ref{35}, as well as using eq. \ref{Rsecond}, yields
\begin{equation}\label{kr}
    \nabla_{\nu}\nabla_{\alpha}\chi_{\mu}=-R_{\alpha\mu\nu}\mathstrut^{\beta}\chi_{\beta}.
\end{equation}
By taking the trace over the first and third indices of the Riemann tensor, one finds a relation between the Ricci tensor and the Killing vector,
\begin{equation}\label{kill}
    \nabla_{\nu}\nabla^{\nu}\chi^{\mu}=-R^{\mu\nu}\chi_{\nu}.
\end{equation}
The significance of eq. \ref{kill} is given by the following proposition and line of reasoning.

\begin{proposition}\label{PropKomar}
Let $A_{\mu\nu}$ be a two-form and $\ell^{\nu}=\nabla_{\mu}A^{\mu\nu}$. Then $\nabla_{\mu}\ell^{\mu}=0$.
\end{proposition}
\begin{proof}
Contracting the Levi-Civita connection with the relation between $A_{\mu\nu}$ and $\ell^{\mu}$ gives
\begin{equation}\label{charge}
    \underbrace{\nabla_{\nu}\nabla_{\mu}A^{\mu\nu}}_{(I)}=\nabla_{\nu}\ell^{\nu}.
\end{equation}
By rewriting (I) using the antisymmetry of $A_{\mu\nu}$ and relabeling the indices, one obtains
\begin{equation}\label{1200}
    \nabla_{\nu}\nabla_{\mu}A^{\mu\nu}=-\nabla_{\mu}\nabla_{\nu}A^{\mu\nu}\implies 2\nabla_{\nu}\nabla_{\mu}A^{\mu\nu}=(\nabla_{\nu}\nabla_{\mu}-\nabla_{\mu}\nabla_{\nu})A^{\mu\nu}.
\end{equation}
Using eq. \ref{Rie2}, eq. \ref{1200} then results in
\begin{equation}
    (\nabla_{\nu}\nabla_{\mu}-\nabla_{\mu}\nabla_{\nu})A^{\mu\nu}=-R_{\nu\mu\alpha}\mathstrut^{\mu}A^{\alpha\nu}-R_{\nu\mu\alpha}\mathstrut^{\nu}A^{\mu\alpha}=-2R_{\nu\alpha}A^{\alpha\nu}=0,
\end{equation}
as $R_{\nu\alpha}A^{\alpha\nu}$ is the contraction of a symmetric tensor with an antisymmetric one. Thus, by eq. \ref{charge}, $\nabla_{\mu}\ell^{\mu}=0$.
\end{proof}

A vector that satisfies the relation $\nabla_{\mu}\ell^{\mu}=0$ is said to represent a conserved ``current”, as it respects the covariant form of the continuity equation. The conserved ``charge” associated with this current can be evaluated by analyzing a compact and bounded region, $\mathcal{V}$, of spacetime such that all sources of this current are inside its spacelike segments. Consider a region such that $\partial \mathcal{V}=\{(t,t')\times\partial\Sigma\}\cup\Sigma_{t}\cup\Sigma_{t'}$, where each $\Sigma$ is a region of a spacelike hypersurface and $\{(t,t')\times\partial\Sigma\}$ is the hypersurface connecting $\Sigma_t$ to $\Sigma_{t'}$, closing the region $\mathcal{V}$, as illustrated in fig. \ref{fig:komar}.\begin{figure}[h]
\centering
\includegraphics[scale=1.55]{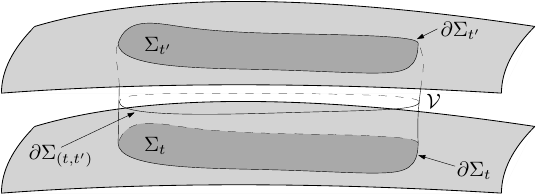} 
\caption{Bounded and compact region, $\mathcal{V}$, for the analysis of the conserved charge.}
\caption*{Source: By the author.}
\label{fig:komar}
\end{figure} Since all sources are inside its spacelike segments, one clearly has that $\ell^{\mu}=0$ on $\{(t,t')\times\partial\Sigma\}$. Now, by integrating $\nabla_{\mu}\ell^{\mu}$ on $\mathcal{V}$ (see appendix \ref{integration}), one obtains
\begin{equation}\label{Komar1}
    \int_{\mathcal{V}}\epsilon_{\mu\nu\alpha\beta}\nabla_{\lambda}\ell^{\lambda}=0,
\end{equation}
but from Stokes' theorem (see eqs. \ref{stokes2} and \ref{stokes3}) and the relation in prop. \ref{PropKomar}, one also has
\begin{equation}\label{Komar2}
    \begin{aligned}[b]
        \int_{\mathcal{V}}\epsilon_{\mu\nu\alpha\beta}\nabla_{\lambda}\ell^{\lambda} &= \int_{\mathcal{\partial V}}\epsilon_{\mu\nu\alpha\beta}\ell^{\beta}\\
        & = \int_{\Sigma_{t'}}\epsilon_{\mu\nu\alpha\beta}\ell^{\beta}-\int_{\Sigma_{t}}\epsilon_{\mu\nu\alpha\beta}\ell^{\beta}+\int_{\partial\Sigma_{(t,t')}}\epsilon_{\mu\nu\alpha\beta}\ell^{\beta}\\
        &= \int_{\Sigma_{t'}}\epsilon_{\mu\nu\alpha\beta}\nabla_{\lambda}A^{\lambda\beta}-\int_{\Sigma_{t}}\epsilon_{\mu\nu\alpha\beta}\nabla_{\lambda}A^{\lambda\beta}\\
        &= \int_{\partial\Sigma_{t'}}\epsilon_{\mu\nu\alpha\beta}A^{\alpha\beta}-\int_{\partial\Sigma_{t}}\epsilon_{\mu\nu\alpha\beta}A^{\alpha\beta},
    \end{aligned}
\end{equation}

Hence, from eq. \ref{Komar1}, one can then identify each term in the last line of eq. \ref{Komar2} as the conserved charge. Then, by eq. \ref{kill}, $A^{\mu\nu}=\nabla^{\mu}\chi^{\nu}$ is the antisymmetric tensor associated with the conserved current $-R^{\mu\nu}\chi_{\nu}$. As a consequence, the global conserved quantity, $c_{\chi}$, associated with a Killing vector, $\chi^{\mu}$, is given by
\begin{equation}\label{Komar}
    c_{\chi}=\int_{\partial\Sigma}\epsilon_{\mu\nu\alpha\beta}\nabla^{\alpha}\chi^{\beta}.
\end{equation}
The integral on the right hand side of eq. \ref{Komar} is known as a \textit{Komar integral} \cite{Komar1959,Komar1963}. Note that $c_{\chi}$ is independent of the choice of spacelike hypersurface, $\Sigma$, the only requirement being that the surface $\partial\Sigma$ is such that the flux of the conserved current over it vanishes. Finally, it should be noted that this form of a conserved quantity is not restricted to Killing vectors, in the sense that any physically significant quantity that gives rise to a two-form will result in a conserved quantity. An example of this is the electromagnetic tensor, $F_{\mu\nu}$, whose associated conserved quantity is the electric charge \cite{Wald1984}.

\section{Causal structure}\label{causal}

The characterization of the causal structure of spacetime is made by \textit{causal} curves, which are curves such that their tangent vector is everywhere either timelike or null. Since observers and light rays can only travel on causal curves, it is possible to establish a notion of causality by studying which events can be connected by them. In other words, two events are said to be \textit{causally connected} if an observer or a light ray emanating from one of them can reach the other. Consequently, the causal structure of spacetime can be regarded as the characterization of the sets that can be interpreted as the future and past of events. In this section, these sets are defined, their boundaries are analyzed, and general results and possible pathologies are discussed.

In special relativity, each event, $a$, in Minkowski spacetime has a \textit{light cone} associated with it, which is delimited by null geodesics and contains the events that are connected by timelike geodesics to it. One labels half of these events as the \textit{future} of $a$, as they can be influenced by a light ray or an observer emanating from $a$. Similarly, one labels the other half as the \textit{past} of $a$, as those events can influence $a$. Due to the trivial topology of Minkowski spacetime, the causal structure in special relativity can be completely described by light cones. Since every event, $a$, in an arbitrary spacetime, $(M,g_{\mu\nu})$, has a \textit{convex normal neighborhood}, i.e., a neighborhood, $S$, of $a$ such that for every $a',a''\in S$ there exists a unique geodesic connecting $a'$ and $a''$ and contained entirely in $S$ (A rigorous proof of this can be found in \cite{Hawking1973}, while a convenient statement can be found in \cite{Wald1984}), one can conclude that the causal structure of an arbitrary spacetime is \textit{locally} identical to that of Minkowski spacetime. However, significant changes can arise globally due to the nontrivial topology of an arbitrary spacetime. Evidently, analyzing the details of the causal structure of an arbitrary spacetime is futile, as it is dependent on its topology. Nevertheless, with a few restrictions that are justified by physical assumptions, it is possible to derive very important results valid in general. Our goal is to then analyze how restrictions on arbitrary spacetimes can give information about the sets that one would justifiably label as the past and future of events in $(M,g_{\mu\nu})$.

As one is interested in results concerning physically reliable spacetimes, it is necessary to first make a restriction regarding a global property. One would like to be able to make a continuous choice of past and future, that is, for each event, $a\in M$, one wishes to be able to consistently and continuously identify which vectors are directed to the future and which are directed to the past. On physical grounds, such a designation has to be related to a well-defined “arrow of time”, for instance, the one given by ``natural'' thermodynamic processes. If a continuous choice can be made, $(M,g_{\mu\nu})$ is said to be \textit{time orientable} (this is analogous, but not equivalent, to the notion of orientability when one defines integration on manifolds). Conversely, if a continuous timelike vector can be chosen in $(M,g_{\mu\nu})$, then it is time orientable. Such a timelike vector can be interpreted as the geometrical quantity associated with the “arrow of time”. This timelike vector, $t^{\mu}$, is said to provide a \textit{time orientation}, and can be used to divide nonspacelike vectors at each point into two sets. More precisely, by identifying $t^{\mu}$ as being directed to the future, a timelike or null vector, $s^{\mu}$, is said to be \textit{future directed} if $s^{\mu}t_{\mu}< 0$, and \textit{past directed} if $s^{\mu}t_{\mu}> 0$. Note that in order to identify that such vectors are pointing in the same direction, it is necessary to take into account the minus sign in the ``squared distance'' in the timelike direction, hence the respective inequalities. The following result states an important property of the interior of these sets. 

\begin{proposition}\label{lightcone}
    Let $(M,g_{\mu\nu})$ be a time orientable spacetime. The set of future (past) directed timelike vectors in the tangent vector space of any $a\in M$ is path-connected.
\end{proposition}
\begin{sproof}
    Let $\mathcal{F}_a$ denote the set of future (past) timelike directed vectors in the tangent vector space of $a\in M$, i.e., $s^{\mu}\in\mathcal{F}_a$ if $s^{\mu}t_{\mu}<0\;(s^{\mu}t_{\mu}>0)$ and $s^{\mu}s_{\mu}<0$, and consider a map $\psi:[0,1]\to\mathcal{F}_a$ defined as
    \begin{equation}
        \psi(\lambda)=\lambda s^{\mu}+(1-\lambda)w^{\mu},
    \end{equation}
    with $s^{\mu},\;w^{\mu}\in\mathcal{F}_a$. To show that this map is continuous, first note that for all $\lambda\in[0,1]$, the vector $\ell^{\mu}(\lambda)=\lambda s^{\mu}+(1-\lambda)w^{\mu}$ is in $\mathcal{F}_a$. In other words,
    \begin{equation}
        \begin{aligned}
           \ell^{\mu}(\lambda)t_{\mu}&=(\lambda s^{\mu}+(1-\lambda)w^{\mu})t_{\mu}\\
           &=\lambda s^{\mu}t_{\mu}+(1-\lambda)w^{\mu}t_{\mu}
        \end{aligned}
    \end{equation}
is negative (positive) for all $\lambda\in[0,1]$ and $s^{\mu},\;w^{\mu}\in\mathcal{F}_a$. Similarly,
\begin{equation}
        \begin{aligned}
           \ell^{\mu}\ell_{\mu}(\lambda)&=(\lambda s^{\mu}+(1-\lambda)w^{\mu})(\lambda s_{\mu}+(1-\lambda)w_{\mu})\\
           &=\lambda^2s^{\mu}s_{\mu}+(1-\lambda)^2w^{\mu}w_{\mu}+2\lambda(1-\lambda)s^{\mu}w_{\mu},
        \end{aligned}
    \end{equation}
    which is also negative for all $\lambda\in[0,1]$ and $s^{\mu},\;w^{\mu}\in\mathcal{F}_a$. Using the topology induced by the metric, $g_{\mu\nu}$, on $\mathcal{F}_a$, one can deduce that the inverse image of any open set will be mapped into either the empty set of $[0,1]$, an open interval of $[0,1]$ or $[0,1]$. Thus, the set of future (past) directed timelike vectors in the tangent vector space of any $a\in M$ is path-connected.
\end{sproof}

The notion of orientability can be extended to curves in a similar way. A differentiable curve, $\lambda$, is said to be a \textit{future directed timelike curve} if its tangent vector is everywhere future directed timelike. Similarly, if the tangent vector is everywhere either future directed timelike or null, $\lambda$ is said to be a \textit{future directed causal curve}. It is important to note that, by curve, it is meant one that contains multiple points, so that the trivial case $\lambda(t)=a$ is excluded. In addition, it is also useful to have a precise definition of when a curve “ends”. Let $\lambda$ be a future directed causal curve. An event, $a$, is said to be a \textit{future endpoint} of $\lambda$ if for every neighborhood, $A$, of $a$ there exists a $t_0$ such that $\lambda(t)\in A$ for all $t>t_0$. By the Hausdorff property of $M$, which states that for any two events it is always possible to find a neighborhood of one that does not contain the other, a curve can have, at most, one future endpoint. Likewise, a curve is said to be \textit{future inextendible} if it has no future endpoint. For example, in fig. \ref{fig:231} one can identify the point $a\in\lambda(t)$ as being its future endpoint. In contrast, the curve in fig. \ref{fig:sd3} is future inextendible, since it runs into an artificially removed point.\begin{figure}[h]
  \begin{subfigure}[b]{0.5\textwidth}
  \centering
    \includegraphics[scale=1.5]{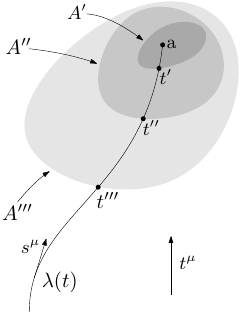}
    \caption{Endpoint of the future directed causal curve.}
    \label{fig:231}
  \end{subfigure}
  \hfill
  \begin{subfigure}[b]{0.5\textwidth}
  \centering
    \includegraphics[scale=1.5]{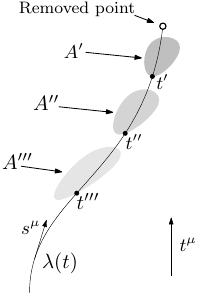}
    \caption{Future directed inextendible causal curve.}
    \label{fig:sd3}
  \end{subfigure}
  \caption{Examples of future directed causal curves.}
  \caption*{Source: By the author.}
  \label{fig123456}
\end{figure} To understand why such a curve has no future endpoint, note that one can always ``zoom in'' closer to the removed point and find neighborhoods such as those exemplified. Consequently, no event in spacetime would respect the properties of a future endpoint. Lastly, analogous definitions and properties apply to past directed curves, past endpoints and past inextendibility, by interchanging future with past in the definitions. 

With these ideas, one can precisely define the sets that can be interpreted as future and past of an event in a time orientable spacetime. The \textit{chronological future} of $a\in M$, $I^+(\{a\})$, is defined as the set of events that can be reached by a future directed timelike curve starting from $a$. In other words, all the events that can be reached by an observer emanating from $a$. Note that, unless spacetime possesses closed timelike curves, $a\not\in I^+(\{a\})$. Similarly, the \textit{causal future} of $a$, $J^+(\{a\})$, is defined as the union of $a$ and the set of events that can be reached by a future directed causal curve starting from $a$. That is, all the events that can be reached by an observer or a light ray emanating from $a$ and $a$ itself. The chronological and causal past of $a$, $I^-(\{a\})$ and $J^-(\{a\})$, are defined likewise. Similarly, for any subset $S\subset M$, 
\begin{equation}\label{chrono}
    I^+(S)=\bigcup_{a\;\in\;S}I^+(\{a\}),
\end{equation}
\begin{equation}
    J^+(S)=\bigcup_{a\;\in\;S}J^+(\{a\}),
\end{equation}
and the chronological and causal past of $S$, $I^-(S)$ and $J^-(S)$, are defined analogously. Evidently, from the definitions of the chronological and causal future,
\begin{equation}\label{future1}
    I^+(S)\subset J^+(S).
\end{equation}
Indeed, this result is also true for the chronological and causal pasts. In the following, the results discussed are valid for the pasts and futures, but for simplicity of notation and space, they will be expressed only for the future or the past, and the notation $I^+(\{a\})=I^+(a)$ and $J^+(\{a\})=J^+(a)$ will be adopted.

Let $a,a'\in M$, $a'\in I^+(a)$ and $\gamma$ denote the future directed timelike curve connecting $a$ to $a'$. It is possible to use the timelike nature of $\gamma$ to deform it slightly and connect $a$ to any point in a given neighborhood of $a'$ while maintaining its timelike nature. This process would correspond to ``tilting'' the tangent vector to $\gamma$ to a null direction at each point, yielding another future directed timelike curve. Evidently, this can be done for any event in $I^+(a)$, since any timelike vector can have its norm arbitrarily reduced while still being timelike. Hence, for any event $a'\in I^+(a)$, it is always possible to find a neighborhood, $A$, of $a'$ such that $A\subset I^+(a)$. Fig. \ref{fig:7} illustrates this property\footnote{In such an illustration, a two-dimensional manifold with Minkowski metric, $ds^2=-c^2(dx^0)^2+(dx^1)^2$, was considered. Generalizations of these diagrams for higher dimensions is straightforward, but for simplicity they will only be presented for the two-dimensional case with a Minkowski metric. In the representations that follow, the coordinate system presented in fig. \ref{fig:7} will be adopted.}, in which the event $a''\in A$ is connected to $a$ by the future directed timelike curve $\gamma'$. In particular, this means that the chronological future of any set event is always open, and from eq. \ref{chrono}, one then has that
\begin{equation}\label{open}
    I^+(S)=\langle I^+(S) \rangle,
\end{equation}
where $\langle I^+(S) \rangle$ denotes the interior of the set $S$ (see appendix \ref{A1}). 

\begin{figure}[h]
\centering
\includegraphics[scale=1.4]{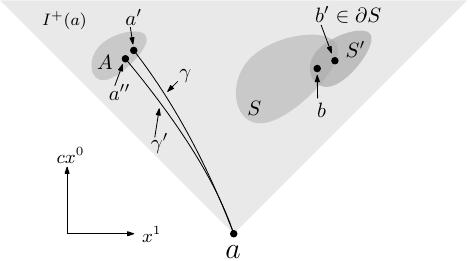} 
\caption{Spacetime diagram of the future of an event.}
\caption*{Source: By the author.}
\label{fig:7}
\end{figure}

Additionally, one can also consider a set $S\subset I^+(a)$, an event, $b\in S$, and an neighborhood of it, $S'$, that intersects $\partial S$ (see fig. \ref{fig:7}). Now, applying the same logic of deformation of a timelike curve to $b$ and an event $b'\in\partial S$, one can conclude that $\overline{S}\subset I^+(a)$. Hence, this means that every event that can be connected to $S$ by a timelike curve can also be connected to $\overline{S}$ by a timelike curve, which allows one to state that
\begin{equation}\label{435}
    I^+(S)=I^+(\overline{S}),
\end{equation}
where $\overline{S}$ denotes the closure of the set $S$.

\begin{proposition}\label{prop2345}
    Let $(M,g_{\mu\nu})$ be a time orientable spacetime, $a,a'\in M$ and $a'\in J^+(a)$. Then $I^+(a')\subset I^+(a)$.
\end{proposition}
\begin{proof}
    Let $\gamma$ denote the future directed causal curve connecting $a$ to $a'$. Consider a point $a''\in I^+(a')$ and let $\gamma'$ denote the future directed timelike curve connecting $a'$ to $a''$. A curve connecting $a$ to $a''$ can be made by joining $\gamma$ and $\gamma'$ at $a'$, and it is possible to vary it slightly to produce a future directed timelike curve, $\gamma''$, connecting $a$ to $a''$. Fig. \ref{fig:prop2345} illustrates an example of this process, and since it can be done for any $a''\in I^+(a')$, then $I^+(a')\subset I^+(a)$. 
\end{proof}
\begin{figure}[h]
\centering
\includegraphics[scale=1.4]{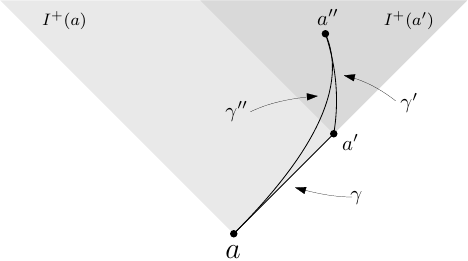} 
\caption{Spacetime diagram for the proof of proposition \ref{prop2345}.}
\caption*{Source: By the author.}
\label{fig:prop2345}
\end{figure}
Note that the requirement of time orientability in this result and the ones that follow is not a geometrical one, in the sense that one could state the following results for more general spacetimes. However, the corresponding concept of the chronological and causal sets would not be of physical interest, since they would present unphysical behavior, e.g., discontinuous change in the notion of going to the future or past.  

It follows from prop. \ref{prop2345} that for any set $S\subset M$, $I^+(S)$ can be expressed as the union of all $I^+(a)$ such that $a\in J^+(S)$. It also allows one to state the following.
\begin{proposition}\label{prop23456}
    Let $(M,g_{\mu\nu})$ be a time orientable spacetime, $a\in M$ and $\gamma$ be a future directed causal curve emanating from $a$ such that for all $a'\in\gamma$, $a'\not\in I^+(a)$. Then $\gamma$ is a null geodesic.
\end{proposition}
\begin{proof}
    Since the causal structure of $(M,g_{\mu\nu})$ is locally the same as of Minkowski spacetime, the causal future of any event is locally delimited by null geodesics emanating from it. If at any point $a'\in\gamma$ the curve $\gamma$ fails to be a null geodesic, then $\gamma$ must enter $I^+(a')$ since the future of $a'$ must be generated, locally, by null geodesics. However, from prop. \ref{prop2345}, this would contradict the condition that points in $\gamma$ are not in $I^+(a)$. Thus, if $\gamma$ is a causal curve emanating from $a$ and does not enter $I^+(a)$, it must be a null geodesic.
\end{proof}
Hence, the property that the causal future of a point is delimited by null geodesics holds globally in any time orientable spacetime. With these properties, it is also possible to show \cite{Penrose1972} that for any set $S\subset M$,
\begin{equation}
    J^+(S)\subset\overline{I^+(S)},
\end{equation}
and combining with eq. \ref{future1} yields
\begin{equation}\label{future2}
    \overline{J^+(S)}=\overline{I^+(S)}.
\end{equation}
Furthermore, from eqs. \ref{future1} and \ref{open}, one also has
\begin{equation}
    I^+(S)\subseteq\langle J^+(S)\rangle.
\end{equation}
However, for any $a\in\langle J^+(S)\rangle$, one can find a neighborhood of it, $S'$, such that $S'\subset J^+(S)$. Using prop. \ref{prop2345}, one can then show that there exists a future directed timelike curve from $S$ to $a$, which means that
\begin{equation}\label{future3}
    I^+(S)=\langle J^+(S)\rangle.
\end{equation}
Finally, by using eqs. \ref{future2}, \ref{future3}, and the fact that $I^+(S)$ is open, one obtains
\begin{equation}\label{future4}
    \partial J^+(S)= \partial I^+(S).
\end{equation}

The next theorem states an important result regarding the inextendibility of null geodesics in $\partial J^+(S)\backslash S$, a proof of which can be found in \cite{Penrose1972}. 

\begin{theorem}\label{the3}
Let $(M,g_{\mu\nu})$ be a time orientable spacetime and $S\subset M$ be nonempty closed set. Then every point $a\in\partial J^+(S)\backslash S$ lies on a null geodesic which lies entirely in $\partial J^+(S)$ and is either past inextendible or has a past endpoint on $S$.
\end{theorem}

Theorem \ref{the3} states that the past directed null geodesic $\gamma$ emanating from $a'\in\partial J^+(S)\backslash {S}$ is contained entirely in $\partial J^+(S)\backslash {S}$ and either reaches $S$ or possesses no past endpoint. An exemplification of this result is given in fig. \ref{fig:prop23455}. \begin{figure}[h]
\centering
\includegraphics[scale=1.4]{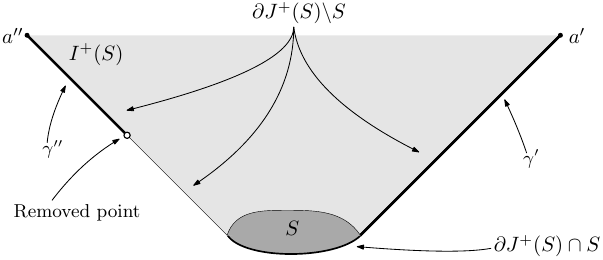} 
\caption{Spacetime diagram of a closed set and its future (see theorem \ref{the3}).}
\caption*{Source: By the author.}
\label{fig:prop23455}
\end{figure}In such an illustration, note that the event $a'$ is causally connected to $S$ by the null geodesic $\gamma'$. However, the event $a''$ is not causally connected to $S$, since the null geodesic $\gamma''$ which is entirely contained in $\partial J^+(S)\backslash {S}$ runs into a removed point. Thus, $\gamma''$ does not reach $S$ and does not have a past endpoint. Furthermore, the conclusion that this null geodesic has no past endpoint translates to the fact that although $a''\in\partial J^+(S)$, $a''\notin J^+(S)$. Finally, note that because of eqs. \ref{435} and \ref{future4}, the requirement of $S$ to be closed implies no loss of generality.

Moving on, two properties of closed sets that will be useful for the developments that follow are now defined. Given a closed set, $S$, its \textit{edge}, $e(S)$, is the set of points $a\in S$ such that every neighborhood of $a$ contains a point $a'\in I^+(a)$, a point $a''\in I^-(a)$ and a future directed timelike curve from $a''$ to $a'$ that does not intersect $S$. If $e(S)=\emptyset$, $S$ is said to be \textit{edgeless}. An example of an edgeless set is a null or spacelike hypersurface that extends indefinitely in all directions. Similarly, a set $S\subset M$ is said to be \textit{achronal} if $I^+(S)\cap S=\emptyset$ (or, equivalently, if $I^-(S)\cap S=\emptyset$). For instance, it is possible to see that the set $S$ in fig. \ref{fig:prop23455} is not achronal. In contrast, fig. \ref{fig:888} illustrates a closed achronal set, $S'$, with edge.\begin{figure}[h]
\centering
\includegraphics[scale=1.5]{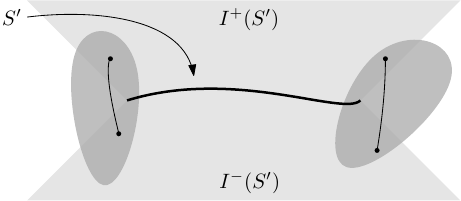} 
\caption{Spacetime diagram illustrating the edge of a closed achronal set, $S'$.}
\caption*{Source: By the author.}
\label{fig:888}
\end{figure} In particular, the achronality property is of importance because, given a three-dimensional achronal set $S'$, it is possible to construct a homeomorphism from neighborhoods of $a\in S'$ to $\mathbb{R}^{3}$. Details on this construction can be found in \cite{Wald1984}, which yields the following result. 

\begin{theorem}\label{2t1}
Let $(M,g_{\mu\nu})$ be a time orientable spacetime and $S\subset M$ be a nonempty edgeless achronal set. Then $S$ is a $C^0$ hypersurface.
\end{theorem}

The significance of this theorem can be exemplified by considering the boundary of the causal future of a closed set, $S\subset M$. From eq. \ref{future4}, it is not possible to find two points, $a,a'\in\partial J^+(S)$, such that $a'\in I^+(a)$, because from prop. \ref{prop2345}, this would contradict the fact that $a'\in\partial J^+(S)$. Thus, $\partial J^+(S)$ is achronal. Moreover, it is also always closed because it is a boundary (see appendix \ref{A1}). Similarly, the arguments of theorem \ref{the3} can be used to show that $e(\partial J^+(S))=\emptyset$ \cite{Penrose1972}, and hence, $\partial J^+(S)$ respects the requirements\footnote{The edgeless property is not necessary to conclude that $\partial J^+(S)$ must be a topological hypersurface, which is a consequence of the fact that $\partial J^+(S)$ is an \textit{achronal boundary} \cite{Penrose1972}. Nonetheless, the theorem was stated as such because it will be used later for an achronal edgeless set which is not an achronal boundary.} of theorem \ref{2t1} and is a topological hypersurface. Of course, $\partial J^+(S)\backslash S$ must contain only null geodesics. Therefore, this segment of the hypersurface must be null. This can be deduced from the fact that its achronality ensures that any timelike vector cannot be tangent to it, otherwise one could follow the integral curves of such a vector in $\partial J^+(S)$ and contradict its achronality. Also, a timelike vector cannot be normal to this segment, as the vector space normal to it would be spanned only by spacelike vectors, and thus, not allowing null geodesics to be contained in it. The same line of reasoning applies to the segments $\partial J^+(S)\cap {S}$ of the hypersurface, which means that they must be null or spacelike.

In order to discuss more aspects of the causal structure, it is necessary to analyze further restrictions for causally and deterministically ``well behaved'' spacetimes. An example of a pathology that would make one deem a spacetime to be ``badly behaved'' is possessing a closed timelike curve. Physically, this condition translates to the possibility of ``time travel'', i.e., for an observer to reach an event in its past. Indeed, a closed timelike curve would also imply that it would not be possible to physically distinguish between a cause and an effect, meaning that the notion of causality, and time itself, would be questionable. In light of this, a spacetime that has no closed timelike curve is said to respect the \textit{chronology condition}. Nevertheless, a spacetime, $(M,g_{\mu\nu})$, can obey the chronology condition but still be deemed to be pathological if a perturbation of the metric results in a closed timelike curve. For instance, if one considers a timelike vector $s^{\mu}$ at $a\in M$, and define the metric 
\begin{equation}
    g'_{\mu\nu}=g_{\mu\nu}-s_{\mu}s_{\nu},
\end{equation}
it is possible to see that any timelike or null vector of $g_{\mu\nu}$ is a timelike vector of $g'_{\mu\nu}$, i.e., the light cone of $g'_{\mu\nu}$ is slightly larger than that of $g_{\mu\nu}$. As such, if the spacetime $(M,g'_{\mu\nu})$ were to have a closed timelike curve, one would deem $(M,g_{\mu\nu})$ physically unreliable even though $(M,g_{\mu\nu})$ itself does not have closed timelike curves. In light of this, a spacetime, $(M,g_{\mu\nu})$, is said to be \textit{stably causal} if there exists a continuous nonvanishing timelike vector, $s^{\mu}$, such that the spacetime $(M,g'_{\mu\nu})$ possesses no closed timelike curve. It can be shown \cite{Hawking1973} that in order for a spacetime to be stably causal there must exist a differentiable function, $f$, on $(M,g_{\mu\nu})$ such that $\nabla^{\mu}f$ is everywhere timelike. The requirement that $\nabla^{\mu}f$ is timelike means that it can be regarded as the “arrow of time”, and thus, establishes a time orientation on $(M,g_{\mu\nu})$. This would imply that along every future (past) directed causal curve $f$ must strictly increase or decrease, and hence, $(M,g'_{\mu\nu})$ with $g'_{\mu\nu}$ constructed from $\nabla^{\mu}f$ could not have any closed timelike curve. 

The notion of the causal structure being ``well behaved'' can also be analyzed through the set of events, $D(S)$, that can be completely determined by conditions on another set, $S$. Hence, the set $D(S)$ is one such that every observer or light ray that can reach it must have passed through $S$ \cite{Geroch1970}. More precisely, given a closed achronal set, $S$, the \textit{future Cauchy development}, $D^+(S)$, is defined as the set of events such that every past inextendible causal curve that passes through it intersects $S$. The \textit{past Cauchy development}, $D^-(S)$, is defined similarly by interchanging future and past. The \textit{Cauchy development} is defined as
\begin{equation}
    D(S)=D^+(S)\cup D^-(S).
\end{equation}
Thus, the Cauchy development of a closed achronal set, $S$, is the set of events that can be completely determined by information on $S$, which evidently contains $S$, as any inextendible causal curve emanating from $S$ crosses $S$. Note that no generality is lost by considering a closed set, as the physical conditions on any open set, $S'$, are expected to be continuous, and thus, should suffice to determine the conditions on $\overline{S'}$. Concerning the achronality requirement, a notion of a Cauchy development can still be defined for non-achronal sets, but its physical usefulness is not as interesting since parts of the set would be in its future \cite{Penrose1972}. Accordingly, the concept of determinism present in the definition of Cauchy development is then a tool to analyze the evolution laws of physical fields in spacetime. A precise description of the \textit{initial-value formulation}, which analyses how the conditions on a set can be used to determine the physical conditions in its domain of dependence, can be found in, e.g., \cite{Wald1984}. 

With the definition of $D^+(S)$, it is clear that its boundary delimits the region from which data can be completely predicted from data on $S$. Such a boundary is known as the \textit{future Cauchy horizon} of $S$, $H^+(S)$. More precisely, $a\in H^+(S)$ if for all $a'\in D^+(S)$, $a\not\in I^-(a')$. The \textit{past Cauchy horizon}, $H^-(S)$, is defined analogously. The \textit{Cauchy horizon} is the union of $H^+(S)$ and $H^-(S)$, which can also be written as 
\begin{equation}
    H(S)=\partial D(S).
\end{equation}
Fig. \ref{fig:24} \begin{figure}[h]
  \begin{subfigure}[b]{0.5\textwidth}
  \centering
    \includegraphics[scale=1.3]{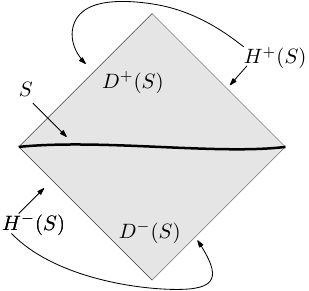}
    \caption{Spacetime with no removed points.}
    \label{fig:242}
  \end{subfigure}
  \hfill
  \begin{subfigure}[b]{0.5\textwidth}
  \centering
    \includegraphics[scale=1.35]{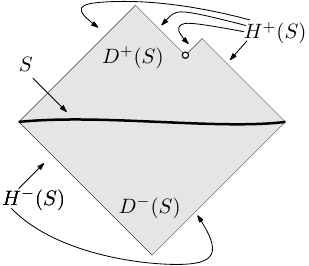}
    \caption{Spacetime with a removed point.}
    \label{fig:241}
  \end{subfigure}
  \caption{Spacetime diagrams illustrating the Cauchy development and Cauchy horizon of a closed achronal set, $S$.}
  \label{fig:24}
\caption*{Source: By the author.}
\end{figure}illustrates examples of the Cauchy development and Cauchy horizon for a closed achronal set, $S$, and the effects of a removed point on these sets. The significance of $D(S)$ and $H(S)$ becomes clear by the following simple, but remarkable result.

\begin{theorem}\label{2p1}
     Let $(M,g_{\mu\nu})$ be a time orientable spacetime. A nonempty closed achronal set, $\Sigma$, has Cauchy development equal to $M$, $D(\Sigma)=M$, if and only if $H(\Sigma)=\emptyset$.
\end{theorem}
\begin{proof}
    If $H(\Sigma)=\partial D(\Sigma)=\emptyset$, then $\overline{D(\Sigma)}=\langle D(\Sigma)\rangle=D(\Sigma)$, so $D(\Sigma)$ is both open and closed. Thus, since $D(\Sigma)\supset\Sigma\neq\emptyset$ and $M$ is connected, one concludes that $D(\Sigma)=M$.
\end{proof}

A closed achronal set, $\Sigma$, for which $D(\Sigma)=M$ is called a \textit{Cauchy hypersurface}. Since a set with an edge implies that there are causal curves that do not intersect it, it follows that any Cauchy hypersurface must be edgeless. Hence, by theorem \ref{2t1}, every Cauchy hypersurface is a topological hypersurface. The nature of this hypersurface can be deduced from its achronal characteristic, since it means that no timelike vector can be tangent to it. Therefore, a Cauchy hypersurface must be made of spacelike or null segments. In addition, a spacetime which possesses a Cauchy hypersurface is said to be \textit{globally hyperbolic}. Consequently, in globally hyperbolic spacetimes, the entire development of spacetime can be analyzed by conditions at $\Sigma$. Finally, note that the Cauchy horizon measures the failure of a set to be a Cauchy hypersurface, as given by theorem \ref{2p1}. 

It can also be deduced that no closed timelike curve can exist in a globally hyperbolic spacetime. A closed timelike curve which intersects $\Sigma$ would violate its achronality, and one which does not intersect it would mean that $D(S)\neq M$, violating global hyperbolicity. These arguments can be strengthened to conclude that a globally hyperbolic spacetime must be stably causal \cite{Wald1984}. Moreover, since globally hyperbolic spacetimes can be interpreted as possessing no “deterministic violations” \cite{Geroch1979}, spacetimes that do not possess a Cauchy hypersurface can be regarded as having sources of “uncontrollable influences” where information is created or destroyed, such as pathological regions or removed points.  Hence, globally hyperbolic spacetimes are those that are causally and deterministically ``well behaved'', in the sense that evolution laws can be regarded as ``well posed'' problems (i.e., evolution laws have a unique solution that changes continuously with variations of initial conditions). Lastly, it is possible to use the integral curves of the timelike vector from the stably causal property of a globally hyperbolic spacetime to construct a homeomorphism between Cauchy hypersurfaces. This construction leads to the following result \cite{Hawking1973}, allowing one to interpret $\Sigma$ as an “instant of time”. 
\begin{theorem}\label{timeins}
    Let $(M,g_{\mu\nu})$ be a globally hyperbolic spacetime. Then the topology of $M$ is $\mathbb{R}\times\Sigma$, where $\Sigma$ denotes any Cauchy hypersurface.
\end{theorem}

We conclude our discussion of causal structure by stating a result regarding the causal future of a compact set in globally hyperbolic spacetimes \cite{Galloway2014}. 
\begin{theorem}\label{theorem7}
    Let $(M,g_{\mu\nu})$ be globally hyperbolic and let $S\subset M$ be compact. Then $J^+(S)$ is closed. 
\end{theorem}

Thus, for any compact set $S$ in a globally hyperbolic spacetime, $J^+(S)=\overline{J^+(S)}$, which implies that $\partial J^+(S)\subset J^+(S)$, and hence, from theorem \ref{the3}, in a globally hyperbolic spacetime every point $a\in\partial J^+(S)\backslash{S}$ can be connected by a past directed null geodesic to $S$. More specifically, this means that the situation presented in fig. \ref{fig:prop23455}, in which a null geodesic in $\partial J^+(S)$ has no past endpoint on $S$, cannot occur for a compact set in a globally hyperbolic spacetime. 

\section{Energy conditions}\label{energycon}

The purpose of this section is to present possible restrictions on the energy-momentum tensor, which are nothing more than ways one would expect a physically reasonable distribution of energy to behave (see, e.g., \cite{Hawking1973}). Such conditions are of significance because, even though one may not have information about the explicit form of the energy-momentum tensor, using Einstein's equation one can infer how the geometry of spacetime is expected to affect the motion of observers and light rays. The interpretation of these conditions in the context of spacetime curvature will be exemplified in the discussion of dynamics of null geodesics in \S\;\ref{null}.

An energy-momentum tensor, $T_{\mu\nu}$, such that
\begin{equation}
    T_{\mu\nu}s^{\mu}s^{\nu}\geq 0,\;\forall\;s^{\mu}\text{ timelike},
\end{equation}
is said to satisfy the \textit{weak energy condition}. This condition can be interpreted as stating that the energy density measured by any observer must be non-negative. In a classical sense, this condition can be seen to hold, as the notion of negative energy density only arises when one considers quantum theories. 

Similarly, $T_{\mu\nu}$ is said to satisfy the \textit{strong energy condition} if
\begin{equation}\label{dominant}
    T_{\mu\nu}s^{\mu}s^{\nu}\geq -\frac{1}{2}T,\;\forall\;s^{\mu}\text{ unit timelike},
\end{equation}
where $T=T_{\mu}\mathstrut^{\mu}$. This condition simply puts a restriction on how strong the stress of matter can become compared with the energy density, as measured by any observer. In particular, the numerical factor on the right hand side of eq. \ref{dominant} is a direct consequence of the explicit form of Einstein's equation. To see this, one contracts Einstein's equation with the inverse metric, $g^{\mu\nu}$, which leads to 
\begin{equation}
    R=-\frac{8\pi G}{c^4}T,
\end{equation}
so that it is possible to write Einstein's equation, equivalently, as
\begin{equation}\label{form37}
    R_{\mu\nu}=\frac{8\pi G}{c^4}\left(T_{\mu\nu}-\frac{1}{2}Tg_{\mu\nu}\right).
\end{equation}
Thus, given eq. \ref{form37}, the strong energy condition implies the non-negativity of the scalar $R_{\mu\nu}s^{\mu}s^{\nu}$. As will be discussed in more detail in \S\;\ref{null}, this non-negativity signifies that timelike or null geodesics tend to get closer together, which in turn can be interpreted as the attractive nature of gravity.

In addition, for all future directed timelike vectors, $s^{\mu}$, the \textit{dominant energy condition} states that the vector $-T^{\mu}\mathstrut_{\nu}s^{\nu}$ is future directed timelike or null. In other words, it states that the flow of energy as measured by any observer must be, at most, at the speed of light and directed to the future \cite{Wald1984}. It should be noted that apart from the dominant energy condition implying the weak one, these conditions are independent mathematical hypotheses. Finally, due to continuity, the weak and the strong energy conditions imply the \textit{null energy condition},
\begin{equation}
    T_{\mu\nu}\ell^{\mu}\ell^{\nu}\geq 0,\;\forall\;\ell^{\mu}\text{ null}.
\end{equation}
Alternatively, one can also deduce that the weak and strong energy conditions imply the null one by explicit analysis of the components of the energy-momentum tensor (see, e.g., \cite{Poisson2004}).

\section{Null geodesic congruences}\label{null}

A \textit{congruence} in an open set, $S\subset M$, is a family of curves such that through each $a\in S$ passes exactly one curve of this family. Consequently, a congruence gives rise to a smooth vector field in $S$, by taking it to be the tangent to the curves in the congruence at each point, and the converse is also true. A null geodesic congruence\footnote{See fig. \ref{fig:sub2} for an example of a congruence.} is one that the curves are null geodesics. One's interest in studying them arises from the fact that, through fairly general arguments, it is possible to derive results concerning the dynamics of null geodesics. Such results can then be used, for example, to study the behavior of the null segments of the boundary of the past or future of a set. 

The analysis of a congruence lies in the study of the transverse vector space, which is spanned by vectors orthogonal to the one that generates the congruence, known as deviation vectors (see appendix \ref{A3}). The failure of a deviation vector to be parallel transported along the curves of the congruence gives information about the dynamics of the curves, as it will show how they become closer or further apart, or rotate around each other. For a null geodesic congruence, the deviation vectors can be visualized as the separation between two “infinitesimally nearby” null geodesics that are being followed by light rays that were emitted by a source at the same time. The analysis of the deviation vectors in a null geodesic congruence can be made as follows.

Let $\ell^{\mu}$ be the tangent to a null geodesic congruence in $S$, which is parametrized by affine parameter $\lambda$ and let $V_a$ denote the tangent vector space at $a\in S$. The deviation vectors, $s^{\mu}$, lie in the vector subspace spanned by vectors orthogonal to $\ell^{\mu}$, denoted by $\overline{V}_a$. However, this is not the vector space of interest, because $\ell^{\mu}\in\overline{V}_a$. Namely, it is not possible to analyze the deviation vectors by restricting one's attention to  $\overline{V}_a$, as the projection of tensors to  $\overline{V}_a$ would be associated with a degenerate metric (see eq. \ref{deg}). Thus, it is necessary to isolate the purely transverse part of $\overline{V}_a$, i.e., elements of $\overline{V}_a$ that are not proportional to $\ell^{\mu}$. To do so, one must use an auxiliary vector, $\eta^{\mu}$, such that $\eta^{\mu}\ell_{\mu}\neq 0$, which means that $\ell^{\mu}$ does not belong to the vector space spanned by vectors orthogonal to $\eta^{\mu}$, denoted by $\overline{V'}_a$. Hence, the deviation vectors lie in the vector subspace spanned by the vectors that are orthogonal to both $\ell^{\mu}$ and $\eta^{\mu}$, denoted by $\hat{V}_a$. Since this vector subspace is spanned only by vectors that are orthogonal to $\ell^{\mu}$ and does not include $\ell^{\mu}$ itself, it is precisely the vector space of interest. In this manner, it is very convenient to use the gauge freedom of $\eta^{\mu}$ to make it null and obey $\eta^{\mu}\ell_{\mu}=-1$. Fig. \ref{fig:null} illustrates the vectors under these conditions and their respective orthogonal spaces at a point $a\in\gamma$, where $\gamma$ is a null geodesic of the congruence. Note that one dimension is suppressed, so that $\overline{V}_a$ and $\overline{V'}_a$ are planes and $\hat{V}_a$ is a line. \begin{figure}[h]
\centering
\includegraphics[scale=1.9]{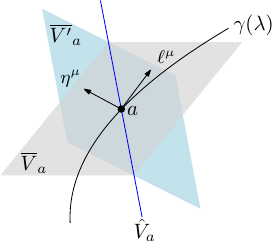} 
\caption{Diagram of the tangent vector spaces of interest at a point on a null geodesic.}
\caption*{Source: By the author.}
\label{fig:null}
\end{figure} With this choice of $\eta^{\mu}$ and considering the properties of an orthogonal projection operator presented in appendix \ref{integration}, the metric that acts on $\hat{V}_a$ is given by
\begin{equation}\label{transversemetric}
    h_{\mu\nu}=g_{\mu\nu}+2\ell_{(\mu}\eta_{\nu)},
\end{equation}
which can be used to project tensors on $\hat{V}_a$, isolating their purely transverse part. By projection of tensors, it is meant the contraction of all its indices with $h^{\mu}\mathstrut_{\nu}$.

The characterization of $\hat{V}_a$ becomes clear if one considers a \textit{local Lorentz frame}, that is, a coordinate system, $\{cx^0,x^1,x^2,x^3\}$, in a neighborhood of $a$ such that the metric is given by $\text{diag}(-1,1,1,1)$\footnote{This coordinate system is convenient to exemplify our analysis, but, as any other, it is otherwise meaningless and has no effect on scalars, which are the quantities of interest.}. Note that this is possible due to the property that any spacetime is locally flat, i.e., locally homeomorphic to Minkowski spacetime. In this coordinate system, one can write the tangent to the null geodesic congruence and an auxiliary null vector as
\begin{equation}\label{eq40}
    \ell^{\mu}=\frac{1}{\sqrt{2}}(1,1,0,0), 
\end{equation}
\begin{equation}\label{eq41}
    \eta^{\mu}=\frac{1}{\sqrt{2}}(1,-1,0,0),
\end{equation}
which indicate the Riemannian character of $h_{\mu\nu}$. Evidently, a basis of $\hat{V}_a$, $\{e^{\mu}_2,e^{\mu}_3\}$, is
\begin{equation}\label{eq43}
    e^{\mu}_2=(0,0,1,0), 
\end{equation}
\begin{equation}\label{eq42}
    e^{\mu}_3=(0,0,0,1).
\end{equation}

Since $\{\ell^{\mu},\eta^{\mu},e^{\mu}_2,e^{\mu}_3\}$ is a basis of $V_a$ and $s^{\mu}$ can be written as a linear combination of $e^{\mu}_2$ and $e^{\mu}_3$, it is clear that $\ell^{\mu}$ and $s^{\mu}$ commute (see eqs. \ref{aeq11} and \ref{commu}). Denoting $B_{\mu\nu}=\nabla_{\nu}\ell_{\mu}$ and using eq. \ref{aeq10}, one finds 
\begin{equation}
    s^{\mu}\nabla_{\mu}\ell^{\nu}=\ell^{\mu}\nabla_{\mu}s^{\nu}=B^{\nu}\mathstrut_{\mu}s^{\mu}.
\end{equation}
Hence, the map $B^{\nu}\mathstrut_{\mu}$ measures the failure of deviation vectors to be parallel propagated, and therefore, it is the operator of interest to analyze the dynamics of the congruence. Clearly, this map is orthogonal to $\ell^{\mu}$, in the sense that contraction of any of its indices with $\ell^{\mu}$ vanishes. However, it is not orthogonal to $\eta^{\mu}$, which means that the vector $B^{\nu}\mathstrut_{\mu}s^{\mu}$ may still be proportional to $\ell^{\mu}$. In order to isolate the purely transverse part of $B_{\mu\nu}$, one uses the transverse metric to project it on $\hat{V}_a$\footnote{This process can be interpreted as restricting the action of $B^{\nu}\mathstrut_{\mu}$ to vectors in $\hat{V}_a$.},
\begin{equation}\label{bhat}
\begin{aligned}[b]
    \hat{B}_{\mu\nu}& =h^{\alpha}\mathstrut_{\mu}h^{\beta}\mathstrut_{\nu}B_{\alpha\beta}\\
    &= B_{\mu\nu}+\ell_{\mu}\eta^{\alpha}B_{\alpha\nu}+\ell_{\nu}\eta^{\alpha}B_{\mu\alpha}+\ell_{\mu}\ell_{\nu}\eta^{\alpha}\eta^{\beta}B_{\alpha\beta}.
    \end{aligned}
\end{equation}
In this manner, $\hat{B}_{\mu\nu}$ is the operator that gives information about the purely transverse behavior of the null congruence. Due to the dimension of the vector space of interest, it is clear that $\hat{B}_{\mu\nu}$ is effectively a $2\times2$ matrix, so that it can be decomposed as
\begin{equation}\label{B}
    \hat{B}_{\mu\nu}=\frac{1}{2}\theta h_{\mu\nu}+\sigma_{\mu\nu}+\omega_{\mu\nu},
\end{equation}
where
\begin{align}
\begin{split}
    \theta& =\hat{B}^{\mu\nu}h_{\mu\nu},
\end{split}\\
\begin{split}
    \sigma_{\mu\nu}& =\hat{B}_{(\mu\nu)}-\frac{1}{2}\theta h_{\mu\nu},
\end{split}\\
    \omega_{\mu\nu}& =\hat{B}_{[\mu\nu]}.
\end{align}

Consequently, the action of $\hat{B}_{\mu\nu}$ on the deviation vectors can be analyzed by interpreting the terms of its decomposition, as exemplified in \cite{Poisson2004}. The trace of $\hat{B}_{\mu\nu}$, $\theta$, is associated with the rate of change of the congruence cross-section area. Therefore, $\theta>0$ means that the geodesics are diverging, while $\theta<0$ means that they are converging. Similarly, the symmetric tracefree part, $\sigma_{\mu\nu}$, is associated with the rate of change of the shape of the cross-section, and the antisymmetric part, $\omega_{\mu\nu}$, is associated with the rotation of the cross-section. Because of these interpretations, $\theta$ is referred to as the \textit{expansion}, $\sigma_{\mu\nu}$ as the \textit{shear tensor} and $\omega_{\mu\nu}$ as the \textit{vorticity tensor}.

The quantity of most physical significance is the expansion, which tells one how geodesics in the congruence move closer or further apart. In particular, one would like to calculate how it changes as one moves along the curves in the congruence. This can be done by first considering how $B_{\mu\nu}$ changes along the integral curves of $\ell^{\mu}$,
\begin{equation}\label{eq12}
\begin{aligned}[b]
    \ell^{\alpha}\nabla_{\alpha}B_{\mu\nu} &=\ell^{\alpha}\nabla_{\alpha}\nabla_{\nu}\ell_{\mu}\\&=\ell^{\alpha}\nabla_{\nu}\nabla_{\alpha}\ell_{\mu}+R_{\alpha\nu\mu}\mathstrut^{\beta}\ell^{\alpha}\ell_{\beta}\\
    & = \nabla_{\nu}(\ell^{\alpha}\nabla_{\alpha}\ell_{\mu})-(\nabla_{\nu}\ell^{\alpha})(\nabla_{\alpha}\ell_{\mu})+R_{\alpha\nu\mu}\mathstrut^{\beta}\ell^{\alpha}\ell_{\beta} \\
    & = -B^{\alpha}\mathstrut_{\nu}B_{\mu\alpha}-R_{\nu\alpha\mu}\mathstrut^{\beta}\ell^{\alpha}\ell_{\beta}.
\end{aligned}
\end{equation}
Now, using the explicit form of the transverse metric, one can verify that $\hat{B}_{\mu\nu}h^{\mu\nu}=B_{\mu\nu}g^{\mu\nu}$ and $B^{\mu\nu}B_{\mu\nu}=\hat{B}^{\mu\nu}\hat{B}_{\mu\nu}$. Considering eq. \ref{B}, one then finds
\begin{equation}
    \hat{B}^{\mu\nu}\hat{B}_{\mu\nu}=\frac{1}{2}\theta^2+\sigma^{\mu\nu}\sigma_{\mu\nu}+\omega^{\mu\nu}\omega_{\mu\nu},
\end{equation}
and thus, taking the trace of eq. \ref{eq12} yields 
\begin{equation}\label{eq13}
    \ell^{\mu}\nabla_{\mu}\theta=\frac{d\theta}{d\lambda}=-\frac{1}{2}\theta^2-\sigma^{\mu\nu}\sigma_{\mu\nu}+\omega^{\mu\nu}\omega_{\mu\nu}-R_{\mu\nu}\ell^{\mu}\ell^{\nu}.
\end{equation}

Eq. \ref{eq13} is known as \textit{Raychaudhuri's equation}, and it dictates the behavior of the expansion along the null geodesics in the congruence. To analyze it, it is of interest to study the sign of the non-expansion terms. The term $-\sigma_{\mu\nu}\sigma^{\mu\nu}$ is manifestly nonpositive due to the fact that the shear tensor is orthogonal to both $\ell^{\mu}$ and $\eta^{\mu}$, i.e., it is purely ``spacelike''. The sign of the Ricci tensor term can be analyzed under the assumptions of the energy conditions. If one contracts Einstein's equation twice with the tangent vector field to the null congruence, one obtains
\begin{equation}\label{eq17}
    R_{\mu\nu}\ell^{\mu}\ell^{\nu}=\frac{8\pi G}{c^4} T_{\mu\nu}\ell^{\mu}\ell^{\nu},
\end{equation}
which gives a relation to the Ricci tensor term in Raychaudhuri's equation, and it will be nonpositive if the weak or strong energy condition holds, as they imply the null energy condition. This can be interpreted as the attractive nature of gravity, as if such a condition is satisfied, null geodesics will tend to converge due to the contribution to the dynamics of the expansion. This interpretation also follows for a timelike congruence, in which the corresponding term in the Raychaudhuri's equation is the one analyzed in \S\;\ref{energycon}. Hence, the attractive nature of gravity follows from suitable energy conditions, as they represent an ``attractive'' contribution to the dynamics of geodesics, be it a timelike or a null one.

In order to analyze the term $\omega^{\mu\nu}\omega_{\mu\nu}$, one should consider the following result, known as \textit{Frobenius' theorem}, a proof of which can be found in \cite{Wald1984}. 
\begin{theorem}\label{frobenius}
    A congruence of curves, whose tangent vector is $\ell^{\mu}$, is orthogonal to a hypersurface (in the sense that it is proportional to the normal of a family of hypersurfaces described by $\Gamma(x^{a})=a$) if and only if $(\ell\wedge d\ell)_{\mu\nu\alpha}=0$, i.e., $\ell_{[\mu}\nabla_{\nu}\ell_{\alpha]}=0$.
\end{theorem}

One can deduce an immediate consequence of Frobenius' theorem by considering a hypersurface described by $\Gamma(x^{a})=a$, where $a$ is a constant scalar. With this characterization, the normal vector to this hypersurface is $w^{\mu}=\nabla^{\mu}\Gamma$, as it “points” in the direction of increasing $a$ and is orthogonal to the directions where $a$ is constant. However, if $w^{\mu}=\nabla^{\mu}\Gamma$ is a null vector, then it is also tangent to the null hypersurface. This result is of significance because, using eq. \ref{aeq3}, one readily obtains
\begin{equation}
    w^{\nu}\nabla_{\nu}w_{\mu}=w^{\nu}\nabla_{\mu}w_{\nu}=\frac{1}{2}\nabla_{\mu}(w^{\nu}w_{\nu}).
\end{equation}
In essence, since $w^{\nu}w_{\nu}$ vanishes on $\Gamma$, its gradient, $\nabla_{\mu}(w^{\nu}w_{\nu})$, must be proportional to its normal vector, $w^{\mu}$. This means that $w^{\nu}\nabla_{\nu}w^{\mu}\propto w^{\mu}$, which is the geodesic equation in a non-affinely parametrized form. Hence, $w^{\mu}$ is the tangent to the null geodesics that lie within $\Gamma$. More precisely, the tangent to a null geodesic congruence is normal to the hypersurface in which the null geodesics lie. Because of this, the null geodesics are referred to as the \textit{generators} of the null hypersurface. Thus, Frobenius' theorem implies that the tangent to any null geodesic congruence will obey $\ell_{[\mu}\nabla_{\nu}\ell_{\alpha]}=0$. In summary, from a physical perspective, each null hypersurface can be used to describe the propagation of the wavefront of the light rays following the null geodesics that generate it, and the action of the map $\hat{B}_{\mu\nu}$ on deviation vectors gives information about the behavior of the wavefront over time.

In light of this, the requirement for a null vector to be hypersurface orthogonal can be shown to be related to the vorticity tensor of the null congruence, which one can verify by contracting $\eta^{\mu}$ with the condition in Frobenius' theorem, 
\begin{equation}
    \begin{aligned}[b]
    0 &=\ell_{[\mu}\nabla_{\nu}\ell_{\alpha]}\eta^{\mu}\\
    &=(\ell_{\mu}\nabla_{[\nu}\ell_{\alpha]}+\ell_{\nu}\nabla_{[\alpha}\ell_{\mu]}+\ell_{\alpha}\nabla_{[\mu}\ell_{\nu]})\eta^{\mu}\\
    &=-B_{[\alpha\nu]}+\ell_{\nu}\eta^{\mu}B_{[\mu\alpha]}+\ell_{\alpha}\eta^{\mu}B_{[\nu\mu]}\\
    &=-B_{[\alpha\nu]}+B_{\mu[\alpha}\ell_{\nu]}\eta^{\mu}+\ell_{[\alpha}B_{\nu]\mu}\eta^{\mu}\\
    &=\hat{B}_{[\alpha\nu]}\\
    &=\omega_{\alpha\nu},
    \end{aligned}
\end{equation}
where eq. \ref{bhat} was used in the fourth line. Hence, Frobenius' theorem implies that the vorticity tensor of any null geodesic congruence must vanish\footnote{Such a strong result is not valid for timelike geodesic congruences. Nonetheless, Frobenius' theorem also implies that in order for a congruence of timelike curves to be orthogonal to a hypersurface, its vorticity tensor must also vanish \cite{Poisson2004}.}. These conclusions lead to the following result.

\begin{proposition}\label{prop1}
    Let $\ell^{\mu}$ denote the tangent vector to a null geodesic congruence and $R_{\mu\nu}\ell^{\mu}\ell^{\nu}\geq 0$. If the expansion of the congruence attains the negative value $\theta_0$ at any point on a geodesic in the congruence, then $\theta\to -\infty$ along that geodesic in affine parameter $\lambda\leq 2/|\theta_0|$, assuming that such a geodesic extends that far.
\end{proposition}
\begin{proof}
Under the assumption that $R_{\mu\nu}\ell^{\mu}\ell^{\nu}\geq 0$ and that the congruence is of null geodesics, Raychaudhuri's equation implies that
\begin{equation}
    \frac{d\theta}{d\lambda}+\frac{1}{2}\theta^2\leq 0,
\end{equation}
which can be rewritten as
\begin{equation}
    \frac{d}{d\lambda}(\theta^{-1})\geq 0,
\end{equation}
and hence
\begin{equation}\label{eq15}
    \theta^{-1}\geq \theta_0^{-1}+\frac{1}{2}\lambda,
\end{equation}
where $\theta_0$ is the initial value of $\theta$. If the congruence is initially converging, eq. \ref{eq15} implies that $\theta\to -\infty$ within an affine parameter $\lambda\leq 2/|\theta_0|$.
\end{proof}

This result can be interpreted as stating that, in a spacetime where Einstein's equation and the weak or strong energy condition hold, a congruence of null geodesics will develop caustics in a finite affine parameter. A \textit{caustic} is a point at which some curves in the congruence intersect, which is merely a singularity of the congruence (i.e., a point in which it is not defined), and implies nothing regarding pathologies in the structure of spacetime.

We conclude this section by noting that the assumption that the null geodesics are affinely parametrized is merely a way to simplify our analysis. In particular, the following version of Raychaudhuri's equation is valid for arbitrarily parametrized null geodesic congruence, 
\begin{equation}\label{eqq2}
    \frac{d\theta}{d\lambda}=\kappa\theta-\frac{1}{2}\theta^2-\sigma^{\mu\nu}\sigma_{\mu\nu}+\omega^{\mu\nu}\omega_{\mu\nu}-R_{\mu\nu}\ell^{\mu}\ell^{\nu},
\end{equation}
which can be readily derived from \ref{eq12} by using the geodesic equation, $\ell^{\nu}\nabla_{\nu}\ell^{\mu}=\kappa\ell^{\mu}$, and that the expansion is now defined as 
\begin{equation}
    \theta=\nabla^{\mu}\ell_{\mu}-\kappa.
\end{equation}
Consequently, the results of prop. \ref{prop1} are still valid, the only difference is in the limit of the parameter for the development of caustics. 

\section{Conjugate points}\label{conj}

Conjugate points are of interest because they represent points that can be almost joined by a family of geodesics. More precisely, consider a null geodesic, $\gamma$, with tangent $\ell^{\mu}$. If a vector field, $s^{\mu}$, is a solution of the Jacobi equation, 
\begin{equation}
   \ell^{\nu}\nabla_{\nu}(\ell^{\alpha}\nabla_{\alpha}s^{\mu})=-R_{\nu\alpha\beta}\mathstrut^{\mu}\ell^{\nu}s^{\alpha}\ell^{\beta},
\end{equation}
it is referred to as a \textit{Jacobi field} on $\gamma$. Since it is a deviation vector, a Jacobi field can be interpreted as the separation of two “infinitesimally nearby” geodesics. A pair of points $a,\;a'\in \gamma$ are said to be \textit{conjugate along $\gamma$} if there exists a Jacobi field which is not identically zero but vanishes at both $a$ and $a'$. Thus, conjugate points can be interpreted as points in which an “infinitesimally nearby” geodesic to $\gamma$ intersects it at both $a$ and $a'$. The following proposition strengthens this interpretation, as it shows that the existence of conjugate points is associated with the behavior of the expansion of a congruence.
\begin{proposition}\label{p1}
    Let $\gamma$ be a null geodesic and $a,\;a'\in\gamma$. Suppose that the congruence of null geodesics of which $\gamma$ is a part emanates from $a$. Then $a'$ is conjugate to $a$ if and only if $\theta\to-\infty$ at $a'$.
\end{proposition}
\begin{proof}
    Let $\ell^{\mu}$ be the tangent to $\gamma$, being parametrized by $\lambda$. The components of the Jacobi field, $s^{\mu}\in\hat{V_a}$, of $\gamma$, obey the linear ordinary differential equations at each $a\in\gamma$
\begin{equation}\label{eq21}
   \frac{d^2 s^a}{d\lambda^2}=-R_{bcd}\mathstrut^{a}\ell^{b}s^{c}\ell^{d},
\end{equation}
and thus, $s^a(\lambda)$ must depend linearly on the initial conditions, i.e.,
\begin{equation}
    s^a(\lambda)=S^a\mathstrut_b(\lambda)\frac{ds^{b}}{d\lambda}(0)+S'^a\mathstrut_b(\lambda)s^b(0).
\end{equation}
However, since the congruence emanates from $a$, $s^a(0)=0$, so that
\begin{equation}\label{eq22}
   s^a(\lambda)=S^a\mathstrut_b(\lambda)\frac{ds^{b}}{d\lambda}(0),
\end{equation}
from which one can deduce that $S^{a}\mathstrut_{b}(0)=0$ and $dS^{a}\mathstrut_{b}/d\lambda(0)=\delta^{a}\mathstrut_b$. By combining eqs. \ref{eq21} and \ref{eq22}, one obtains
\begin{equation}\label{eq23}
   \frac{d^2S^{a}\mathstrut_{b}}{d\lambda^2}=-R_{cde}\mathstrut^{a}\ell^{c}\ell^{e}S^{d}\mathstrut_{b}.
\end{equation}
Since $a'$ will be conjugate to $a$ if and only if there exists nontrivial initial data for which $s^{\mu}=0$ at $a'$, then by eq. \ref{eq22}, the necessary and sufficient condition of conjugacy is that $\text{det}\;(S^{\mu}\mathstrut_{\nu})=0$ at $a'$. From this, it follows that $\text{det}\;(S^{\mu}\mathstrut_{\nu})\neq0$ between conjugate points. As such, it is possible to study the condition on the determinant of $S^{\mu}\mathstrut_{\nu}$ by noting that it must be related to the tensor field $B_{\mu\nu}=\nabla_{\nu}\ell_{\mu}$ of the congruence. This relation can be derived from
\begin{equation}
    \frac{ds^{a}(\lambda)}{d\lambda}=\ell^{b}\nabla_{b}s^{a}=B^{a}\mathstrut_{b}s^{b},
\end{equation}
which follows from the fact that the Jacobi field and $\ell^{\mu}$ commute. From eq. \ref{eq22}, one then has
\begin{equation}
    \frac{dS^{a}\mathstrut_{b}(\lambda)}{d\lambda}=B^{a}\mathstrut_{c}S^{c}\mathstrut_{b},
\end{equation}
which in matrix form reads
\begin{equation}
    \frac{dS(\lambda)}{d\lambda}=BS(\lambda),\quad B=\frac{1}{S}\frac{dS(\lambda)}{d\lambda}.
\end{equation}
In order to proceed, it is necessary to make use of the result that the derivative and trace operator commute, i.e,
\begin{equation}
    \text{tr}\left(\frac{dA(\lambda)}{d\lambda}\right)=\frac{d}{d\lambda}\left[\text{tr}(A(\lambda))\right],
\end{equation}
and the identity $\text{tr}(A(\lambda))=\ln{\left[\text{det}(\exp{A(\lambda)})\right]}$ which is valid for any invertible matrix $A$ \cite{Magnus1999}. The definition of the expansion then yields
\begin{equation}
\begin{aligned}
\theta &=\text{tr}(B)\\
&=\text{tr}\left(\frac{1}{S}\frac{dS(\lambda)}{d\lambda}\right)\\
&=\text{tr}\left(\frac{d}{d\lambda}\left[\ln{S(\lambda)}\right]\right)\\
&=\frac{d}{d\lambda}(\ln|\text{det} (S(\lambda))|)\\
&=\frac{1}{\text{det} (S(\lambda))}\frac{d}{d\lambda}\text{det} (S(\lambda)).
\end{aligned}
\end{equation}
Now, since the map $S^{a}\mathstrut_{b}(\lambda)$ satisfies eq. \ref{eq23}, the derivative of its determinant cannot vanish anywhere on $\gamma$. Hence, $\theta\to-\infty\Leftrightarrow\text{det} (S(\lambda))\to 0$. 
\end{proof}

In other words, conjugate points are associated with caustics in a congruence. Namely, one can use proposition \ref{prop1} to also deduce the existence of a point conjugate to the one at which the congruence emanates if the expansion attains a negative value at any point on the congruence. The next result states an important consequence of the existence of conjugate points on null geodesics, a proof of which can be found in \cite{Penrose1972}.

\begin{theorem}\label{theorem4}
    Let $\gamma$ be a null geodesic and $a,\;a'\in\gamma$. If there is a point conjugate to $a$ in $(a,a')$, then there is a timelike curve connecting $a$ to $a'$.
\end{theorem}

Hence, the existence of conjugate points along a null geodesic implies that it can be smoothly deformed to yield a timelike curve. This can be interpreted as the failure of the null geodesic to remain on the boundary of the causal future (past) of a set, i.e., it must have entered the chronological future (past) of the set. The following remarks and results will be stated for the causal future of a set, but they are also valid for the causal past, with adequate changes in conditions.

A similar notion of conjugacy can be defined for a point and a two-dimensional spacelike submanifold (i.e., a surface), $S$, of a spacetime $(M,g_{\mu\nu})$. At each point $a\in S$, there exist two future directed null vectors that are orthogonal to $S$. An example of such a surface is the one spanned by the deviation vectors of a null geodesic congruence. As discussed in \S\;\ref{null}, $\{e^{\mu}_2,e^{\mu}_3\}$ span such a surface, while $\{\ell^{\mu},\eta^{\mu}\}$ are the future directed null vectors (here it is assumed that a time orientation is provided by $t^{\mu}=(1,0,0,0)$). Due to the choice of coordinates on a local Lorentz frame, $\ell^{\mu}$ is referred to as the tangent to the “outgoing” null congruence, and $\eta^{\mu}$ as the tangent to the “incoming” null congruence. Fig. \ref{fig:surface} exemplifies such a surface, in which one space dimension is suppressed, so that $S$ is depicted as a curve.

\begin{figure}[h]
\centering
\includegraphics[scale=1.2]{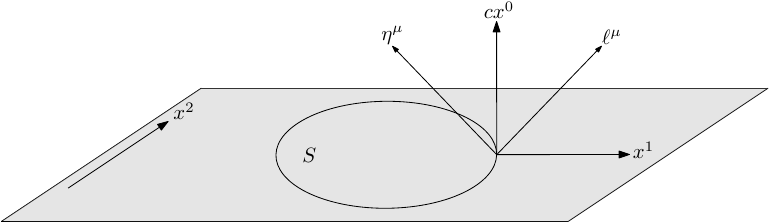} 
\caption{Surface spanned by vectors orthogonal to the ``incoming'' and ``outgoing'' null vectors, or equivalently, to $(\partial_{x^0})^{\mu}$ and $(\partial_{x^1})^{\mu}$.}
\caption*{Source: By the author.}
\label{fig:surface}
\end{figure}

 With these remarks, the notion of conjugacy can be defined in the following manner. Let $\gamma$ be a null geodesic orthogonal to $S$ (which will be the integral curve of $\ell^{\mu}$ or $\eta^{\mu}$) and $a'\in\gamma$ but $a'\not\in S$. The point $a'$ is said to be conjugate to $S$ along $\gamma$ if there exists a Jacobi field, $s^{\mu}$, on $\gamma$ which is nonzero on $S$ but vanishes at $a'$. By the same arguments of the conjugacy of two points, $a'$ will be conjugate to $S$ if and only if the expansion of the congruence of geodesics is orthogonal (the one generated by $\ell^{\mu}$ or $\eta^{\mu}$) to $S$ approaches $-\infty$ at $a$. Using proposition \ref{prop1} and a development analogous to that of the proof of proposition \ref{p1}, one can derive the following result.

\begin{proposition}\label{p2}
    Let $(M,g_{\mu\nu})$ be a spacetime satisfying $R_{\mu\nu}\ell^{\mu}\ell^{\nu}\geq 0$ for all null $\ell^{\mu}$. Let $S$ be a two-dimensional spacelike submanifold of $M$ such that the expansion of an orthogonal null geodesic congruence generated by $\ell^{\mu}$ has the negative value $\theta_0$ at $a\in S$. Then within finite parameter, there exists a point $a'$ conjugate to $S$ along the null geodesic $\gamma$ passing through $a$, assuming that $\gamma$ extends that far.
\end{proposition}

The requirement of $R_{\mu\nu}\ell^{\mu}\ell^{\nu}\geq 0$ for all $\ell^{\mu}$ null is a consequence of the fact that any two null vectors whose inner product does not vanish can be used to construct a surface that respects the properties of the proposition, simply by taking the orthogonal space, $\hat{V}$. Thus, such a condition could be weakened to require that $R_{\mu\nu}\ell^{\mu}\ell^{\nu}\geq 0$ be respected only for the tangent to the null congruence one wishes to analyze. Furthermore, the null geodesics need to be orthogonal to $S$ as a consequence of the definition of conjugacy of a surface and a point. In particular, if $\gamma$ were not orthogonal to $S$, the Jacobi field would not lie within $S$, and the definition of conjugacy would not be adequate. With this in mind, a similar proof to that of theorem \ref{theorem4} yields the following \cite{Penrose1972}.
\begin{theorem}\label{theorem5}
    Let $(M,g_{\mu\nu})$ be a spacetime, let $S$ be a two-dimensional spacelike submanifold of $M$ and let $\gamma$ be a differentiable causal curve from $S$ to $a$. Then the necessary and sufficient condition that $\gamma$ cannot be smoothly deformed to a timelike curve connecting $S$ and $a$ is that $\gamma$ be a null geodesic orthogonal to $S$ with no conjugate point to $S$ between $S$ and $a$.
\end{theorem}

Finally, as a consequence of theorem \ref{theorem5} and the properties of a globally hyperbolic spacetime, it is possible to state and prove the following theorem.
\begin{theorem}\label{theorem6}
    Let $(M,g_{\mu\nu})$ be a globally hyperbolic spacetime and let $S$ be a compact orientable two-dimensional spacelike submanifold of $M$. Then every $a\in\partial J^+(S)$ lies on a future directed null geodesic starting from $S$ which is orthogonal to $S$ and has no conjugate point to $S$ between $S$ and $a$.
\end{theorem}
\begin{proof}
    From the remarks below theorem \ref{theorem7}, it follows that if $a\in\partial J^+(S)$, then it must be connected to $S$ by a null geodesic. By theorem \ref{theorem5}, if this null geodesic were not orthogonal to $S$ or had a conjugate point between $S$ and $a$, then it would be possible to deform it to produce a timelike curve connecting $S$ to $a$, which would mean that $a\not\in\partial J^+(S)$.
\end{proof}

\section{Singularities}\label{sing}

This section is devoted to a brief discussion of a singularity theorem which proves, under certain conditions, the development of singularities in the context of gravitational collapse. First, it should be noted that a precise definition of a singularity is significantly problematic, as none of the many attempts to indicate their presence in a spacetime seem to fully describe it in all the necessary aspects \cite{Wald1984}. Nevertheless, one satisfying way to characterize a spacetime as possessing a singularity is by identifying the “holes” it leaves behind. In other words, one can define a spacetime to be singular by identifying geodesics that reach these “holes”, which one could justifiably assume to be due to the presence of a singularity \cite{Geroch1968}. 

In this manner, it is necessary to give a precise notion of what it means for a geodesic to reach such a pathology. Since geodesics are a property of the intrinsic spacetime structure, the failure of their affine parameters to extend to arbitrarily large values can be associated with an encounter with a pathology. In light of this, a geodesic is defined to be \textit{incomplete} if it is inextendible in at least one direction, but has only a finite range of affine parameter. It is easy to see that this definition precisely accounts for what happens when a geodesic encounters a ``removed point'', as the one exemplified in fig. \ref{fig:prop23455}. Consequently, a spacetime is said to be \textit{singular} if it possesses at least one incomplete geodesic. Although this definition does not give a perfect notion of a singularity or details about its nature, the pathology arising in a spacetime that has at least one incomplete geodesic is evident, and one can reason that this class of spacetimes earns the adjective “singular”. In fact, the notion of geodesic incompleteness is the one present in the singularity theorems \cite{Hawking1973}, and it is reasonable to believe that this concept is enough to conclude that spacetime pathologies $-$or singularities$-$ are predicted by general relativity under certain conditions. Essentially, this line of reasoning is based on the idea that observers or light rays that follow incomplete geodesics will end their existence in a finite affine parameter.

In order to state the pertinent theorem that proves geodesic incompleteness, it is necessary to define the notion of a trapped surface. Let $(M,g_{\mu\nu})$ be a spacetime. A \textit{trapped surface}, $T$, is a closed (i.e., compact and without boundary), two-dimensional spacelike submanifold of $M$ such that the expansion of both sets of orthogonal future directed null geodesics (e.g., the future directed “incoming” and “outgoing” families of null geodesics) is everywhere negative. Fig. \ref{trap234} illustrates a surface, $T$, which is emitting a flash of light. The surfaces $S'$ and $S''$ illustrate the behavior of the ``incoming'' and ``outgoing'' wavefronts of the light emitted, respectively. As exemplified for a single event $a\in T$, such surfaces are constructed by considering all events in $T$. If the areas of $S'$ and $S''$ are both less\footnote{In the following, we shall refer to the incoming and outgoing families of null geodesic without the quote unquote, but the reader should recall that such nomenclature does not necessarily indicate the behavior of the geodesics. For example, for a trapped surface, the outgoing null geodesics, which one would expect to be ``naturally'' diverging from one another, are actually converging.} than the area of $T$, then $T$ is a closed trapped surface, as exemplified in fig. \ref{fig:trap12}. 
\begin{figure}[h]
  \begin{subfigure}[b]{0.5\textwidth}
  \centering
    \includegraphics[scale=1]{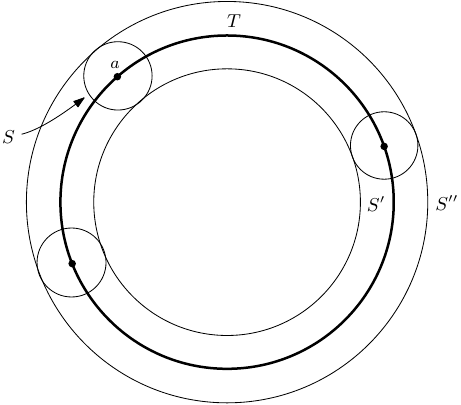}
    \caption{Behavior of the wave fronts of $T$.}
    \label{trap234}
  \end{subfigure}
  \hfill
  \begin{subfigure}[b]{0.5\textwidth}
  \centering
    \includegraphics[scale=1.08]{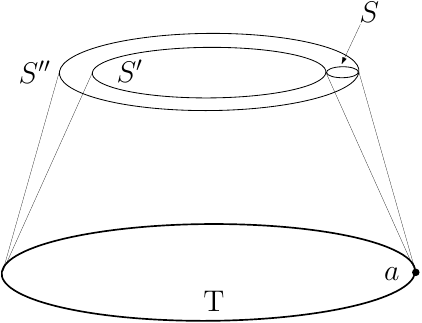}
    \caption{Example of a trapped surface.}
    \label{fig:trap12}
  \end{subfigure}
  \caption{Spacelike surface, $T$, and the orthogonal families of null geodesics.}\label{fig:trap1}
  \caption*{Source: Adapted from HAWKING; ELLIS \cite{Hawking1973}.}
\end{figure}

The requirement of $T$ to be closed is related to the idea that one wishes to interpret trapped surfaces in the context of spacetime curvature. In other words, if one does not require $T$ to be closed, the intersection of the past light cone of any two spacelike separated points in Minkowski spacetime would be an ``open trapped surface” \cite{Frolov1998}. As will be discussed below, trapped surfaces will be associated with regions of strong gravitational field, and the trivial ``open trapped surfaces” that exist in flat spacetime are uninteresting. Not only that, the compact property is essential to the proof of the pertinent theorem. Evidently, a trapped surface satisfies the requirements of propositions \ref{p2} and theorems \ref{theorem5} and \ref{theorem6}, which in combination with restrictions on the causal structure can be related to geodesic incompleteness. Such a relation is given by the next theorem, originally derived in \cite{Penrose1965}. 
\begin{theorem}\label{HawPen}
    Let $(M,g_{\mu\nu})$ be a globally hyperbolic spacetime with a noncompact Cauchy hypersurface satisfying $R_{\mu\nu}\ell^{\mu}\ell^{\nu}\geq 0$ for all null $\ell^{\mu}$. If $(M,g_{\mu\nu})$ contains a trapped surface, then it is not null geodesically complete.
\end{theorem}

The conditions of the theorem above can be interpreted as follows. Concerning the restriction on the contraction of the Ricci tensor, as has already been discussed, it will be respected if Einstein's equation holds and the weak or strong energy condition is respected by $T_{\mu\nu}$. The condition that spacetime contains a trapped surface can be interpreted by considering that the light rays that are following the outgoing null geodesics, which one would expect to be ``naturally'' diverging, are actually converging. In particular, the formation of a trapped surface can be interpreted as the limit at which gravitational collapse can no longer be stopped, since the areas of wave fronts of both incoming and outgoing families of null geodesics will decrease, and the distribution of energy contained inside a trapped surface will necessarily be contained in the region delimited by the outgoing null geodesics. Additionally, the requirement of existence of a noncompact Cauchy hypersurface can be interpreted as the requirement that the universe is ``infinite''. In essence, this means that the spatial section of the universe does not ``close around itself'', e.g., a plane or hyperbolic spatial section. The conclusion of null geodesic incompleteness can then be interpreted as a consequence of the pathological region resulting from the collapse associated with the trapped surface, which will undoubtedly be related to loss of determinism in at least one region of spacetime. Finally, the proof of this theorem consists in obtaining a contradiction between the null geodesic completeness, which implies that $\partial J^+(T)$ must be compact, and the existence of a noncompact Cauchy hypersurface. Details on this development can be found in, e.g., \cite{Wald1984}. 

Arguably, the assumption of global hyperbolicity is by far the strongest, but the restriction on the topology of the Cauchy hypersurface is also severely significant, such that it raises questions regarding the physical reliability of the results of this theorem. Evidently, experimental confirmation of either assumption is impossible, so that although this theorem provides a proof of the development of pathologies in a type of spacetime, its physical relevance is questionable. Nonetheless, developments originally proposed in \cite{Hawking1970} removed the requirement of the existence of Cauchy hypersurfaces. Not only that, the results derived there relate the existence of singularities in the context of gravitational collapse and cosmology with fairly general assumptions, the former still being related to the existence of trapped surfaces. It is in this sense that singularities in physically reliable spacetimes are a genuine prediction of general relativity.

\section{Asymptotic flatness}\label{flat}

The last tool necessary in order to have the complete framework to study black holes is the concept of asymptotic flatness. This concept is of significance because it characterizes isolated systems in the context of general relativity, and it also gives rise to a satisfying notion of what it means for an observer to be “distant” from the sources. For instance, the usefulness of this notion in general relativity can be seen in situations where one wishes to study the character of emitted radiation in a system, as it is equivalent to analyzing the behavior of the metric for large distances and late times. Such analysis can then be used to study what events can be causally connected to the “distant” regions. In essence, since the presence of energy can be detected by the curvature of spacetime (as per Einstein's equation), a region can be characterized as “distant” if it is very close to being flat, that is, if the spacetime behaves similarly to Minkowski spacetime in comparison with some sense of ``central region''.

Roughly speaking, asymptotic flatness can then be summarized as a spacetime whose metric, written in coordinates $\{t,x,y,z\}$, reduces to that of Minkowski as $r\to\infty$ and $ct\to\pm\infty$, where $r=\left(x^2+y^2+z^2\right)^{1/2}$. The issue with simply defining a spacetime to be asymptotically flat as such arises from the fact that this notion is coordinate dependent, and one wishes to be able to analyze the quantities at the limit of asymptotic flatness in a coordinate independent manner. To progress towards a definition that is adequate, it is useful to first analyze the Minkowski metric for large distances and characterize the ``distant'' regions from some ``central region'' given by a coordinate system. Such characterization can then be used to compare with regions of other spacetimes, and identify those that behave similarly. 

To start this analysis, consider the form of the Minkowski metric, $\eta_{\mu\nu}$, in null coordinates,
\begin{equation}\label{minnull}
    ds^2=-dudw+\frac{1}{4}(w-u)^2(d\theta^2+\sin^2\theta d\phi^2),
\end{equation}
where $w=ct+r$, $u=ct-r$ and $\{r,\theta,\phi\}$ are the ordinary spherical coordinates. Note that the pathologies such as $w-u=0$ and $\sin{\theta}=0$ are coordinate dependent, i.e., they are merely a consequence of the limitations of the coordinate system. In particular, they can be removed by imposing adequate coordinate restrictions, e.g., $0<\theta<\pi$, and at the points where the metric given by eq. \ref{minnull} is not defined, one has to use another appropriate coordinate system. Additionally, the lack of terms $dw^2$ and $du^2$ in eq. \ref{minnull} is a consequence of the fact that the coordinate basis vectors $(\partial_w)^{\mu}$ and $(\partial_u)^{\mu}$ are null, thus, the hypersurfaces $\{w=\text{constant}\}$ and $\{u=\text{constant}\}$ are also null. Because of this, such coordinates can then be identified as representing the outgoing and incoming radial null geodesics, respectively. To analyze the behavior of one of these families at large distances, it would be necessary to take the limit of the corresponding null coordinate. However, it is evident that taking the limit of $w$ or $u$ to infinity would yield a badly behaving metric. Furthermore, coordinate transformations that would be able to describe infinity at a finite distance, such as $1/w$, would also result in a pathological metric. Nevertheless, such pathologies can be removed by means of a conformal transformation.

A \textit{conformal transformation} is a map, $\psi:M\to M'$, such that its action on $g_{\mu\nu}$ is given by $\psi^* g_{\mu\nu}=\Omega^2g_{\mu\nu}$. The smooth, nonvanishing, non-negative function $\Omega$ is called \textit{conformal factor}. Additionally, if $\psi$ is a diffeomorphism, it is also said to be a \textit{conformal isometry}. In this context, the metric $\Omega^2g_{\mu\nu}$ is referred to as the \textit{unphysical metric} and $(\psi[M],\Omega^2g_{\mu\nu})$ is the \textit{unphysical spacetime}. These nomenclatures are a consequence of the fact that the curvature tensors are not preserved in a general conformal transformation, so that $\Omega^2g_{\mu\nu}$ will not be a solution to Einstein's equation.

Consider a conformal isometry of Minkowski spacetime,
\begin{equation}
    {\eta'}_{\mu\nu}=\Omega^2\eta_{\mu\nu},
\end{equation}
with conformal factor given by
\begin{equation}
    \Omega^2=4(1+w^2)^{-1}(1+u^2)^{-1}.
\end{equation}
By defining new coordinates, $t'$ and $r'$,
\begin{equation}
    t'=\arctan w+\arctan u,\; r'=\arctan w-\arctan u,
\end{equation}
the metric ${\eta'}_{\mu\nu}$ in coordinates $\{t',r',\theta,\phi\}$ reads
\begin{equation} \label{89}
    (ds')^2=-(dt')^2+(dr')^2+\sin^2 r'(d\theta^2+\sin^2\theta d\phi^2).
\end{equation}
The metric given by eq. \ref{89} is precisely the Lorentz metric on the four-dimensional cylinder, $\mathbb{S}^3\times\mathbb{R}$, but the coordinate ranges of $t'$ and $r'$ are restricted by the coordinate transformation, which are given by
\begin{equation}\label{eq25}
    -\pi < t'+r'< \pi,\quad-\pi < t'-r'< \pi,\quad0\leq r'.
\end{equation}
Clearly, the metric is pathological at $r'=0$, $r'=\pi$, $\theta=0$ and $\theta=\pi$. The removal of such singularities also follows from adequate coordinate restrictions. However, note that the pathological points are not ones associated with large distances \textit{and} times, as was the case of the metric given by eq. \ref{minnull}. Hence, this metric is smooth in the regions of interest.

The conformal transformation that results in the metric given by eq. \ref{89} and coordinate restrictions given by eq. \ref{eq25} is referred to as the \textit{conformal compactification} of the Minkowski spacetime. The \textit{conformal infinity of Minkowski spacetime} is defined as the boundary of the open region, $S$, given by the coordinate restrictions of eq. \ref{eq25} in the $\mathbb{S}^3\times\mathbb{R}$ manifold. Consequently, one may view a conformal compactification as a way to produce a well behaved metric that brings the infinitely ``distant'' regions of a physical spacetime in time or space to a finite region in the unphysical spacetime, which is precisely the boundary $\partial S$. The conformal infinity of Minkowski spacetime can be naturally divided into five parts, as detailed below and illustrated in the conformal compactification, fig. \ref{fig:min1}.\begin{figure}[h]
\centering
\includegraphics[scale=1.2]{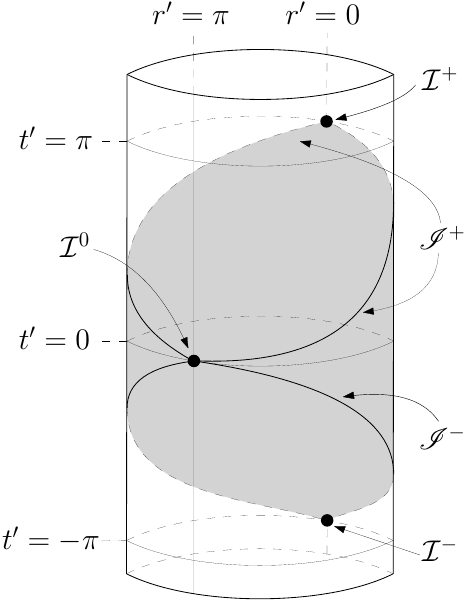} 
\caption{Conformal compactification of Minkowski spacetime in the manifold $S^3\times\mathbb{R}$.}
\caption*{Source: Adapted from WALD \cite{Wald1984}.}
\label{fig:min1}
\end{figure}
\begin{enumerate}
    \item[(1)] The point $\mathcal{I}^-$, called \textit{past timelike infinity}, given by coordinates $t'=-\pi\;,r'=0$, i.e., $ct\to-\infty$ at finite $r$.
    \item[(2)] The null hypersurface $\mathscr{I}^-$, called \textit{past null infinity}, given by $t'=-\pi+r'$ for $0<r'<\pi$, i.e., $ct-r\to-\infty$ at finite $ct+r$.
    \item[(3)] The point $\mathcal{I}^0$, called \textit{spatial infinity}, given by coordinates $t'=0,\;r'=\pi$, i.e., $r\to\infty$ at finite $ct$.
    \item[(4)] The null hypersurface $\mathscr{I}^+$, called \textit{future null infinity}, given by $t'=\pi-r'$ for $0<r'<\pi$, i.e., $ct+r\to\infty$ at finite $ct-r$.
    \item[(5)] The point $\mathcal{I}^+$, called \textit{future timelike infinity}, given by coordinates $t'=\pi\;,r'=0$, i.e., $ct\to\infty$ at finite $r$.
\end{enumerate}

The interpretation of each of these sets and points follows from the analysis of Minkowski spacetime in spherical coordinates. But first, note that a conformal transformation may affect the norm of a vector, but it always leaves its characterization unchanged. More precisely, it ``maps'' (see appendix \ref{derivativeoperators}) timelike, spacelike and null vectors into timelike, spacelike, and null vectors, respectively. Thus, it preserves the causal structure of the transformed spacetime. In a more general manner, one also has the following result \cite{Wald1984}.

\begin{proposition}\label{nulll}
    Let $(M,g_{\mu\nu})$ be a spacetime, $\psi$ denote a conformal isometry, $\gamma$ be a curve with tangent vector $\ell^{\mu}$ and $\nabla'_{\mu}$ denote the connection compatible with $\Omega^2g_{\mu\nu}$. If $\ell^{\mu}$ is null and obeys $\ell^{\nu}\nabla_{\nu}\ell^{\mu}=f\ell^{\mu}$, where $f$ is an arbitrary function on $\gamma$, then $\ell^{\nu}\nabla'_{\nu}\ell^{\mu}\propto\ell^{\mu}$.
\end{proposition}

In particular, this result can be interpreted as stating that null geodesics are \textit{conformally invariant}\footnote{However, an affine parameter for $\gamma$ may not be an affine parameter for $\psi[\gamma]$.}. Although it is not valid, in general, for timelike or spacelike geodesics, one can still deduce that all timelike geodesics of Minkowski spacetime begin at $\mathcal{I}^-$ and end at $\mathcal{I}^+$, and all spacelike sections pass through $\mathcal{I}^0$ \cite{Hawking1973}. Similarly, all null geodesics begin at $\mathscr{I}^-$ and end at $\mathscr{I}^+$. However, note that a non-geodesic timelike curve may end at $\mathscr{I}^+$. By construction, radial null geodesics correspond to $\pm45\degree$ lines in the \textit{conformal diagram} (also known as a \textit{Penrose diagram} \cite{Penrose1963}) of Minkowski spacetime on the $(t',r')$ plane, fig. \ref{fig:min2}.\begin{figure}[h]
\centering
\includegraphics[scale=1.5]{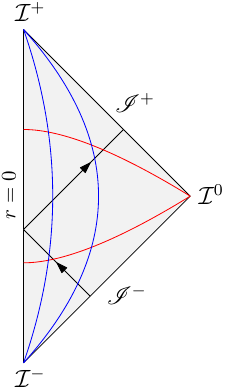} 
\caption{Conformal diagram of Minkowski spacetime.}
\caption*{Source: By the author.}
\label{fig:min2}
\end{figure} Note that the angular coordinates are suppressed, so each point, except for $\mathcal{I}^-$, $\mathcal{I}^+$, $\mathcal{I}^0$ and those with $r=0$, is a two-sphere of radius $r(t',r')$. Furthermore, red lines correspond to surfaces of constant $t$, while blue lines correspond to surfaces of constant $r$, and a radial null geodesic is represented by a black line inside the diagram. Note that the red and black lines are “reflected” on $r=0$, as the conformal diagram can be drawn from a part of the conformal compactification without loss of information. In this sense, one can see that an incoming null geodesic emerges from $\mathscr{I}^-$, reaches the ``center'' of the spacetime and then emerges as an outgoing null geodesic which will eventually reach $\mathscr{I}^+$.

The analysis of the conformal compactification of Minkowski spacetime allows one to give a precise definition of asymptotic flatness, using the following line of reasoning and definitions. A spacetime, $(M,g_{\mu\nu})$, is said to be \textit{asymptotically simple} if it there exists a spacetime, $(M',{g'}_{\mu\nu})$, and a conformal isometry, $\psi:M\to{M'}$, with conformal factor $\Omega$ such that $\Omega=0|_{\partial(\psi[M])}$, $\nabla_{\mu}\Omega\neq0|_{\partial(\psi[M])}$, $\partial(\psi[M])$ is smooth, and every null geodesic has two endpoints on $\partial(\psi[M])$. This definition precisely captures the notion of a spacetime being asymptotically similar to that of Minkowski, as it is possible to see in detail from its conditions. By requiring that $M$ be related by a conformal isometry to an open region of $M'$, one demands that $\psi[M]$ has the same causal structure as $M$. The restrictions on the behavior of the conformal factor are necessary for the correspondence between $\partial(\psi[M])$ and the infinity of the transformed spacetime. In particular, $\Omega=0|_{\partial(\psi[M])}$ implies that the affine parameter of a null geodesic diverges on $\partial(\psi[M])$, which can be interpreted as the infinite rescaling of the affine parameter at such regions. Moreover, by analysis of the differentiability of the Ricci scalar of $g'_{\mu\nu}$ \cite{Hawking1973}, the condition $\nabla_{\mu}\Omega\neq0|_{\partial(\psi[M])}$ implies that $\nabla^{\mu}\Omega|_{\partial(\psi[M])}$ must be a null vector, which is precisely the normal to $\mathscr{I}^-$ and $\mathscr{I}^+$. On the other hand, the requirement of $\partial(\psi[M])$ to be smooth means that $\partial(\psi[M])$ must be the union of the null hypersurfaces $\mathscr{I}^+$ and $\mathscr{I}^-$, which is denoted by $\mathscr{I}$ and is referred to simply as the \textit{infinity} of an asymptotically simple spacetime. Namely, not only is the conformal boundary of $\mathcal{M}$ not smooth in $\mathcal{I}^-$, $\mathcal{I}^+$ and $\mathcal{I}^0$, but it is also uninteresting to analyze the spacetime there. In other words, one is not interested in properties of a spacetime at a finite distance from the sources at early or late times, as well as only at arbitrarily large distances. The limit of interest is precisely the one at large distances at early or late times, which is associated with the behavior of null geodesics. Finally, the last requirement excludes the possibility of regions where the gravitational interaction acts in such a manner as to “trap” null geodesics, meaning that they would not have endpoints in $\partial(\psi[M])$. However, as will be discussed in ch. \ref{chapter2}, there are spacetimes of interest that behave like Minkowski in regions far from the sources, but possess regions of strong gravitational field that result in such entrapment. Thus, it is necessary to generalize the definition of asymptotic simplicity to allow for such cases. 

A spacetime, $(M,g_{\mu\nu})$, is said to be \textit{weakly asymptotically simple} if there exists an asymptotically simple spacetime, $(M',g'_{\mu\nu})$, and a neighborhood, $S'$, of $\partial(\psi[M'])$ such that $\psi^{-1}[S']$, is isometric to an open set $S\subset M$. Hence, a weakly asymptotically simple spacetime can be converted to an asymptotically simple one by “ignoring” the regions where gravity acts in a manner as to ``trap'' null geodesics. Finally, a spacetime is defined to be \textit{asymptotically flat} if it is weakly asymptotically simple and $R_{\mu\nu}=0$ on a neighborhood of $\partial(\psi[M'])$. This last condition is merely a requirement that spacetime obeys the vacuum Einstein's equation at “infinity”, i.e., $T_{\mu\nu}=0$ on a neighborhood of $\partial(\psi[M'])$. It should be noted that, in spacetimes in which there is electric charge, this condition can be modified to allow for electromagnetic radiation near $\mathscr{I}$. In light of this definition, one can interpret asymptotically flat spacetimes as those in which the ``distant'' at early and late times are similar to those of Minkowski spacetime. Furthermore, because of the constraints on the Ricci tensor in such regions, one can justifiably affirm that such behavior is associated with the lack of an energy distribution. Finally, even if the energy distribution in spacetime ends up producing regions such that gravity ``traps'' null geodesics, one can still make a comparison by considering the regions where such effects do not happen. 

These arguments and definitions are illustrated in fig. \ref{fig:min5}. The asymptotically flat spacetime, $(M,g_{\mu\nu})$, has a region, $R$, illustrated in black, for which null geodesics become ``trapped''. Due to our definition of asymptotic flatness, there must be an asymptotically simple spacetime, $(M',g'_{\mu\nu})$, such that there exists an isometry, $\psi':M\to M'$, which maps the open set $S=M\backslash R$, into a region of $M'$.\begin{figure}[h]
\centering
\includegraphics[scale=1.1]{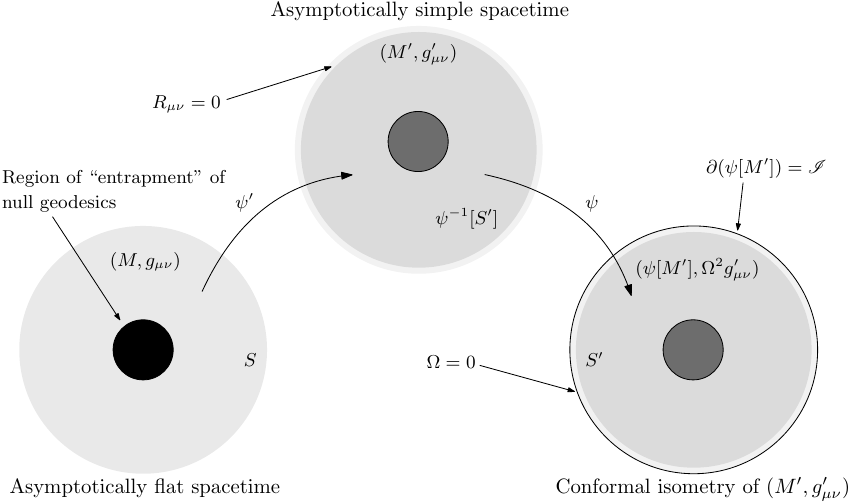} 
\caption{Concept of asymptotically flat spacetimes.}
\caption*{Source: By the author.}
\label{fig:min5}
\end{figure} Additionally, the conformal isometry of $(M',g'_{\mu\nu})$ is such that there is a neighborhood of $\mathscr{I}$, $S'$, for which the region $\psi^{-1}[S']$ is precisely the image of the map $\psi'$. Moreover, the light gray region in $\psi[M']$ is a neighborhood of $\mathscr{I}$ such that its inverse image is a region for which $R_{\mu\nu}=0$. Lastly, note that $\psi[M']$ must necessarily be bounded, as indicated by the black line (and the value of the conformal factor there), but the asymptotically simple and asymptotically flat spacetimes are not. In other words, one should view their illustrations as extending infinitely.

\chapter{Classical aspects of black holes}\label{chapter2}

All the machinery is now in place to describe black holes in the framework of general relativity. Roughly speaking, a black hole is a region of spacetime from which nothing can ``escape''. This “no escape” property can be precisely defined when one is dealing with an asymptotically flat spacetime, in the sense that the future null infinity, $\mathscr{I}^+$, can characterize a region for which observers and light rays can contemplate ``escaping'' to. In essence, for such spacetimes, the ``entrapment'' property is directly related to the behavior of the vector associated with displacements in the radial direction. That is, the distribution of energy will affect the metric in such a way as to produce a region whose geometry does not allow light rays or observers to increase their radial coordinate to a certain value, meaning that they are “trapped” in a region of spacetime. The goal of this chapter is to make these statements precise, investigate the consequences of these regions, and derive properties for the physical quantities that can be measured by observers outside of them.

To do so, we first analyze spacetime outside a spherically symmetric distribution of energy and how the metric behaves depending on the region in which the distribution is contained. In this manner, we will see that there is a limit for the radius for which the distribution can be contained and still allow observers and light rays to ``escape'' from the effects of the gravitational interaction. Although this analysis is of an idealized description of an energy distribution, it provides many important results that can be straightforwardly generalized to other distributions that also give rise to a “no escape” region. With the support of this simplified model, we will then give a precise definition of the black hole region of a spacetime based on the concept of asymptotic flatness. Furthermore, considering the results of causal structure and null geodesic congruences, it will be possible to deduce several properties for the dynamics of the black holes, which will be valid for spacetimes that possess globally hyperbolic regions. These properties will mainly follow from geometrical arguments, but they will also rely on assumptions that can be stated in the form of the so-called cosmic censor conjecture.

After discussing the black hole region and its boundary, we will study relations for the surface gravity as measured by observers at the asymptotic region, as well as derive relations for the mass, angular momentum, and area of black holes that are described by a time invariant configuration. We will also see that the importance of the black hole uniqueness theorems follows from another conjecture, namely, that at sufficiently ``late times'' after its formation, a black hole is expected to reach a time independent state, so that it will be completely characterized by three parameters as seen by observers outside of it: its mass, angular momentum, and electric charge.

\section{Schwarzschild spacetime}\label{schsec}
The analysis of spacetime geometry outside a spherically symmetric distribution of energy illustrates several properties that are useful for the study of black holes. Due to its high symmetry, it is a good starting point to analyze the properties of regions where gravity behaves in such a manner so that nothing can get out. The exterior (i.e., vacuum) geometry of such energy distributions is given by the following result, known as \textit{Birkhoff's theorem} \cite{Misner1973}.\begin{theorem}
Let the geometry of a given region of a spacetime be spherically symmetric and be a solution of the vacuum (i.e., $T_{\mu\nu}=0$) Einstein's equation. Then that geometry is locally isometric to Schwarzschild geometry.
\end{theorem}

This theorem states that spacetime outside a spherically symmetric distribution of energy must be described by the \textit{Schwarzschild metric} (i.e., it must be a piece of Schwarzschild spacetime), which in coordinates $\{t,r,\theta,\phi\}$, takes the form
\begin{equation}\label{sch}
    ds^2=-\left(1-\frac{r_s}{r}\right)c^2dt^2+\left(1-\frac{r_s}{r}\right)^{-1}dr^2+r^2(d\theta^2+\sin^2\theta d\phi^2),
\end{equation}
where $r_s$ is the Schwarzschild radius (see eq. \ref{schr}). In this sense, Birkhoff's theorem establishes the \textit{uniqueness} of Schwarzschild geometry for the exterior region of a spherically symmetric energy distribution. Now, from the fact that the Schwarzschild metric is independent of $t$, it is possible to roughly see that it is also asymptotically flat, as the components of the metric in this coordinate system behave like $g_{ab}=\eta_{ab}+\mathcal{O}(r^{-1})$ for large $r$ for all values of $t$. Details on the asymptotic flat property of the Schwarzschild spacetime following the definition given in \S\;\ref{flat} will be discussed below. Additionally, the term $M$ in $r_s$ (see eq. \ref{schr}) can be identified as the geometrical quantity associated with the mass of the energy distribution, as it is the “charge” associated with the gravitational interaction as measured by an observer at rest at the asymptotic region (i.e., an observer whose radial coordinate is arbitrarily large), where the gravitational effects can be approximated by Newtonian gravity.

Regarding the symmetries of the Schwarzschild spacetime, from the discussion of Lie differentiation in appendix \ref{derivativeoperators}, one knows that in a coordinate system where the components of a tensor are independent of a coordinate, $x^a$, the Lie derivative of such a tensor with respect to $(\partial_{x^a})^{\mu}$ vanishes. Thus, $\xi^{\mu}=(\partial_{t})^{\mu}$ is a Killing vector. In particular, the existence of a timelike Killing vector implies that the orbits of the one-parameter group of diffeomorphisms generated by it are timelike curves, which can be interpreted as the invariance of the metric under time translations. A spacetime with a timelike Killing vector in a neighborhood of infinity, $\mathscr{I}$, is said to be \textit{stationary}. In fact, Schwarzschild spacetime obeys the stronger property of being \textit{static}, as it is stationary and $\xi^{\mu}$ is hypersurface orthogonal (see theorem \ref{frobenius}). As per Frobenius' theorem, this last condition is equivalent to requiring the vanishing of the vorticity tensor of the congruence generated by $\xi^{\mu}$. If that were not the case, the orbits of $\xi^{\mu}$ would favor a direction on the spacelike hypersurfaces $\{t=\text{constant}\}$, meaning that spacetime would not be invariant under a time reflection, i.e., a transformation such as $dt\to-dt$. Hence, Birkhoff's theorem implies that the unique solution to the exterior geometry of a spherically symmetric distribution of energy is also static. This can be interpreted as stating that Einstein's equation implies that there exists no monopole gravitational radiation, in the exact same way as Maxwell's equations imply that there is no monopole electromagnetic radiation. Additionally, as the metric components in eq. \ref{sch} are also independent of $\phi$,  $\psi^{\mu}=(\partial_{\phi})^{\mu}$ is also a Killing vector. Together with the other two Killing vectors that are a linear combination of $\psi^{\mu}$ and $(\partial_{\theta})^{\mu}$, they span the group $\text{SO}(3)$, which is associated with the invariance of the metric under rotations in any angular direction, i.e., spherical symmetry. Because of this, one has that the orbits of this group of isometries in the spacetime will result in two-spheres, so that the action of the metric on these surfaces must characterize the area of a two-sphere, $A$. As such, it is important to note that the coordinate $r$, which is defined by the area of a two-sphere,
\begin{equation}\label{a32}
    r=\left(\frac{A}{4\pi}\right)^{1/2},
\end{equation}
need not represent the physical distance to the center of a two-sphere.

Even though Schwarzschild spacetime is mostly of interest to analyze the spacetime outside a spherically symmetric energy distribution, it is relevant to study its global properties. In this manner, a striking feature of the Schwarzschild metric is that it is pathological for $r=0$ and $r=r_s$. By evaluation of curvature scalars (i.e., quantities associated with curvature that are invariant) \cite{Misner1973}, say,  
\begin{equation}\label{curvature}
    R^{\mu\nu\alpha\beta}R_{\mu\nu\alpha\beta}=\frac{12r^2_s}{r^6},
\end{equation}
it becomes clear that the singular character at $r=r_s$ is merely a consequence of the coordinate system, and in fact, an observer can reach $r<r_s$ in a finite proper time \cite{Misner1973}. However, the singularity at $r=0$ is a \textit{physical} singularity. Hence, it cannot be eliminated by a coordinate transformation. This translates to the fact that it is not possible to find a coordinate system such that the metric is well behaved at $r=0$. Now, since the region $r<r_s$ is accessible to observers, it is of interest to analyze how they evolve there. Notice, however, that in the hypersurface $\{r=r_s\}$, $\xi^{\mu}$ and $\nabla^{\mu}r$ (i.e., the vector orthogonal to hypersurfaces $\{r=\text{constant}\}$) are null, and in the region $r<r_s$, $\xi^{\mu}$ is spacelike and $\nabla^{\mu}r$ is timelike. Thus, for observers that reach the hypersurface $\{r=r_s\}$, the singularity at $r=0$ is no longer a matter of “where”, but “when”. More precisely, the observer will end his existence in a finite proper time due to the incompleteness of his geodesic upon his inevitable encounter with the singularity. Furthermore, note that this behavior is not restricted to the motion of observers, that is, any light ray that reaches the region $r\leq r_s$ will not only not be able to escape to $r>r_s$, but it will also reach the singularity\footnote{Evidently, this follows if one considers that general relativity is valid in arbitrarily high curvature regime. That is, at some point in this analysis the Planck scale will be accessible, and as discussed in ch. \ref{Introduction}, the predictions of general relativity may not hold accurately.} in a finite affine parameter if it goes into $r<r_s$. Hence, it follows that the hypersurface $\{r=r_s\}$ is precisely the delimiter of the region of the Schwarzschild spacetime from which observers and light rays cannot escape.

In order to use the Schwarzschild metric in coordinates $\{t,r,\theta,\phi\}$ to describe the entire spacetime that possesses the region $r<r_s$, it is necessary to remove the hypersurface $\{r=r_s\}$ due to its singular nature. Such a process produces two disjoint open regions, and the requirement for spacetime to be connected will restrict one's analysis to only one of them. Nevertheless, it is possible to perform coordinate transformations to find a metric that is not pathological for $r=r_s$, and then extend its domain to the entire spacetime \cite{Hawking1973}. In particular, the conformal diagram of the Schwarzschild spacetime can then be developed by performing a conformal transformation of the metric in this appropriate coordinate system. This can be done by considering the following transformations and line of reasoning.

First, note that radial null geodesics of the Schwarzschild metric (i.e., $ds^2=0$ for constant $d\phi=d\theta=0$) must respect the relation
\begin{equation}\label{eq33}
    cdt=\left(1-\frac{r_s}{r}\right)^{-1}dr.
\end{equation}
It is then useful to define the \textit{tortoise} coordinate, $r'$, by
\begin{equation}
    dr'=\left(1-\frac{r_s}{r}\right)^{-1}dr,
\end{equation}
so that radial null geodesics obey $cdt=dr'$. The coordinate transformation of interest is one that removes the singular behavior at $r=r_s$ and uses coordinates that are representative of the families of radial null geodesics.

Consider then the transformation
\begin{equation}\label{eq199}
    w=\exp{\left(\frac{ct+r'}{2r_s}\right)},\;u=-\exp{\left(\frac{r'-ct}{2r_s}\right)},
\end{equation}
in which the Schwarzschild metric takes the form
\begin{equation}\label{m3}
    ds^2=-\frac{4r_s^3}{r}\exp{\left(-\frac{r}{r_s}\right)}dwdu+r^2(d\theta^2+\sin^2\theta d\phi^2),
\end{equation}
where $r=r(w,u)$ is given implicitly by 
\begin{equation}
    wu=-\left(\frac{r}{r_s}-1\right)\exp{\left(\frac{r}{r_s}\right)}. 
\end{equation}
The coordinates $(w,u)$ are said to be null, for the same reasons as discussed in the Minkowski case (see \S~\ref{flat}). Note that the Schwarzschild metric in coordinates $\{w,u,\theta,\phi\}$ is not pathological for $r=r_s$, thus, one may \textit{analytically extend} it \cite{Townsend1997} to describe the entire region $0<r<\infty$ by allowing $\{w,u\}$ to take any value in the interval $(-\infty,\infty)$ compatible with $r>0$. The \textit{analytically extended Schwarzschild spacetime} is then described by the following intervals in these coordinates. The region $r>r_s$ corresponds to $w>0,\;u<0$, the hypersurface $\{r=r_s\}$ corresponds to $w=0$ or $u=0$, and the region $r<r_s$ corresponds to $w<0,\;u>0$. Similarly to the Minkowski case, in order to depict the entire spacetime in a finite illustration, it is necessary to bring infinity to a finite range of coordinates. Consider then another coordinate transformation,
\begin{equation}
    w'=\arctan{w},\;
    u'=\arctan{u},
\end{equation}
in which the Schwarzschild metric takes the form
\begin{equation}\label{m4}
    ds^2=-\frac{4r_s^3}{r\cos^2{w'}\cos^2{u'}}\exp{\left(-\frac{r}{r_s}\right)}dw'du'+r^2(d\theta^2+\sin^2\theta d\phi^2).
\end{equation}
Hence, the entire range of coordinates $\{t,r\}$ is depicted by coordinates $\{w',u'\}$ in the interval $(-\pi/2,\pi/2)$. Clearly, the metric given in eq. \ref{m4} presents pathologies in these intervals, but one can reach a well behaved form by performing a conformal transformation with conformal factor
\begin{equation}\label{13}
   \Omega=\cos{w'}\cos{u'},  
\end{equation}
which yields the unphysical metric
\begin{equation}
    (ds')^2=-\frac{4r_s^3}{r}\exp{\left(-\frac{r}{r_s}\right)}dw'du'+r^2\cos^2{w'}\cos^2{u'}(d\theta^2+\sin^2\theta d\phi^2).
\end{equation}

This conformal compactification of the analytically extended Schwarzschild spacetime in coordinates $\{w',u'\}$ is illustrated in fig. \ref{fig:sch1}, where the angular coordinates have been suppressed.\begin{figure}[h]
\centering
\includegraphics[scale=1]{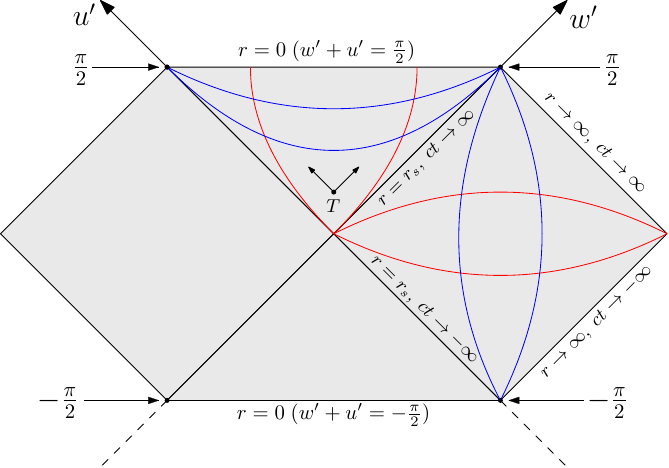} 
\caption{Conformal compactification of the analytically extended Schwarzschild spacetime on the plane $w'\times u'$.}
\caption*{Source: By the author.}
\label{fig:sch1}
\end{figure} Thus, each point represents a two-sphere of radius $r(w',u')$, except for some parts of its boundary. The correspondence between the coordinates and the regions discussed so far is as follows. The hypersurface $\{r=r_s\}$ corresponds to $w'=0$ or $u'=0$, with different values of time coordinate, $t$, values depending on the axis. The singularity at $r=0$ corresponds to $w'+u'=\pm\pi/{2}$. The line $w'=\pi/{2}$ for $-\pi/{2}<u'<0$ corresponds to the region $r\to\infty,\;t\to\infty$, as the line $u'=-\pi/{2}$ for $0<w'<\pi/{2}$ corresponds to the region $r\to\infty,\;t\to-\infty$. In addition, red lines correspond to the hypersurfaces $\{t=\text{constant}\}$, blue lines correspond to the hypersurfaces $\{r=\text{constant}\}$ and $45\degree$ lines correspond to radial null geodesics. The change in shape of the lines of the pertinent hypersurfaces in the region $r<r_s$ is due to the change in sign in the norm of $\xi^{\mu}$ and $\nabla^{\mu}r$. Finally, the two arrows in the region $r<r_s$ are representative of the incoming and outgoing families of future directed null geodesics orthogonal to the spacelike surface, $T$.

The conformal diagram of the analytically extended Schwarzschild spacetime is illustrated in fig. \ref{fig:sch2}. The interpretation of each one of the presented regions follows from fig. \ref{fig:sch1}. \begin{figure}[h]
\centering
\includegraphics[scale=1]{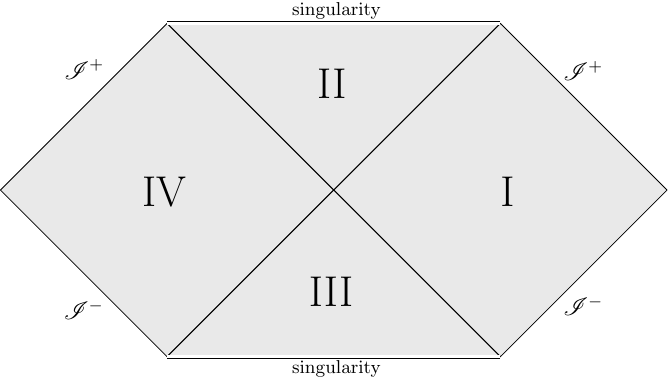} 
\caption{Conformal diagram of the analytically extended Schwarzschild spacetime.}
\caption*{Source: By the author.}
\label{fig:sch2}
\end{figure}Region I is the region $r>r_s$, and will be the representation of the spacetime in the vacuum region of any spherically symmetrical distribution of energy. Its asymptotic behavior is depicted from the limits of the conformal compactification in fig. \ref{fig:sch1}, and identified in fig. \ref{fig:sch2} as the null hypersurfaces $\mathscr{I}^-$ and $\mathscr{I}^+$. With this, one can conclude that Schwarzschild spacetime is asymptotically flat, as from the form of the conformal factor given by eq. \ref{13}, it is possible to see that it obeys all the properties of the definition given in \S\;\ref{flat}. Additionally, region II corresponds to $r<r_s$, and the “no escape” property can be clearly identified, as any observer must decrease its radial coordinate and any light ray (as represented by the two arrows pointing out of the surface $T$ in fig. \ref{fig:sch1}) will also eventually reach the singularity, represented by a thick white line. Since any light ray emitted in region II must decrease its radial coordinate, the area of its wavefront must decrease (see eq. \ref{a32}). Namely, the identification of the coordinates $\{w',u'\}$ as the outgoing and incoming families of future directed null geodesics leads one to the conclusion that any two-sphere in region II is a trapped surface. Explicit verification that the expansion of the congruence of null geodesics generated by $(\partial_{w'})^{\mu}$ and $(\partial_{u'})^{\mu}$ is everywhere negative in the region $r<r_s$ can be found in \cite{Poisson2004}. In contrast, region III has the exact opposite properties of region II, as any light ray in it must have come from the singularity and will eventually leave region III. Finally, region IV represents another asymptotically flat region that is causally disconnected from region I, possessing its own set of null infinities. 

Although the conformal diagram of the analytically extended spacetime gives a better picture for the analysis of the Schwarzschild spacetime, the extension of the coordinates produces regions whose physical significance is questionable. First, note that region II will be a product of a spherically symmetric system if the energy distribution is such that it is contained within a radius $r< r_s$. One can study the plausibility of distributions of energy to obey such a scenario by considering, for example, the interior of stars. A star can be approximately described as a self-gravitating sphere of hydrogen supported by thermal and radiation pressure as a result of the process of nuclear fusion at the core. As the fuel for the fusion depletes in the later stages of the life of a star, the temperature will decrease and so will the pressure. However, the final state of the star with $T\to0$ will not result in $P\to 0$ because of the degeneracy pressure, as a consequence of the Pauli exclusion principle \cite{Sakurai1994}. Nonetheless, the decrease in pressure will result in a contraction that can produce a configuration where the mass is contained in a radius $r\leq r_s$. If such a configuration were to be produced, then a trapped surface must be formed around the distribution, and following theorem \ref{HawPen}, the spacetime will be null geodesically incomplete. It has been shown that these conclusions hold even if there are small deviations of spherical symmetry \cite{Hawking1973}. Even without knowledge of theorem \ref{HawPen}, one can still identify the pathological behavior of gravitational collapse, as the formation of region II will occur if the star is massive enough\footnote{Indeed, the formation of region II can occur in an arbitrarily low curvature regime (see eq. \ref{curvature}).}. The critical mass for which the degeneracy pressure will not be enough to sustain the effects of the gravitational interaction, resulting in the star being contained in an arbitrarily small radius, is given by $M_c\simeq 1.4\; M_{\odot}$, known as the \textit{Chandrasekhar limit}. Details on the derivation of this limit can be found in, e.g., \cite{Townsend1997}.

Consequently, at the late stages of the life of a star with a mass greater than $M_c$ and with small deviations of spherical symmetry, the gravitational effects will not be sustained, resulting in a configuration where region II will form, and a region where the ``no escape'' property will be an undeniable product of the collapse. The collapse of such a distribution of energy can be illustrated in the conformal diagram of the Schwarzschild spacetime, as shown in fig. \ref{fig:sch3}. The exterior region is depicted in light gray, which, from Birkhoff's theorem, must be a piece of the conformal diagram of the analytically extended Schwarzschild spacetime. The interior of the distribution is depicted in gray, where the geometry depends on the details of the energy-momentum tensor, which has to be spherically symmetric. In such an illustration,\begin{figure}[h]
  \begin{subfigure}[b]{0.5\textwidth}
  \centering
    \includegraphics[scale=1.1]{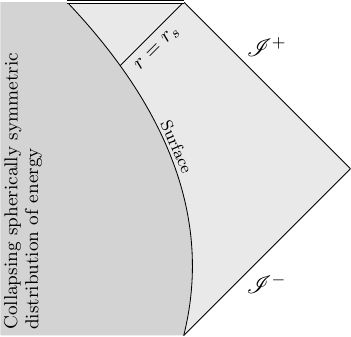}
    \caption{Collapse of a spherically symmetric distribution of energy.}
    \label{fig:sch3}
  \end{subfigure}
  \hfill
  \begin{subfigure}[b]{0.5\textwidth}
  \centering
    \includegraphics[scale=1.2]{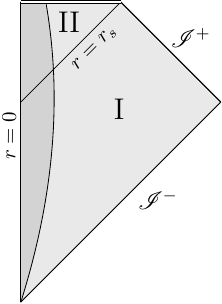}
    \caption{Conformal diagram restricted to the region of interest.}
    \label{fig:sch4}
  \end{subfigure}
  \caption{Collapse of a spherically symmetric distribution of energy in the conformal diagram of the extended Schwarzschild spacetime.}
  \caption*{Source: By the author.} 
\end{figure} the gray area covers regions III and IV, and at some point its surface crosses the Schwarzschild radius, which represents the point of no return when it comes to the collapse, i.e., all the distribution and signals emitted from it will unavoidably reach the singularity. Note that, in the context of gravitational collapse, regions III and IV never form, as a consequence of the fact that the Schwarzschild metric is only a solution to the exterior region. Fig. \ref{fig:sch4} is the conformal diagram of interest, illustrating all the important properties of regions I and II, as well as the collapsing spherically symmetric energy distribution in a finite drawing.

The collapse of a spherically symmetric energy distribution can also be illustrated in a spacetime diagram, fig. \ref{fig:sch5}. In such representation, only the incoming null geodesics are at $45\degree$ \footnote{This can be done by using the \textit{incoming Eddington–Finkelstein} coordinates (see \cite{Hawking1973}).}, and consequently, light cones do not necessarily form $45\degree$. Now that only one space dimension is suppressed, the structure of the collapse becomes even more obvious. The event $a$ indicates the formation of the region with the ``no escape'' property, since no event in its causal future has a radial coordinate greater than $r_s$. Although details about the light cones inside the dark region are dependent on the details of the energy distribution (which evidently dictates how outgoing null geodesics evolve from $a$ to the surface of the distribution), after the outgoing null geodesics exit the distribution, they will either reach the singularity or remain with radial coordinate $r=r_s$. Indeed, this translates to the conclusion that any two-sphere in the region $r<r_s$ will be a trapped surface. As indicated, the surface $T$ is a trapped surface, and one can picture the two arrows coming out of it as the incoming and outgoing null families of geodesics, similar to the ones depicted in fig. \ref{fig:sch1}. 

 \begin{figure}[h]
\centering
\includegraphics[scale=1]{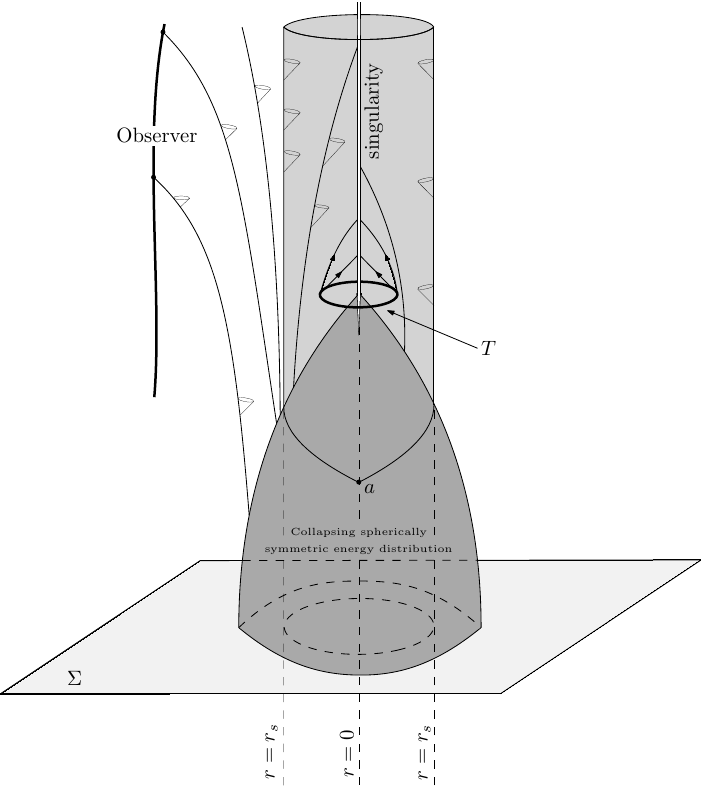} 
\caption{Spacetime diagram of a spherically symmetric gravitational collapse.}
\caption*{Source: Adapted from PENROSE \cite{Penrose1965}.}
\label{fig:sch5}
\end{figure}

Furthermore, the behavior of the outgoing null geodesics that exit the energy distribution just before it reaches $r=r_s$ indicates how observers on the outside perceive the collapse, which can be deduced by the following line of reasoning. Consider an observer with radial coordinate $r>r_s$ whose world line coincides with the integral curves of the timelike Killing vector, $\xi^{\mu}$. Its normalized tangent reads
\begin{equation}\label{ell1}
    \ell^{\mu}=\frac{\xi^{\mu}}{V(r)},
\end{equation}
where 
\begin{equation}\label{reds}
    V(r)=(-c^{-2}\xi^{\mu}\xi_{\mu}(r))^{1/2}=(-c^{-2}g_{tt}(r))^{1/2}.
\end{equation} 
In order words, normalization of $\ell^{\mu}$ translates to
\begin{equation}
    \frac{d\tau}{dt}=V(r),
\end{equation}
where $\tau$ denotes the observer's proper time. Consequently, the factor $V$ is known as \textit{redshift factor}, as it relates the proper time of an observer following the orbit of $\xi^{\mu}$ at $r>r_s$ to the proper time of an observer following the orbit of $\xi^{\mu}$ at the asymptotic region (i.e., one with $r\to\infty$). Namely, a light ray emitted at $r>r_s$ will reach the observer in the asymptotic region, whose proper time coincides with the coordinate time, $t$, redshifted by a factor of $V$. Hence, one can relate the passage of time for an observer following the orbit of $\xi^{\mu}$ at $r$ and $r'$ by
\begin{equation}
    \frac{d\tau}{d\tau'}=\frac{V(r)}{V(r')}.
\end{equation}

Evidently, the behavior of the redshift factor as $r\to r_s$ indicates that the gravitational time dilation diverges (see eq. \ref{sch}), which means that for any observer\footnote{This conclusion generalizes to observers that are not following the integral curves of $\xi^{\mu}$ but remain outside the black hole, as the variation of their spatial coordinates would contribute with a finite time dilation factor.} with radial coordinate $r>r_s$, it takes an infinite amount of time for any energy distribution to reach $r=r_s$. In essence, an observer at $r>r_s$ would perceive the energy distribution apparently slow down as it gets closer to $r=r_s$, and the light coming from it would get arbitrarily redshifted. As commented earlier, observers can reach $r=r_s$ in a finite amount of proper time, and the conclusion that for observers at $r>r_s$ this takes an infinite amount of time is merely another consequence of the conclusion that the notion of simultaneity is observer dependent.

\section{Black holes}\label{bhsec}
In the last section, the analysis of the Schwarzschild spacetime concluded that spherically symmetric distributions of energy contained in a radius $r\leq r_s$ would result in a region such that no observer or light ray that enters it could reach $r>r_s$. This ``no escape'' property can be precisely defined for any asymptotically flat spacetime\footnote{It is also possible to give a satisfying notion of a black hole when one is dealing with a spacetime which is not asymptotically flat but has regions that can be identified as ``infinity'' \cite{Wald1984}.}, if one considers that the causal past of the future null infinity does not contain the entire spacetime, as a consequence of such regions. In other words, there are events in such spacetimes that are not causally connected to $\mathscr{I}^+$. However, in order to evaluate which events are causally connected to the ``infinitely distant'' regions, it is necessary to have knowledge of the development of the entire spacetime. Thus, it is convenient to restrict one's analysis to asymptotic flat spacetimes that are causally and deterministic ``well behaved'', or at least, posses regions that are.

More precisely, consider an asymptotically flat spacetime, $(M,g_{\mu\nu})$, with associated conformal isometry\footnote{The ``middle'' asymptotically simple spacetime (see fig. \ref{fig:min5}) is not mentioned for simplicity, its existence is implied by our definition of asymptotic flatness (see \S\;\ref{flat}). In particular, one may picture the ``associated conformal isometry'', $\psi$ as the composition of the maps $\psi$ and $\psi'$ illustrated in fig. \ref{fig:min5}.}, $\psi$, to $(M',g'_{\mu\nu})$. The spacetime $(M,g_{\mu\nu})$ is said to be \textit{strongly asymptotically predictable} if there exists an open set, $A'\subset M'$, with $\overline{\psi[M]\cap J^-(\mathscr{I^+})}\subset A'$ such that $(A',g'_{\mu\nu})$ is globally hyperbolic. That is, the region $\psi^{-1}[A']$ is a globally hyperbolic region of $(M,g_{\mu\nu})$. The \textit{black hole region} of a strongly asymptotically predictable spacetime, $(M,g_{\mu\nu})$, is defined to be 
\begin{equation}\label{blackhole}
    B=M\backslash \psi^{-1}[J^-(\mathscr{I}^+)].
\end{equation}
This definition can be interpreted as stating that the development of the spacetime on the boundary of the black hole region and outside of it does not depend on its interior. In light of this, an asymptotically flat spacetime which fails to be strongly asymptotically predictable is said to possess a \textit{naked singularity} (see \cite{Penrose1973} for a detailed discussion). 

Note that the singularity in the gravitational collapse of a spherically symmetric energy distribution is not in $J^-(\mathscr{I^+})$ (see the conformal diagram, fig. \ref{fig:sch4}), as a consequence of the black hole region ``clothing'' it. A naked singularity, on the other hand, can be causally connected to $\mathscr{I^+}$. For example, the singularity in region III of the conformal diagram of the analytically extended Schwarzschild spacetime (see fig \ref{fig:sch2}) is in $J^-(\mathscr{I^+})$, which is unphysical in the context of gravitational collapse (in fact, because it has the exact opposite properties of region II, region III is referred to as a \textit{white hole}). Many considerations over the past decades have been made regarding the general physical plausibility of naked singularities, and the lack of successful counterexamples (see, e.g., \cite{Wald1997a}) has led to a conjecture, in which one of its formulations takes the form stated below (usually referred to as the \textit{weak version}).

\begin{cosmic*}

All physical spacetimes are globally hyperbolic.
 
\end{cosmic*}

Since a spacetime possessing an arbitrary naked singularity cannot be globally hyperbolic, the above formulation of the cosmic censor conjecture can be interpreted as stating that naked singularities, except for a possible initial cosmological one (which under certain conditions, is also predicted by results given in \cite{Hawking1970}), cannot be present in any physically reasonable spacetime, since they can be regarded as a way to violate predictability of the spacetime (unless, of course, one has the knowledge of how to describe spacetime singularities and impose adequate boundary conditions). In this sense, the cosmic censor conjecture can be understood as the requirement that gravitational collapse always results in a predictable black hole\footnote{Note that since the definition of geodesic completeness relies on the affine parameter range and such a concept is not necessarily conformally invariant (see \S~\ref{flat}), we will also require geodesic completeness in the physical spacetime.}. Although the cosmic censor conjecture is believed to be an excellent assumption in an extensive way mainly due to indirect evidence, proof of it is an open problem in general relativity. Lastly, more precise formulations can be made regarding the predictability from data on Cauchy hypersurfaces, which can be found in \cite{Wald1984}. The reader can also find a more comprehensive and philosophical discussion in \cite{Landsman2021}.

By restricting the analysis of black holes to strongly asymptotically predictable spacetimes, the interpretation of a Cauchy hypersurface as an “instant of time” can be used to define a black hole at a given time. Let $\psi^{-1}[A']\subset M$ be the globally hyperbolic region of a strongly asymptotically predictable spacetime, $(M,g_{\mu\nu})$, and let $\Sigma$ denote a Cauchy hypersurface. The \textit{black hole region at a time} $\Sigma$ is defined to be $B\cap\Sigma$, and each connected component of $B\cap\Sigma$ is a \textit{black hole at a time} $\Sigma$. The following theorem gives a property of the evolution of black holes over Cauchy hypersurfaces.

\begin{theorem}\label{45678}
    Let $(M,g_{\mu\nu})$ be a strongly asymptotically predictable spacetime and let $\Sigma_2$ and $\Sigma_1$ be Cauchy hypersurfaces for $\psi^{-1}[A']\subset M$, with $\Sigma_2\subset I^+(\Sigma_1)$. Let $\mathscr{B}_1$ be a nonempty connected component of $B\cap\Sigma_1$. Then $J^+(\mathscr{B}_1)\cap \Sigma_2\neq \emptyset$ and is contained in a single connected component of $B\cap\Sigma_2$.
\end{theorem}
\begin{proof}
$J^+(\mathscr{B}_1)\cap \Sigma_2\neq \emptyset$ follows from the condition that $\Sigma_1$ and $\Sigma_2$ are Cauchy hypersurfaces and $I^+(\Sigma_1)\supset\Sigma_2$. From the definition of the black hole region, one has that $J^+(\mathscr{B}_1)\subset B$, which means that $J^+(\mathscr{B}_1)\cap \Sigma_2$ is contained in $B\cap\Sigma_2$. Now, if $J^+(\mathscr{B}_1)\cap \Sigma_2$ were not connected, then it would be possible to find disjoint open sets, $S$ and $S'$ contained in $\Sigma_2$ such that $S\cap J^+(\mathscr{B}_1)\neq\emptyset$, $S'\cap J^+(\mathscr{B}_1)\neq\emptyset$ and $S\cup S'= J^+(\mathscr{B}_1)\cap \Sigma_2$. Then, one would have that $\mathscr{B}_1\cap I^-(S)\neq\emptyset$, $\mathscr{B}_1\cap I^-(S')\neq\emptyset$ and $\mathscr{B}_1\subset I^-(S)\cap I^-(S')$. However, no point $a\in\mathscr{B}_1$ could lie in both $I^-(S)$ and $I^-(S')$, as it would be possible to divide timelike vectors at $a$ into two nonempty disjoint sets, contradicting the connectedness of the set of future directed vectors at $a$ (see theorem \ref{02} and proposition \ref{lightcone}). Hence, it would be possible to write $\mathscr{B}_1$ as the union of the disjoint open sets $I^-(S)\cap\Sigma_1$ and $I^-(S')\cap\Sigma_1$, contradicting its connectedness. Thus, $J^+(\mathscr{B}_1)\cap \Sigma_2$ is connected.
\end{proof}

This result states that a black hole cannot disappear from the strongly asymptotically predictable spacetime, i.e., $J^+(\mathscr{B}_1)\cap \Sigma_2\neq \emptyset$, and it cannot bifurcate, as $J^+(\mathscr{B}_1)\cap \Sigma_2$ must be contained in a single connected component of the black hole region in $\Sigma_2$. It should be noted that it has no dependence on Einstein's equation or conditions respected by the energy distribution in spacetime, as it is a consequence of the properties of Cauchy hypersurfaces and that $B\cap \psi^{-1}[J^-(\mathscr{I}^+)]=\emptyset$. The following result relates the existence of trapped surfaces with the black hole region, proof of which can be found in \cite{Hawking1973}.

\begin{proposition}\label{trapped}
    Let $(M,g_{\mu\nu})$ be a strongly asymptotically predictable spacetime satisfying $R_{\mu\nu}\ell^{\mu}\ell^{\nu}\geq 0$ for all null $\ell^{\mu}$. Suppose $M$ contains a trapped surface, $T$. Then $T\subset B$.
\end{proposition}

Thus, in strongly asymptotically predictable spacetimes where Einstein's equation and the weak or strong energy condition hold, a trapped surface implies that the spacetime possesses a non-empty black hole region. However, the converse is not true, a black hole region does not imply the existence of a trapped surface. Furthermore, from theorem \ref{HawPen}, black holes can then be associated with singularities under certain conditions, which will, of course, not be naked. Lastly, note that although the definition of the black hole region requires knowledge of the entire development of spacetime, one can use local properties (in time and space) to locate a connected component of the black hole region if one can detect trapped surfaces. 

The definition of a variation of the concept of a trapped surface will also be useful for the discussion of further properties of black holes. Let $(M,g_{\mu\nu})$ be a strongly asymptotically predictable spacetime. An \textit{outer trapped surface} is a compact spacelike two-dimensional submanifold of $M$ such that the expansion of the orthogonal outgoing future directed null geodesic congruence to it is everywhere nonpositive, i.e., $\theta\leq0$. That is, an outer trapped surface differs from a trapped surface simply because the expansion of the ``outgoing'' future directed null congruence may vanish, instead of being manifestly negative. With this definition, the \textit{trapped region}, $\mathcal{T}$, at a time $\Sigma$, is the set of all points $a\in\Sigma$ such that there is an outer trapped surface in $\Sigma$ through $a$. The \textit{apparent horizon}, $\mathscr{A}$, is the boundary of the trapped region, $\mathscr{A}=\partial \mathcal{T}$. The outer (inner) boundary $\partial \mathcal{T}$ is called the outer (inner) apparent horizon. The following result relates the existence of these sets to the black hole region at a time $\Sigma$, proof of which can be found in \cite{Hawking1973}. 

\begin{proposition}
    Let $(M,g_{\mu\nu})$ be a strongly asymptotically predictable spacetime satisfying $R_{\mu\nu}\ell^{\mu}\ell^{\nu}\geq 0$ for all null $\ell^{\mu}$. Let $\Sigma$ be a Cauchy hypersurface for the globally hyperbolic region $\psi^{-1}[A']\subset M$ and $\mathcal{T}\subset\Sigma$ be a trapped region. Then $\mathcal{T}\subset\mathscr{B}_{\Sigma}$ and the outer apparent horizon is a surface such that the expansion of its outgoing orthogonal future directed null congruence vanishes.
\end{proposition}

Under the same assumptions of prop. \ref{trapped}, this proposition shows that the trapped region, and thus, the apparent horizon, must lie inside the black hole region. These properties of the trapped region and the apparent horizon can be seen in the conformal diagram of the analytically extended Schwarzschild spacetime, fig. \ref{fig:sch2}. Evidently, in such case, the trapped region coincides with the black hole region, and the outer apparent horizon coincides with the boundary of $\mathscr{B}_{\Sigma}$, which is precisely the intersection of the hypersurface $\{r=r_s\}$ with the Cauchy hypersurface, $\Sigma$. 

\section{Event horizon}
Let $(M,g_{\mu\nu})$ be a strongly asymptotically predictable spacetime. The \textit{event horizon}, $H$, of $(M,g_{\mu\nu})$ is the boundary of the black hole region (see eq. \ref{blackhole}) in the physical spacetime, $H=\partial(\psi^{-1}[J^-(\mathscr{I^+})])$. Since the event horizon does not contain $\mathscr{I}^+$, it is made only of the null geodesics (see proposition \ref{nulll}) in $\partial(J^-(\mathscr{I^+}))\backslash{\overline{\mathscr{I}^+}}$, and thus, $H$ is a null hypersurface. In fact, one can verify that the event horizon is a part of the black hole region by the following line of reasoning. Let $a,a'\in\mathscr{I^+}$, $a'\in J^+(a)$ (with the causal past evaluated at the unphysical spacetime) and another event, $a''$, such that $a''\in \psi[M]\cap J^-(a)$. Then, following the same argumentation as the one presented in the proof of proposition \ref{prop2345}, one can conclude that $a''\in I^-(a')$. Since one may apply this logic to any two events located in arbitrarily large parameter of the null geodesic generators of $\mathscr{I^+}$, one then concludes that $\psi^{-1}[J^-(\mathscr{I^+})]=\psi^{-1}[I^-(\mathscr{I^+})]$. Hence, the black hole region is closed in the physical spacetime, i.e., $H\subset B$. In particular, this means that one may define the black hole region, equivalently, as $B=M\backslash \psi^{-1}[I^-(\mathscr{I}^+)]$. 

Additionally, by theorem \ref{the3}, the event horizon must be generated by future inextendible null geodesics contained entirely in it, as no generator of $H$ can have a future endpoint on $\mathscr{I^+}$. Consequently, the generators of $H$ cannot develop caustics, since that would mean that the generator has left $\partial(\psi^{-1}[J^-(\mathscr{I^+})])$. Nonetheless, a generator may enter $H$ at a caustic (which can happen if the black hole absorbs matter or radiation, or when it is formed, as exemplified by the event $a$ in fig. \ref{fig:sch5}), but once it enters, it can never leave, since it would be in contradiction with theorem \ref{the3}. Lastly, note that $\psi[H]\subset\overline{\psi[M]\cap J^-(\mathscr{I^+})}\subset A'$, where $A'$ is the globally hyperbolic region of the strongly asymptotically predictable spacetime. As stated above, this condition does not exclude the possibility of, for example, the null geodesic generators of $H$ in the physical spacetime to be future incomplete. In this sense, one should regard this as an additional assumption for the physical spacetime.

Similarly, as the intersection of the black hole region with a Cauchy hypersurface was used to define a black hole at a given time, one can define the event horizon at a time, $\Sigma$, to be $ H\cap \Sigma$, which is a spacelike two-dimensional submanifold of $\Sigma$. Similarly, each connected component of $ H\cap \Sigma$ is the event horizon of a black hole at a time $\Sigma$. The next theorem gives an important result regarding the evolution of the event horizon over Cauchy hypersurfaces, originally derived in \cite{Hawking1971}\footnote{A more general version of this theorem can be found in \cite{Chruściel2001}.}.

\begin{theorem}\label{area}
    \hypertarget{Area}{}Let $(M,g_{\mu\nu})$ be a strongly asymptotically predictable spacetime satisfying $R_{\mu\nu}\ell^{\mu}\ell^{\nu}\geq0$ for all $\ell^{\mu}$ null. Let $\Sigma_1$ and $\Sigma_2$ be Cauchy hypersurfaces for the globally hyperbolic region $\psi^{-1}[A']\subset M$ with $\Sigma_2\subset I^+(\Sigma_1)$ and let $\mathscr{H}_1=H\cap\Sigma_1$ and $\mathscr{H}_2=H\cap\Sigma_2$. Then the area of $\mathscr{H}_2$ is greater or equal to the area of $\mathscr{H}_1$.
\end{theorem}
\begin{proof}
Since through each point in $\mathscr{H}_1$ passes precisely one null geodesic generator of $H$, one can construct a map $f:\mathscr{H}_1\to\mathscr{H}_2$ with $f[\mathscr{H}_1]\subseteq\mathscr{H}_2$ by following the generators from $\Sigma_1$ to $\Sigma_2$. As $I^+(\Sigma_1)\supset\Sigma_2$, the variation of the area mapped by following the generators from $\Sigma_1$ to $\Sigma_2$ is given by the expansion of the generators of $H$. Hence, it suffices to show that the expansion of the null geodesic generators of the event horizon is non-negative. Suppose $\theta<0$ at $a\in\mathscr{H}_1$, such that through it passes the null geodesic $\gamma$. Assuming that the null geodesics in $H$ are complete, it follows from proposition \ref{p2} that within a finite affine parameter there will be a point $a'\in\gamma$ conjugate to $a$, which means that points in $\gamma$ beyond $a'$ are timelike related to $a$. Thus, $\gamma$ must have left $H$ at $a'$, contradicting the result that the null generators of the horizon have to be contained entirely in it, as per theorem \ref{the3}. Hence, $\theta\geq0$ everywhere on $H$.
\end{proof}

Note that this theorem relies on the assumption that the null geodesic generators of the event horizon are complete, which will be the case if one considers that singularities do not develop on $H$. Furthermore, it also requires that $R_{\mu\nu}\ell^{\mu}\ell^{\nu}\geq 0$ for all null $\ell^{\mu}$, which will be the case if Einstein's equation and the strong or weak energy condition hold. In particular, the development of caustics on $H$ can only happen if the area of the event horizon at a given time increases, i.e., a generator entering $H$ rather than leaving it. Finally, this result concerns the evolution of $\mathscr{H}_{\Sigma}$, which is the event horizon of the black hole region at a time $\Sigma$. Because this is a global property (in terms of a spatial section) of the black hole region, one can ask if for a spacetime containing multiple black holes at a time $\Sigma$ one of them could decrease the area of its event horizon while being “compensated” by a greater increase in the area of others. In other words, one can ask if the results of the theorem hold locally. It has been shown that this is indeed the case \cite{Giulini2003}, as if a connected component of $\mathscr{H}_2$ has a smaller area than any connected component of $\mathscr{H}_1$, then it must have formed at a time $\Sigma$ such that $\Sigma_1<\Sigma<\Sigma_2$, since theorem \ref{45678} states that black holes cannot bifurcate.

Fig. \ref{fig:haw1} illustrates the dynamics of the event horizon over Cauchy hypersurfaces. Namely, at a time $\Sigma$, there exists a black hole whose event horizon area is given by $S_1(\Sigma)$, and at the event $a\in\Sigma$, another black hole forms. Following, at a time $\Sigma'$, the black hole formed at a time $\Sigma$ and the other one merge, and at a time $\Sigma''$, there are again two black holes. Now, if the area of the event horizon of the smaller black hole at $\Sigma''$, given by $S_3(\Sigma'')$, is smaller than both $S_1(\Sigma')$ and $S_2(\Sigma')$, then it must have formed at a time between $\Sigma'$ and $\Sigma''$.

 \begin{figure}[h]
\centering
\includegraphics[scale=1.2]{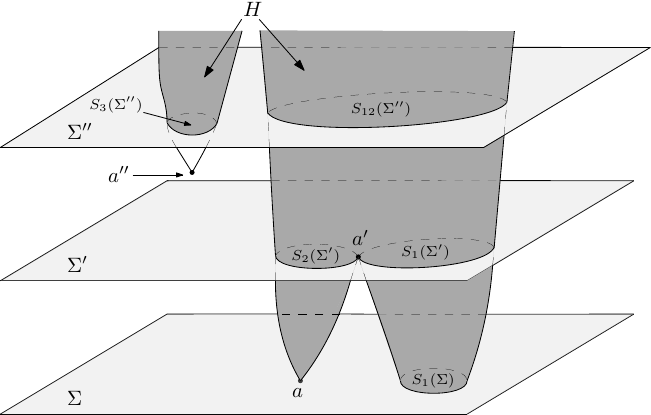} 
\caption{Spacetime diagram illustrating the dynamics of connected components of $\mathscr{H}$.}
\caption*{Source: By the author.}
\label{fig:haw1}
\end{figure}

To conclude this section, a property of the event horizon of stationary spacetimes is presented. To state the pertinent result, it is necessary to first define the notion of a Killing horizon. A null hypersurface whose normal vector is a Killing vector is said to be a \textit{Killing horizon}. The next result, derived originally in \cite{Hawking1972a}, states that the event horizon of a stationary black hole must be a Killing horizon, and is known as the \textit{strong rigidity theorem} (see also \cite{Wald1994} for a convenient statement). 

\begin{theorem}\label{theoremkilling} 
    Let $(M,g_{\mu\nu})$ be a strongly asymptotically predictable stationary spacetime which obeys Einstein's equation such that the metric and the matter fields are analytic. Then $H$ is a Killing horizon.
\end{theorem}

Theorem \ref{theoremkilling} is merely a consequence of the fact that, in stationary spacetimes, the expansion of the null geodesic generators of the event horizon must vanish identically, so that the area of $\mathscr{H}$ is invariant under time translations. More precisely, the Killing vector associated with the time translation symmetry, $\xi^{\mu}$ (see \S~\ref{schsec}), must be tangent to $H$, which means that it must be spacelike or null. If that were not the case, $H$ would not be mapped into itself along the orbits of $\xi^{\mu}$, thus contradicting the stationary property of spacetime. For instance, $\xi^{\mu}$ is the normal to the event horizon of a Schwarzschild black hole (i.e., a black hole described by the Schwarzschild metric). Consequently, since the expansion of the null vector normal to $H$ vanishes, the event horizon and outer apparent horizon of a stationary black hole coincide. 

It should be noted that the Killing horizon characteristic of the event horizon of a stationary black hole can also be deduced without invoking Einstein's equation, relying on purely geometrical arguments, defined as follows. An \textit{axisymmetric} spacetime possesses a spacelike Killing vector, $\psi^{\mu}$, associated with rotations at infinity, that is, its orbits are closed curves of length $2\pi$ in a neighborhood of $\mathscr{I}$. A spacetime is said to be \textit{stationary and axisymmetric} if the Killing vectors associated with these isometries commute. A stationary and axisymmetric spacetime is said to possess the $t$-$\phi$ orthogonality property if the two planes spanned by $\xi^{\mu}$ and $\psi^{\mu}$ are orthogonal to a family of surfaces. More details on these properties can be found in \cite{Wald1984}.

In particular, it can be shown \cite{Heusler1996} that in a static or stationary spacetime with the $t$-$\phi$ orthogonality property, the event horizon must be a Killing horizon. This result is known as the \textit{weak rigidity theorem}. This result is of significance because, although theorem \ref{theoremkilling} is an interesting result, the assumption that the metric and matter fields are analytic has no physical justification. In this sense, a purely geometrical result provides another rationale for this property.

\section{Surface gravity}\label{kappas}

The purpose of this section is to introduce a scalar that can be defined for null hypersurfaces and extended to Killing horizons, which will be of great importance to the analysis of classical and semiclassical aspects of black holes. Now, the definition of this quantity does not rely on the assumption that the event horizon is a Killing horizon, but several physically significant properties of it do. Hence, it is of interest to restrict this discussion to stationary black holes. In particular, we will also see that the physical interpretation of this quantity can be easily investigated in the case of a Schwarzschild black hole. For simplicity, geometrized units will be adopted in this section.

Since the event horizon is a null hypersurface, its normal vector, $\chi^{\mu}$, must respect
\begin{equation}
    \chi^{\mu}\chi_{\mu}\overset{H}{=}0,
\end{equation}
where the sign $H$ was used on the equality to state that it is only valid on $H$.
Since the vector $\nabla^{\mu}(\chi^{\nu}\chi_{\nu})$ must also be normal to the horizon (see the remarks below theorem \ref{frobenius}), one can define a scalar, $\kappa$, to be the proportionality factor between these two vectors,
\begin{equation}\label{k7}
    \nabla^{\mu}(\chi^{\nu}\chi_{\nu})\overset{H}{=}-2\kappa\chi^{\mu},
\end{equation}
which is equivalent to writing
\begin{equation}\label{k1}
    \chi^{\nu}\nabla_{\nu}\chi^{\mu}\overset{H}{=}\kappa\chi^{\mu},
\end{equation}
as per Killing's equation. It is also possible to find an explicit relation for $\kappa$ by noting that since $\chi^{\mu}$ is hypersurface orthogonal on the horizon, it obeys $\chi_{[\mu}\nabla_{\nu}\chi_{\alpha]}\overset{H}{=}0$. As a consequence of Killing's equation, Frobenius' theorem yields 
 \begin{equation}\label{fk}
     \chi_{\mu}\nabla_{\nu}\chi_{\alpha}\overset{H}{=}-2\chi_{[\nu}\nabla_{\alpha]}\chi_{\mu}.
 \end{equation}
By contracting eq. \ref{fk} with $\nabla^{\nu}\chi^{\alpha}$, one finds
 \begin{equation}
 \begin{aligned}[b]
     \chi_{\mu}(\nabla^{\nu}\chi^{\alpha})(\nabla_{\nu}\chi_{\alpha})&\overset{H}{=}(\nabla^{\nu}\chi^{\alpha})(\chi_{\nu}\nabla_{\alpha}\chi_{\mu}-\chi_{\alpha}\nabla_{\nu}\chi_{\mu})\\
     &\overset{H}{=}\chi_{\nu}(\nabla^{\nu}\chi^{\alpha})(\nabla_{\alpha}\chi_{\mu})-\chi_{\alpha}(\nabla^{\nu}\chi^{\alpha})(\nabla_{\nu}\chi_{\mu})\\
     &\overset{H}{=} 2(\chi_{\nu}\nabla^{\nu}\chi^{\alpha})(\nabla_{\alpha}\chi_{\mu}), 
 \end{aligned}
 \end{equation}
 in which Killing's equation was used, and the indices were relabeled. By using eq. \ref{k1} repeatedly, one obtains
 \begin{equation}\label{kd2}
     \kappa^2\overset{H}{=}-\frac{1}{2}(\nabla^{\mu}\chi^{\nu})(\nabla_{\mu}\chi_{\nu}),
 \end{equation}
 which is the desired relation. Note that this relation is equivalent to eq. \ref{k1} in the case where a null hypersurface is a Killing horizon, but is generally more efficient to evaluate the explicit form of $\kappa$ through eq. \ref{kd2}. 

To investigate the physical interpretation of $\kappa$ for an event horizon, it is convenient to analyze the case of a Schwarzschild black hole. Consider the observer with four-velocity $\ell^{\mu}$, as given by eq. \ref{ell1}. One can verify if such an observer at $r>r_s$ follows a geodesic by evaluating if its tangent respects the geodesic equation. This evaluation yields
\begin{equation}
\begin{aligned}[b]
    a^{\mu}&=\ell^{\nu}\nabla_{\nu}\ell^{\mu}\\
    &=\frac{\xi^{\nu}}{V}\nabla_{\nu}\frac{\xi^{\mu}}{V}\\
    &=\frac{\xi^{\nu}\nabla_{\nu}\xi^{\mu}}{V^2}+\frac{\xi^{\nu}{\xi^{\mu}}}{V}\nabla_{\nu}\frac{1}{V}\\
    &=-\frac{\xi^{\nu}\nabla^{\mu}\xi_{\nu}}{V^2}-\frac{\xi^{\nu}{\xi^{\mu}}}{2V^4}\nabla_{\nu}V^2\\
    &=\frac{\nabla^{\mu}(-\xi^{\nu}\xi_{\nu})}{2V^2}+\frac{{\xi^{\mu}\xi^{\nu}\xi^{\alpha}}}{V^4}\nabla_{\nu}\xi_{\alpha}\\
    &=\frac{\nabla^{\mu}V}{V}\\
    &=\frac{\xi^{\nu}\nabla_{\nu}\xi^{\mu}}{(-\xi^{\alpha}\xi_{\alpha})},
    \end{aligned}
\end{equation}
as the last term on the third line vanishes due to the contraction of a symmetric tensor with an antisymmetric one, and Killing's equation as well as $\nabla^{\mu}V^2/2V^2=\nabla^{\mu}V/V$ were used. As it can be readily verified that $\xi^{\mu}$ is not proportional to $\nabla^{\mu}V$, this result shows that an observer following an orbit of $\xi^{\mu}$ which is not in the asymptotic region is not in geodesic motion. That is, it is necessary to apply an acceleration for his world line to remain an integral curve of $\xi^{\mu}$, whose magnitude is given by
\begin{equation}\label{local}
    a=(a^{\mu}a_{\mu})^{1/2}=\frac{(\xi^{\nu}\nabla_{\nu}\xi^{\mu}\xi^{\alpha}\nabla_{\alpha}\xi_{\mu})^{1/2}}{(-\xi^{\beta}\xi_{\beta})}.
\end{equation}
To continue, it is useful to consider that this acceleration is being applied by an observer with four-velocity $\ell^{\mu}$ in the asymptotic region through a massless inextendible string. In order to evaluate the acceleration the observer at infinity must be exerting on the string, consider the following line of reasoning.

Let $a_{r}$ denote the magnitude of the local acceleration and let $a_{\infty}$ denote the magnitude of the acceleration measured at infinity, as illustrated in Fig. \ref{fig:kappa}. In order to “pull” or “release” a unit mass test point-like body that is following an integral curve of $\xi^{\mu}$, the observer at infinity must apply a force such that the work done is $dW_{\infty}=a_{\infty}d\ell$. Similarly, the work applied to the unit mass body is measured locally to be $dW_{r}=a_{r}d\ell$, since the string is ideal and thus, $d\ell$ must be the same in both cases. Clearly, these two variations of energy cannot be equal, otherwise energy could be created by sending light rays from one end of the string to the other (as per the redshift factor, $V$). The proportionality factor between $a_{\infty}$ and $a_r$ can be found precisely from this fact, as the passage of time for the observer at $r$ is not the same as for the one at infinity. Therefore, for energy to be conserved, it is necessary that
\begin{equation}
    a_{\infty}=V(r)a_r.
\end{equation}

\begin{figure}[h]
\centering
\includegraphics[scale=1.2]{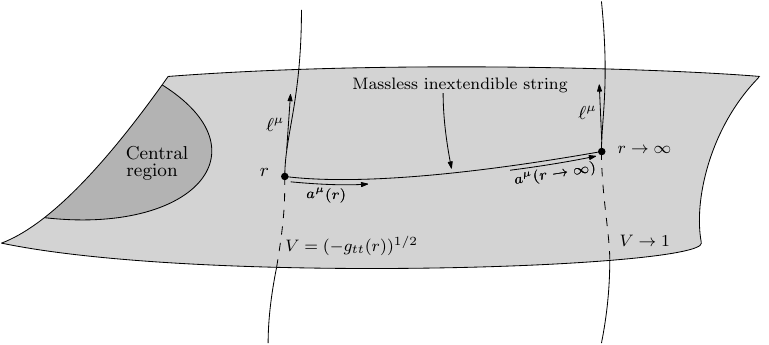} 
\caption{Spacetime diagram for the evaluation of $a_{\infty}$.}
\caption*{Source: By the author.}
\label{fig:kappa}
\end{figure}

To proceed, note that any Killing vector must obey 
\begin{equation}\label{46}
\begin{aligned}[b]
    \chi_{[\mu}\nabla_{\nu}\chi_{\alpha]}&=\frac{2!}{3!}(\chi_{\mu}\nabla_{[\nu}\chi_{\alpha]}+\chi_{\nu}\nabla_{[\alpha}\chi_{\mu]}+\chi_{\alpha}\nabla_{[\mu}\chi_{\nu]})\\
    &=\frac{1}{3}(\chi_{\mu}\nabla_{\nu}\chi_{\alpha}+\chi_{\nu}\nabla_{\alpha}\chi_{\mu}+\chi_{\alpha}\nabla_{\mu}\chi_{\nu})\\
    &=\frac{1}{3}(\chi_{\mu}\nabla_{\nu}\chi_{\alpha}+\chi_{\nu}\nabla_{\alpha}\chi_{\mu}-\chi_{\alpha}\nabla_{\nu}\chi_{\mu}).
    \end{aligned}
\end{equation}
Contracting eq. \ref{46} with itself with all the indices raised yields
\begin{equation}\label{199}
    3(\chi^{[\mu}\nabla^{\nu}\chi^{\alpha]})(\chi_{[\mu}\nabla_{\nu}\chi_{\alpha]})=\chi^{\mu}\chi_{\mu}(\nabla^{\nu}\chi^{\alpha})(\nabla_{\nu}\chi_{\alpha})-2\chi^{\mu}\chi_{\nu}(\nabla^{\nu}\chi^{\alpha})(\nabla_{\mu}\chi_{\alpha}).
\end{equation}
If one divides eq. \ref{199} by the norm of $\chi^{\mu}$ and takes the limit $r\to r_s$, one finds that the right hand side must vanish. This can be deduced from the fact that, if $\kappa\neq0$, then $\nabla_{\mu}(\chi^{\nu}\chi_{\nu})\neq 0$ on the horizon (see the remarks below theorem \ref{frobenius}), while Frobenius' theorem implies that 
\begin{equation}
    \nabla_{\beta}(\chi^{[\mu}\nabla^{\nu}\chi^{\alpha]})(\chi_{[\mu}\nabla_{\nu}\chi_{\alpha]})=2\chi^{[\mu}\nabla^{\nu}\chi^{\alpha]}\nabla_{\beta}\chi_{[\mu}\nabla_{\nu}\chi_{\alpha]}\overset{H}{=}0.
\end{equation}
Thus, L'Hopital's rule \cite{Stewart2007} implies that the limit of the left hand side of eq. \ref{199} divided by $\chi^{\mu}\chi_{\mu}$ as $r\to r_s$ results in
\begin{equation}
\lim_{r\to r_s}{\left[(\nabla^{\nu}\chi^{\alpha})(\nabla_{\nu}\chi_{\alpha})\right]}=\lim_{r\to r_s}{\left[\frac{2\chi^{\mu}\chi_{\nu}(\nabla^{\nu}\chi^{\alpha})(\nabla_{\mu}\chi_{\alpha})}{(-\chi^{\beta}\chi_{\beta})}\right]},
\end{equation}
which by using eq. \ref{kd2} reduces to
\begin{equation}
    \kappa=\lim_{r\to r_s}{\left[\frac{(\chi_{\nu}\nabla^{\nu}\chi^{\alpha})^{1/2}(\chi^{\mu}\nabla_{\mu}\chi_{\alpha})^{1/2}}{(-\chi^{\beta}\chi_{\beta})^{1/2}}\right]}.
\end{equation}
Finally, for a Schwarzschild black hole, $\chi^{\mu}=\xi^{\mu}$, and using eqs. \ref{reds} and \ref{local} yields
\begin{equation}
    \kappa=\lim_{r\to r_s}{[V(r)a_r]}.
\end{equation}

Hence, $\kappa$ is exactly the acceleration measured by an observer following the orbit of $\xi^{\mu}$ in the asymptotic region to hold a unit mass test point-like body at a constant radial coordinate as $r\to r_s$, which can be interpreted as the gravitational acceleration\footnote{Although $\kappa$ is finite, the local acceleration, $a_{r_s}$, diverges.} at the event horizon of a Schwarzschild black hole. Due to this interpretation, $\kappa$ is referred to as the \textit{surface gravity} of a stationary black hole. It should be noted that the acceleration necessary to keep a particle “at rest” at the event horizon of a black hole, as measured from the asymptotic region, can be evaluated for any black hole (but details about the spacetime would be necessary). However, most of the interesting properties of $\kappa$ rely on the assumption that the black hole is stationary. 

The value of $\kappa$ for the Schwarzschild black hole can be computed using eq. \ref{kd2}. Since the null Killing vector normal to the event horizon is $\xi^{\mu}$, eq. \ref{kd2} reads,
\begin{equation}\label{kd3}
    \kappa^2\overset{H}{=}-\frac{1}{2}(\nabla^{\mu}\xi^{\nu})(\nabla_{\mu}\xi_{\nu}).
\end{equation}
To evaluate the surface gravity, one makes use of eq. \ref{Christoffel dual} and notes that
\begin{equation}
    \nabla^{\mu}\xi^{\nu}=g^{\mu\alpha}\nabla_{\alpha}\xi^{\nu}=g^{\mu\alpha}(\partial_{\alpha}\xi^{\nu}+\Gamma^{\nu}\mathstrut_{\alpha\beta}\xi^{\beta}).
\end{equation}
Since $\xi^{\mu}=(\partial_t)^{\mu}$, it is not difficult to conclude\footnote{This conclusion follows from the Christoffel symbols presented in appendix \ref{B1}. The metric presented there is a more general one, but it reduces to the Schwarzschild metric when the factor $a=0$. More details on this metric are given in \S\;\ref{kerr}.} that the only nonvanishing independent component of $\nabla^{\mu}\xi^{\nu}$ is $\nabla^{r}\xi^{t}=-r_s/2r^2$. As $\nabla^{r}\xi^{t}=-\nabla^{t}\xi^{r}$, eq. \ref{kd3} reads
\begin{equation}
    \kappa^2=-g_{rr}g_{tt}(\nabla^{r}\xi^{t})^2\Big|_{r=r_s}=\frac{r_s^2}{4r^4}\Big|_{r=r_s}.
\end{equation}
Thus, restoring the constants, the surface gravity of the Schwarzschild black hole is 
\begin{equation}\label{ksch}
    \kappa=\frac{c^2}{2r_s}.
\end{equation}

It follows from the explicit form of the surface gravity of the Schwarzschild black hole that it is constant over the event horizon. In fact, this result can be proven to be much more general. In order to investigate explicitly the variation of $\kappa$ over the event horizon of any stationary black hole, one needs an adequate operator to do so. Since the surface gravity is defined only on $H$, it can only be differentiated in directions tangent to the event horizon. One would assume that a choice of operator to measure the change of $\kappa$ over $H$ would be the projector operator associated with its metric, but since $H$ is a null hypersurface, no such natural operator exists. Nevertheless, such differentiation can be done by considering the tensor\footnote{This tensor should not be thought of as the volume element of $H$, since it does not obey eq. \ref{form1}. Evidently, there is no such tensor.} $\epsilon^{\mu\nu\alpha\beta}\chi_{\beta}$, where $\epsilon_{\mu\nu\alpha\beta}$ is the volume element of $M$. Since this tensor is tangent to the horizon (in the sense that it is orthogonal to any vector normal to it), one may apply $\epsilon^{\mu\nu\alpha\beta}\chi_{\beta}\nabla_{\alpha}$ to any equation holding there, and its action will give information about the variation of said quantity over the horizon. Moreover, due to the antisymmetric property of $\epsilon^{\mu\nu\alpha\beta}$, one can apply, equivalently, $\epsilon^{\mu\nu\alpha\beta}\chi_{[\beta}\nabla_{\alpha]}$. For simplicity, the development of calculations can be made by applying only $\chi_{[\beta}\nabla_{\alpha]}$ and then contracting with the volume element at the end.
 
 By applying $\chi_{[\beta}\nabla_{\alpha]}$ to eq. \ref{k1}, one obtains
 \begin{equation}\label{k2}
     \chi_{\mu}\chi_{[\beta}\nabla_{\alpha]}\kappa\overset{H}{=}\underbrace{\chi_{[\beta}\nabla_{\alpha]}(\chi^{\nu}\nabla_{\nu}\chi_{\mu})}_{(\text{I})}-\kappa\chi_{[\beta}\nabla_{\alpha]}\chi_{\mu}.
 \end{equation}
 Note that (I) can be written as
 \begin{equation}
 \begin{aligned}[b]
     \chi_{[\beta}\nabla_{\alpha]}(\chi^{\nu}\nabla_{\nu}\chi_{\mu}) & =(\chi_{[\beta}\nabla_{\alpha]}\chi^{\nu})(\nabla_{\nu}\chi_{\mu})+\chi^{\nu}\chi_{[\beta}\nabla_{\alpha]}\nabla_{\nu}\chi_{\mu}\\
     & \overset{H}{=} \kappa\chi_{[\beta}\nabla_{\alpha]}\chi_{\mu}-\chi^{\nu}R_{\nu\mu[\alpha}\mathstrut^{\sigma}\chi_{\beta]}\chi_{\sigma}, 
\end{aligned}
 \end{equation}
 where eqs. \ref{kr}, \ref{k1} were used and eq. \ref{fk} was used twice. Due to the antisymmetry of the first two indices of the Riemann tensor, as per eq. \ref{Rfirst}, eq. \ref{k2} then reads
 \begin{equation}\label{k3}
     \chi_{\mu}\chi_{[\beta}\nabla_{\alpha]}\kappa\overset{H}{=}\chi^{\nu}R_{\mu\nu[\alpha}\mathstrut^{\sigma}\chi_{\beta]}\chi_{\sigma}.
 \end{equation}
One's goal now is to rewrite the right hand side of eq. \ref{k3} to find a relation for $\chi_{[\beta}\nabla_{\alpha]}\kappa$. To proceed, apply $\chi_{[\beta}\nabla_{\alpha]}$ to eq. \ref{fk}, which yields
\begin{equation}\label{k4}
    \underbrace{(\chi_{[\beta}\nabla_{\sigma]}\chi_\mu)(\nabla_{\nu}\chi_{\alpha})}_{(\text{II})}+\chi_{\mu}\chi_{[\beta}\nabla_{\sigma]}\nabla_{\nu}\chi_{\alpha}=\underbrace{-2(\chi_{[\beta}\nabla_{\sigma]}\chi_{[\nu})\nabla_{\alpha]}\chi_{\mu}}_{(\text{III})}-2(\chi_{[\beta}\nabla_{\sigma]}\nabla_{[\alpha}\chi_{|\mu|})\chi_{\nu]}.
\end{equation}
Eq. \ref{k4} can be simplified by noting that
\begin{equation}
\begin{aligned}[b]
    \text{(II)}-\text{(III)} & =(\chi_{[\beta}\nabla_{\sigma]}\chi_\mu)(\nabla_{\nu}\chi_{\alpha})+(\chi_{[\beta}\nabla_{\sigma]}\chi_{\nu})\nabla_{\alpha}\chi_{\mu}-(\chi_{[\beta}\nabla_{\sigma]}\chi_{\alpha})\nabla_{\nu}\chi_{\mu}\\
               & = -\frac{1}{2}[(\chi_{\mu}\nabla_{\beta}\chi_\sigma)(\nabla_{\nu}\chi_{\alpha})+(\chi_{\nu}\nabla_{\beta}\chi_{\sigma})\nabla_{\alpha}\chi_{\mu}-(\chi_{\alpha}\nabla_{\beta}\chi_{\sigma})\nabla_{\nu}\chi_{\mu}]\\   
               & = -\frac{1}{2}(\nabla_{\beta}\chi_\sigma)[(\chi_{\mu}\nabla_{\nu}\chi_{\alpha}+2\chi_{[\nu}\nabla_{\alpha]}\chi_{\mu}]\overset{H}{=}0.
\end{aligned}   
\end{equation}

Consequently, using eq. \ref{kr}, eq. \ref{k4} yields
\begin{equation}\label{k5}
    -\chi_{\mu}R_{\nu\alpha[\sigma}\mathstrut^{\delta}\chi_{\beta]}\chi_{\delta}=2\chi_{[\nu}R_{\alpha]\mu[\sigma}\mathstrut^{\delta}\chi_{\beta]}\chi_{\delta}.
\end{equation}
By contracting $g^{\sigma\mu}$ with eq. \ref{k5}, one finds
\begin{equation}\label{sssq}
\begin{aligned}[b]
    -\chi^{\mu}R_{\nu\alpha[\mu}\mathstrut^{\delta}\chi_{\beta]}\chi_{\delta}&=2\chi_{[\nu}R_{\alpha]}\mathstrut^{\mu}\mathstrut_{[\mu}\mathstrut^{\delta}\chi_{\beta]}\chi_{\delta},\\
    -\chi_{\beta}R_{\nu\alpha\mu}\mathstrut^{\delta}\chi^{\mu}\chi_{\delta}+\chi_{\delta}R_{\nu\alpha\beta}\mathstrut^{\delta}\chi_{\mu}\chi^{\mu}&=\chi_{[\nu}R_{\alpha]}\mathstrut_{\mu}\mathstrut^{\mu\delta}\chi_{\beta}\chi_{\delta}-\chi_{[\nu}R_{\alpha]}\mathstrut_{\mu\beta}\mathstrut^{\delta}\chi^{\mu}\chi_{\delta}.
\end{aligned}   
\end{equation}Note that the left hand side of eq. \ref{sssq} vanishes at $H$, as the first term is the contraction of a symmetric tensor with an antisymmetric one, whereas the second term is multiplied by the norm of a null vector. Therefore, by using that $R_{\alpha\mu}\mathstrut^{\mu\delta}=-R_{\alpha}\mathstrut^{\delta}$, eqs. \ref{Rthird}, \ref{Rfirst} and \ref{Rsecond}, it is possible to find
\begin{equation}
    -\chi_{[\nu}R_{\alpha]}\mathstrut^{\delta}\chi_{\beta}\chi_{\delta}\overset{H}{=}\chi^{\mu}R_{\beta\mu[\alpha}\mathstrut^{\delta}\chi_{\nu]}\chi_{\delta},
\end{equation}
which is the desired relation to use in eq. \ref{k3}, allowing one to write
\begin{equation}\label{kap}
    \chi_{[\beta}\nabla_{\alpha]}\kappa\overset{H}{=}-\chi_{[\beta}R_{\alpha]}\mathstrut^{\mu}\chi_{\mu}.
\end{equation}

Eq. \ref{kap} was derived using only geometrical arguments, being a consequence of the fact that $\kappa$ was defined by eq. \ref{k1} and that $\chi^{\mu}$ is the normal to a Killing horizon. Eq. \ref{kap} is of interest because one can use it to relate the variation of $\kappa$ over the horizon with the energy-momentum tensor. More specifically, considering Einstein's equation, one can apply $\chi^{\mu}$ twice with the adequate choice of indices, so that one finds
\begin{equation}
    \chi_{\alpha}R_{\mu\nu}\chi^{\nu}-\frac{1}{2}R\chi_{\alpha}g_{\mu\nu}\chi^{\nu}=8\pi\chi_{\alpha}T_{\mu\nu}\chi^{\nu}.
\end{equation}
By antisymetrizing over $(\alpha\mu)$, one obtains
\begin{equation}\label{eq37}
\chi_{[\alpha}R_{\mu]}\mathstrut^{\nu}\chi_{\nu}=8\pi\chi_{[\alpha}T_{\mu]}\mathstrut^{\nu}\chi_{\nu}.
\end{equation}
Now, since the black hole is stationary, the expansion of the null geodesic generators of $H$ must vanish, and from the Killing horizon property of $H$, one has that $\nabla_{(\mu}\chi_{\nu)}=0$. Since the generators of a null geodesic congruence must have a vanishing vorticity tensor, Raychaudhuri's equation yields
\begin{equation}\label{eq78}
    R_{\mu\nu}\chi^{\mu}\chi^{\nu}\overset{H}{=}0.
\end{equation}
From Einstein's equation, eq. \ref{eq78} then implies that $T_{\mu\nu}\chi^{\mu}\chi^{\nu}\overset{H}{=}0$. In essence, this means that the vector $-T^{\mu\nu}\chi_{\nu}$ is orthogonal to $\chi^{\mu}$ on the horizon, thus, it must be spacelike or null on it. However, if the dominant energy condition holds, then $-T^{\mu\nu}\chi_{\nu}$ must be proportional to $\chi^{\mu}$, so that the right hand side of eq. \ref{eq37} vanishes. Hence, if Einstein's equation and the dominant energy condition hold, the surface gravity is constant over the event horizon of a stationary black hole, i.e., 
\begin{equation}
\epsilon^{\mu\nu\alpha\beta}\chi_{[\mu}\nabla_{\nu]}\kappa\overset{H}{=}0.
\end{equation}

Finally, it should be noted that one can arrive at the same conclusion without invoking Einstein's equation, as shown in \cite{Rácz1996}. Such derivation consists in purely geometrical arguments, and is made by considering that a static, or merely stationary spacetime, obeys the $t$-$\phi$ orthogonality property, in a very similar fashion as for the geometrical arguments given for the weak rigidity theorem. For an extensive review of the many formulations of the constancy of $\kappa$ over Killing horizons, see \cite{Heusler1996}.

\section{Stationary black holes}\label{kerr}

In the last sections, several properties of the black hole region and the event horizon, $H$, were discussed, including the Killing horizon property of $H$ of a stationary black hole and a derivation of the generality of the constancy of the surface gravity over $H$ in such cases. The reason behind the interest in properties of stationary black holes lies in the idea that, although details of the gravitational collapse\footnote{It should be noted that other sources of gravitational collapse, such as inhomogeneities in the early universe, will produce black holes that, viewed as isolated systems at sufficient ``late times'', will also reach a stationary final state.} that results in a black hole can greatly affect the geometry of the spacetime, at sufficiently ``late times'' after its formation, the black hole is expected to reach a stationary final state due to the emission of gravitational waves and interaction with others sources of energy in the spacetime. This assumption is corroborated by the behavior of other systems, e.g., electromagnetic ones, in which a time dependent configuration of charges radiates away the higher order multipole moments and eventually ``settles down'' to a stationary final state  \cite{Townsend1997,Wald1984}. These results and arguments have led to the following conjecture. 

\begin{state*}
Let $(M,g_{\mu\nu})$ be a strongly asymptotically predictable spacetime, $\Sigma_{t}$ be a Cauchy hypersurface for the globally hyperbolic region $\psi^{-1}[A]\subset M$ (see \S~\ref{bhsec}) and $B$ denote the black hole region of $(M,g_{\mu\nu})$. If $B\cap\Sigma_{t'}=\emptyset$ for all $t'<t$ and $B\cap\Sigma_{t'}\neq\emptyset$ for all $t'\geq t$, then at $\Sigma_{t'}$, with $t'\gg t$, the spacetime will be described by a stationary metric.

\end{state*}

In other words, at sufficiently ``late times'' after its formation, a black hole will be in a stationary state. In particular, the usage of quote unquote for ``late times'' is due to the fact\footnote{In the following we shall drop the quote unquote for such statements.} that one cannot simply use the coordinate time, or more precisely, the proper time of any observer outside a black hole to make such statements. As was exemplified in the discussion of Schwarzschild spacetime, the coordinate time is not an appropriate coordinate to make statements regarding $H$, and in fact, the Schwarzschild black hole never forms in the frame of reference of outside observers. Consequently, it is more appropriate to make such statements regarding the evolution of black holes over ``time'' by means of Cauchy hypersurfaces, which is precisely what was done for theorems \ref{45678} and \ref{area}. Indeed, although Cauchy hypersurfaces are not unique, since a globally hyperbolic spacetime has topology $\mathbb{R}\times\Sigma$ (see theorem \ref{timeins}), one can make use of them to make adequate statements.

Now, the stationary state conjecture is of physical significance because it is possible to show that the description of a stationary black hole solution of the Einstein–Maxwell equation (i.e., a solution of eq. \ref{eq1} with the energy momentum tensor given by the electromagnetic field \cite{Wald2022}) is made uniquely by the \textit{Kerr-Newman metric}. This metric can be written in Boyer–Lindquist coordinates (see, e.g., \cite{Giulini2009}) as
\begin{equation}\label{kerrnewman}
    \begin{split}
   ds^2=&-\left(\frac{\Delta-a^2\sin^2\theta}{\Sigma}\right)c^2dt^2-\frac{2a\sin^2\theta(r^2+a^2-\Delta)}{\Sigma}cdtd\phi \\
&+\left[\frac{(r^2+a^2)^2-\Delta a^2\sin^2\theta}{\Sigma}\right]\sin^2\theta d\phi^2+\frac{\Sigma}{\Delta}dr^2+\Sigma d\theta^2,
\end{split}
\end{equation}
where
\begin{equation}
    a=\frac{L}{Mc},\; \Sigma=r^2+a^2\cos^2{\theta},\;\Delta=r^2-r_sr+a^2+e^2,\; e^2=\frac{q^2G}{4\pi\epsilon_0c^4},
\end{equation}
$L$ is the angular momentum of the black hole, $q$ its electric charge and $\epsilon_0$ is the \textit{vacuum permittivity}. The length parameters related to angular momentum and electric charge are referred to as the \textit{Kerr parameter}, $a$, and \textit{length electric charge}, $e$. 

The uniqueness of the Kerr-Newman metric to describe stationary black holes\footnote{Even though the Kerr-Newman metric is the most general stationary black hole solution of the Einstein-Maxwell equation, there is no analogue of Birkhorff's theorem to it, i.e., it is not the most general description of a star \cite{Townsend1997}.} is often referred to as the result from the black hole \textit{uniqueness theorems}, which can be summarized as follows\footnote{Unless stated otherwise, in following we will restrict our analysis to black hole solutions of the Einstein-Maxwell equation.}. It was first shown that every stationary black hole must be static or axisymmetric, which was followed by the proof of uniqueness of the Schwarzschild spacetime as the only static vacuum black hole solution. It was then proved that a stationary axisymmetric black hole must be characterized only by its mass and angular momentum, and generalization of these statements to charged distributions of energy was then derived. A qualitative review of the highly technical arguments from which these conclusions were made and the references of the original derivations can be found in \cite{Robinson2009}, while details on the quantitative aspects can be found in \cite{Heusler1996} and the derivation of the Kerr-Newman metric can be found in \cite{Frolov1998}. See also \cite{Landsman2021} for an interesting discussion.

This series of results that prove the uniqueness of the Kerr-Newman metric as a description of stationary black holes in the presence of suitable matter fields (see \S~\ref{nature} for further discussion) clearly states, together with the stationary state conjecture, that at sufficiently late times after its formation, a black hole\footnote{By black hole, we mean the region $B\cap\Sigma$, where $\Sigma$ denotes a Cauchy hypersurface.} will have only three degrees of freedom: its mass, angular momentum, and electric charge. As a consequence, stationary black holes are said to have no ``hair'', i.e., no distinguishable characteristics of the energy distribution that gave rise to it are accessible to outside observers after it has ``settled down'' to a stationary state other than the parameters $(r_s,a,e)$. Lastly, as will be discussed in detail below, the Kerr-Newman metric describes a connected black hole region. It seems that the only stationary spacetime for which one would have a non-connected black hole region is one describing multiple black holes with $a=0$ and $2e=r_s$, which would have to be a static configuration \cite{Hartle1972,Neugebaur2009}. As implied by our focus on the analysis of the Kerr-Newman metric, this section is restricted to cases in which the black hole region is connected, since the exception mentioned above is non interesting as it is highly non-physical.

The region in the exterior of the energy distribution described by the Kerr-Newman metric is known as \textit{electrovac} due to the fact that it is vacuum with the exception of possible electromagnetic fields as a consequence of the charge of the energy distribution. However, in physical scenarios, if the contributions of the electromagnetic field are of considerable order, over time the black hole will attract charges of opposite sign, eventually decreasing the significance of these contributions. Hence, astrophysical bodies can be treated as respecting $e\simeq 0$ to a very good order of approximation, which restricts their description to the Kerr family of metrics, given by eq. \ref{kerrnewman} with $e=0$. Consequently, any Kerr black hole will be uniquely described by its mass and angular momentum. Details on the Kerr metric, such as Christoffel symbols and its inverse, are presented in appendix \ref{B1}. In addition, if $a=0$ and $e\neq0$, eq. \ref{kerrnewman} reduces to the Reissner–Nordström solution, which describes the spacetime outside a charged spherically symmetric distribution of energy. 

Mathematically speaking, the Kerr metric differs very little from the Kerr-Newman metric. In fact, the qualitative nature of the analysis of the Kerr metric is mostly unaffected by the presence of electric charge. Thus, most of the results concerning a Kerr black hole can be straightforwardly generalized to a Kerr-Newman one. Because of this, for the remainder of this section the focus will be on the qualitative structure of the Kerr spacetime and its black hole region, but there will be remarks when the relevant quantitative structure changes if one were to consider $e\neq0$. Starting, we will first discuss its symmetries, and evaluate conserved quantities associated with the Killing vectors. We then turn our attention to the black hole region, computing the radial coordinate at which observers and light rays cannot escape to infinity and analyzing its properties as measured by observers in the asymptotic region. Next, using the tangent to the null geodesic generators of the event horizon, we will derive a relation between the mass, area and angular momentum of the black hole, and study how such quantities vary when the black hole is perturbed. Finally, we discuss the physical plausibility of the inequalities of the mass and Kerr parameter, $a$. 

The symmetries of the Kerr spacetime can be identified directly from its form in Boyer–Lindquist coordinates. First, as expected, none of its coordinate components depend on $t$, translating to the fact that $\xi^{\mu}=(\partial_{t})^{\mu}$, which is timelike in a neighborhood of infinity, $\mathscr{I}$, is a Killing vector. From this and the fact that the metric components behave as $g_{ab}=\eta_{ab}+\mathcal{O}(r^{-1})$ for large $r$, one can roughly see that it is asymptotically flat. A detailed demonstration of this property of the Kerr spacetime can be found in \cite{Ashtekar1978}. However, although the Kerr metric is stationary, it is not static due to the presence of the term $dtd\phi$. This also implies that it is not spherically symmetric, as the angular momentum of the energy distribution privileges a direction in space. This direction can be easily identified from the fact that the metric components also do not depend on $\phi$, so $\psi^{\mu}=(\partial_{\phi})^{\mu}$ is the Killing vector associated with translation in the direction of the rotation of the energy distribution that gave rise to the black hole. Since this Killing vector can be reparametrized so that its orbits are closed curves of length $2\pi$, the Kerr metric is also axisymmetric. Indeed, the Kerr-Newman metric is a stationary and axisymmetric spacetime that obeys the $t$-$\phi$ orthogonality property \cite{Heusler1996}.

The Killing vectors discussed give rise to quantities conserved along geodesics, 
\begin{equation}\label{energy}
    E=-s^{\mu}\xi_{\mu},
\end{equation}
\begin{equation}\label{angularmomentum}
    L=s^{\mu}\psi_{\mu},
\end{equation}
which can be interpreted as energy per unit rest mass and angular momentum per unit rest mass, respectively. It should be noted that the scalars given by eqs. \ref{energy} and \ref{angularmomentum} can still be interpreted as such, even if $s^{\mu}$ is not the tangent to a geodesic. That is, the scalar may not be conserved, but its interpretation as the energy or angular momentum of the observer with four-velocity $s^{\mu}$ is still physically justified. Additionally, one expects that these interpretations should also be valid for the conservation laws that give rise to the conserved charges given by the Komar integrals (see \S~\ref{symmetry}). More precisely, one expects that the conserved quantity given by the Komar integral of the timelike Killing vector to be related to the total mass of the spacetime, while the Komar integral of the spacelike Killing vector is expected to be related to its total angular momentum. 

It is possible to find these relations by evaluating the Komar integral over a two-sphere at the asymptotic region. For simplicity, calculations for $\xi^{\mu}$ and $\psi^{\mu}$ will be performed simultaneously. This computation is based on the fact that the integrand of the Komar integral must be proportional to the volume element on a two-sphere, as the vector space of two-forms over a two-dimensional vector space is one-dimensional (see theorem \ref{propform}). This statement translates to
\begin{equation}\label{eq1300}
    \epsilon_{\mu\nu\alpha\beta}\nabla^{\alpha}(\partial_a)^{\beta}=f_a\epsilon_{\mu\nu\alpha\beta}\zeta^{\alpha}\tau^{\beta},
\end{equation}
where $a$ is an index to represent each Killing vector, and $\tau^{\mu}$, $\zeta^{\mu}$ are the normal vectors to the two-sphere at sufficiently large $r$, i.e., they are normalized and orthogonal to the hypersurfaces $\{t=\text{constant}\}$,
\begin{equation}\label{normalt}
    \tau^{\mu}=\frac{\nabla^{\mu}t}{\left(-\nabla^{\mu}t\nabla_{\mu}t\right)^{1/2}}=\frac{g^{\mu\nu}\nabla_{\nu}t}{\left(-g^{tt}\right)^{1/2}}=\frac{g^{\mu\nu}\partial_{\nu}t}{\left(-g^{tt}\right)^{1/2}}=\frac{g^{\mu t}}{\left(-g^{tt}\right)^{1/2}},
\end{equation}
and $\{r=\text{constant}\}$,
\begin{equation}\label{normalr}
    \zeta^{\mu}=\frac{\nabla^{\mu}r}{\left(\nabla^{\mu}r\nabla_{\mu}r\right)^{1/2}}=\frac{g^{\mu\nu}\nabla_{\nu}r}{\left(g^{rr}\right)^{1/2}}=\frac{g^{\mu\nu}\partial_{\nu}r}{\left(g^{rr}\right)^{1/2}}=\frac{g^{\mu r}}{\left(g^{rr}\right)^{1/2}}.
\end{equation}
Denoting the volume element on the two-sphere by $\epsilon_{\mu\nu}$, one can evaluate the scalar $f_{a}$ by applying $\epsilon_{\mu\nu}$ to both sides of eq. \ref{eq1300}. Using Killing's equation, eqs. \ref{Christoffel vector}, \ref{form1}, and \ref{form2}, one obtains
\begin{equation}
\begin{aligned}[b]
f_a\epsilon^{\mu\nu}\epsilon_{\mu\nu}&=\epsilon^{\mu\nu}\epsilon_{\mu\nu\alpha\beta}\nabla^{\alpha}(\partial_a)^{\beta},\\
 2f_a&=-4\delta^{[{\lambda}}\mathstrut_{\alpha}\delta^{\sigma]}\mathstrut_{\beta}\zeta_{\lambda}\tau_{\sigma}\nabla^{\alpha}(\partial_a)^{\beta},\\
f_a&=-2\zeta^{\lambda}\tau_{\sigma}\nabla_{{\lambda}}(\partial_a)^{\sigma}\\
&=-2(-g^{tt}g^{rr})^{-1/2}g^{\lambda r}\delta^{t}\mathstrut_{\sigma}\left(\partial_{\lambda}(\partial_a)^{\sigma}+\Gamma^{\sigma}\mathstrut_{\lambda \rho}(\partial_a)^{\rho}\right)\\
&=-2(-g^{tt}g^{rr})^{-1/2}g^{rr}\Gamma^{t}\mathstrut_{r a}.
\end{aligned}    
\end{equation}

By considering the explicit form of the volume element on the two-sphere, given by eq. \ref{volumeexplicit2}, one can write the Komar integral as
\begin{equation}
    \int_{\partial\Sigma}\epsilon_{\mu\nu\alpha\beta}\nabla^{\alpha}(\partial_a)^{\beta}=\int_{\partial\Sigma}f_a\epsilon_{\mu\nu}=\int_{\partial\Sigma}f_ar^2d\Omega.
\end{equation}
Since the Komar integral is independent of $\partial\Sigma$\footnote{For a Kerr-Newman black hole, the Komar integral will have an extra term due to electric charge contribution. Namely, since the current $-R^{\mu\nu}\chi_{\nu}$ will not vanish due to the electrovac region, one would have to take into account its contributions.}, one can take it to be a two-sphere at the asymptotic region, i.e., $r\to\infty$. Evaluation of the scalar $f_a$ in this limit for the pertinent cases, considering only terms of dominant order, yields
\begin{equation}
    f_t=-2(-g^{tt}g^{rr})^{-1/2}g^{rr}\Gamma^{t}\mathstrut_{r t}\implies\lim_{r\to\infty}f_t=-\frac{cr_s}{r^2},
\end{equation}
\begin{equation}
    f_{\phi}=-2(-g^{tt}g^{rr})^{-1/2}g^{rr}\Gamma^{t}\mathstrut_{r \phi}\implies\lim_{r\to\infty}f_{\phi}=\frac{3r_sa\sin{\theta}}{r^2}.
\end{equation}
Thus, the Komar integrals for the Killing vectors of the Kerr metric read
\begin{equation}
    \int_{\partial\Sigma}\epsilon_{\mu\nu\alpha\beta}\nabla^{\alpha}\xi^{\beta}=-cr_s\int_{\partial\Sigma}d\Omega,
\end{equation}
\begin{equation}
    \int_{\partial\Sigma}\epsilon_{\mu\nu\alpha\beta}\nabla^{\alpha}\psi^{\beta}=3r_sa\int_{\partial\Sigma}\sin{\theta}d\Omega,
\end{equation}
in which one finds the expected relations for the total mass,
\begin{equation}\label{MKerr}
    M=-\frac{c}{8\pi G}\int_{\partial\Sigma}\epsilon_{\mu\nu\alpha\beta}\nabla^{\alpha}\xi^{\beta},
\end{equation}
and angular momentum,
\begin{equation}\label{JKerr}
   L=\frac{c^3}{16\pi G}\int_{\partial\Sigma}\epsilon_{\mu\nu\alpha\beta}\nabla^{\alpha}\psi^{\beta}.
\end{equation}

To discuss the black hole region of the Kerr spacetime, first note that the Kerr metric is a solution of the vacuum Einstein's equation for any value of mass, $r_s$, and Kerr parameter, $a$. For the case $a=0$, the metric reduces to the Schwarzschild metric and the pertinent results have already been discussed. Moreover, note that the sign of the Kerr parameter is not relevant, as if $a<0$, one can simply perform the transformation  $\phi\to-\phi$ (which results in $a\to-a$), so one can consider $a>0$ without loss of generality. In this manner, it is clear that the Kerr metric is pathological for $\Sigma=0$ and $\Delta=0$. By evaluating curvature scalars \cite{Poisson2004}, it becomes clear that the singular character at $\Delta=0$ is coordinate dependent, given by the values
\begin{equation}\label{coorKerr}
    r_\pm=\frac{1}{2}\left(r_s\pm\left(r_s^2-4a^2\right)^{1/2}\right).
\end{equation}
Consequently, one can relate the Kerr parameter and the Schwarzschild radius to the coordinate singularities by
\begin{equation}\label{radius12}
    a^2=r_+r_-,
\end{equation}
\begin{equation}\label{radius13}
    r_s=r_++r_-.
\end{equation}
For the case $2a<r_s$, known as \textit{slow} Kerr, there are two coordinate singularities. The case $2a=r_s$ is known as the \textit{extreme} Kerr\footnote{Analogous cases arise if $e\neq0$. For instance, the extreme Kerr-Newman black hole obeys $4(a^2+e^2)=r_s^2$.}, and has only one coordinate singularity. For the case $2a>r_s$, known as \textit{fast} Kerr, there are no coordinate singularities. However, the singularity at $\Sigma=0$, which is present in all of these cases, is a true, physical singularity, given by
\begin{equation}
    r^2+a^2\cos^2\theta=0.
\end{equation}
For details on the analytical extension of the metric to describe the entire spacetime for each of the cases above, as well as their conformal diagrams, see \cite{ONeill1995}. In the following, only the slow Kerr case will be considered. The analysis of the fast and extreme case will be made at the end of this section. 

We now show that the outer coordinate singularity, $r_+$, delimits the black hole region of Kerr spacetime. As introduced before, the vector $\nabla^{\mu}r$ (see eq. \ref{normalr}) is normal to the hypersurfaces $\{r=\text{constant}\}$. From its norm, $\nabla^{\mu}r\nabla_{\mu}r=g^{rr}=\Delta/\Sigma$, one can see that it will be spacelike in the regions $r>r_+$ and $r<r_-$, timelike in $r_-<r<r_+$ and null on the hypersurfaces $\{r=r_{\pm}\}$. Note that the fact that $\nabla^{\mu}r$ is timelike in the region delimited by the coordinate singularities implies that $s^{\mu}\nabla_{\mu}r=dr/d\lambda$ will be negative for any future directed timelike or null vector, $s^{\mu}$, in $r_-<r<r_+$, where $\lambda$ is the parameter of the causal curve. Hence, any light ray or observer passing through the hypersurface $\{r=r_+\}$ will inevitably reach the hypersurface $\{r=r_-\}$, with the region $r>r_+$ no longer being accessible to them. More precisely, no observer or light ray in $r<r_+$ can be in the past of events in the region $r>r_+$, and thus, $r_+$ delimits the black hole region of the Kerr spacetime. From this, one can evaluate the area of the event horizon\footnote{By area of the event horizon, it is meant the area of $\mathscr{H}$. Because of the null nature of the hypersurface, its ``volume'' is zero.} by integrating the two-dimensional metric that arises from the Kerr metric with $dt=0$ and $r=r_+$, 
\begin{equation}\label{areakerr}
    \begin{aligned}[b]
    A &=\int_{\mathscr{H}}(g_{\theta\theta}g_{\phi\phi})^{1/2}d\theta d\phi\\
         &=(r_+^2+a^2)\int_0^{2\pi}d\phi\int_0^{\pi}\sin{\theta}d\theta\\
         &=4\pi(r_+^2+a^2).   
    \end{aligned}
\end{equation}
Lastly, although the hypersurface $\{r=r_-\}$ is not as significant to the black hole region, it also earns the adjective ``horizon'' due to the behavior of the vector $\nabla^{\mu}r$, that is, it also acts as a one-way membrane. Indeed, analysis of the extended Kerr spacetime confirms that the hypersurface $\{r=r_-\}$ is a Cauchy horizon, marking the region from which predictions regarding the evolution of physical fields can be made from outside the black hole (see \cite{Hawking1973,Wald1984} for further discussion). 

Moving on, from the fact that the Kerr spacetime is stationary, $\xi^{\mu}$ must be null or spacelike on the event horizon. Since $\xi^{\mu}\xi_{\mu}=g_{tt}$, one can conclude that $\xi^{\mu}$ becomes spacelike in the region 
\begin{equation}\label{534}
    \frac{1}{2}\left(r_s-\left(r_s^2-4a^2\cos^2{\theta}\right)^{1/2}\right)<r<\frac{1}{2}\left(r_s+\left(r_s^2-4a^2\cos^2{\theta}\right)^{1/2}\right).
\end{equation}
Evidently, the upper limit of this region does not coincide with $r_+$ if $\cos^2{\theta}\neq 1$. This means that there is a region outside the black hole for which observers cannot follow the integral curves of $\xi^{\mu}$. Such region, corresponding to $r_+<r<r_e$, is known as the \textit{ergoregion}, and the surface $\{r=r_e,t=\text{constant}\}$, where $r_e$ denotes the upper limit in eq. \ref{534}, is known as the \textit{ergosphere}. To analyze the possible behavior of world lines in the ergoregion, one can consider the constraints on the components of the tangent vector, $s^{\mu}$, of a future directed causal curve parametrized by $\lambda$,
\begin{equation}
    s^{\mu}s_{\mu}=g_{tt}\left(\frac{dt}{d\lambda}\right)^2+2g_{t\phi}\left(\frac{dt}{d\lambda}\right)\left(\frac{d\phi}{d\lambda}\right)+g_{rr}\left(\frac{dr}{d\lambda}\right)^2+g_{\phi\phi}\left(\frac{d\phi}{d\lambda}\right)^2+g_{\theta\theta}\left(\frac{d\theta}{d\lambda}\right)^2\leq0,
\end{equation}
in which one can see that all terms are manifestly positive, except for the second one. Now, in the ergoregion, $g_{t\phi}<0$ and $dt/d\lambda=s^{\mu}\nabla_{\mu}t>0$, since $\nabla^{\mu}t$ is past directed timelike, in order for $t$ to increase as one moves to the future. Hence, one concludes that $d\phi/d\lambda>0$ for any timelike or null vector in the region $r_+<r<r_e$, which means it is impossible for an observer or signal in this region to not rotate in the direction of the black hole. This effect is known as \textit{frame dragging}.

Even though it is not possible to follow the orbits of $\xi^{\mu}$ once an observer reaches the hypersurface $\{r=r_e\}$, they still can maintain a fixed $r$ and $\theta$ and move following the orbits of a Killing vector that is a linear combination of $\xi^{\mu}$ and $\psi^{\mu}$. Such observers would, of course, see no variation of the metric as they evolve. We now show that there exists a family of observers that follow such orbits and move orthogonal to the hypersurfaces $\{t=\text{constant}\}$. That is, the tangent vector to their world lines is proportional to $\nabla^{\mu}t$. It is easy to see that the failure of the world lines of these observers to coincide with the orbits of $\xi^{\mu}$ is directly related to the nonvanishing of the $g_{\phi t}$ term in the Kerr metric. Indeed, an interesting property of this family of observers, whose tangent vector will be denoted by $\chi^{\mu}$, is that they have no angular momentum,
\begin{equation}
    L=\psi_{\mu}\chi^{\mu}\propto\psi_{\mu}\tau^{\mu}=0,
\end{equation}
but have a nonvanishing coordinate angular velocity which is a function of their radial coordinate,
\begin{equation}
    \Omega(r)=\frac{\chi^{\phi}}{\chi^t}=\frac{d\phi}{dt}=\frac{\tau^{\phi}}{\tau^{t}}=\frac{g^{\phi t}}{g^{tt}}=-\frac{g_{\phi t}}{g_{\phi\phi}},
\end{equation}
where $g^{\nu\alpha}g_{\alpha\mu}=\delta^{\nu}\mathstrut_{\mu}$ was used for both results. From the Kerr metric, one obtains
\begin{equation}\label{angularv}
    \Omega(r)=\frac{ca(r^2+a^2-\Delta)}{(r^2+a^2)^2-\Delta a^2\sin^2{\theta}}.
\end{equation}
Finally, normalizing the tangent vector to the world line of these observers yields
\begin{equation}\label{tangentobs}
    \chi^{\mu}=\frac{\xi^{\mu}+\Omega(r)\psi^{\mu}}{\left(-g_{tt}-2\Omega(r) g_{t\phi}-\Omega^2(r) g_{\phi\phi}\right)^{1/2}}.
\end{equation}

Note that the coordinate angular velocity given by eq. \ref{angularv} is the one measured by observers following an orbit of $\xi^{\mu}$ at the asymptotic region, i.e., those whose proper time coincides with the coordinate time. The significance of the family of observers that follow the orbits of $\chi^{\mu}$ in the ergoregion is given by the following line of reasoning. Since $\xi^{\mu}$ is spacelike on the hypersurface $\{r=r_+\}$, the null Killing vector tangent to the null geodesic generators of the event horizon must be a linear combination of $\xi^{\mu}$ and $\psi^{\mu}$. By requiring that such a linear combination be null, a possible solution\footnote{Any other solution is proportional to $\chi^{\mu}$. This one is ``preferred'' for our analysis since it is the tangent vector that rises when one parametrizes the null geodesics generators of $H$ by the time coordinate, $t$.} is $\chi^{\mu}=\xi^{\mu}+\Omega(r_+)\psi^{\mu}$, where 
\begin{equation}
    \Omega(r_+)=\frac{ca}{r_+^2+a^2}.
\end{equation}
Now, no observer would be able to follow the orbits of $\chi^{\mu}$ at $r_+$, but as $\Omega(r)$ is the angular velocity of the particles as measured by ``distant'' observers with four velocity $\xi^{\mu}$, as $r\to r_+$, these observers can identify $\Omega(r_+)$ as the angular velocity of the event horizon, which will be denoted simply as $\Omega$. Thus, the tangent vector to the null geodesic generators of the event horizon of the Kerr black hole, when the null geodesics are parametrized by $t$, is 
\begin{equation}\label{tangentH}
    \chi^{\mu}=\xi^{\mu}+\Omega\psi^{\mu}.
\end{equation}

One can straightforwardly deduce that the family of observers that follow the orbits of $\chi^{\mu}$ is not in geodesic motion. Evidently, the acceleration necessary to keep such observers following the orbits of $\chi^{\mu}$ as $r\to r_+$, as measured by observers in the asymptotic region following the orbits of $\xi^{\mu}$, is precisely the surface gravity of a Kerr black hole. However, note that such a measurement cannot be made by an ideal string as proposed for the Schwarzschild black hole, since $\chi^{\phi}\neq 0$. The measurement process would be made by local observers at $r\to r_+$ and converted to a measure at infinity by a redshift factor, which is given by the norm of the vector in eq. \ref{tangentobs}. As such, from the explicit form of $\chi^{\mu}$, one can obtain \cite{Chruściel2020} 
\begin{equation}\label{kappakerr}
    \kappa=\frac{c^2\left(r_s^2-4a^2\right)^{1/2}}{r_s\left(r_s+\left(r_s^2-4a^2\right)^{1/2}\right)}.
\end{equation}
Furthermore, it follows from the explicit form of $\Omega$ that
\begin{equation}\label{chixi}
    \chi^{\mu}\xi_{\mu}\overset{H}{=}0,
\end{equation}
\begin{equation}\label{chipsi}
    \chi^{\mu}\psi_{\mu}\overset{H}{=}0,
\end{equation}
and from the fact that $\xi^{\mu}$ and $\psi^{\mu}$ are elements of a coordinate basis, $[\xi,\chi]^{\mu}$ and $[\xi,\psi]^{\mu}$ must vanish, which is equivalent to
\begin{equation}\label{chixi1}
    \chi^{\nu}\nabla_{\nu}\xi^{\mu}=\xi^{\nu}\nabla_{\nu}\chi^{\mu},
\end{equation}
\begin{equation}\label{chipsi2}
    \chi^{\nu}\nabla_{\nu}\psi^{\mu}=\psi^{\nu}\nabla_{\nu}\chi^{\mu}.
\end{equation}

In the developments that follow, it will be necessary to make use of the transverse metric of $H$ with respect to $\chi^{\mu}$. In particular, the auxiliary null vector used to define the transverse metric in \S\;\ref{null} which obeys $\eta^{\mu}\chi_{\mu}=-1$, will be chosen for convenience to respect further relations, as detailed below. In a local Lorentz frame at the event horizon, with coordinates $\{cx^0,x^1,x^2,x^3\}$, one can write
\begin{equation}
    \chi^{\mu}=(1,1,0,0),
\end{equation}
\begin{equation}
    \xi^{\mu}=(0,0,1,0),
\end{equation}
as a consequence of eq. \ref{chixi}. From eq. \ref{tangentH}, these choices fix $\psi^{\mu}$ to be
\begin{equation}
    \psi^{\mu}=\left(\frac{1}{\Omega},\frac{1}{\Omega},-\frac{1}{\Omega},0\right).
\end{equation}
A vector, $\vartheta^{\mu}$, given by
\begin{equation}
    \vartheta^{\mu}=(0,0,0,1),
\end{equation}
together with $\xi^{\mu}$ and $\psi^{\mu}$, clearly yields a basis for the vector space spanned by vectors orthogonal to $\chi^{\mu}$. Thus, by choosing 
\begin{equation}\label{choice}
    \eta^{\mu}=(1,0,-1,0),
\end{equation}
it will also respect the following relations,
\begin{equation}\label{Auxxi}
    \eta^{\mu}\xi_{\mu}=-1,
\end{equation}
\begin{equation}\label{Auxpsi}
    \eta^{\mu}\psi_{\mu}=0,
\end{equation}
\begin{equation}
    \eta^{\mu}\vartheta_{\mu}=0.
\end{equation}
Lastly, the transverse metric still takes the form given by eq. \ref{transversemetric}, since it is a consequence of the condition $\eta^{\mu}\chi_{\mu}=-1$. 

The next property of interest of a Kerr black hole is the relationship between the variations of mass, angular momentum and area of the event horizon of two slightly different Kerr black holes. To derive such a relation, consider first how these quantities relate for a single black hole, which can be found by applying the Levi-Civita connection to eq. \ref{tangentH} and integrating over the event horizon. Adopting geometrized units throughout this calculation, the integration yields
\begin{equation}\label{34}
    \int_{\mathscr{H}}\epsilon_{\mu\nu\alpha\beta}\nabla^{\alpha}\chi^{\beta}=\int_{\mathscr{H}}\epsilon_{\mu\nu\alpha\beta}\nabla^{\alpha}\xi^{\beta}+\Omega\int_{\mathscr{H}}\epsilon_{\mu\nu\alpha\beta}\nabla^{\alpha}\psi^{\beta}.
\end{equation}
Now, one can identify the integrals on the right hand side of eq. \ref{34} with the Komar integrals over $\partial\Sigma=\mathscr{H}$, which is possible due to the fact that it is independent of $\partial\Sigma$, provided that the flux of the conserved current vanishes over it. Thus, one has
\begin{equation}\label{mass2}
    \underbrace{\int_{\mathscr{H}}\epsilon_{\mu\nu\alpha\beta}\nabla^{\alpha}\chi^{\beta}}_{\text{(I)}}=-8\pi M+16\pi \Omega L.
\end{equation}

In order to evaluate (I), one should note that the volume element restricted to $\mathscr{H}$ is (see appendix \ref{integration} and \S~\ref{null})
\begin{equation}
    \epsilon_{\mu\nu}=\epsilon_{\mu\nu\alpha\beta}\eta^{\alpha}\chi^{\beta},
\end{equation}
where $\eta^{\mu}$ is the auxiliary null vector which in a local Lorentz frame at $H$ takes the form given by eq. \ref{choice}. Since $\mathscr{H}$ is a surface, one can use the linearity of two-forms over it to compute (I) in terms of $\epsilon_{\mu\nu}$ and the proportionality factor, $f$, i.e.,
\begin{equation}\label{eqx}
    f\epsilon_{\mu\nu}=\epsilon_{\mu\nu\alpha\beta}\nabla^{\alpha}\chi^{\beta}.
\end{equation}
The scalar $f$ can be evaluated by contracting $\epsilon^{\mu\nu}$ with eq. \ref{eqx}, which yields
\begin{equation}
    f=\frac{1}{2}\epsilon^{\mu\nu}\epsilon_{\mu\nu\alpha\beta}\nabla^{\alpha}\chi^{\beta}.
\end{equation}
It is possible to simplify $f$ by using that $\eta^{\mu}\chi_{\mu}=-1$, which coupled with the fact that $\chi^{\mu}$ is a Killing vector, the definition of $\kappa$ and eq. \ref{form2}, yields
\begin{equation}
\begin{aligned}[b]
    f & =\frac{1}{2}\eta_{\lambda}\chi_{\delta}\epsilon_{\mu\nu\alpha\beta}\epsilon^{\mu\nu\lambda\delta}\nabla^{\alpha}\chi^{\beta}\\
    &= -2\eta_{\lambda}\chi_{\delta}\delta^{[\lambda}\mathstrut_{\alpha}\delta^{\delta]}\mathstrut_{\beta}\nabla^{\alpha}\chi^{\beta}\\
    & =-2\eta_{\lambda}\chi_{\delta}\nabla^{[\lambda}\chi^{\delta]}\\
    & = 2\eta_{\lambda}\chi_{\delta}\nabla^{\delta}\chi^{\lambda}\\
    &= 2\kappa\eta_{\lambda}\chi^{\lambda}\\
    & = -2\kappa.
\end{aligned}
\end{equation}
Thus, one has
\begin{equation}
    \text{(I)}=\int_{\mathscr{H}}f\epsilon_{\mu\nu}=-2\int_{\mathscr{H}}\kappa\epsilon_{\mu\nu},
\end{equation}
but since a Kerr black hole is stationary, $\kappa$ is constant over $H$ and $\int_{\mathscr{H}}\epsilon_{\mu\nu}$ is just the area of the event horizon (see eq. \ref{areakerr}), hence, eq. \ref{mass2} reads
\begin{equation}\label{mass3}
    M=\frac{1}{4\pi}\kappa A+2\Omega L.
\end{equation}

Eq. \ref{mass3} is known as \textit{Smarr's formula}, and is the desired relation for the quantities of interest of a Kerr black hole. For a Kerr-Newman black hole, Smarr's formula has an extra term corresponding to the electric charge contribution, $q\Phi $, where $\Phi$ is the electric potential at the event horizon\footnote{Which can also be shown to be constant over $H$ \cite{Heusler1996}.}. Such a relation would be derived in the exact same manner, and the additional term would show in the mass, as the total mass of the spacetime would be equal to the mass of the black hole plus the contributions of the energy-momentum tensor in the electrovac region (see \cite{Heusler1996} for details). 

To proceed, one needs to perturb Smarr's formula, where one finds
\begin{equation}\label{Smarrper}
    \delta M=\frac{1}{4\pi}(\kappa\delta A+\delta\kappa A)+2(\delta\Omega L+\Omega \delta L),
\end{equation}
which has the unwanted variations $\delta\kappa$ and $\delta\Omega$. As the relation of interest for the perturbed version of eq. \ref{mass3} is for variations only on the mass, angular momentum and event horizon area, it is necessary to turn one's attention to perturbations of the Kerr metric. The idea is to consider two neighboring configurations of Kerr black holes \cite{Bardeen1973}, which are presented in more detail in appendix \ref{B2}. As such, the variation of quantities of two nearby stationary solutions is denoted $\delta$. 

To remove the unwanted terms in eq. \ref{Smarrper}, one starts by contracting eq. \ref{k7} with the auxiliary vector $\eta^{\mu}$, which yields
\begin{equation}\label{k16}
    \kappa=\frac{1}{2}\eta^{\mu}\nabla_{\mu}(\chi^{\nu}\chi_{\nu}).
\end{equation}
By perturbing eq. \ref{k16}, one obtains
\begin{equation}\label{k2345}
    \delta\kappa=\frac{1}{2}[(\delta\eta^{\mu})\nabla_{\mu}(\chi^{\nu}\chi_{\nu})+\underbrace{\eta^{\mu}\nabla_{\mu}(\chi_{\nu}\delta\chi^{\nu}+\chi^{\nu}\delta\chi_{\nu})}_{\text{(I)}}].
\end{equation}
Using eqs. \ref{chipsi2}, \ref{Beq2} and Killing's equation, one can deduce that (I) is equal to
\begin{equation}
    \begin{aligned}[b]
    \text{(I)} & =\eta^{\mu}\chi^{\nu}\nabla_{\mu}(\delta\chi_{\nu})+\eta^{\mu}(\delta\chi_{\nu})\nabla_{\mu}\chi^{\nu}+\eta^{\mu}\chi_{\nu}\nabla_{\mu}(\delta\chi^{\nu})+\eta^{\mu}(\delta\chi^{\nu})\nabla_{\mu}\chi_{\nu} \\
    & = \eta^{\mu}\chi^{\nu}\nabla_{\mu}(\delta\chi_{\nu})+\eta^{\mu}(\delta\chi_{\nu})\nabla_{\mu}\chi^{\nu}+\eta^{\mu}\chi_{\nu}\nabla_{\mu}(\delta\Omega\psi^{\nu})+\eta^{\mu}(\delta\Omega\psi^{\nu})\nabla_{\mu}\chi_{\nu}\\
    &= \eta^{\mu}\chi^{\nu}\nabla_{\mu}(\delta\chi_{\nu})+\eta^{\mu}(\delta\chi_{\nu})\nabla_{\mu}\chi^{\nu}+\eta^{\mu}\chi_{\nu}(\delta\Omega)\nabla_{\mu}\psi^{\nu}+\eta^{\mu}(\delta\Omega)\chi^{\nu}\nabla_{\mu}\psi_{\nu}\\
    &= 2\eta^{(\mu}\chi^{\nu)}\nabla_{\mu}(\delta\chi_{\nu})+\eta^{\mu}(\delta\chi_{\nu})\nabla_{\mu}\chi^{\nu}+2\eta^{\mu}\chi_{\nu}(\delta\Omega)\nabla_{\mu}\psi^{\nu}-\eta^{\nu}\chi^{\mu}\nabla_{\mu}(\delta\chi_{\nu})\\
    &= 2\eta^{(\mu}\chi^{\nu)}\nabla_{\mu}(\delta\chi_{\nu})+2\eta^{\mu}(\delta\chi_{\nu})\nabla_{\mu}\chi^{\nu}+2(\delta\Omega)\eta^{\mu}\chi_{\nu}\nabla_{\mu}\psi^{\nu}.
    \end{aligned}
\end{equation}
where the term $\eta^{\nu}\chi^{\mu}\nabla_{\mu}(\delta\chi_{\nu})$ was added and subtracted in the fourth line. This allows one to write eq. \ref{k2345} as
\begin{equation}\label{k17}
    \delta\kappa=\underbrace{(\delta\eta^{\mu})\chi_{\nu}\nabla_{\mu}\chi^{\nu}}_{\text{(II)}}+\underbrace{\eta^{(\mu}\chi^{\nu)}\nabla_{\mu}(\delta\chi_{\nu})}_{\text{(III)}}+\underbrace{\eta^{\mu}(\delta\chi_{\nu})\nabla_{\mu}\chi^{\nu}}_{\text{(IV)}}+(\delta\Omega)\eta^{\mu}\chi_{\nu}\nabla_{\mu}\psi^{\nu}.
\end{equation}

Eq. \ref{k17} can be simplified by noting that, as a consequence of eq \ref{Beq8}, one has
\begin{equation}
    \text{(II)}+\text{(IV)}=\nabla^{\mu}\chi^{\nu}((\delta\eta_{\mu})\chi_{\nu}+\eta^{\mu}(\delta\chi_{\nu}))=0.
\end{equation}
Similarly, using eqs. \ref{Beq3}, \ref{Beq9} and the transverse metric on $\mathscr{H}$, eq. \ref{transversemetric}, one can rewrite $\text{(III)}$, 
\begin{equation}
\begin{aligned}[b]
\eta^{(\mu}\chi^{\nu)}\nabla_{\mu}(\delta\chi_{\nu}) &= \frac{1}{2}(h^{\mu\nu}-g^{\mu\nu})\nabla_{\mu}(\delta\chi_{\nu})\\
& = \frac{1}{2}(\alpha h^{\mu\nu}\nabla_{\mu}\chi_{\nu}+h^{\mu\nu}\chi_{\nu}\nabla_{\mu}\alpha-\nabla^{\nu}(\delta\chi_{\nu}))\\
& = -\frac{1}{2}\nabla^{\nu}(\gamma_{\nu\mu}\chi^{\mu}+(\delta\Omega)\chi_{\nu})\\
& = -\frac{1}{2}(\gamma_{\nu\mu}\nabla^{\nu}\chi^{\mu}+\chi^{\mu}\nabla^{\nu}\gamma_{\nu\mu}+(\delta\Omega)\nabla^{\nu}\chi_{\nu})\\
& = -\frac{1}{2}\chi^{\mu}\nabla^{\nu}\gamma_{\nu\mu}\\
& = \frac{1}{2}\chi_{\mu}\nabla_{\nu}\gamma^{\nu\mu},
\end{aligned}    
\end{equation}
where the terms $h^{\mu\nu}\nabla_{\mu}\chi_{\nu}$ and $\gamma_{\nu\mu}\nabla^{\nu}\chi^{\mu}$ vanish by being the contraction of a symmetric tensor with an antisymmetric one, and $\nabla^{\nu}\chi_{\nu}=0$, as per Killing's equation. Thus, eq. \ref{k17} reads
\begin{equation}
    \delta\kappa=\frac{1}{2}\chi_{\mu}\nabla_{\nu}\gamma^{\nu\mu}+(\delta\Omega)\eta^{\mu}\chi_{\nu}\nabla_{\mu}\psi^{\nu},
\end{equation}
and by integrating over $\mathscr{H}$, one obtains
\begin{equation}\label{k18}
    \delta\kappa A=\frac{1}{2}\int_{\mathscr{H}}\epsilon_{\alpha\beta}\chi_{\mu}\nabla_{\nu}\gamma^{\nu\mu}+(\delta\Omega)\underbrace{\int_{\mathscr{H}}\epsilon_{\alpha\beta}\eta^{\mu}\chi_{\nu}\nabla_{\mu}\psi^{\nu}}_{\text{(V)}},
\end{equation}
where $\epsilon_{\mu\nu}$ is the volume element on $\mathscr{H}$. 

One can evaluate (V) by considering the scalar $f'$ that relates the two-form $\epsilon_{\mu\nu\alpha\beta}\nabla^{\alpha}\psi^{\beta}$ to its integrand,
\begin{equation}\label{123456}
     f' \epsilon_{\mu\nu\alpha\beta}\nabla^{\alpha}\psi^{\beta}=\epsilon_{\mu\nu}\eta^{\alpha}\chi_{\beta}\nabla_{\alpha}\psi^{\beta},
\end{equation}
which can be computed by applying $\epsilon^{\mu\nu}$ to both sides of eq. \ref{123456},
\begin{equation}
\begin{aligned}[b]
  f' \epsilon^{\mu\nu\sigma\rho}\eta_{\sigma}\chi_{\rho}\epsilon_{\mu\nu\alpha\beta}\nabla^{\alpha}\psi^{\beta}&=\epsilon^{\mu\nu}\epsilon_{\mu\nu}\eta_{\alpha}\chi_{\beta}\nabla^{\alpha}\psi^{\beta}, \\
 -4f'\delta^{[\sigma}\mathstrut_{\alpha}\delta^{\rho]}\mathstrut_{\beta} \eta_{\sigma}\chi_{\rho}\nabla^{\alpha}\psi^{\beta}&=2\eta_{\alpha}\chi_{\beta}\nabla^{\alpha}\psi^{\beta}, \\
-4f'\eta_{\sigma}\chi_{\rho}\nabla^{\sigma}\psi^{\rho}&=2\eta_{\alpha}\chi_{\beta}\nabla^{\alpha}\psi^{\beta}, \\
 f'&=-\frac{1}{2},
\end{aligned}    
\end{equation}
and from the Komar integral of the spacelike Killing vector, one has
\begin{equation}
    \text{(V)}=-\frac{1}{2}\int_{\mathscr{H}}\epsilon_{\mu\nu\alpha\beta}\nabla^{\alpha}\psi^{\beta}=-8\pi L.
\end{equation}
Thus, eq. \ref{k18} reduces to
\begin{equation}
    \underbrace{\int_{\mathscr{H}}\epsilon_{\alpha\beta}\chi_{\mu}\nabla_{\nu}\gamma^{\nu\mu}}_{\text{(VI)}}=2A\delta\kappa +16\pi L\delta\Omega.
\end{equation}

In order to proceed, it is necessary to make use of the general result stated by eq. \ref{generalresult}. Such a result is useful because one can evaluate (VI) by relating their integrands. More precisely, the scalar, $f''$, that relates the integrand of (VI) to that of the integral over $\mathscr{H}$ of eq. \ref{generalresult} is given by
\begin{equation}
\begin{aligned}[b]
f''\epsilon_{\mu\nu\alpha\beta}\xi^{\beta}\nabla_{\rho}(\gamma^{\alpha\rho}-g^{\alpha\rho}\gamma) &= \epsilon_{\mu\nu}\chi_{\alpha}\nabla_{\beta}\gamma^{\beta\alpha},\\
 f''\epsilon^{\mu\nu\rho\lambda}\eta_{\rho}\chi_{\lambda}\epsilon_{\mu\nu\alpha\beta}\xi^{\beta}\nabla_{\rho}(\gamma^{\alpha\rho}-g^{\alpha\rho}\gamma)&=\epsilon^{\mu\nu}\epsilon_{\mu\nu}\chi_{\alpha}\nabla_{\beta}\gamma^{\beta\alpha},\\
 -4f''\delta^{[\rho}\mathstrut_{\alpha}\delta^{\lambda]}\mathstrut_{\beta}\eta_{\rho}\chi_{\lambda}\xi^{\beta}\nabla_{\rho}(\gamma^{\alpha\rho}-g^{\alpha\rho}\gamma)&=2\chi_{\alpha}\nabla_{\beta}\gamma^{\beta\alpha},\\
 -2f''\eta_{[\alpha}\chi_{\beta]}\xi^{\beta}\nabla_{\rho}(\gamma^{\alpha\rho}-g^{\alpha\rho}\gamma)&=\chi_{\alpha}\nabla_{\beta}\gamma^{\beta\alpha},\\
 -f''\chi_{\alpha}\nabla_{\rho}(\gamma^{\alpha\rho}-g^{\alpha\rho}\gamma)&=\chi_{\alpha}\nabla_{\beta}\gamma^{\beta\alpha},\\
 -f''\chi_{\alpha}\nabla_{\rho}\gamma^{\rho\alpha}&=\chi_{\alpha}\nabla_{\beta}\gamma^{\beta\alpha},\\
 f''&=-1,
\end{aligned}    
\end{equation}
where $\eta^{\mu}\chi_{\mu}=-1$ and eqs. \ref{chixi}, \ref{Auxxi}, and \ref{Beq5} were used. Therefore, one has
\begin{equation}
    \underbrace{\int_S\epsilon_{\nu\alpha\beta\delta}\xi^{\delta}\nabla_{\rho}(\gamma^{\beta\rho}-g^{\beta\rho}\gamma)}_{\text{(VII)}}=-2A\delta\kappa-16\pi L\delta\Omega.
\end{equation}
Using eq. \ref{admmass}, (VII) is evaluated to be $8\pi \delta M$. Hence, one obtains a relation for the variation of the surface gravity, mass and angular velocity,
\begin{equation}\label{kappavar}
    \delta M=-\frac{1}{4\pi}A\delta\kappa-2L\delta\Omega,
\end{equation} 
which can be summed with eq. \ref{Smarrper} to yield the sought result. Restoring the constants, one finds
\begin{equation}\label{first}
    c^2\delta M=\frac{\kappa c^2}{8\pi G}\delta A+\Omega\delta L.
\end{equation}

Thus, the variation in first order of mass, area and angular momentum of a stationary uncharged black hole over a perturbation respects eq. \ref{first}. For a Kerr-Newman black hole, there is an extra term corresponding to the variation of electric charge, $\Phi \delta q$. Indeed, since all terms have units of energy, such an equation may simply be interpreted as an ``energy conservation law'' applied to the geometrical properties of a black hole. Finally, note that the geometrical interpretation for the total mass and angular momentum of a black hole follows from the fact that these quantities are defined from the Komar integrals. In other words, they rise from Killing vectors, which are derived from the metric, and thus, are a consequence of spacetime geometry.

The last property of interest of a stationary black hole in the classical framework that will be discussed is the physical plausibility of the extreme and fast Kerr cases. As noted in the derivation of the coordinate singularities, the extreme case has only one coordinate singularity, while the fast case has none. In the extreme case, the behavior of the hypersurface $\{r=r_+\}$ is very similar, as it acts as a way one membrane. However, the absence of a coordinate singularity in the fast case means that the fast Kerr metric does not describe a black hole, but rather, a naked singularity. 

Assuming that the cosmic censor conjecture holds, one would then conclude that the fast case is unphysical. In particular, since perturbations of the extreme case would also produce a naked singularity, one would also expect that the extreme case should not be physically accessible, in the sense that no continuous process applied to a slow Kerr black hole could produce an extreme Kerr black hole. This notion of attainability of a configuration that one would deem to be physical, but which could produce an unphysical one if it is perturbed, is very similar to the one discussed in \S\;\ref{causal} for the definition of a stably causal spacetime. In essence, since the surface gravity of an extreme Kerr black hole must vanish (see eq. \ref{kappakerr}), these ideas can be stated as the impossibility of reducing the surface gravity of a Kerr black hole to zero. Note that this is also the case for the surface gravity of an extreme Kerr-Newman or Reissner–Nordström ($r_s=2e$) black hole \cite{Wald1984}. Additionally, note that this line of reasoning is based only on the assumption that the cosmic censor conjecture should hold, but it has been shown that it is possible to reach extremal cases of Kerr black holes by a finite number of processes \cite{Farrugia1979} (although this is not a violation of the cosmic censor conjecture). Nonetheless, these exceptions to the attainability of $\kappa=0$ can be used to point to a more appropriate version of this property of the surface gravity \cite{Sullivan1980}, which one could argue to be nothing more than a particular case of the cosmic censor conjecture. Such a development was proposed in \cite{Israel1986} and takes the form of the theorem below.

\begin{theorem}\label{thekappa}
    Let $(M,g_{\mu\nu})$ be a strongly asymptotically predictable spacetime with $g_{\mu\nu}$ continuous and piecewise $C^3$, and $\Sigma_t$ denote Cauchy hypersurfaces for the globally hyperbolic region $\psi^{-1}[A']\subset M$. If $\Sigma_t$ contains trapped surfaces for all $t<t'$ but none for $t>t'$, then the weak energy condition does not hold in a neighborhood of the apparent horizon of $\Sigma_{t'}$.
\end{theorem}

Details on the proof of this theorem can be found in \cite{Israel1986} and references therein. Its content can be interpreted by considering that the slow Kerr black hole possesses trapped surfaces, while the extreme Kerr black hole does not (see \cite{ONeill1995} for details on the analytical extension of the Kerr metric and analysis of the expansion of the incoming and outgoing null congruences). In particular, one can visualize the “extremization” of a slow Kerr black hole as a series of continuous processes in which the black hole “captures” distributions of energy with parameters that contribute to reducing the inequality $r_s>2a$. Consequently, any absorption (excluding negative energy, as per the weak energy condition in theorem \ref{thekappa}) will “inflate” the inner apparent horizon, up to the point where it merges with the outer horizon. At that point, there would be no trapped surfaces, which would correspond to the extreme case. This is precisely what the result above restricts, which says that unless the weak energy condition is violated, it is not possible to “eliminate” all trapped surfaces by a finite amount of processes. 

In fact, the analysis of particular cases of the attainability of extreme black holes was considered in detail in \cite{Wald1974,Sorce2017}. In such an analysis, an explicit examination was made of the behavior of test point-like massive objects falling into a black hole, possessing parameters that, if absorbed, would extremize it or produce a naked singularity. More specifically, in the case of a rotating black hole, the gravitational spin interaction \cite{Wald1972} would result in a repulsion that would stop the body from entering the black hole region. Similarly, for the case of a charged black hole, electric repulsion would also play a role in restricting the absorption of charged bodies that would extremize the black hole.

\chapter{Semiclassical aspects of black holes}\label{chapter3}

The next step on the road to the black hole information problem is the effective emission of particles by black holes, which arises as a prediction from quantum field theory in curved spacetime. Although semiclassical gravity is expected to be only an approximation at the scales where spacetime structure can still be described classically, it is reasonable to anticipate that the predictions that arise from this limit of application of both of these highly successful theories will provide insights into quantum effects in gravitation as well as information regarding the development of an adequate theory of quantum gravity. Nevertheless, in extreme regions, such as the vicinity of a singularity or the early universe, in which the curvature scales are of order of $\ell_p^{-2}$ (see eq. \ref{planck}), the pertinent phenomena will have to be analyzed in light of a complete theory of quantum gravity. However, these extreme regimes will not appear for most of our analysis, and semiclassical gravity will be an adequate approximation to analyze most of the physics of black holes whose mass is much greater than the Planck mass (i.e., black holes that obey $r_s\gg\ell_p$).

We first present a precise description of the formalism of quantum field theory and a review of its basic elements in Minkowski spacetime. The first section of this chapter is devoted to the presentation of these concepts as well as their generalization to curved spacetime. In this context, we study the solutions of a free scalar field in Minkowski spacetime in order to develop the calculations in curved spacetime, while also discussing the ambiguity in the definition of a vacuum state for general spacetimes. Using these developments, we show how a dynamical gravitational field can lead to the creation of particles and study this effect in detail in the case of a Kerr black hole. Additionally, we analyze the classical properties of stationary black holes in light of these predictions and thermodynamic arguments. We then conclude the chapter by discussing the deep implication that entanglement (see appendix \ref{information} for a review of the necessary concepts that lead to this implication) has for quantum field theory and black holes. Namely, this conclusion follows from the assumption that any physically reliable quantum state must satisfy a condition in order for it to be possible to define a nonsingular ``energy-momentum expectation value''.

\section{Quantum field theory in curved spacetime}\label{qft}

In the \textit{Lagrange formalism}, a system of point-like massive bodies is characterized by a set of discrete generalized coordinates, $q_{\alpha}$, which are a representation of the degrees of freedom of the system. The \textit{principle of least action} leads one to the \textit{Euler-Lagrange equations}, which can be used to derive the equations of motion of bodies \cite{Greiner1996}. In \textit{classical field theory}, the object of study is now a \textit{field} (which will not necessarily be a scalar), denoted by $\psi(t,\textbf{x})$, where $\textbf{x}=(x^1,x^2,x^3)$. In particular, a \textit{classical} field is characterized by its value at each point of spacetime. Consequently, a system described by a field possesses infinite degrees of freedom, as the dynamical variables of the system are now the values of the field at each point of the manifold. An example of a classical field is the electromagnetic field, $A^{\mu}(t,\textbf{x})$, which is perhaps the best example in nature in which a phenomenon -- the electromagnetic interaction -- is fundamentally described by a field.

Let $\psi(t,\textbf{x})$ denote a classical field over Minkowski spacetime. The \textit{Lagrangian}, $L$, of the field can be written as a volume integral of a density function, known as the \textit{Lagrange density},
\begin{equation}
    L(t)=\int d^3x\;\mathcal{L}(\psi(t,\textbf{x}),\partial_{\mu}\psi(t,\textbf{x})),
\end{equation}
where the restrictions on the possible dependence of $\mathcal{L}$ on the field and its derivatives have been found to be adequate to describe phenomena represented by field theories \cite{Greiner1996}. The principle of least action,
\begin{equation}
    \delta\int_{t}^{t'} dt'' \int d^3x\;\mathcal{L}(\psi(t,\textbf{x}),\partial_{\mu}\psi(t,\textbf{x}))=0,
\end{equation}
leads one to the \textit{Euler-Lagrange field equations},
\begin{equation}
    \frac{\partial \mathcal{L}}{\partial\psi}-\partial_{\mu}\left(\frac{\partial \mathcal{L}}{\partial(\partial_{\mu}\psi)}\right)=0.
\end{equation}

The Hamiltonian formalism can also be applied to a field theory by defining the \textit{canonically conjugate field},
\begin{equation}
    \pi(\textbf{x},t)=\frac{\partial \mathcal{L}}{\partial(\partial_t{\psi})},
\end{equation}
in which the \textit{Hamiltonian} can be written as an integral of the \textit{Hamiltonian density}, $\mathscr{H}(t,\textbf{x})$,
\begin{equation}
    H(t)=\int d^3x\;\mathscr{H}(t,\textbf{x}),\quad \mathscr{H}(t,\textbf{x})=\pi(t,\textbf{x})\partial_t{\psi}(t,\textbf{x})-\mathcal{L}(\psi(t,\textbf{x}),\partial_{\mu}\psi(t,\textbf{x})).
\end{equation}
By requiring that a variation of $H$ vanishes, one obtains \textit{Hamilton's equations of motion},
\begin{equation}
    \partial_t{\psi}(t,\textbf{x})=\frac{\delta H}{\delta \pi},\quad \partial_t{\pi}(t,\textbf{x})=-\frac{\delta H}{\delta \psi}.
\end{equation}

Lastly, by defining the \textit{Poisson brackets} for the functionals (see, e.g., \cite{Greiner1996}) $A[\phi,\pi]$ and $B[\phi,\pi]$, 
\begin{equation}
    \{A,B\}=\int d^3x\;\left(\frac{\delta A}{\delta\phi(x)}\frac{\delta B}{\delta\pi(x)}-\frac{\delta A}{\delta\pi(x)}\frac{\delta B}{\delta\phi(x)}\right),
\end{equation}
one can find relations for the field and canonically conjugate field,
\begin{equation}
    \{\psi(t,\textbf{x}),\pi(t,\textbf{x}')\}=\delta^3(\textbf{x}-\textbf{x}'),
\end{equation}
\begin{equation}
    \{\psi(t,\textbf{x}),\psi((t,\textbf{x}')\}=\{\pi(t,\textbf{x}),\pi(t,\textbf{x}')\}=0,
\end{equation}
where $\delta^3(\textbf{x}-\textbf{x}')$ is the three-dimensional Dirac delta distribution \cite{Butkov1973}.

The process of quantization of a classical field is analogous to that of a classical harmonic oscillator. Namely, the classical field and the canonically conjugate field are promoted to operators, $\hat{\psi}(t,\textbf{x})$ and $\hat{\pi}(t,\textbf{x})$, which are postulated to satisfy the \textit{equal-time commutation relations}
\begin{equation}\label{equal1}
    [\hat{\psi}(t,\textbf{x}),\hat{\pi}(t,\textbf{x}')]=i\hbar\delta^3(\textbf{x}-\textbf{x}')\hat{I},
\end{equation}
\begin{equation}\label{equal2}
    [\hat{\psi}(t,\textbf{x}),\hat{\psi}(t,\textbf{x}')]=[\hat{\pi}(t,\textbf{x}),\hat{\pi}(t,\textbf{x}')]=0,
\end{equation}
where $\hat{I}$ is the identity operator of the pertinent Hilbert space. By postulating these relations, it is implied that the field is a \textit{bosonic} one. In contrast, a \textit{fermionic} field would be postulated to obey corresponding anticommutator relations. The reason behind these nomenclatures will become clear below. In either case, the self-adjoint operators, $\hat{\psi}(t,\textbf{x})$ and $\hat{\pi}(t,\textbf{x})$, act on the possible quantum states of the system, which are elements of the pertinent Hilbert space.

These ideas can be straightforwardly generalized to curved spacetimes which are globally hyperbolic, simply by identifying the field as a tensor over $M$, using the covariant volume element in the integrals and applying the Levi-Civita connection instead of the ordinary partial derivative. In particular, the globally hyperbolic property is necessary in order to ensure\footnote{Nevertheless, one can argue that suitable conditions on spacelike hypersurfaces that fail to be a Cauchy hypersurface should suffice for the evolution of the field to be a ``well posed'' problem (see \S~\ref{evaporation}).} that the evolution of the field is a ``well posed'' problem (see \S~\ref{causal}). The development of this work is based on a free neutral scalar field, whose Lagrangian reads 
\begin{equation}\label{lagkg}
    \mathcal{L}=-\frac{1}{2}\sqrt{-g}\left[\nabla^{\mu}\psi\nabla_{\mu}\psi-\left(\frac{m^2c^2}{\hbar^2}+\xi R\right)\psi^2\right],
\end{equation}
where $m$ is the field mass, $R$ is the Ricci scalar and $\xi$ a real constant. By ``free'', it is meant that the field does not interact with itself or other fields, except with the spacetime background, i.e., by means of gravitational interaction. Interaction with other fields would yield a scalar potential term in the Lagrangian, and the gravitational interaction is made explicit by the presence of the Levi-Civita connection and the term $\xi R\psi^2$. This last term is present due to the fact that when $m=0$ and $\xi=1/6$, the action is invariant under conformal transformations\footnote{Another reason for including this term is due to the \textit{renormalization} of the theory when one considers an interaction term such as $\lambda\psi^4$ \cite{Parker2009}.}. Hence, one can argue that, in the general case, simply taking $\xi=0$ is not necessarily the most physically adequate choice. Nevertheless, for developments of interest, this term will not affect the results, and for simplicity, it will be considered that the field is \textit{minimally coupled}, i.e., $\xi=0$. Applying the Lagrangian of eq. \ref{lagkg} in the Euler-Lagrange Field equations yields the \textit{Klein-Gordon} equation,
\begin{equation}\label{kleingordon}
    \left(\nabla^{\mu}\nabla_{\mu}+\frac{m^2c^2}{\hbar^2}+\xi R\right)\psi=0.
\end{equation}
Further details on this equation, as well as how it can be deduced that it describes spinless particles can be found in \cite{Greiner2000}. In the developments that follow, we will adopt geometrized units.

We will use the Heisenberg picture, so that the dynamics of the system is given entirely by the evolution of the fundamental operator of the system, the field operator. In order to study how the field operator acts on the vector states, it is useful to first consider ``classical'' (i.e., scalar) solutions of the Klein-Gordon equation. Let $f_a$ and $f_b$ be scalar solutions of the Klein-Gordon equation. Consider the vector 
\begin{equation}
    \ell^{\mu}=f_a^*\nabla^{\mu}f_b-f_b\nabla^{\mu}f_a^*,
\end{equation}
where $^*$ represents complex conjugation. One can readily verify that $\ell^{\mu}$ respects the covariant form of the conservation equation, and, if the solutions vanish as $r\to\infty$, following the same line of reasoning as the one that led to the Komar integral in \S\;\ref{symmetry}, one can show that the conserved charge associated with the conserved current $\ell^{\mu}$ is simply the integral of $\epsilon_{\mu\nu\alpha\beta}\ell^{\beta}$ over a spacelike hypersurface, $\Sigma$. Consequently, this conserved quantity can be used to define the \textit{Klein-Gordon inner product}, $(f_a,f_b)$, as
\begin{equation}\label{inner}
    (f_a,f_b)=-i\int_{\Sigma} \epsilon_{\mu\nu\alpha\beta}\left(f_a^*\nabla^{\beta}f_b-f_b\nabla^{\beta}f_a^*\right),
\end{equation}
which, clearly, is independent of the choice of $\Sigma$. It is easy to see that this inner product is linear with respect to the second argument and antilinear with respect to the first. One can also readily verify that
\begin{equation}\label{inner1}
(f_a,f_b)^*=(f_b,f_a)=-(f_a^*,f_b^*).
\end{equation}

Consider, now, a solution of the minimally coupled Klein-Gordon equation in Minkowski spacetime which takes the form
\begin{equation}\label{KG1234}
    f_p=c_p e^{ip_{\mu}x^{\mu}},
\end{equation}
with $p^{\mu}=(\omega_p,\textbf{p})$, $\textbf{p}=(p^{x^1},p^{x^2},p^{x^3})$, $\omega_p=\sqrt{|\textbf{p}|+m^2}$ and $c_p$ a constant. Note that the solution index, $p$, is three dimensional, being associated with the \textit{mode} $\textbf{p}$, but the vector notation in the index has been removed for clearer notation. A solution is said to be a \textit{positive frequency solution} if
\begin{equation}\label{timeKG}
    \xi^{\mu}\nabla_{\mu}f_p=-i\omega_p f_p\;\text{for }\omega_p>0,
\end{equation}
where $\xi^{\mu}$ is a timelike Killing vector of Minkowski spacetime. The complex conjugate of a positive frequency solution is said to be a \textit{negative frequency solution}. It should be noted that a more rigorous nomenclature would be positive or negative norm solution (with a norm given by the Klein-Gordon inner product), but we will adopt the customary nomenclature of frequencies to identify solutions.

In this manner, the most general solution of the minimally coupled Klein-Gordon equation is a linear combination of positive frequency solutions, $f_p$, and negative frequency solutions, $f_p^*$, i.e., $\{f_p,f_p^*\}$ comprises a basis of the vector space of solutions. The structure of an inner product gives rise to a notion of orthogonality between the elements of a basis, namely,
\begin{equation}\label{inner2}
    (f_p,f_{p'})=\delta^3(\textbf{p}-\textbf{p}'),
\end{equation}
\begin{equation}\label{inner3}
    (f_p,f_{p'}^*)=0.
\end{equation}
Eq. \ref{inner2} can then be used to find $c_p=(2\omega_p)^{-1/2}(2\pi)^{-3/2}$. The expansion of the field operator, $\hat{\psi}$, in terms of this basis is given by
\begin{equation}\label{exp1234}
    \hat{\psi}=\int^{\infty}_{-\infty} d^3p\left(\hat{a}_pf_p+\hat{a}^{\dag}_pf_p^*\right),
\end{equation}
where $\hat{a}_p$ and $\hat{a}_p^{\dag}$ are the operators in the expansion of $\hat{\psi}$ in terms of the basis $\{f_p,f_p^*\}$. Evidently, these operators can also be written as
\begin{equation}
    \hat{a}_p=(f_p,\hat{\psi}),
\end{equation}
\begin{equation}
    \hat{a}^{\dag}_p=-(f^*_{p},\hat{\psi}).
\end{equation}
Additionally, from the orthogonality of solutions, eqs. \ref{inner2} and \ref{inner3}, it also follows the commutation relations,
\begin{equation}\label{commua}
    [\hat{a}_p,\hat{a}^{\dag}_{p'}]=\delta^3(\textbf{p}-\textbf{p}')\hat{I},
 \end{equation}
\begin{equation}\label{commua1}
    \left[\hat{a}_p,\hat{a}_{p'}\right]=[\hat{a}^{\dag}_p,\hat{a}^{\dag}_{p'}]=0.
\end{equation}

To interpret these operators, it is convenient to restrict the solutions $f_p$ to a volume with periodic boundary conditions. This can be done by considering that the system is confined to a cube of side $L$, such that the solution given in eq. \ref{KG1234} now takes the form
\begin{equation}
    f_p=(2L^3\omega_p)^{-1/2} e^{ip_{\mu}x^{\mu}},
\end{equation}
where
\begin{equation}
    p^{a}=\frac{2\pi n}{L},\quad n=0,\pm 1,\pm 2,...,\quad a=x^1,x^2,x^3.
\end{equation}
By doing so, the Dirac delta appearing in the orthogonality condition becomes a Kronecker delta, which then results in
\begin{equation}\label{inner25}
    [\hat{a}_p,\hat{a}^{\dag}_{p'}]=\delta_{\textbf{p}\textbf{p}'}.
\end{equation}
Now, returning to a classical perspective and using a spatial Fourier decomposition, it is possible to see that the minimally coupled Klein-Gordon equation implies that the complex coefficients, $a_p$, individually satisfy the harmonic oscillator equation. In analogy with the quantization of such a system, one can then identify the operators $\hat{a}_p$ and $\hat{a}_p^{\dag}$ in eq. \ref{exp1234} as \textit{annihilation} and \textit{creation} operators for a particular mode, respectively. Hence, one can define the \textit{vacuum} state, $|0\rangle$, to be the one which is ``destroyed'' by the action of any $\hat{a}_p$, i.e.,
\begin{equation}\label{vacuum}
    \hat{a}_p|0\rangle=0,\;\forall\;\textbf{p},
\end{equation}
and also being normalized as $\langle 0|0\rangle=1$. 

One can also define the \textit{number operator} to be
\begin{equation}
    \hat{N}_p=\hat{a}_p^{\dag}\hat{a}_p,
\end{equation}
which measures the number of ``quanta'' or ``excitations'' in the mode $\textbf{p}$. Following the \textit{Fock representation} (i.e., states are represented by their ``quanta'' content), the Hilbert space of a system is represented using the eigenvalues of the number operator. Note that this is viable due to the fact that any of its elements can be constructed by acting $\hat{a}_p^{\dag}$ on $|0\rangle$. In particular, an orthonormal basis of the Hilbert space is given by the set of vectors that can be written as
\begin{equation}\label{statevec}
    |\alpha_{p_1},\beta_{p_2},\ldots\rangle=(\alpha!\beta!\ldots)^{-1/2}(\hat{a}^{\dag}_{p_1})^{\alpha}(\hat{a}^{\dag}_{p_2})^{\beta}\ldots|0\rangle,
\end{equation}
where an arbitrary vector such as the one in eq. \ref{statevec} can be interpreted to represent a state with $\alpha$ quanta in mode $\textbf{p}_1$, $\beta$ quanta in mode $\textbf{p}_2$ and so forth. This is precisely the interpretation that quantum field theory provides of ``particles''\footnote{In the following, we shall drop the quote unquote in the word particles, but the reader should recall that this is merely a compact way of referring to the ``excitations'' of a quantum field.}, in other words, they are merely ``excitations'' of the quantum field. 

Hence, in this interpretation lies the root of the field nomenclature given above for the commutation relations. Namely, a field whose quanta obey Bose-Einstein statistics \cite{Sakurai1994} will be a bosonic one, while one whose quanta obey Fermi-Dirac statistics will be a fermionic one \cite{Parker2009}. Additionally, note that the ``confinement'' of the field in a box in order to obtain the proper delta in the commutation relation, eq. \ref{inner25}, is merely a mathematical procedure. In particular, one may always perform this process in order to discretize the index of the solutions, which can also be done by constructing an orthonormal set of wave packets from $f_p$. Conversely, one may always take the limit $L\to\infty$ in order to reinstate the continuous index. Finally, the interpretation of these operators generalizes straightforwardly to curved spacetime, as the field operator can always be expanded in the form given in eq. \ref{exp1234} when one considers a prescription to identify positive frequency solutions.

It is useful to investigate how two bases of the vector space of solutions of the Klein-Gordon (not necessarily the minimally coupled) equation, $\{f_{\omega},f_{\omega}^*\}$ and $\{u_{\omega},u_{\omega}^*\}$, relate. Namely, it is evident that the construction of annihilation and creation operators is directly connected with the choice of basis, and one wishes to analyze how such operators can be expressed in terms of one another. To do so, note that $\hat{\psi}$ can be written in terms of each basis independently, 
\begin{equation}
\begin{aligned}[b]
    \hat{\psi} & =\int d\omega(\hat{a}_{\omega}f_{\omega}+\hat{a}^{\dag}_{\omega}f_{\omega}^*)\\
    & =\int d\omega(\hat{b}_{\omega}u_{\omega}+\hat{b}^{\dag}_{\omega}u_{\omega}^*),
\end{aligned}
\end{equation}
where the index $\omega$ should be interpreted as the set of pertinent indices of the system, i.e., it represents a number of continuous and/or discrete indices, and the integral is merely a symbol to denote the pertinent sums. Evidently, one can also write $u_{\omega}$ in terms of $\{f_{\omega},f_{\omega}^*\}$,
\begin{equation}\label{bogo}
u_{\omega}=\int d\omega'(\alpha_{\omega\omega'}f_{\omega'}+\beta_{\omega\omega'}f_{\omega'}^*).
\end{equation}
The coefficients $\alpha_{\omega\omega'}$ and $\beta_{\omega\omega'}$ are known as \textit{Bogolubov coefficients} and the change of basis represented by them is known as a \textit{Bogolubov transformation}. 

By applying $(f_{\omega'},\cdot)$ and $(f^*_{\omega'},\cdot)$ to eq. \ref{bogo}, one finds
\begin{equation}\label{coef1}
    \alpha_{\omega\omega'}=(f_{\omega'},u_{\omega}),
\end{equation}
\begin{equation}\label{coef2}
    \beta_{\omega\omega'}=-(f^*_{\omega'},u_{\omega}).
\end{equation}
Using the normalization conditions and the properties of the inner product, one can then readily verify that the inverse transformation reads 
\begin{equation}
    f_{\omega}=\int d\omega'(\alpha_{\omega'\omega}^*u_{\omega'}-\beta_{\omega'\omega}u_{\omega'}^*).
\end{equation}
The coefficients can also be used to relate the operators associated with each basis by using that $\hat{b}_{\omega}=(u_{\omega},\hat{\psi})$ and $\hat{a}_{\omega}=(f_{\omega},\hat{\psi})$, which yield
\begin{equation}\label{bogo1}
    \hat{b}_{\omega}=\int d\omega'(\alpha^*_{\omega\omega'}\hat{a}_{\omega'}-\beta^*_{\omega\omega'}\hat{a}_{\omega'}^{\dag}),
\end{equation}
\begin{equation}\label{bogo2}
    \hat{a}_{\omega}=\int d\omega'(\alpha_{\omega'\omega}\hat{b}_{\omega'}+\beta_{\omega'\omega}^*\hat{b}_{\omega'}^{\dag}).
\end{equation}
By applying $(u_{\omega'},\cdot)$ and $(u^*_{\omega'},\cdot)$ to eq. \ref{bogo}, one finds that the coefficients also respect
\begin{equation}\label{coeff1}
    \int d\omega''(\alpha_{\omega\omega''}\alpha^*_{\omega'\omega''}-\beta_{\omega\omega''}\beta^*_{\omega'\omega''})=\delta({\omega-\omega'}),
\end{equation}
\begin{equation}
     \int d\omega''(\alpha_{\omega\omega''}\beta_{\omega'\omega''}-\beta_{\omega\omega''}\alpha_{\omega'\omega''})=0.
\end{equation}

Note that there are infinitely many choices of basis of solutions of the Klein-Gordon equation, and since the vacuum state is defined by identifying the positive frequency solution on a basis, one might ask if the definition of the vacuum state is dependent on the choice of basis. First, note that in order to label a solution as a positive frequency one in a ``natural'' manner, it is necessary for spacetime to possess a timelike Killing vector, and to evaluate the condition given by eq. \ref{timeKG}. Thus, adopting a different basis but maintaining the prescription to split solutions into positive and negative ones will yield the same vacuum state. This can be readily verified by noting that, from eqs. \ref{bogo1} and \ref{bogo2}, the creation operators defined by different bases are related by the coefficients $\alpha_{\omega\omega'}$, but they may also be related to the annihilation operators if the coefficients $\beta_{\omega\omega'}$ do not vanish. Now, if the creation operators defined by a basis are a linear combination of creation and annihilation operators defined by another basis, then it is clear that the vacuum state defined by \ref{vacuum} will be basis dependent, i.e., there will be a disagreement in the measurement of ``quanta'' of the vacuum state defined by a given basis. 

This disagreement will be derived explicitly in the next section, but to see this in more detail now, consider the vacuum state as defined in an inertial frame of reference, $\mathscr{A}$, in Minkowski spacetime. Due to the static nature of Minkowski spacetime, there is a timelike Killing vector that can be used to define the vacuum state. It is straightforward to deduce \cite{Dowker} that the positive frequency solutions in any other inertial frame which is related to $\mathscr{A}$ by an orthochronous Lorentz transformation can be written as a linear combination of the positive frequency solutions in $\mathscr{A}$. Thus, the coefficients $\beta_{\omega\omega'}$ vanish in the Bogolubov transformation relating the bases of solutions, which means that the definition of the vacuum state is independent of the choice of inertial frame of reference. This ``natural'' vacuum state that arises due to the static property of Minkowski spacetime is known as \textit{static vacuum}\footnote{The definition of a ``natural'' vacuum state associated with a timelike Killing vector does not become more ``special'' if such a vector is also hypersurface orthogonal. In other words, the argumentation presented would not change if the spacetime in consideration were only stationary.}. However, uniformly accelerated observers do not perceive the static vacuum as one with no quanta. As we will see in the next section, the nature of this phenomenon lies in the fact that positive frequency solutions in such a non-inertial frame are related to positive \textit{and} negative frequency solutions in an inertial frame. In fact, this is known as the \textit{Fulling-Davies-Unruh effect} \cite{Crispino2007}, which states that a uniformly accelerated observer perceives the static vacuum as a thermal state. 

Finally, for the discussion of the propagation of the scalar field in spacetimes, it will be necessary to consider the \textit{geometric optics approximation}. This approximation is of use when the wavelength of a wave propagating in spacetime is orders of magnitude smaller than the typical length scale of the analysis. Namely, the radius of curvature of the spacetime\footnote{The typical radius of the curvature of the spacetime can be evaluated by $\left(R^{\mu\nu\alpha\beta}R_{\mu\nu\alpha\beta}\right)^{-1/4}$.} and the length over which the amplitude of the wave varies. For example, if one considers a wave of the form $\psi=A(x^{a})e^{iS(x^a)}$, then the wave equation (which is also the massless minimally coupled Klein-Gordon equation) under the assumption that $\nabla_{\mu}S\gg\nabla_{\mu}A$ implies that \cite{Fabbri2005}
\begin{equation}\label{optics}
    \nabla_{\mu}S\nabla^{\mu}S=0.
\end{equation}
Since the propagation of a wave is given by the curves normal to the surfaces of constant phase, $S$, eq. \ref{optics} can be interpreted as stating that the wave vector is null. Moreover, following the same line of reasoning as the remarks presented below theorem \ref{frobenius}, the propagation of the wave must be along null geodesics. Thus, in order to study the propagation of the solutions of the massless minimally coupled Klein-Gordon equation under the geometric optics approximation, it suffices to know the behavior of the null geodesics of spacetime.

\section{Effective particle creation by black holes}\label{creation}

The content of the effective particle creation effect by black holes is directly related to the failure of observers in different regions of spacetime to agree on a definition of a vacuum state. In particular, because of the necessity of a timelike Killing vector to define a \textit{preferred} vacuum state, perhaps the most notable difference between quantum field theory in Minkowski spacetime and general spacetimes is the ambiguity in the definition of a \textit{physical} vacuum state. This is a consequence of the fact that without a ``natural'' definition of positive frequency solutions, i.e., a ``natural'' symmetry, the vacuum state that arises from any choice of basis has no physical meaning. 

Nonetheless, for stationary spacetimes, or merely those that possess stationary regions, it is possible to use the timelike Killing vector associated with the time translation symmetry to have a preferred selection of positive frequency solutions. In the context of black holes, this is the fundamental aspect behind the process of effective particle creation. In essence, in a spacetime possessing a black hole that resulted from gravitational collapse, one expects that at ``early times''\footnote{The way ``early time'' is referred to here is in the same sense as the one presented in the stationary state conjecture (see \S~\ref{kerr}). More precisely, one identifies ``early times'' by means of Cauchy hypersurfaces. The reader should recall that these time comparisons are related to a ``slicing'' of the spacetime in Cauchy hypersurfaces.} the energy distribution that gave rise to the black hole should be described by a metric that is asymptotically flat (in the sense that the spacetime has a region similar to $\mathscr{I}^-$), and thus, stationary. Following considerations of the stationary final state conjecture, at late times after the formation of the black hole, the spacetime will also be described by an asymptotically flat metric, so that it is possible to use the prescription given by the timelike Killing vector in $\mathscr{I}^-$ and $\mathscr{I}^+$ to compare the vacuum state defined in them. Such a comparison is simply made through the Bogolubov coefficients relating a basis in each of these regions. In this manner, one finds that the vacuum state in past null infinity, $|0\rangle_{\mathscr{I}^-}$, will not be equal to the vacuum state in future null infinity, $|0\rangle_{\mathscr{I}^+}$, as illustrated in fig. \ref{fig:qft1}. A detailed analysis of this disagreement between definitions of the vacuum state will be made for a black hole that was the result of gravitational collapse. Natural units (see ch. \ref{Introduction}) will be adopted throughout this development.

\begin{figure}[h]
\centering
\includegraphics[scale=1.25]{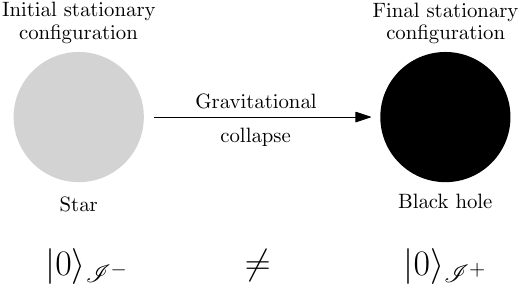} 
\caption{Change in the definition of the vacuum state due to the process of gravitational collapse.}
\caption*{Source: Adapted from FABBRI; NAVARRO-SALAS \cite{Fabbri2005}.}
\label{fig:qft1}
\end{figure}

As discussed, it will be assumed that the spacetime possesses an early time and, at least, distant stationary region, $\mathscr{I}^-$, and due to the stationary state conjecture, a late time stationary region, $\mathscr{I}^+$. The analysis starts by considering linear independent sets of solutions in both regions. Let $\{f_{\omega},f_{\omega}^*\}$ denote such a set on $\mathscr{I}^-$, and $\{u_{\omega},u_{\omega}^*\}$ one on $\mathscr{I}^+$, where the index $\omega$ should be understood to represent the pertinent number of continuous and/or discrete indices. Evidently, one can write the field operator unambiguously in terms of $\{f_{\omega},f_{\omega}^*\}$ if it is complete, where the basis is defined using the timelike Killing vector in $\mathscr{I}^-$, 
\begin{equation}\label{7890}
    \hat{\psi}=\int d\omega(\hat{a}_{\omega}f_{\omega}+\hat{a}^{\dag}_{\omega}f_{\omega}^*),
\end{equation}
but the corresponding expansion does not hold for $\{u_{\omega},u_{\omega}^*\}$. This is a consequence of the fact that all the null geodesics of spacetime have a past endpoint on $\mathscr{I}^-$, but the same cannot be said for $\mathscr{I}^+$. In particular, there are null geodesics that never reach $\mathscr{I}^+$ because they enter the black hole region. Thus, it is actually $H\cup\mathscr{I}^+$, where $H$ denotes the event horizon, that is the hypersurface that all null geodesics encounter. Namely, in order to expand $\hat{\psi}$ in terms of $\{u_{\omega},u_{\omega}^*\}$, it is necessary to also specify a set of solutions at $H$ and use both sets simultaneously.

Let $\{u_{\omega},u_{\omega}^*\}$ denote a complete set of solutions on $\mathscr{I}^+$ with zero Cauchy data on $H$, i.e., they vanish on $H$ as well as their gradients, and let $\{s_{\omega},s_{\omega}^*\}$ denote a complete set of solutions on $H$ with zero Cauchy data on $\mathscr{I}^+$ (see fig. \ref{fig:qft27}). Then $\{u_{\omega},u_{\omega}^*,s_{\omega},s_{\omega}^*\}$ is a basis of the vector space of solutions. Hence, the field operator can be written as
\begin{equation}
    \hat{\psi}=\int d\omega(\hat{b}_{\omega}u_{\omega}+\hat{c}_{\omega}s_{\omega}+\hat{b}^{\dag}_{\omega}u^*_{\omega}+\hat{c}^{\dag}_{\omega}s^*_{\omega}).
\end{equation}
Note that the split of $\{u_{\omega},u_{\omega}^*\}$ is unambiguous due to the timelike Killing vector in $\mathscr{I}^+$, but the same cannot be said for $\{s_{\omega},s_{\omega}^*\}$, as on $H$ there is no timelike Killing vector.\begin{figure}[h]
\centering
\includegraphics[scale=1.25]{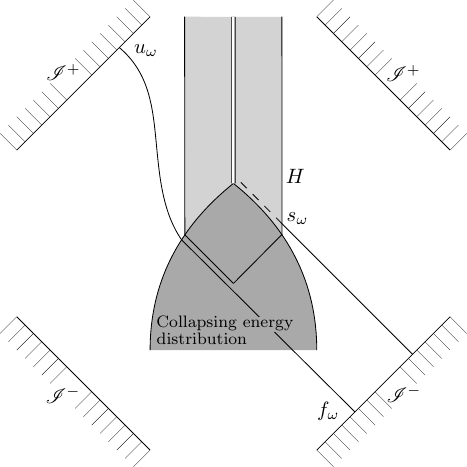} 
\caption{Evolution of field modes in the process of black hole formation.}
\caption*{Source: Adapted from HAWKING \cite{Hawking1976}.}
\label{fig:qft27}
\end{figure} This means that there is no natural definition of a vacuum state at $H$, but this is not relevant to the calculation of the particle creation effect \cite{Wald1975}, as the analysis of the vacuum state is only in $\mathscr{I}^+$. Each of the sets $\{u_{\omega},u_{\omega}^*\}$ and $\{s_{\omega},s_{\omega}^*\}$ obeys its respective orthogonality relations individually, but since they are also in disjoint regions at late times, one must have
\begin{equation}(s_{\omega},u_{\omega'})=(s_{\omega},u^*_{\omega'})=0.
\end{equation}
To proceed, it is necessary to analyze the Bogolubov transformation relating the bases $\{u_{\omega},u_{\omega}^*,s_{\omega},s_{\omega}^*\}$ and $\{f_{\omega},f_{\omega}^*\}$. In particular, the relation of interest is the one between the operator $\hat{b}_{\omega}$ and the operators $\hat{a}_{\omega}$ and $\hat{a}_{\omega}^{\dag}$. This relation can then be used to construct the number operator on $\mathscr{I}^+$ and act it on the state $|0\rangle_{\mathscr{I}^-}$, which will give information about how the vacuum state of $\mathscr{I}^-$ is perceived at $\mathscr{I}^+$. 

The expansion of $u_{\omega}$ in terms of $\{f_{\omega},f_{\omega}^*\}$ is given by eq. \ref{bogo}, and the operator $\hat{b}_{\omega}^{\dag}$ relates to the operators $\hat{a}_{\omega}$ and $\hat{a}_{\omega}^{\dag}$ by eq. \ref{bogo1}. Taking the conjugate transpose of eq. \ref{bogo1} yields
\begin{equation}\label{bogotrans}
    \hat{b}_{\omega}^{\dag}=\int d\omega'(\alpha_{\omega\omega'}\hat{a}^{\dag}_{\omega'}-\beta_{\omega\omega'}\hat{a}_{\omega'}).
\end{equation}
From eqs. \ref{bogo1} and \ref{bogotrans}, one can then construct the number operator on $\mathscr{I}^+$, $\hat{N}_{\omega}=\hat{b}_{\omega}^{\dag}\hat{b}_{\omega}$, which reads
\begin{multline}\label{Nbqft}
    \hat{N}_{\omega}=\int d\omega'\int d\omega''\left(\alpha_{\omega\omega''}\alpha^*_{\omega\omega'}\hat{a}_{\omega''}^{\dag}\hat{a}_{\omega'}-\alpha_{\omega\omega''}\beta^*_{\omega\omega'}\hat{a}_{\omega''}^{\dag}\hat{a}_{\omega'}^{\dag}-\alpha^*_{\omega\omega'}\beta_{\omega\omega''}\hat{a}_{\omega''}\hat{a}_{\omega'}\right.\\ \left. +\beta_{\omega\omega''}\beta^*_{\omega\omega'}\hat{a}_{\omega''}\hat{a}_{\omega'}^{\dag}\right) .
\end{multline}
Using this construction, it is possible to evaluate $_{{\mathscr{I}^-}}\langle0|\hat{N}_{\omega}|0\rangle_{\mathscr{I}^-}$. The action of each of the individual operators in \ref{Nbqft} follows from their interpretation of creation and annihilation operators individually for each mode (see, e.g., \cite{Sakurai1994} for a deduction of the action of these operators for a simple harmonic oscillator). Hence, from the condition that the state vectors given in eq. \ref{statevec} form an orthonormal basis, one finds
\begin{equation}
    _{{\mathscr{I}^-}}\langle0|\hat{a}_{\omega''}^{\dag}\hat{a}_{\omega'}|0\rangle_{\mathscr{I}^-}=0,
\end{equation}
\begin{equation}
    _{{\mathscr{I}^-}}\langle0|\hat{a}_{\omega''}^{\dag}\hat{a}_{\omega'}^{\dag}|0\rangle_{\mathscr{I}^-}=0,
\end{equation}
\begin{equation}
    _{{\mathscr{I}^-}}\langle0|\hat{a}_{\omega''}\hat{a}_{\omega'}|0\rangle_{\mathscr{I}^-}=0,
\end{equation}
\begin{equation}
    _{{\mathscr{I}^-}}\langle0|\hat{a}_{\omega''}\hat{a}_{\omega'}^{\dag}|0\rangle_{\mathscr{I}^-}={}_{{\mathscr{I}^-}}\langle 1_{\omega''}|1_{\omega'}\rangle_{\mathscr{I}^-}=\delta(\omega'-\omega'').
\end{equation}
Thus,
\begin{equation}\label{par1}
    _{{\mathscr{I}^-}}\langle0|\hat{N}_{\omega}|0\rangle_{\mathscr{I}^-}=\int d\omega'|\beta_{\omega\omega'}|^2.
\end{equation}
As discussed at the end of \S~\ref{qft}, the information about how the vacuum state in $\mathscr{I}^-$ is perceived by observers in the future asymptotic region is given by the coefficients $\beta_{\omega\omega'}$. Indeed, eq. \ref{par1} shows that if the coefficients $\beta_{\omega\omega'}$ do not vanish, then the state $|0\rangle_{\mathscr{I}^-}$ will not be perceived as one with no particles at $\mathscr{I}^+$. In other words, the failure of the positive frequency solutions in $\mathscr{I}^-$ to evolve to linear combinations of \textit{only} positive frequency solutions in $\mathscr{I}^+$ is precisely the reason why these two regions have different definitions of vacuum states. 

Therefore, to evaluate the content of the perceived state in $\mathscr{I}^+$, it suffices to find an explicit relation to the integral on the right hand side of eq. \ref{par1}. However, evaluating the Bogolubov coefficients explicitly and integrating is not trivial in general, and one can justifiably expect that they will be highly dependent on details of the collapse. Nevertheless, it will soon become clear that details of the collapse are negligible, as the effect of the black hole region and the stationary state conjecture will make it so that only the parameters $(r_s,a,e)$ are of significance to the spectrum of measured particles at $\mathscr{I}^+$. In this manner, the path that will be taken is to find a relation for $|\beta_{\omega\omega'}|^2$ and $|\alpha_{\omega\omega'}|^2$ and then use eq. \ref{coeff1} to evaluate eq \ref{par1}, as first done originally in \cite{Hawking1975}. To find how the coefficients relate, it is necessary to solve the Klein-Gordon equation and make use of the explicit forms of $f_{\omega}$ and $u_{\omega}$. 

Consider the massless Klein-Gordon equation, for simplicity, in the Kerr spacetime, which, from eq. \ref{Christoffel divergence}, reads
\begin{equation}\label{KG}
    \frac{1}{\sqrt{-g}}\partial_{\mu}\left[\sqrt{-g}g^{\mu\nu}\partial_{\nu}\psi\right]=0.
\end{equation}
 Using the Kerr metric and its inverse (see appendix \ref{B1}), eq. \ref{KG} can be written as
\begin{multline}\label{KG2}
\left\{\frac{1}{\Delta}\left[(r^2+a^2)^2-\Delta a^2\sin^2{\theta}\right]\partial^2_t-\left(\frac{\Delta-a^2\sin^2{\theta}}{\Delta\sin^2{\theta}}\right)\partial^2_{\phi}+\left(\frac{2r_sar}{\Delta}\right)\partial_t\partial_{\phi}-\partial_r(\Delta\partial_r)\right.\\ \left.-\frac{1}{\sin{\theta}}\partial_{\theta}\left(\sin{\theta}\partial_{\theta}\right)\right\}\psi=0.
\end{multline}
Supposing a solution $\psi=R(r)\Theta(\theta)e^{im\phi}e^{-i\omega t}$, eq. \ref{KG2} separates and yields equations for $R(r)$ and $\Theta(\theta)$. Using $\lambda$ as the separation constant, one obtains
\begin{equation}\label{radialKG}
    \Delta\frac{d}{dr}\left(\Delta\frac{dR(r)}{dr}\right)+\left[\omega(r^2+a^2)^2+m^2a^2-2r_sarm\omega-(\omega^2a^2+\lambda)\Delta\right]R(r)=0,
\end{equation}
\begin{equation}\label{angularKG}
\frac{1}{\sin{\theta}}\frac{d}{d\theta}\left(\sin{\theta}\frac{d\Theta(\theta)}{d\theta}\right)+\left(\lambda+\omega^2a^2\cos{\theta}-\frac{m^2}{\sin^2{\theta}}\right)\Theta(\theta)=0.
\end{equation}

The angular equation, eq. \ref{angularKG}, is an eigenvalue equation whose solution is known to give rise to the \textit{oblate spheroidal harmonics} \cite{Ford1975}. Indeed, this equation, coupled with the harmonic oscillator equation for the coordinate $\phi$, is very similar to those which yield the spheroidal harmonics, the only difference being the additional $\omega^2a^2\cos{\theta}\Theta(\theta)$ term. This small difference makes it so that the eigenvalues, $\lambda_{\ell m}$, associated with the oblate spheroidal harmonics, $S_{\ell m}(\theta,\phi)=\Theta_{\ell m}(\theta)e^{im\phi}$, have a nontrivial dependence on the integers $\ell(=0,1,2\ldots)$ and $m(=-\ell,\ell+1,\ldots,\ell-1,\ell)$. Nevertheless, these functions are orthogonal \cite{Iyer1979} in the sense that by imposing adequate regularity conditions \cite{Frolov2011}, they can be shown to obey 
\begin{equation}
    \int^{2\pi}_0 d\phi\int^{\pi}_0 d\theta  \sin{\theta}S^*_{\ell m}(\theta,\phi) S_{\ell' m'}(\theta,\phi)=\delta_{\ell\ell'}\delta_{mm'}.
\end{equation}

As for the radial equation, eq. \ref{radialKG}, it is easier to analyze it through a more adequate coordinate. In appendix \ref{congrukerr}, the derivation of principal null congruences for the Kerr spacetime is presented, and an adequate coordinate,   $r'$, is defined (see eq. \ref{tortoiseKerr}). By also defining
\begin{equation}
    U(r)=R(r)(r^2+a^2)^{1/2},
\end{equation}
eq. \ref{radialKG} takes the form
\begin{equation}
    \frac{dU(r)}{dr'}+V(r)=0,
\end{equation}
with 
\begin{multline}\label{potential}
    V(r)=\omega^2+\frac{1}{(r^2+a^2)^2}\left[m^2a^2-2r_sarm\omega-(\omega^2a^2+\lambda)\Delta\right]\\+\frac{\Delta}{(r^2+a^2)^3}\left[\Delta+r(2r-r_s)\right]-\frac{3r^2\Delta^2}{(r^2+a^2)^4}.
\end{multline}
Due to the interest in analyzing the solutions of the Klein-Gordon equation in the asymptotic regions, it suffices to study the behavior of $V(r)$ as $r\to\infty$. In addition, because the hypersurface $H\cup\mathscr{I}^+$ is the one that all null geodesics encounter, it will also be useful to analyze the behavior of the potential as $r\to r_+$. In these limits, one finds
\begin{equation}\label{eq5}
    V(r)=
     \begin{cases}
      \omega^2, & r\to \infty, \\
      (\omega-m\Omega)^2, & r\to r_+.
    \end{cases}
\end{equation}

Thus, considering the definition of the null coordinates in the Kerr spacetime, eqs. \ref{null1kerr} and \ref{null2kerr}, as well as the well behaved angular coordinate, $\phi'$, at $H$, given by eq. \ref{phiKerr}, one obtains scalar solutions of eq. \ref{KG2} with form
\begin{equation}
    \psi_{\mathscr{I}^-}=\frac{S_{\ell m}(\theta,\phi)}{r}(\underbrace{e^{-i\omega w}}_{\text{(I)}}+\underbrace{e^{-i\omega u}}_{\text{(II)}}),
\end{equation}
\begin{equation}
    \psi_{\mathscr{I}^+}=\frac{S_{\ell m}(\theta,\phi)}{r}(\underbrace{e^{-i\omega w}}_{\text{(III)}}+\underbrace{e^{-i\omega u}}_{\text{(IV)}}),
\end{equation}
\begin{equation}\label{waveH}
    \psi_{H}=\frac{e^{im\phi'}\Theta_{\ell m}(\theta)}{(r_+^2+a^2)^{1/2}}(\underbrace{e^{-i(\omega-m\Omega) w}}_{\text{(V)}}+\underbrace{e^{-i(\omega-m\Omega) u}}_{\text{(VI)}}).
\end{equation}
The interpretation of each of these solutions is straightforward when one considers that solutions that evolve with a constant value of outgoing coordinate, $w$, correspond to null geodesics of the incoming principal null congruence. Evidently, the same line of reasoning applies to solutions with a constant value of the incoming null coordinate, $u$, that is, they correspond to null geodesics of the outgoing principal null congruence. Hence, the solutions (I), (III) and (V) are incoming at the corresponding regions, while the solutions (II), (IV) and (VI) are outgoing, as illustrated in fig. \ref{fig:qft10}. \begin{figure}[h]
\centering
\includegraphics[scale=1.5]{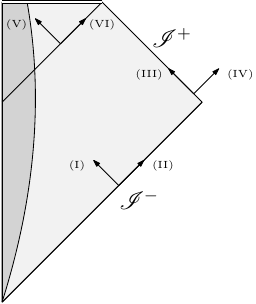} 
\caption{Solutions of the Klein-Gordon equation in the regions of interest in the conformal diagram of the analytically extended Schwarzschild spacetime.}
\caption*{Source: By the author.}
\label{fig:qft10}
\end{figure}
Note that such an illustration is made using the conformal diagram of a black hole that formed due to the collapse of a spherically symmetric distribution of energy, which is adequate to illustrate the qualitative causal structure on or outside the event horizon. 

Although all of these solutions are physical, not all of them will be relevant to the computation of the Bogolubov coefficients. More precisely, the observers in the future asymptotic region will only have access to the solutions (III) and (IV), but (III) is just the ``natural'' evolution of (I) along the asymptotic region, while (IV) will be the result of the ``natural'' evolution of (II) along the asymptotic region as well as the contributions from (I) as it is scattered  (as a consequence of the potential barrier given by eq. \ref{potential}) and ``reflected'' by the spacetime. More precisely, only the solutions (I), (IV) and (V) will be of significance to the computation of the coefficients, the latter being necessary only for the analysis of the phases of the solutions as they are ``reflected'' ``close'' to the formation of $H$. Therefore, normalized solutions \cite{Parker2009} of eq. \ref{KG2} in the asymptotic regions read
\begin{equation}\label{wave-}
    f_{\omega\ell m}=\frac{e^{-i\omega w}}{C_{\pi}r(2\omega)^{1/2}}S_{\ell m}(\theta,\phi),
\end{equation}
\begin{equation}\label{wave+}
    u_{\omega\ell m}=\frac{e^{-i\omega u}}{C_{\pi}r(2\omega)^{1/2}}S_{\ell m}(\theta,\phi),
\end{equation}
where $C_{\pi}$ is a factor containing the $2\pi$ terms, and now the explicit number of indices necessary to label the solutions has been shown, so that the index $\omega$ is now one-dimensional. Namely, an equation such as eq. \ref{7890} actually reads
\begin{equation}
    \hat{\psi}=\sum_{\ell, m}\int_{0}^{\infty} d\omega(\hat{a}_{\omega\ell m}f_{\omega\ell m}+\hat{a}^{\dag}_{\omega\ell m}f_{\omega\ell m}^*),
\end{equation}
in which the discrete indices $\ell$ and $m$ can be interpreted as the quantum numbers associated with angular momentum \cite{Parker2009}. Finally, as discussed above, a wave packet formed by a superposition of the $f_{\omega\ell m}$ is incoming and localized at large $r$ at $ct\to-\infty$, while one formed by a superposition of $u_{\omega\ell m}$ is outgoing and localized at large $r$ at $ct\to\infty$. For simplicity of notation, in the developments that follow, the discrete quantum numbers will be suppressed. 

In order to evaluate the Bogolubov coefficients with the explicit forms of solutions of the Klein-Gordon equation, it is useful to consider the ``time reversed'' evolution of the solutions $u_{\omega}$ that reached $\mathscr{I}^+$ at large constant values of incoming null coordinate. When following such a solution backward in time, part of it will be scattered by the potential barrier (i.e., due to $V(r)$) and reach $\mathscr{I}^-$ with a frequency similar to the one it had when it was ``emitted'' in $\mathscr{I}^+$. The other part of $u_{\omega}$ that was not scattered will pass through the center of the collapsing energy distribution and ``emerge'' as an incoming solution\footnote{This corresponds to the process of ``reflection'' of a null geodesic in the conformal diagram. } that will reach $\mathscr{I}^-$ having a large frequency $\omega'\gg\omega$. This significant blueshift is a consequence of the fact that when $u_{\omega}$ reaches the collapsing distribution, it will be in a much more dense configuration in comparison with how it was when it ``left'' it as an incoming wave. Indeed, this can be thought of as the ``time inversion'' of the gravitational redshift effect that arises when one receives signals emitted from a source closer to an energy distribution (see \S~\ref{schsec} for details on this for a spherically symmetric distribution of energy). Because of this large blueshift and the consideration of high frequency modes\footnote{In this approximation, one can justifiably neglect interactions of the wave solution with the energy distribution.}, the propagation of the solutions from $\mathscr{I}^+$ to $\mathscr{I}^-$ is made through null geodesics, as per the geometric optics approximation. Now, since the expansion of $u_{\omega}$ in terms of $\{f_{\omega'},f_{\omega'}^*\}$ is valid in the entire spacetime, one can use the ``traced backwards'' form of $u_{\omega}$ at $\mathscr{I}^-$ to evaluate the moduli of the Bogolubov coefficients, $\alpha_{\omega\omega'}$ and $\beta_{\omega\omega'}$. However, due to the large blueshift at the event horizon, the expansion of $u_{\omega}$ in terms of $\{f_{\omega'},f_{\omega'}^*\}$ will be made by coefficients obeying $\omega'\gg\omega$. Consequently, the spectrum of effectively created particles is determined by $\beta_{\omega\omega'}$ (see eq. \ref{par1}) for arbitrarily large $\omega'$.  

The pertinent analysis is then made by viewing the solutions $u_{\omega}$ not as a function of their incoming null coordinate, $u_{\omega}(u)$, as given by eq. \ref{wave+}, but actually, as a function of the outgoing null coordinate, $u_{\omega}(w)$, that it has when it ``reaches'' $\mathscr{I}^-$. Therefore, one has to evaluate how the incoming coordinate of a pertinent solution at $\mathscr{I}^+$ (i.e., a solution (IV) in fig. \ref{fig:qft10}) relates to the outgoing coordinate of a pertinent solution at $\mathscr{I}^-$ (i.e., a solution (I) in fig. \ref{fig:qft10}). This function, $u(w)$, relates the incoming null coordinate of an outgoing null geodesic to the outgoing coordinate of an incoming null geodesic, both belonging to the respective principal null congruences. In contrast, the function $u_{\omega}(w)$ is the form of a solution at $\mathscr{I}^+$ when traced to $\mathscr{I}^-$ in terms of the outgoing null coordinate of the null geodesic that reaches $\mathscr{I}^-$. Using developments of \cite{Parker2009}, consider the following line of reasoning to evaluate $u(w)$.

Let $w_0$ be the outgoing null coordinate of the incoming null geodesic that generates the event horizon. This null geodesic and any other incoming one with outgoing null coordinate $w>w_0$ which is not scattered by the potential barrier will, evidently, never reach $\mathscr{I}^+$, since it will either generate or cross the event horizon. In this context, the analysis of interest is how an incoming null geodesic with outgoing coordinate $w<w_0$ evolves into an outgoing null geodesic with incoming coordinate $u(w)$. To do so, consider an incoming null geodesic with $w'\gg w_0$, parametrized by the affine parameter $\lambda$. \begin{figure}[h]
  \centering
    \includegraphics[scale=1.4]{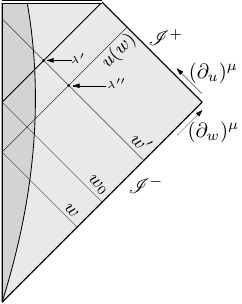}
  \caption{Incoming and outgoing null geodesics in the collapse of spherically symmetric energy distribution for the evaluation of $u(w)$.}\label{fig:qft45}
\caption*{Source: By the author.}
\end{figure} Because $w'\gg w_0$, the part of this solution that reaches the event horizon will do so a ``long time'' after its formation\footnote{In the sense that the outgoing parameter of the incoming null geodesic is numerically of orders of magnitude bigger than the one that generates the event horizon.}, and one can consider that at that point the black hole would have ``settled down'' to a Kerr black hole. Let $\lambda'$ and $\lambda''$ denote the affine parameter value of this curve as it intersects the outgoing null geodesic with incoming null coordinate $u(w_0)$ and $u(w)$, respectively, as illustrated in fig. \ref{fig:qft45}. Now, if one considers that
\begin{equation}\label{geo4}
    w_0-w\ll 1,
\end{equation}
then, physically, this means that the outgoing null geodesic with incoming coordinate $u(w)$ leaves the collapsing distribution ``just before'' the formation of the event horizon\footnote{In the sense that the incoming parameter of the outgoing null geodesic minus the one that generates the event horizon is orders of magnitude smaller than $-1$.}. Thus, the affine parameter, $\lambda$, can be chosen so that the evolution of the incoming parameter, $u(\lambda)$, of the incoming null geodesic with outgoing coordinate $w'$ as it approaches $H$ is approximated by eq. \ref{incoming}. 

Indeed, one can characterize the incoming coordinate, $u$, of an outgoing null geodesic by the value of the affine parameter, $\lambda''$, that the incoming null geodesic with outgoing coordinate $w'$ had when it crossed it. For example, an incoming null geodesic that obeys eq. \ref{geo4} will ``emerge'' as an outgoing null geodesic with incoming coordinate given by
\begin{equation}\label{incoming1}
    u(\lambda'')\approx-\frac{1}{\kappa}\ln\left(\frac{\lambda''}{C'}\right).
\end{equation}
Similarly, one evidently has that $u(\lambda')\to\infty$, as $\lambda'$ was adopted to vanish at the event horizon (see appendix \ref{congrukerr}). Therefore, the affine separation distance between the outgoing null geodesics with incoming coordinates $u(w_0)$ and $u(w)$, $\lambda''-\lambda'$, as measured by any incoming null geodesic that crosses the event horizon, is then simply $\lambda''$.

Note that the affine parameter distance between incoming null geodesics, with outgoing coordinates $w$ and $w_0$, can be chosen to be constant throughout the spacetime \cite{Parker2009}, i.e., when they ``emerge'' as outgoing null geodesics with incoming coordinates $u(w)$ and $w(w_0)$, they will still have the same affine parameter distance. Additionally, due to the asymptotically flat nature of the Kerr spacetime, the parameter $w$ itself is an affine parameter\footnote{This can be deduced by the same developments that led to eq. \ref{dev}.} at $\mathscr{I}^-$, i.e., along the outgoing null geodesics with incoming coordinates $u\to-\infty$ that generate $\mathscr{I}^-$. This translates to
\begin{equation}\label{parakerr}
    w_0-w=C\lambda'',
\end{equation}
where $C$ is a negative constant. Hence, one can apply eq. \ref{parakerr} in eq. \ref{incoming1}, resulting in
\begin{equation}\label{uw}
    u(w)\approx-\frac{1}{\kappa}\ln\left(\frac{w_0-w}{C'}\right),
\end{equation}
where all constants have been absorbed into the positive constant $C'$. This relation is what determines the spectrum of created particles due to the gravitational collapse of an energy distribution that results in a black hole. More precisely, the stationary state conjecture and the arbitrarily high blueshift (for the ``traced backwards'' solutions at $\mathscr{I}^+$) at the event horizon make it so that the Bogolubov coefficients are determined, predominantly, by solutions that obey eq. \ref{uw}. It is in this sense that the details of the gravitational collapse will have negligible influence on the spectrum of created particles.

Lastly, in order to compute the exact final form of $u_{\omega}(w)$ at $\mathscr{I}^-$, it is also necessary to study how its phase changes as it passes close to $H$. Evidently, from eq. \ref{waveH}, the effective frequency at $H$ is $\omega-m\Omega$, so that when one traces the surfaces of constant phase of the solutions $u_{\omega}$ from $H$ to $\mathscr{I}^-$, they will reach $\mathscr{I}^-$ with the same effective frequency. Thus, by tracing back $u_{\omega}$ from $\mathscr{I}^+$ to $\mathscr{I}^-$, one finds that it reaches $\mathscr{I}^-$ with the form
\begin{equation}\label{wave1}
    u_{\omega}(w)=
     \begin{cases}
      r^{-1}S(\theta,\phi)\left[2(\omega-m\Omega)\right]^{-1/2}\exp\left[{\frac{i(\omega-m\Omega)}{\kappa}\ln\left({\frac{w_0-w}{C'}}\right)}\right], & w<w_0,\\
      0, & w>w_0, 
    \end{cases}
\end{equation}
where $w_0$ is the outgoing null coordinate of the incoming null geodesic that generates the event horizon. The vanishing of the solution for $w<w_0$ is a consequence of the fact that scattered solutions will reach $\mathscr{I}^-$ with a relatively small frequency, and those that are not scattered will cross $H$, being irrelevant to created particles at $\mathscr{I}^+$.
  
From the explicit form of $u_{\omega}(w)$ on $\mathscr{I}^-$, eq. \ref{wave1}, one can evaluate the Bogolubov coefficients as follows. Consider the Fourier transform of $u_{\omega}(w)$ \cite{Butkov1973},
\begin{equation}\label{Fourier}
    \mathcal{U}_{\omega}(\omega')=\frac{1}{\sqrt{2\pi}}\int^{\infty}_{-\infty}u_{\omega}(w)e^{i\omega'w}dw,
\end{equation}
and its inverse,
\begin{equation}\label{Fourierinverse}
    u_{\omega}(w)=\frac{1}{\sqrt{2\pi}}\int^{\infty}_{-\infty}\mathcal{U}_{\omega}(\omega')e^{-i\omega'w}d\omega'.
\end{equation}
From eq. \ref{Fourierinverse}, it is possible to write
\begin{equation}\label{u123456}
    \begin{aligned}[b]
    u_{\omega}(w) & = \frac{1}{\sqrt{2\pi}}\left(\int^{\infty}_{0}\mathcal{U}_{\omega}(\omega')e^{-i\omega'w}d\omega'+\int^{0}_{-\infty}\mathcal{U}_{\omega}(\omega')e^{-i\omega'w}d\omega'\right)\\
    & = \frac{1}{\sqrt{2\pi}}\left(\int^{\infty}_{0}\mathcal{U}_{\omega}(\omega')e^{-i\omega'w}d\omega'+\int_{0}^{\infty}\mathcal{U}_{\omega}(-\omega')e^{i\omega'w}d\omega'\right),
    \end{aligned}
\end{equation}
where the variable change $\omega'\to-\omega'$ was performed in the second integral of the second line. Now, from eq. \ref{bogo} and the asymptotic form of the solutions  $f_{\omega}$, given by eq. \ref{wave-}, one also has
\begin{equation}\label{u12345}
    u_{\omega}=\int^{\infty}_{0} d\omega'\left(\frac{\alpha_{\omega\omega'}S(\theta,\phi)}{C_{\pi}r\sqrt{2\omega'}}e^{-i\omega'w}+\frac{\beta_{\omega\omega'}S(\theta,\phi)}{C_{\pi}r\sqrt{2\omega'}}e^{i\omega'w}\right).
\end{equation}
Comparison of eqs. \ref{u123456} and \ref{u12345} yields
\begin{equation}
    \alpha_{\omega\omega'}=\frac{C_{\pi}r\sqrt{4\pi\omega'}}{S(\theta,\phi)}\mathcal{U}_{\omega}(\omega'),
\end{equation}
\begin{equation}
    \beta_{\omega\omega'}=\frac{C_{\pi}r\sqrt{4\pi\omega'}}{S(\theta,\phi)}\mathcal{U}_{\omega}(-\omega').
\end{equation}
Finally, using the Fourier transform, eq. \ref{Fourier}, and the ``traced backwards'' form of $u_{\omega}$, eq. \ref{wave1}, one obtains
\begin{equation}\label{alpha1}
    \alpha_{\omega\omega'}=\int^{w_0}_{-\infty}\left(\frac{\omega'}{\omega-m\Omega}\right)^{1/2}e^{i\omega'w}\exp\left[{\frac{i(\omega-m\Omega)}{\kappa}\ln\left({\frac{w_0-w}{C'}}\right)}\right]dw,
\end{equation}
\begin{equation}\label{beta1}
    \beta_{\omega\omega'}=\int^{w_0}_{-\infty}\left(\frac{\omega'}{\omega-m\Omega}\right)^{1/2}e^{-i\omega'w}\exp\left[{\frac{i(\omega-m\Omega)}{\kappa}\ln\left({\frac{w_0-w}{C'}}\right)}\right]dw.
\end{equation}

The goal now is to work with these integrals in order to find a relation for the modules of the coefficients. Consider first the change of variable $s=w_0-w$ in eq. \ref{alpha1}, and $s=w-w_0$ in eq. \ref{beta1}, which yields
\begin{equation}\label{alpha2}
    \alpha_{\omega\omega'}=-\int^{0}_{\infty}\left(\frac{\omega'}{\omega-m\Omega}\right)^{1/2}e^{i\omega'w_0}e^{-i\omega's}\exp\left[{\frac{i(\omega-m\Omega)}{\kappa}\ln\left({\frac{s}{C'}}\right)}\right]ds,
\end{equation}
\begin{equation}\label{beta2}
    \beta_{\omega\omega'}=\int^{0}_{-\infty}\left(\frac{\omega'}{\omega-m\Omega}\right)^{1/2}e^{-i\omega'w_0}e^{-i\omega's}\exp\left[{\frac{i(\omega-m\Omega)}{\kappa}\ln\left({-\frac{s}{C'}}\right)}\right]ds.
\end{equation}
To simplify these integrals, it is useful to make use of complex analysis. Since their integrands are analytic and proportional to $e^{-i\omega's}$ with $\omega'>0$, one can relate the integral along the real axis to one along the imaginary axis by studying a closed contour in the lower half of the circle, $|z|=R$. The adequate choices of contours to each integral are illustrated in figs. \ref{fig:contour1} and \ref{fig:contour2}. \begin{figure}[h]
  \begin{subfigure}[b]{0.5\textwidth}
  \centering
    \includegraphics[scale=1.2]{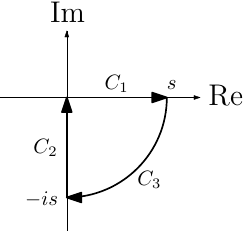}
    \caption{Contour for $\alpha_{\omega\omega'}$.}
    \label{fig:contour1}
  \end{subfigure}
  \hfill
  \begin{subfigure}[b]{0.5\textwidth}
  \centering
    \includegraphics[scale=1.2]{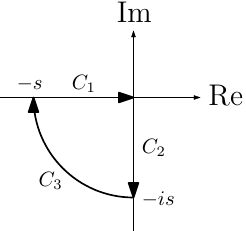}
    \caption{Contour for $\beta_{\omega\omega'}$.}
    \label{fig:contour2}
  \end{subfigure}
  \caption{Contour of integration on the complex plane for evaluation of $\alpha_{\omega\omega'}$ and $\beta_{\omega\omega'}$, with $s\to\infty$.}
\caption*{Source: By the author.}
\end{figure} In this manner, the residue theorem \cite{Churchill1974} can then be used to show that the integral in eqs. \ref{alpha2} and \ref{beta2} over the respective closed contours for $\alpha_{\omega\omega'}$ and $\beta_{\omega\omega'}$ vanishes. Moreover, Jordan's lemma \cite{Churchill1974} can be used to show that the integral along the contour $C_3$ also vanishes. Hence, one finds that it is possible to rewrite both integrals with limits along the imaginary axis,
\begin{equation}
    \alpha_{\omega\omega'}: \int^{0}_{\infty}=\int^{0}_{-i\infty},
\end{equation}
\begin{equation}
    \beta_{\omega\omega'}: \int^{0}_{-\infty}=\int^{0}_{-i\infty}.
\end{equation}

These changes in limits of integration, together with another variable substitution, $s= is'$, for both integrals, yield
\begin{equation}
    \alpha_{\omega\omega'}=-i\int^{0}_{-\infty}\left(\frac{\omega'}{\omega-m\Omega}\right)^{1/2}e^{i\omega'w_0}e^{\omega's'}\exp\left[{\frac{i(\omega-m\Omega)}{\kappa}\ln\left({\frac{is'}{C'}}\right)}\right]ds',
\end{equation}
\begin{equation}
    \beta_{\omega\omega'}=i\int^{0}_{-\infty}\left(\frac{\omega'}{\omega-m\Omega}\right)^{1/2}e^{-i\omega'w_0}e^{\omega's'}\exp\left[{\frac{i(\omega-m\Omega)}{\kappa}\ln\left({-\frac{is'}{C'}}\right)}\right]ds'.
\end{equation}
By restricting the action of the natural logarithm function for the region $s'<0$ so that it can be a single valued function,
\begin{equation}
    \ln{\left(\pm\frac{is'}{C'}\right)}=\mp i\left(\frac{\pi}{2}\right)+\ln{\left(\frac{|s'|}{C'}\right)},
\end{equation}
one finds
\begin{equation}\label{alpha3}
    \alpha_{\omega\omega'}=-i\frac{e^{i\omega'w_0}}{\sqrt{2\pi}}\exp\left[{\frac{\pi(\omega-m\Omega)}{2\kappa}}\right]\mathcal{I},
\end{equation}
\begin{equation}\label{beta3}
    \beta_{\omega\omega'}=i\frac{e^{-i\omega'w_0}}{\sqrt{2\pi}}\exp\left[{-\frac{\pi(\omega-m\Omega)}{2\kappa}}\right]\mathcal{I},
\end{equation}
where
\begin{equation}
    \mathcal{I}=\int^{0}_{-\infty}\left(\frac{\omega'}{\omega-m\Omega}\right)^{1/2}e^{\omega's'}\exp{\left[{\frac{i(\omega-m\Omega)}{\kappa}\ln\left({\frac{|s'|}{C'}}\right)}\right]}ds'.
\end{equation}

Eqs. \ref{alpha3} and \ref{beta3} can be used to write the desired relation between the modules of the Bogolubov coefficients,
\begin{equation}
    |\alpha_{\omega\omega'}|^2=\exp\left[{\frac{2\pi(\omega-m\Omega)}{\kappa}}\right]|\beta_{\omega\omega'}|^2.
\end{equation}
In this manner, one can use eq. \ref{coeff1} with $\omega=\omega'$, which yields a relation for the right hand side of eq. \ref{par1}, resulting in 
\begin{equation}\label{delta12}
    _{{\mathscr{I}^-}}\langle0|N_{\omega}|0\rangle_{\mathscr{I}^-}=\left\{\exp\left[{\frac{2\pi(\omega-m\Omega)}{\kappa}}\right]-1\right\}^{-1}\delta(0).
\end{equation}
The infinity that arises in eq. \ref{delta12} due to the Dirac delta term $\delta(0)$ \cite{Butkov1973} is a physical one. Evidently, since one is analyzing the expected number of particles at $\mathscr{I}^+$, a steady flux from a stationary configuration will result in an infinite number of particles\footnote{Neglecting change in mass of the black hole. This issue will be discussed in detail in \S\;\ref{evaporation}.}. Hence, the quantity of interest is the expected number of particles per unit time. An \textit{heuristic} way of evaluating such a quantity from eq. \ref{delta12} is to make use of the integral representation of the Dirac delta and isolate its infinity in a ``time variable''. More precisely, consider the identity
\begin{equation}
    \delta(x-a)=\lim_{t\to\infty}\frac{1}{2\pi}\int_{-\frac{t}{2}}^{\frac{t}{2}}e^{i(x-a)t'}dt',
\end{equation}
which is just a more convenient way of writing the integral representation. It can then be used to isolate the infinity in a ``time variable'', in the sense that 
\begin{equation}
    \delta(0)=\lim_{t\to\infty}\frac{t}{2\pi},
\end{equation}
which allows one to obtain
\begin{equation}\label{flux}
    \frac{d}{dt}\left(_{{\mathscr{I}^-}}\langle0|N_{\omega}|0\rangle_{\mathscr{I}^-}\right)=\frac{1}{2\pi}\left\{\exp\left[{\frac{2\pi (\omega-m\Omega)}{\kappa}}\right]-1\right\}^{-1}.
\end{equation}

A more rigorous way of obtaining this result is by considering wave packets constructed from $u_{\omega}$ or by confining the solutions to a finite volume (see \cite{Parker2009} for details on this procedure). As discussed in \S~\ref{qft}, these are mathematical procedures that reduce the delta in eq. \ref{coeff1} to a Kronecker one. Physically, this can be interpreted as a way to localize the modes so that the infinity arising from the steady flux is removed. Additionally, one should include a fraction term to the distribution presented in eq. \ref{flux} to account for the scattering of waves discussed earlier, as a consequence of the potential barrier\footnote{The exact form of this term can be evaluated by relaxing the approximation of geometric optics.}. This term, denoted by $\Gamma_{\ell m}(\omega)$, is related to the part of the waves that, when traced back from $\mathscr{I}^+$ would reach the collapsing body ``just before'' the formation of the event horizon, so that one now has
\begin{equation}\label{flux1}
    \frac{d}{dt}\left(_{{\mathscr{I}^-}}\langle0|N_{\omega\ell m}|0\rangle_{\mathscr{I}^-}\right)=\frac{\Gamma_{\ell m}(\omega)}{2\pi}\left\{\exp\left[{\frac{2\pi (\omega-m\Omega)}{\kappa}}\right]-1\right\}^{-1}.
\end{equation}
As the black hole evolves to a stationary state through the emission of gravitational radiation, this fraction becomes closer to the one that dictates the amount of waves that would cross the white hole horizon (see \S~\ref{bhsec}) if one were to consider the analytical extension of the spacetime. Because the scattering properties from $\mathscr{I}^+$ to $\mathscr{I}^-$ are symmetrical in such extension, for $u\to\infty$, one may interpret $\Gamma_{\ell m}(\omega)$ as the probability that an incoming wave emitted in $\mathscr{I}^-$ will cross the black hole event horizon. Therefore, eq. \ref{flux1} can be interpreted as stating that the expected number of particles perceived at $u\to\infty$ in $\mathscr{I}^+$ between $\omega$ and $\omega+d\omega$, and with angular momentum quantum numbers $\ell$ and $m$ per unit time, is given by a black body spectrum \cite{Swendsen2020} with temperature 
\begin{equation}\label{temperature}
    T=\frac{\hbar \kappa}{2\pi ck_B}.
\end{equation}
For a Schwarzschild black hole, considering eq. \ref{ksch}, one has
\begin{equation}
    T\approx 6.18\;10^{-8}\left(\frac{M}{M_\odot}\right)^{-1}\;\text{K},
\end{equation}
i.e., the temperature of a Schwarzschild black hole is inversely proportional to its mass. 

In essence, these results translate to the conclusion that a stationary black hole will behave as a gray body of absorptivity $\Gamma_{\ell m}(\omega)$ and temperature proportional to its surface gravity. This effective particle creation effect by black holes is also referred to as the \textit{Hawking effect}, while the approximately\footnote{An affirmation of an \textit{exactly} thermal character would only be justifiable in a development taking account the effects of emitted particles on the metric, \textit{physical optics} (i.e., exact wave propagation), as well as trans-Planckian physics.} thermal radiation predicted by it is referred to as \textit{Hawking radiation}. Finally, the term $-m\Omega$ accompanying the frequency in the black body spectrum in eq. \ref{flux1} can be interpreted as a ``chemical potential'' term, which physically means that for a given mode, the effective emission of particles with angular momentum $m$ is more likely than that of particles with angular momentum $-m$. For a Kerr-Newman black hole, there will be an extra term corresponding to the electric charge contribution, such that the ``chemical potential'' term takes the form $-m\Omega-e\Phi$ (see the remarks below eq. \ref{mass3}). The same interpretation follows for its presence, that is, the effective emission of particles will also favor those with the same sign of electric charge as that of the black hole. Consequently, the spectrum of emitted particles from a black hole tends to carry away angular momentum and electric charge.

\section{Black holes and thermodynamics}\label{thermo}

In this section, we will discuss how the classical properties derived in ch. \ref{chapter2} can be interpreted in light of the semiclassical particle creation effect. Particularly, the discussion about the nature of the classical properties of black holes arose even before the derivation of the effective particle creation effect, where it was suggested that black holes must possess a physical entropy in order to preserve the second law of thermodynamics \cite{Bekenstein1972, Bekenstein1973, Bekenstein1974}. In fact, the impossibility of reducing the surface gravity of a black hole to zero (see theorem \ref{thekappa}) was first assumed \cite{Bardeen1973} as a way to make a complete correspondence between the classical properties of black holes, which are rigorous results of differential geometry \cite{Bekenstein1980,Wald2001}, and the laws of thermodynamics, which are approximations of macroscopic properties of a system. To start this discussion, we first give a brief review of the concepts necessary to state the laws of thermodynamics.

By \textit{laws of thermodynamics}, we refer to the postulate construction presented in \cite{Swendsen2020} which leads to the relations and properties of entropy, temperature and energy of a composite system (which will be referred to simply as a system), i.e., a system composed of subsystems with some constraint between them. In such a construction, the entropy is viewed as a function of \textit{extensive parameters}, i.e., the parameters that can be used to fully characterize a time independent state of a system, known as an \textit{equilibrium state}. The entropy \textit{may} be a homogeneous first order function of the extensive parameters and additive over systems, both of which are approximations of the interactions and properties of the system. One can relate the variation of the entropy with the variation of energy, $U$, volume, $V$, and mole number, $N$, of a system by 
\begin{equation}\label{firstlaw}
    \delta U=T\delta S-P\delta V+\mu\delta N,
\end{equation}
commonly referred to as the \textit{first law}, which is merely a consequence of energy conservation. Here, $T$ is the temperature, $P$ is the pressure, $\mu$ is the chemical potential, and the $\delta$ is representative of variations of an equilibrium state to a neighboring one, which is made by virtue of a \textit{quasi-static process}. From this relationship, the \textit{zeroth law} can be stated following the net flow of energy between constituents of a system. In other words, a system in equilibrium has constant temperature. The \textit{second law} is related to the time asymmetry of ``natural'' physical processes, which can be related to the ``natural'' flow of energy between two systems after a constraint is released, corresponding to the mathematical relation
\begin{equation}
    \delta S\geq0.
\end{equation}
The \textit{third law} states that the entropy of a system goes to a constant as $T\to0$. 

Although these general, macroscopic properties (i.e., properties of systems containing a large number of particles such that measurements have negligible statistical fluctuations) are in excellent agreement with experiments, they are not useful for the derivation of \textit{intrinsic} properties of a system. That is, they do not describe how unique microscopic properties amount to these macroscopic relations. A proposal to give this description is through postulating that the entropy is a multiple of the logarithm of the probability of the macroscopic state of a system. Consequently, the entropy of a system can be interpreted as a measure of its degrees of freedom, i.e., the number of quantum states accessible to the energy distribution that describes it. Following this postulate and using \textit{quantum} and \textit{classical statistical mechanics}, one can provide a justification for the postulates and the laws of thermodynamics. Because of this proposed definition of entropy, the second law of thermodynamics can then be interpreted as the natural evolution of a system from a less to a more probable macroscopic state. Lastly, it should be noted that the fact that $T=0$ is unattainable is \textit{not} a consequence of the third law \cite{Swendsen2020}. In reality, both quantum and classical statistical mechanics state that in order to achieve absolute zero, it would be necessary an infinite number of processes.

Regardless of the nature of the interactions present in a system, one expects that the laws of thermodynamics are applicable to it. In order to investigate how they may apply to black holes, it is useful to consider how a black hole is perceived by an observer outside of it. Such an observer would make measurements of the energy distribution outside the black hole and would assign an entropy to it. However, due to the uniqueness theorems, the information about the energy distribution inside the black hole is limited to the three parameters that describe it, $(r_s,a,e)$, and it is not clear how such an observer would go about assigning an entropy to the black hole, if it even has a nonvanishing one. Suppose, first, a scenario in which the black hole has a vanishing entropy. This would clearly be problematic when one considers the second law of thermodynamics, as it would be possible to reduce the entropy of the universe simply by allowing energy to cross the event horizon. Thus, in order to preserve the second law of thermodynamics, black holes must have a nonvanishing physical entropy.

This being the case, one then must have that the entropy assigned to a stationary black hole must be dependent only on the parameters that characterize its ``equilibrium state'', much like the dependence of the ordinary entropy on the extensive parameters of a system in equilibrium. Now, consider the relation in first order of the parameters of a stationary black hole when it is perturbed, eq. \ref{first}. Notice that the first law of thermodynamics has a similar form, where one has a $TdS$ term followed by arbitrary ``work'' terms. The association of a possible physical temperature of a black hole with a multiple of $\kappa$, as given by the Hawking effect (see eq. \ref{temperature}), then leads to the association of the entropy with a multiple of the area of the event horizon, namely, 
\begin{equation}\label{entropy}
    S_B=\frac{k_B A}{4\ell_p^2}=\frac{c^3k_BA}{4\hbar G},
\end{equation}
which for a Schwarzschild black hole, reads
\begin{equation}
    S_B\approx 3.61\;10^{53}\left(\frac{M}{M_\odot}\right)^2\;\frac{\text{J}}{\text{K}}.
\end{equation}
Intriguingly, in order for eq. \ref{first} to have the exact same form of eq. \ref{firstlaw} given eq. \ref{temperature}, $S_B$ has to be defined with the Planck length. Evidently, such an association gains a deeper physical meaning when one considers the Hawking effect, as otherwise a simple prescription of this assignment of entropy would have no physical meaning, other than an attempt to preserve the second law of thermodynamics. That is, because the black hole is effectively emitting radiation with an approximate black body spectrum, one can arguably affirm that the Hawking effect is a justification for a \textit{physical} temperature of a black hole. 

In light of this, if one considers that the event horizon area is in fact related to the \textit{physical} entropy of a black hole, then the entropy of a system containing a black hole is given by the \textit{generalized entropy}
\begin{equation}
    S=S'+S_B,
\end{equation}
where $S'$ is the entropy of the energy outside of the black hole. Thus, the second law of thermodynamics can be restated as the \textit{generalized second law} \cite{Bekenstein1972},
\begin{equation}
    \delta S\geq0.
\end{equation}
Indeed, when semiclassical analysis is taken into account, purely classical violations of the generalized second law do not hold \cite{Wald1984}, so that the full consideration of the Hawking effect leads one to believe that the association of $\kappa$ with the temperature of a black hole is a physical one (see, e.g., \cite{Matsas2005} for a detailed discussion), rather than just a mathematical analogy. It is in this sense that one may conclude that the classical properties derived in ch. \ref{chapter2} are merely the laws of thermodynamics applied to a system containing a black hole.

Still, in the context of quantum field theory in curved spacetime, the assumptions made in order to derive the classical properties may not hold. A clear example of this is theorem \ref{area}, which relies on the condition that $R_{\mu\nu}\ell^{\mu}\ell^{\nu}\geq 0$ for all null $\ell^{\mu}$. As discussed, this condition will be satisfied if Einstein's equation and null energy condition hold. More precisely, due to Raychaudhuri's equation, it is known that one can interpret such a condition as the attractive nature of gravity, which also follows from the weak and strong energy conditions. But, in the semiclassical framework, it is straightforward to find examples where these conditions are violated, e.g., the \textit{Casimir effect} \cite{Casimir1948 ,Parker2009} (still, ``averaged'' energy conditions can be satisfied in such cases \cite{Fewster2012}). Although this is a clear indicator that theorem \ref{area} is no longer valid in the semiclassical depiction of a black hole, this is in agreement with the expectation that a black hole should lose energy due to the Hawking effect, reducing its area, and thus, its entropy (more details on this will be given in \S~\ref{evaporation}). Evidently, this is also in agreement with the generalized second law.

Concerning the other properties, first note that the stationary state conjecture is also physically justifiable by analogous behavior of thermodynamic systems, in which arbitrary states are expected to ``settle down'' to a final, time independent, equilibrium state, being described uniquely by extensive parameters. Hence, the constancy of $\kappa$ over the event horizon of a stationary black hole can be argued to be analogous, or some perspectives might even say, equivalent, to the constancy of $T$ in the constituents of a system in equilibrium. Recall that this property may be derived by using the Killing horizon property of $H$ and assuming the validity of Einstein's equation and the dominant energy condition (which is associated with the speed limit of observers and signals), or from purely geometrical arguments following from the $t$-$\phi$ orthogonality property. As discussed, one can also interpret the relation of the variations of the parameters of a black hole simply as the first law of thermodynamics applied to it (in fact, it can be seen merely as an ``energy conservation law''). In particular, note that the derivation of eq. \ref{first} followed from the constancy of $\kappa$ over $H$, the Killing horizon nature of the event horizon, variations over neighboring stationary solutions and the asymptotic properties of the Kerr spacetime. In other words, such a result can also be derived through purely geometrical arguments. Notwithstanding, it has also been shown \cite{Iyer1994} that a more general form of eq. \ref{first} holds in any metric theory of gravity (see ch. \ref{Introduction}) whose field equations are derived from a diffeomorphism covariant Lagrangian \cite{Wald1993, Rácz1992, Iyer1995}. In this framework, the variation relation (to first order) is seen as a direct consequence of the variation identity of the Noether current, and the black hole entropy is seen as the Noether charge arising from the symmetry represented by the diffeomorphism. Similarly, it is straightforward to see that black holes also obey that their area goes to a constant as $\kappa\to0$ (see eqs. \ref{areakerr} and \ref{kappakerr}). Furthermore, following considerations of the cosmic censor conjecture, one expects that it would be impossible to reduce the surface gravity of a black hole to zero in a finite amount of time, much like statistical mechanics states that it is impossible to achieve $T=0$ by a finite number of processes. 

Because of the geometrical arguments and conjectures associated with the derivations of classical properties of black holes, it is evident that they do not necessarily depend on the specific content of the field equations (although they may also be derived from them). This generality points to a possible connection between gravitation and the description of heat, which is, in fact, not exclusive to black holes. Other thermodynamic properties of null hypersurfaces, not necessarily event horizons, have also been developed in the last decades (see, e.g., \cite{Gibbons1977,Hawking1995, Jacobson2003, Padmanabhan2010,Guedens2012}). Moreover, under certain hypotheses, one can also derive Einstein's equation from the thermodynamic relation of heat, temperature and entropy \cite{Jacobson1995}. Although these developments are promising, no underlying explanation for such connections has been found, and there is still an important question regarding the microscopic derivation of these properties and laws. For instance, if the geometrical quantities of a black hole are in fact associated with its thermodynamic properties, one can only guess as to what, and where, are the degrees of freedom responsible for them.

\section{Energy-momentum expectation values and Hadamard states}\label{hadmard}

To conclude this chapter, we will discuss an important implication of entanglement for quantum field theory, which mainly follows from analysis of a construction of a suitable ``energy-momentum expectation value''. As it is well known, in classical theories, the mathematical object that contains all the information about an energy distribution is the energy momentum tensor, $T_{\mu\nu}$. Of course, this is precisely the object that acts as a source for the gravitational interaction, as postulated by Einstein's equation, eq. \ref{eq1}. Consequently, in the semiclassical framework, one would expect that some notion of ``semiclassical Einstein's equation'' \cite{Wald1994}, i.e.,
\begin{equation}\label{semiefe}
    R_{\mu\nu}-\frac{1}{2}Rg_{\mu\nu}=\frac{8\pi G}{c^4}\langle \psi|\hat{T}_{\mu\nu}|\psi\rangle,
\end{equation}
to be an adequate description of the effect of the quantum field on the spacetime metric for an arbitrary state, $|\psi\rangle$. However, it is far from clear how one would go about constructing the operator $\hat{T}_{\mu\nu}$ for each $a\in M$, since the most ``natural'' approach of doing so by using the corresponding classical tensor would yield ill-defined operations. Namely, this is a consequence of the fact that in the process of quantization, (see eqs. \ref{equal1} and \ref{equal2}), the quantum field operator is actually defined as an \textit{operator-valued distribution} \cite{Wald1994,Wald1984}.

For instance, if one tries to construct $\hat{T}_{\mu\nu}$ for the classical Klein-Gordon field from its energy momentum tensor\footnote{The energy-momentum tensor of classical fields can be evaluated from the variation of their action with respect to the spacetime metric \cite{Parker2009}.},
\begin{equation}
    T_{\mu\nu}=\nabla_{\mu}\psi\nabla_{\nu}\psi-\frac{1}{2}g_{\mu\nu}(\nabla_{\alpha}\psi\nabla^{\alpha}\psi+\frac{m^2c^2}{\hbar^2}\psi^2),
\end{equation}
it is clear that this would not yield a well defined operator as a consequence of the nonlinear operations on $\hat{\psi}$. Focusing on the terms $\hat{\psi}^2$, a way to deal with such a complication is to first consider the well defined \textit{bi-distribution} $\hat{\psi}(a)\hat{\psi}(a')$ and then take the limit to the corresponding event, i.e., $a\to a'$. However, for any state with finitely many particles, the expectation value $\langle \psi|[\hat{\psi}(a)]^2|\psi\rangle$ would diverge. In essence, this can be seen by substituting the formal expression for $\hat{\psi}$, eq. \ref{exp1234}, in
\begin{equation}
   \lim_{a\to a'}\hat{\psi}(a)\hat{\psi}(a'),
\end{equation}
which yields an infinite sum of terms $\hat{a}_{p}\hat{a}_{p}^{\dag}$ evaluated at the event $a$.

In Minkowski spacetime, this divergence can be traced back to the interpretation that this calculation is merely the sum of the zero-point energies of the infinite number of harmonic oscillators that give rise to the field. Evidently, by interpreting the divergence as such, one can then make use of a ``vacuum energy subtraction'' to define a smooth function of $a$ and $a'$ by
\begin{equation}\label{eqsubs}
    F(a,a')=\langle \psi|\hat{\psi}(a)\hat{\psi}(a')|\psi\rangle-\langle0|\hat{\psi}(a)\hat{\psi}(a')|0\rangle,
\end{equation}
 where $|0\rangle$ denotes the Minkowski static vacuum. Hence, one can then define 
\begin{equation}
    \langle \psi|[\hat{\psi}(a)]^2|\psi\rangle=\lim_{a\to a'}F(a,a'),
\end{equation}
and using a similar logic, one may construct an adequate notion of $\langle \psi|\hat{T}_{\mu\nu}|\psi\rangle$ in Minkowski spacetime. However, because of the absence of a preferred physical definition of a vacuum state in an arbitrary spacetime, the notion of a ``vacuum energy subtraction'' loses its ``natural'' meaning. Nevertheless, it is possible to construct an axiomatic approach to establish the uniqueness of $\langle \psi|\hat{T}_{\mu\nu}|\psi\rangle$ up to an addition of a conserved local curvature term, so that one can single-out a condition for what classes of states one would deem physical, i.e., those that would lead to a physically adequate definition of $\langle \psi|\hat{T}_{\mu\nu}|\psi\rangle$ in any spacetime.

This physical condition can be defined by requiring that the subtraction in eq. \ref{eqsubs} is made not by some notion of a vacuum state, but rather, a \textit{locally constructed bi-distribution} with similar characteristics as $\langle0|\hat{\psi}(a)\hat{\psi}(a')|0\rangle$. Namely, this locally constructed bi-distribution should have a divergence of leading order as to mimic the singular character of $\langle0|\hat{\psi}(a)\hat{\psi}(a')|0\rangle$, such as those that are proportional, in leading order, to the inverse squared geodesic distance between the events $a$ and $a'$. Thus, the \textit{Hadamard ansatz} for this bi-distribution, $H(a,a')$, can be written as
\begin{equation}\label{hard}
    H(a,a')=\frac{U(a,a')}{(2\pi)^2\sigma(a,a')}+V(a,a')\ln{\sigma}+W(a,a'),
\end{equation}
where $U(a,a')$, $V(a,a')$ and $W(a,a')$ are smooth functions that equal one when $a=a'$, and $\sigma(x,x')$ is the squared geodesic distance between the unique geodesic connecting $a$ and $a'$. Note that the existence of such a unique geodesic is a consequence of the result that every event in a spacetime has a convex normal neighborhood (see \S~\ref{causal}), while the particular characteristics of the smooth functions in eq. \ref{hard} can be found by requiring that $H(a,a')$ obeys the Klein-Gordon equation, eq. \ref{kleingordon}. More details on the construction of this bi-distribution, such as the subtleties that come into play when the geodesic connection $a$ and $a'$ is null, can be found in \cite{Wald1994}. 

Therefore, given the Hadamard ansatz of eq. \ref{hard}, one can define
\begin{equation}\label{32}
    F(a,a')=\langle \psi|\hat{\psi}(a)\hat{\psi}(a')|\psi\rangle-H(a,a'),
\end{equation}
so that one obtains a unique, physical prescription to define $\langle \psi|[\hat{\psi}(a)]^2|\psi\rangle$ (and thus, $\langle \psi|\hat{T}_{\mu\nu}|\psi\rangle$), provided that the state be one such that $F(a,a')$ is a smooth (or sufficiently differentiable) function of $a$ and $a'$. The condition of order of differentiability can be traced back to the physical assumption that the energy momentum tensor be locally conserved \cite{Wald1994}. In essence, these ideas for the definition of $\langle \psi|\hat{T}_{\mu\nu}|\psi\rangle$ on arbitrary spacetimes can be interpreted as a requirement that the ``short distance singularity structure'' of the Hadamard ansatz (corresponding to the first term on the right hand side of eq. \ref{hard}) to be similar to that of $\langle0|\hat{\psi}(a)\hat{\psi}(a')|0\rangle$. Consequently, one can argue that it is physically reasonable to require that, for a state to be considered physically acceptable, $\langle \psi|\hat{\psi}(a)\hat{\psi}(a')|\psi\rangle$ must exist and have a ``short distance singularity structure'' of the Hadamard ansatz, eq. \ref{hard}. States satisfying this condition, known as \textit{Hadamard condition}, are referred to as \textit{Hadamard states}. 

The above discussion can be summarized in the statement that $\langle \psi|\hat{T}_{\mu\nu}|\psi\rangle$ is defined up to a curvature term and is non singular for all Hadamard states. Conversely, $\langle \psi|\hat{T}_{\mu\nu}|\psi\rangle$ should be singular for any non-Hadamard state. Furthermore, one can verify that for a massive Klein-Gordon field in any static, globally hyperbolic spacetime, the static vacuum state is a Hadamard state. It then follows that all states with only finitely many quanta in each mode of the field satisfy the Hadamard condition. Thus, there is a wide class of Hadamard states in globally hyperbolic spacetimes. Also, it can be shown that the Hadamard condition is preserved under dynamical evolution. More precisely, if a state satisfies the Hadamard condition in a neighborhood of any Cauchy hypersurface, then it satisfies the Hadamard condition throughout spacetime. For details and proofs of these properties, see \cite{Wald1994} and references therein.

The importance of Hadamard states for quantum field theory can be analyzed in the following manner \cite{Unruh2017}. Let $(M,g_{\mu\nu})$ be a globally hyperbolic spacetime and $\Sigma$ denote a Cauchy hypersurface. The full system described by the quantum field is then given by the quantum field observables in a neighborhood of $\Sigma$, so that one can use dynamical evolution laws to obtain the states of the field throughout spacetime. Now, consider a division $\Sigma$ into two subsystems, where each one is given by the quantum field observables in the disjoint open regions $\Sigma_1\subset \Sigma$ and $\Sigma_2\subset \Sigma$ that have a common boundary, $S$, such that $\Sigma_1\cup\Sigma_2\cup S=\Sigma$. Let $V_1$ and $V_2$ denote the globally hyperbolic regions with Cauchy hypersurfaces $\Sigma_1$ and $\Sigma_2$, respectively. Consequently, the subsystem $\alpha$ will consist of the field observables in the globally hyperbolic region $V_{\alpha}$, with $\alpha=1,2$, as illustrated in fig. \ref{fig:hada1}. It is then straightforward to see that for any Hadamard state, the two subsystems will be entangled.

\begin{figure}[h]
\centering
\includegraphics[scale=1.5]{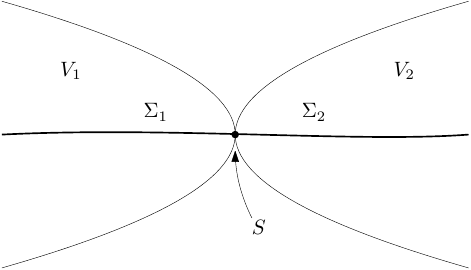} 
\caption{Subsystems with globally hyperbolic regions $V_1$ and $V_2$ in the globally hyperbolic spacetime $(M,g_{\mu\nu})$.}
\caption*{Source: Adapted from UNRUH; WALD \cite{Unruh2017}.}
\label{fig:hada1}
\end{figure}

This can be deduced by considering events $a\in S$, $a'\in\Sigma_1$ and $a''\in\Sigma_2$. The entanglement of the two subsystems can be shown by studying the condition given by eq. \ref{ent} when both $a'$ and $a''$ approach $a$. More precisely, if $|\psi\rangle$ is a Hadamard state and $\hat{\psi}$ is the field operator of a free scalar field, then by the Hadamard condition (see eqs. \ref{hard} and \ref{32}),
\begin{equation}
    \lim_{a',a''\to a}\left[\langle\psi|\hat{\psi}(a')\hat{\psi}(a'')|\psi\rangle\right]=\lim_{a',a''\to a}\frac{U(a',a'')}{(2\pi)^2\sigma(a',a'')},
\end{equation}
which is clearly divergent, as $\sigma(a',a'')$ is the squared geodesic distance between the events $a'$ and $a''$. However, for any physically acceptable state, one also has
\begin{equation}\label{remarks}
    \lim_{a',a''\to a}\left[\langle\psi|\hat{\psi}(a')|\psi\rangle\langle\psi|\hat{\psi}(a'')|\psi\rangle\right]=\left[\langle\psi|\hat{\psi}(a)|\psi\rangle\right]^2,
\end{equation}
which is finite, since the state $|\psi\rangle$ is a superposition of states with finitely many particles (see the expansion of the field operator in terms of annihilation and creation operators, eq. \ref{exp1234}). Consequently, one has
\begin{equation}
    \lim_{a',a''\to a}\left[\langle\psi|\hat{\psi}(a')\hat{\psi}(a'')|\psi\rangle\right]\neq\lim_{a',a''\to a}\left[\langle\psi|\hat{\psi}(a')|\psi\rangle\langle\psi|\hat{\psi}(a'')|\psi\rangle\right].
\end{equation}

By eq. \ref{ent}, one can then conclude that entanglement between states in causally complementary regions always occurs in quantum field theory for any spacetime and Hadamard state. In the context of black holes, the application of this result is evident if one considers the regions $\langle B\rangle\cap\Sigma$ and $\Sigma\backslash B$, with the common boundary being $H$. In particular, for any physically acceptable (Hadamard) state, the field observables inside and outside the black hole at a given time will be entangled.

\newpage

\chapter{Black hole information problem}\label{chapter4}

The results and discussions of chs. \ref{chapter3} and \ref{chapter4} as well as those in appendix \ref{information} give us all the necessary tools to precisely formulate the black hole information problem. In this chapter, we will discuss the consequences of the Hawking effect on the dynamical evolution of the black hole, which mainly follow from the expectation that the emitted particles will carry away energy, angular momentum, and electric charge in such a way as to reduce the parameters of the black hole, causing it to slowly reduce in size. With the assumption that the evolution of the black hole is approximately given by a quasi-static process and that no deviation from semiclassical predictions occurs when the black hole reaches the Planck scale (i.e., $r_s\sim\ell_p$), we will see that information loss is a genuine prediction of semiclassical gravity. In a classical sense, this would correspond to the idea that, after the black hole has disappeared due to the emission of radiation, observers would not be able to access most of the distinguishable information about the energy distribution that gave rise to it other than the three parameters, $(r_s,a,e)$, i.e., information would have been lost. From a quantum mechanics perspective, the process of black hole formation and complete evaporation corresponds to the evolution of a pure state to a mixed one. We will then discuss some proposals for alternatives that would result in a process in which unitary evolution is preserved, but mainly at the cost of questioning the validity of semiclassical predictions in regimes in which one expects it to be an adequate description of the fundamental interactions. Finally, we will review the assumptions and hypotheses that lead to the black hole information problem.

\section{Consequences of particle creation by black holes}\label{evaporation}

In this section, we will analyze the consequences of the Hawking effect for the dynamical evolution of $B\cap\Sigma$ (which we refer to simply as a black hole), where $\Sigma$ denotes a Cauchy hypersurface. Although calculating the precise effect of the emitted particles on the metric (also known as \textit{back reaction} effects) is not a trivial task for four-dimensional spacetimes, one expects that the main consequence of their effective emission will be to reduce the parameters of the black hole at a given time, i.e., its mass, angular momentum and electric charge. In particular, the energy reduction is merely a consequence of the fact that the positive energy flux to infinity implies that there must exist a negative energy flux going into the $B\cap\Sigma$, while the reduction of the Kerr parameter, $a$, and the length electric charge, $e$, are expected due to the fact that created particles tend to carry away such physical properties, as discussed in the remarks below eq. \ref{temperature}. In fact, it can be shown that the reduction in angular momentum and electric charge is much faster than the one in mass, so that a Kerr-Newman black hole will quickly (in comparison with time scales of interest) become a Schwarzschild black hole (see \cite{Page1976, Page1976a, Page1977} for quantitative details). Because of this, in the following discussions we will only consider a Schwarzschild black hole. 

As per eqs. \ref{ksch} and \ref{temperature}, the temperature of the spectrum associated with a Schwarzschild black hole is inversely proportional to its mass, which means that the energy of created particles tends to increase as a Schwarzschild black hole shrinks. More specifically, this relation between the particle energy spectrum and black hole mass ensures that created particles will have a small influence on the spacetime metric for most of their evolution, so that neglecting back reaction effects will be a good approximation when the black hole has a mass much greater than the Planck mass. In this sense, the details of the back reaction effects will have little influence on the event horizon, $H$, and the exterior region, $\Sigma\backslash B$. Consequently, the geometry of the spacetime can be described by a sequence of quasi-static processes in which the mass of the black hole, $r_s$, decreases slowly, with the process of energy loss being approximated by Stefan's law \cite{Swendsen2020} with temperature given by eq. \ref{temperature}. This argumentation leads one to the conclusion that a black hole should completely \textit{evaporate}, i.e., radiate away its mass, in a finite amount of time. Evidently, the Planck scale will be accessible to $H$ at some finite time. Hence, this prediction is only valid if one considers that no deviations from general relativity and quantum field theory occur on that scale. For the moment, let us consider that to be the case and analyze the consequences of a complete evaporation process. 

The conformal diagram of a Schwarzschild black hole that resulted from a spherically symmetric collapse and evaporates completely is depicted in fig. \ref{fig:eva5} \cite{Hawking1975,Wald1984b}.\begin{figure}[h]
  \begin{subfigure}[b]{0.5\textwidth}
  \centering
    \includegraphics[scale=1.3]{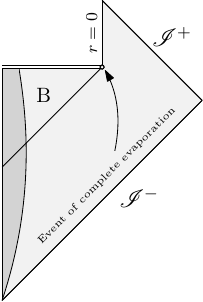}
    \caption{Illustration of the spacetime.}
    \label{fig:eva1}
  \end{subfigure}
  \hfill
  \begin{subfigure}[b]{0.5\textwidth}
  \centering
    \includegraphics[scale=1.3]{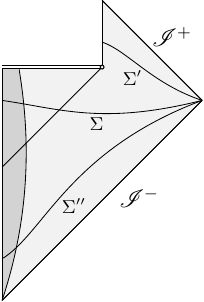}
    \caption{Illustration of spacelike hypersurfaces.}
    \label{fig:eva2}
  \end{subfigure}
  \caption{Conformal diagram of a Schwarzschild black hole that formed from a spherically symmetric distribution of energy and evaporated completely.}
\caption*{Source: By the author.}
\label{fig:eva5}
\end{figure} In this representation of the causal structure of such spacetime, one can see that ``after'' complete evaporation, events located at the origin of spatial coordinates are causally connected to $\mathscr{I}^+$. In contrast, all causal curves that cross $H$ do not reach any spacelike hypersurface ``after'' complete evaporation. This signifies that knowledge of the field observables in any spacelike hypersurface ``after'' the event of complete evaporation, such as $\Sigma'$ in fig. \ref{fig:eva2}, will not suffice to determine conditions in the entire spacetime. Indeed, if one considers any spacelike hypersurface ``before'' complete evaporation, such as $\Sigma$ or $\Sigma''$ in fig. \ref{fig:eva2}, then $D(\Sigma')$ does not contain it. Conversely, one can also conclude that $\Sigma'$ is also not contained in either $D(\Sigma)$ or $D(\Sigma'')$ \cite{Kodama1979}. However, given suitable initial conditions, one can argue that data at either $\Sigma$ or $\Sigma''$ may very well be enough to determine conditions at $\Sigma'$ \cite{Wald1984b}. In this sense, adequate spacelike hypersurfaces ``before'' complete evaporation can be regarded as ones from which conditions in the entire spacetime can be determined. Thus, one can argue that the initial value problem for the matter fields earns the adjective ``well posed'' in the entire spacetime\footnote{Nevertheless, this \textit{heuristic} argumentation does not suffice to argue that black hole evaporation is not in conflict with predictability \cite{Lesourd2019}.}. 

Additionally, the hypothesis that the black hole disappears is in accord with theorem \ref{45678}, whose results are based on the condition that the evolution of the black hole is given along Cauchy hypersurfaces. In any case, this supposed ``loss of determinacy'' (in the sense that knowledge of conditions at $\Sigma'$ will certainly not be sufficient to describe the entire region of the spacetime ``before'' the complete evaporation) is clearly independent of the details of what happens to the distribution that follows the causal curves that cross $H$, relying only on the assumption that they will no longer be accessible to any observer after the black hole has completely evaporated. This ``past-indeterminacy'' at $\Sigma'$ can be precisely quantified by the following line of reasoning. 

Using the representation of fig. \ref{fig:eva2}, consider a time\footnote{This hypersurface can be considered as an ``instant of time'' because the region ``before'' complete evaporation is globally hyperbolic.}, $\Sigma''$, for which the spherically symmetric energy distribution has not collapsed to form a black hole yet, and consider that maximum possible knowledge of its details has been acquired, i.e., it is described by a pure state. Consider now a time, $\Sigma$, in which the black hole has already formed and ``settled down'' to a stationary configuration. In accord with quantum field theory, the quantum field states inside and outside the black hole are entangled (see \S~\ref{hadmard}), which means that the state of the system outside of the black hole can only be described by a density operator (see appendix \ref{informationtheory}). Nonetheless, the full state of the system at $\Sigma$ still has the same purity as the one at $\Sigma''$. At this point, the black hole is emitting radiation as a gray body with an approximately thermal spectrum, with temperature given by eq. \ref{temperature} and absorptivity $\Gamma$, and slowly evaporating due to loss of energy, in accord with semiclassical properties. If the evaporation occurs completely, then the state of the field at a ``time'' $\Sigma'$ will still be mixed, as it remains entangled with the states inside the black hole, even though it no longer exists. In essence, the system will evolve from a pure state at $\Sigma''$ to a mixed state at $\Sigma'$, which in the vocabulary of mixtures, corresponds to a \textit{loss of quantum coherence}. Hence, the information about the quantum field state at $\Sigma'$ (which is the only region accessible to observers after complete evaporation) will not be enough to determine the state of the system at a time $\Sigma''$ or $\Sigma$, i.e., information will be lost. 

Generalization of this argumentation for an energy distribution which is not spherically symmetric or stationary follows from the results and discussions in chs. \ref{chapter2} and \ref{chapter3}. More precisely, the stationary state conjecture, the black hole uniqueness theorems, and the fact that the event horizon acts in a way as to make sure that the details of the collapse pose negligible influence over the spectrum of created particles measured at late times at $\mathscr{I}^+$ ensure that the same conclusion holds. 

The description of the process of black hole formation and evaporation stated above merits two important remarks. First, the loss of information can be clearly traced back to the interpretation that the evolution of the black hole ``removes'' the degrees of freedom of the quantum field from observers outside the black hole, as its complete evaporation means that those degrees of freedom are no longer accessible. In essence, the complete evaporation of a black hole can be interpreted as leaving a lasting ``deterministic pathology'' on the spacetime, signified by the fact that conditions on ``late time'' spacelike hypersurfaces do not suffice to entirely determine those ``before'' complete evaporation. The second remark concerns the evolution from a pure state to a mixed one\footnote{The point of this analysis is that, regardless of the purity of the initial state of the energy distribution that gave rise to the black hole, the final state will be in a mixed state.}, which can be wrongfully interpreted as a breakdown of postulates of quantum mechanics. In fact, the evolution from a pure state to a mixed one should not be confused with a lack of conservation of probability, which would in turn be extremely problematic. Note that this is not the case because the dynamical evolution of the quantum field in the process of black hole formation and complete evaporation is not given along Cauchy hypersurfaces. Precisely, because conditions at $\Sigma'$ \textit{certainly} do not suffice to determine conditions at $\Sigma''$, evolution of the quantum state from $\Sigma''$ to $\Sigma'$ is expected to be non-unitary (see eqs. \ref{eqr}, \ref{eqs}, and \ref{45}). In this sense, the formation and complete evaporation of a black hole may be interpreted as producing an \textit{open system}, which would result in a \textit{physical} non-unitary evolution.

It should be noted that this conclusion of loss of quantum coherence is in accord with quantum field theory and, in fact, it even happens in any ``well behaved'', globally hyperbolic spacetime for suitable choices of ``initial'' and ``final'' hypersurfaces. In essence, one can obtain the same conclusion if one considers the dynamical evolution of the quantum field observables for a massless Klein-Gordon field in Minkowski spacetime from a Cauchy hypersurface (e.g., a hyperplane) to, say, an ``asymptotically null'' hyperboloid (see fig. \ref{fig:min56}). The exact same phenomenon of loss of quantum coherence occurs in this globally hyperbolic spacetime, in which the hypersurface $\Sigma'$ fails to be a Cauchy hypersurface, and thus, evolution of the state from $\Sigma$ to $\Sigma'$ corresponds to the evolution of a pure state to a mixed state, i.e., a non-unitary evolution (see \cite{Unruh2017} for a physical example of this process). In this sense, physical non-unitary evolution is a prediction of quantum field theory for suitable choices of ``initial'' and ''final'' hypersurfaces.

\begin{figure}[h]
\centering
\includegraphics[scale=1.3]{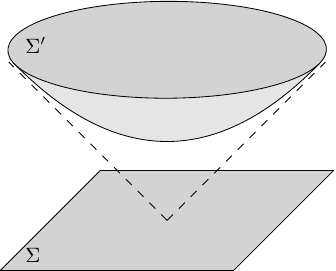} 
\caption{Spacetime diagram of Minkowski spacetime showing a Cauchy hypersurface, $\Sigma$, and a hyperboloid, $\Sigma'$.}
\caption*{Source: Adapted from WALD \cite{Wald1994}.}
\label{fig:min56}
\end{figure} 

In the exact same manner, during the entire process of black hole formation and evaporation, the quantum field respects deterministic equations of motion, and presents no kind of pathological behavior. The same can be said about the structure of spacetime and the physics governing it, except perhaps at the singularity. Thus, under the assumption that a black hole evaporates completely, the result of information loss is in complete accord with predictions of quantum field theory in curved spacetime. 

\section{Alternatives to information loss}

There are, however, proposals of alternatives which revolve around possibilities regarding the physics of black holes in order to restore unitary evolution. In this section, we briefly summarize some of the most notable of them, and discuss how and where the picture stated in the last section fails.

One of the proposals to restore unitarity is known as \textit{remnants} (see, e.g., \cite{Bekenstein1994}), which states that the evaporation process does not continue when the black hole reaches the Planck scale, so that a Planck-sized object containing all the information would be the final state of a black hole (see region I in fig. \ref{fig:bh1} for an illustration of where the conformal diagram of fig. \ref{fig:eva5} would fail in this framework). Another perspective to avoid information loss is known as \textit{firewalls} (see, e.g., \cite{Almheiri2013}), which states that the evaporation process does occur as simply as predicted by semiclassical gravity\footnote{The firewall proposal was mainly motivated by the concept of \textit{black hole complementarity} \cite{Susskind1993, Susskind1994}, which in summary, states that information that enters the event horizon is accessible on the outside.}, in particular, that there is less entanglement (or none at all) between pairs of created particles (i.e., those that reach infinity and those that represent the negative energy flux into the black hole). As a consequence, an infalling observer would perceive high energy field quanta at the event horizon, and there would be no evolution from a pure to mixed state in the process of evaporation (see region II in fig. \ref{fig:bh1}). Additionally, there is the view of \textit{fuzzballs} \cite{Mathur2005}, which is named after the proposed structure that should follow from gravitational collapse, rather than a black hole. In this view, the whole concept of evaporation and information loss is avoided by simply stating that black holes do not form, as some new physics would prevent it (see region III in fig. \ref{fig:bh1}). 

\begin{figure}[h]
\centering
\includegraphics[scale=1.3]{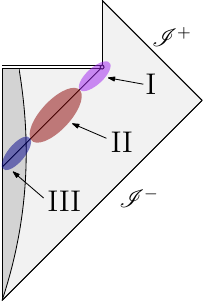} 
\caption{Regions of possible failure of semiclassical gravity in the conformal diagram of a spherically symmetric black hole formed from gravitational collapse.}
\caption*{Source: By the author.}
\label{fig:bh1}
\end{figure} 

Although these proposals lead to a more complete discussion and shed light on the reliable fundamental structure of the theories involved in the process of black hole formation and evaporation, there are significant flaws in arguing that semiclassical predictions do not hold at arbitrary scales. More specifically, proposals that seek to invalidate predictions of general relativity and quantum field theory in regimes where one justifiably expects them to be an adequate description of nature are problematic. For example, in order for the fuzzball proposal to be an adequate possibility, it would be necessary to contradict predictions of general relativity concerning gravitational collapse. In other words, since the formation of an event horizon does not necessarily require high energy or high curvature conditions, proposing that semiclassical gravity presents flaws in arbitrary energy scales to prevent the formation of \textit{any} black hole is contradictory. Recall that the scales of curvature at the event horizon of a Schwarzschild black hole are proportional to $r_s^{-2}$ (see eq. \ref{curvature}), which means that supermassive collapses would have an event horizon form in arbitrarily low curvature regimes\footnote{Indeed, because an event horizon has no local significance, arguments based on some mechanism that would stop the formation of any black hole are inconsistent.}. Hence, it is hard to see how these claims would be useful to describe the impossibility of the formation of an event horizon in the collapse of such distributions, which is an instance of scales where semiclassical predictions are expected to hold very accurately. Regarding the firewall proposal, the same argument comes into play, since the explanation for the reduction or lack of entanglement (i.e., deviation from the semiclassical picture) would have to hold at arbitrarily low energy and curvature regimes.

That being said, there are no similar remarks to the remnant proposal, since the scale at which the evaporation process is proposed to stop is one where details on the quantum gravitational phenomena would be necessary to accurately predict the fate of the black hole. Nevertheless, following such a hypothesis, the Planck-sized object would need to possess an extremely high entanglement entropy (and in fact, arbitrarily high, since \textit{any} black hole would have to result in a remnant), which would clearly need to exceed the ``ordinary'' proposal of a black hole entropy, as given by eq. \ref{entropy}. In addition, in order for observers to physically access the information, an interaction with the remnants would be necessary, and due to the arbitrariness present in the amount of entropy it would have, it is unclear how such interaction would work in the general case. In contrast, if remnants do not interact with the rest of the universe, then information can be basically assumed to be inaccessible. Lastly, the exact same reasoning applies to proposals that seek to invalidate the conclusion of information loss by arguing that it goes into a ``baby universe'' (say, if the singularity is actually a ``bridge'' between two universes). That is, if information is inaccessible to any observer at a time ``after'' complete evaporation, then it is hard to argue why it would not merit the conclusion of being lost.

Nonetheless, there is no way to make a precise prediction concerning the result of the evaporation process without a complete theory of quantum gravity, which will allow for a complete evaluation of the particle creation, back reaction effects, and determination of black hole physics beyond Planck scales. This is precisely the content of the \textit{black hole information problem}, which is this question regarding the result of the process of black hole evaporation. In particular, one would like to know if information is truly lost, or how this inadequate conclusion can shed light on the regime in which semiclassical analysis is expected to fail. That is, although general relativity and quantum field theory provide a good description of black hole physics on most of the accessible energy scales, and in the author's opinion, point to the most plausible result of information loss, it is possible only to speculate as to how this will be addressed in an adequate theory of quantum gravity.

\section{Nature of the black hole information problem}\label{nature}

The evolution of the black hole region in light of semiclassical developments points to the physical prediction of information loss. Evidently, the formulation of this prediction and the state of affairs regarding it relies on several assumptions and hypotheses that arise from the theories involved in the semiclassical framework of gravity. The purpose of this section is to discuss these assumptions and hypotheses. It should be noted that an extensive discussion concerning the physical principles behind the theories involved is beyond the scope of this work. In particular, our starting point will be to assume the validity of the postulates of general relativity (see chs. \ref{Introduction} and \ref{chapter1}) and the quantization process in quantum field theory (see \S~\ref{qft}). The reader can find an extensive discussion on these concepts, as well as experimental evidence for their validity, in \cite{Will2014,Weinberg1995, Zee2023}.

With these remarks, consider, first, some results that arise in the purely classical analysis in the framework of general relativity. Arguably, the main point is the physical plausibility of the existence of black holes. As discussed, gravitational collapse and pathological regions are genuine predictions of the theory (see \S~\ref{sing}), which mainly follow from analysis of the stability of astronomical bodies and the eventual formation of trapped surfaces (see \S~\ref{schsec}). Considering the cosmic censor conjecture (see \S~\ref{bhsec}), the existence of black holes can then be seen as a way to preserve determinism in spacetimes, rather than a generic prediction of the theory. In particular, the physical assumption behind this comes from the expectation of global hyperbolicity (see \S~\ref{causal}), namely, that a physically reliable spacetime is one in which the dynamical evolution of physical fields is a ``well posed'' problem. Additionally, by defining black holes in the context of asymptotically flat spacetimes (see \S~\ref{flat}), one makes use of the frame of reference of observers at the asymptotic region to analyze physically significant parameters. As it is customary in the analysis of physical systems, one considers that the system is isolated for a simpler analysis. Although no physical black hole would constitute an isolated system, a theoretical treatment of one as such is clearly justified by the expectation that an idealized analysis would yield results that are of significance for observations.

With the existence and definition of black holes in the framework of general relativity being well justified, the most notable property to the formulation of the information problem comes from the uniqueness theorems (see \S~\ref{kerr}), which are also rigorous results in differential geometry and topology. However, it is evident that its physical relevance (i.e., applicability to real scenarios) follows directly from the assumption that black holes reach a stationary final state. Indeed, the expectation that physical properties not associated with invariance over a one-parameter group of isometries (see \S~\ref{symmetry}) should be radiated away is precisely the justification for this conjecture. Although one might argue that the stationary state conjecture does not rise as a consequence of a fundamental physical assumption such as determinism in spacetimes, it is undeniable that the similar character of other dynamical systems and the prediction of gravitational waves justify it. In summary, the cosmic censor and stationary state conjecture are reasonable assumptions in the classical realm.

When quantum field theory in curved spacetime is taken into account, it is natural to question if these conjectures and properties hold in the same manner. With regard to the cosmic censor conjecture, the assumption that gravitational collapse should result in a black hole remains reasonable, since versions of the singularity theorems with weakened energy conditions (e.g., quantum energy inequalities \cite{Fewster2011}) are valid\footnote{However, general developments taking into account back reaction effects and interactions may still change these conclusions.}. As a matter of fact, considering that violations of the weak version of the cosmic censor conjecture can be found in quantum mechanics \cite{Matsas2009}, and that theorem \ref{thekappa} need not hold in such a framework, one can justifiably challenge its physical reliability in a more general manner. In other words, singularities, either naked or ``concealed'' by event horizons, may very well be a genuine feature of a complete theory of gravitation.

Now, the main concern with singularities is not necessarily with their existence, but rather, with what their existence implies in a theory in which one does not have the necessary tools to describe them in a satisfactory manner. In other words, the intrinsic quantum process for which the weak version of the cosmic censor conjecture may not hold, and the expectation of determinism (or at least, predictability\footnote{See details below.}) in a physical theory support the hypothesis that, with a complete theory of quantum gravity, an adequate description of singularities might be possible \cite{Penrose1979}. This is meant not only in the sense that one would have information about the quantum gravitational phenomena that occur in the vicinity of a singularity, but also, details about the appropriate boundary conditions one would impose there. Evidently, this does not invalidate predictions of black holes or their physical importance in the semiclassical context, but simply points to the possibility that not every singularity has to be concealed by an event horizon in a complete theory of quantum gravity. 

Currently, as far as the author is aware, there are no arguments that suggest that the stationary state conjecture should not hold in the semiclassical framework, but the same cannot be said about the black hole uniqueness theorems. Indeed, it is well known that the uniqueness theorems summarized in \S~\ref{kerr} are valid only when suitable classical fields are present \cite{Wald1984}, i.e., they are valid for Einstein-Maxwell equation, so that full consideration of more general equations can bring to light some ``hair''\footnote{In the sense that time independent black holes may not be described uniquely by $(r_s,a)$ and finitely many ``charge'' variables.}. Nonetheless, one can argue that such time independent solutions not described by the Kerr-Newman metric would be subjected to some sort of \textit{generalized uniqueness theorems} (see \cite{Chruściel2012} for an extensive review of solutions and arguments supporting this hypothesis). 

Turning our attention now to the pertinent effects that arise in the semiclassical framework, perhaps the most important assumptions and hypotheses are those that lead to the effective particle creation effect and the conclusion of entanglement between field observables inside and outside the black hole. First, one has the requirement of global hyperbolicity, which is necessary for one to have a ``well posed'' (see \S~\ref{causal}) problem for the dynamical evolution of the quantum field observables (see \S~\ref{qft}). For instance, note that even though it can be deduced that a hypersurface such as $\Sigma$ in fig. \ref{fig:eva2} does not obey $D(\Sigma)=M$, this does not mean that $M$ is not globally hyperbolic. More precisely, it only means that $\Sigma$ is not a Cauchy hypersurface. Additionally, since it can be argued that data on $\Sigma$ can still be regarded to be enough to determine conditions at $\Sigma'$, evolution of the field observables from $\Sigma$ to $\Sigma'$ could still constitute a ``well posed'' problem, regardless of global hyperbolicity\footnote{However, one would still expect that suitable causality conditions are satisfied \cite{Lesourd2019}.}. Second, one has to consider the geometric optics approximation, which is well justified for analysis of scalar solutions of the minimally coupled Klein-Gordon equation with high frequency (see \S~\ref{creation}). In particular, the analysis of high frequency solutions is only relevant to the ``traced backward'' form of solution from $\mathscr{I}^+$ to $H$, as from $H$ to $\mathscr{I}^-$ the geometric optics approximation is justified by the arbitrary blueshift at the event horizon.

In fact, such a property is also responsible for ensuring that details of the energy distribution that gave rise to the black hole pose negligible influence over the spectrum of created particles. Because of this and the rigorous arguments that are involved in the derivation of the Hawking effect, one is tempted to believe that perhaps the only flaw in its prediction lies in the fact that the exponential redshift suffered by the outgoing particles means that they originate from modes with extremely high frequency. In other words, the approximately thermal spectrum measured at $\mathscr{I}^+$ should possess much higher wavenumbers when the particles are located closer to the event horizon, and may very well constitute a situation in which Planck scales are accessible. Not only that, consideration of only high frequency modes from $\mathscr{I}^+$ to $H$, and neglect of back reaction effects in the derivation of the Hawking effect lead one to the question of how much information Hawking radiation would actually carry, as well as if its deduction is independent of Planckian physics. 

Supposing for the moment that one is satisfied with the semiclassical arguments that give rise to Hawking radiation, the prediction of information loss follows from the conclusion of entanglement between the quantum field observables inside and outside the black hole. As briefly commented in appendix \ref{entanglement}, entanglement is an intrinsic property of quantum mechanics derived from the theory alone, and experimental evidence for it is abundant. Nonetheless, it sparked a powerful discussion regarding its role in a physical theory that is both local and realistic \cite{Einstein1935}. By \textit{realistic}, it is meant that the physical quantities that are predicted by the theory have definite values, i.e., the observables of the theory are consistent with an objective reality, independent of the process of measurement. By \textit{local}, it is meant that measurements made on spacelike separated (i.e., causally disconnected) systems cannot be relevant to one another. Although these philosophical perspectives concerning the nature of a successful theory are (arguably) justifiable, one can deduce that quantum mechanics is not, in general, in accord with such principles, as per \textit{Bell's inequalities}\footnote{These inequalities can be understood as a way to predict constraints on experiments using a local realistic theory, and then compare them with predictions of quantum mechanics. See \cite{Nielsen2010} for a detailed discussion.} \cite{Bell1964}. Certainly, it is not difficult to see that entanglement implies that measurements made on arbitrarily distant entangled systems can influence one another. Not only that, experiments that show violations of Bell's inequalities (see, e.g., \cite{Aspect1981, Aspect1982}) support the conclusion that local realistic theories are not adequate for the description of quantum phenomena. In this purely theoretical sense, entanglement is a necessary consequence of a physical theory that accurately describes quantum phenomena.

In the context of quantum field theory in curved spacetime, entanglement rises as an intrinsic feature of causally complementary regions due to Hadamard states (see \S~\ref{hadmard}). As discussed, the physical assumption behind this class of states is that they are those for which the expectation value $\langle \psi|\hat{T}_{\mu\nu}|\psi\rangle$ is non singular. In light of the expectation that quantum gravitational effects are to be governed by equations whose source is directly related to some notion of ``energy-momentum expectation values'', the requirement for physically acceptable classes of states to be those that obey the Hadamard condition is well justified \cite{Wald1994}. In fact, entanglement between the field observables inside and outside the black hole, as predicted by the requirement that states obey the Hadamard condition (see the remarks below eq. \ref{remarks}), is the justification for the nature of the Hawking radiation, i.e., it justifies its approximately thermal character. Precisely, the same exact explanation arises when one considers the Fulling-Davies-Unruh effect (see \S~\ref{qft}) in Minkowski spacetime, in which the thermal radiation measured by uniformly accelerated observers originates from the entanglement of the field observables in the left and right Rindler wedges \cite{Wald1994}. Certainly, in the context of black holes, the particles emitted to $\mathscr{I}^+$ at late times are (most likely) weakly correlated with each other (i.e., there are weak correlations between measurements of particles emitted in different modes), but the fact that they are correlated with particles that enter the black hole (i.e., those corresponding to the negative energy flux) is a direct consequence of the nature of the entanglement \cite{Unruh2005, Unruh2017}.

Given that entanglement and the approximately thermal character of Hawking radiation are well justified, the result is that the process of black hole formation and complete evaporation will constitute the evolution from a pure state to a mixed one. The conclusion of information loss follows precisely from this loss of quantum coherence, as quantified by the von Neumann entropy (see \S~\ref{von}), but the physical plausibility of such evolution has also been the topic of much discussion. Namely, it has been argued that this process could be stated in the requirement that the dynamical evolution of $\hat{\rho}$ has to be given by a more general equation than the Schrödinger equation (see, e.g., \cite{Sakurai1994}), which in turn, would result in violations of energy-momentum conservation \cite{Banks1984,Ellis1984}. Indeed, no physical process known to date provides an account for violation of unitarity (in the sense that a pure state evolves into a mixed state) for dynamical evolution of a system over Cauchy hypersurfaces. However, we stress again that in the process of black hole formation and evaporation, the violation of unitarity is a consequence of semiclassical predictions in light of the fact that any spacelike ``late time'' hypersurface will not suffice to determine conditions in all spacetime. Consequently, it should be noted that, in this context, the failure of the evolution of the system from a pure to a mixed state does not imply that probability is not conserved (which, arguably, is a fundamental principle), but rather, that at ``late times'' no observer will be able to make a complete description of the system via dynamical evolution to the past. It is in this sense that information loss is in complete accord with semiclassical predictions following the assumption of complete evaporation \cite{Unruh2017}, and one can also argue that such evolution may happen in more general processes in light of a complete theory of quantum gravity \cite{Unruh1995b}.

The last point worthy of discussion is not about an assumption that contributes directly to the formulation of the black hole information problem, but rather, how one can interpret semiclassical properties of black holes from a ``full picture'' perspective. Namely, since one may interpret the Hawking effect as a manifestation of the physical temperature of a black hole (see \S~\ref{thermo}), it is of interest to understand which relevant assumptions about the information problem corroborate or invalidate such interpretation. First, it should be noted that, by temperature, we refer to a behavior of a system in which its microscopic degrees of freedom result in an emission of radiation in accord with a black body spectrum. Although the predominant behavior of the spectrum of created particles is given by a black body spectrum (see eq. \ref{flux1}), most likely deviations will arise from a detailed calculation (in the sense that considerations of back reaction effects, higher range of frequency and quantum gravitational phenomena may change its precise form). Nonetheless, given the constraints that come into play in the development presented in \S~\ref{creation}, one can show that the content of the Hawking radiation is in agreement with a correspondence with black body radiation. Indeed, in Wald's words \cite{Wald1975}:

\noindent``\textit{The density matrix for emission of particles to infinity at late times by spontaneous particle creation resulting from spherical gravitational collapse to a black hole is identical in all aspects to that of black body thermal emission at temperature $kT=\hbar\kappa/2\pi c$}''.

Generalizations of such a statement to a Kerr-Newman black hole follow from the stationary state conjecture and most of the arguments presented in the derivation of the Hawking effect in ch. \ref{chapter3} (see \cite{Hawking1975, Wald1975} for a detailed discussion). Thus, the agreement of the predominant part of the Hawking radiation with a black body spectrum is not just numerical (see eq. \ref{flux1}), it is identical in all aspects (clearly, the gray body factor does not affect the nature of the spectrum). This corroborates the interpretation of $\kappa$, as given by eq. \ref{temperature}, as the physical temperature of a black hole, and the discussion in \S~\ref{thermo} supports the idea that one may interpret the stationary final state of a black hole as physically equivalent to any other system in thermodynamic equilibrium \cite{Wald2001}. In this sense, the semiclassical properties are known to constitute the alleged \textit{black hole thermodynamics}. 

However, if that is truly the case, then the physical assumptions that give rise to the semiclassical properties derived in chs. \ref{chapter2} and \ref{chapter3} should have a deeper meaning. For instance, when one considers the property that $\kappa=0$ should be unattainable, a possible interpretation is that this is just a restatement of the cosmic censor conjecture. Similarly, when one considers the constancy of $\kappa$ over the event horizon, one may derive it from Einstein's equation and the dominant energy condition or from purely geometrical arguments. Additionally, the fact that one can derive an analogous ``energy conservation law'' (see eq. \ref{first}) for more general metric theories of gravity, and the arguments surrounding the generalized second law corroborate the ideas stated above. Consequently, if one were to consider a physical association between the geometrical and thermodynamic properties of black holes, it would be natural to question what the underlying physical correspondences are between the arguments that lead to these sets of relations, both from a geometrical and thermodynamic perspective. In essence, it is unclear what would be the microscopic properties that give rise to the thermodynamic properties of black holes, and although many developments were proposed \cite{Wald2001}, it is safe to say that this is still an open question.

In light of the discussions above, it is clear that analysis of the black hole information problem is of extreme interest to studying fundamental properties of the gravitational interaction. Thus, it is critical to the development of a complete quantum theory of gravity. Evidently, this does not follow because it is ``paradoxical'', as many references tend to label it, but in fact, because it provides some of the best clues available to study gravitation in regimes where semiclassical gravity is expected to breakdown. In essence, considering the success of both general relativity and quantum field theory to describe phenomena at currently accessible energy scales, it is natural to suppose that they are merely consequences (or effective theories) of a deeper, more comprehensive theory of the fundamental interactions. It is in this sense that the conclusion of information loss and the state of affairs surrounding it are expected to point to aspects of these theories that will help one reevaluate the assumptions of reality. In a more provocative tone, it is tempting to believe that the black hole information problem may lead to progress towards the development of a new conceptual framework.

\chapter{Conclusions and perspectives}\label{conclusion} 

In this work, we reviewed the classical and semiclassical description of black holes and presented the black hole information problem. With a precise formulation of the problem, we were able to identify the many assumptions and hypotheses that give rise to it and discuss their physical plausibility. The main conclusion of this analysis is that the conjectures and hypotheses necessary for the classical description of black holes and the derivation of the semiclassical particle creation effect are well justified in the framework of quantum field theory in curved spacetime. Consequently, under the ``natural'' assumption of complete evaporation and considering the Einstein-Maxwell equation, the conclusion of information loss is in complete accord with semiclassical predictions, and does not stand in contradiction with any known phenomena. However, it is clear that by studying the evolution of the black hole, at some point one will require knowledge of physics at the Planck scale. Without the development of a complete, satisfying theory of quantum gravity, it is the author's opinion that the conclusion of information loss stands as the most adequate one given the current knowledge about gravitation.

Nevertheless, the developments presented in this work point to several open questions, and the perspective for future research is extensive. We briefly summarize below some of the most intriguing of them and possible directions for progress in black hole physics and in the development of a quantum theory of gravity.

\section{Hawking effect and information in Hawking radiation}

Although the predominant character of Hawking radiation is thermal, the geometric optics approximation, neglect of back reaction effects, and lack of details of quantum gravitational effects at the Planck scale raise questions about how much information it actually carries. Indeed, although there are clear constraints on the thermal character of Hawking radiation \cite{Visser2015}, a quantification of its precise dependence on details of the gravitational collapse can only be analyzed in light of a complete theory of quantum gravity. Similarly, these approximations also lead to questions regarding the universality of the Hawking effect \cite{Unruh2005}, as well as the precise origin of Hawking radiation \cite{Giddings2016}. Perhaps the most notable progress towards answers to these questions comes from considerations of \textit{analogue gravity} (see, e.g., \cite{Barceló2011} for an extensive review), which investigates analogue behavior of gravitation in other physical systems. First proposed as an analysis of \textit{moving mirrors} \cite{Fulling1976, Birrell1984} and \textit{sonic analogue} black holes \cite{Unruh1981, Unruh1995}, the prediction of a thermal character in such systems can be understood as corroboration of the Hawking effect and its existence for physical black holes. Evidently, observation and experimentation of analogue black holes is a much more suitable task for investigation of the Hawking effect \cite{Modugno2021, Kolobov2021, Wilson2011}, and one can justifiably expect that it may lead to new insights into fundamental aspects of semiclassical gravitation. 

In a more general context for physical black holes, one can use the expected contributions of Hawking radiation to the observable effects in the $\gamma$-ray background, and then evaluate the possible constraints on the population of black holes in a given mass range. Most notably, one can analyze the possibility of black holes that formed in the early epochs of the universe to be candidates for \textit{dark matter} (see, e.g., \cite{Green2021}). Although experimental evidence for black holes (see, \cite{Genzel2010, Event2019, Narayan1995} for notable observations and \cite{Rees1998} for an extensive review) supports the prediction of the Hawking effect, experimental confirmation of it may be out of reach for the foreseeable future. This is clearly related to the fact that the most notable (i.e., supermassive) black holes have an extremely small associated temperature. In this sense, purely theoretical advancements in a theory of quantum gravity stand as the most promising direction for general investigations of the Hawking effect in gravitation.

\section{Degrees of freedom in black holes}

Following the interpretation of the Hawking effect as a manifestation of a physical temperature of a black hole, many more considerations arise for its thermodynamic properties. Notably, lying at the center of these considerations is the question of the degrees of freedom responsible not only for black hole entropy, but also for its temperature. The first direct calculation of black hole entropy followed from considerations of \textit{Euclidean quantum gravity} \cite{Gibbons1977b,Gibbons1993}, which provided an interpretation of metric contributions to the partition function. The generality of these results has also been discussed further in \cite{Iyer1995,Hawking2005}. Different perspectives arising from \textit{quantum geometry} \cite{Ashtekar1998} and entanglement entropy
\cite{Bombelli1986, Solodukhin2011} have also been considered as interesting approaches to reproduce the ``macroscopical'' result of an entropy that is proportional to the area of an event horizon. However, arguably the most quantitatively successful calculation is the one following from \textit{string theory} \cite{Strominger1996,Maldacena1997,Horowitz1998,Hubeny2015}, which managed to reproduce the numerical factor of $1/4$. Although these developments are promising, the nature and location of the degrees of freedom of a black hole is still an open question \cite{Wald2001}. Because of this, one can also raise questions regarding Boltzmann's interpretation of entropy. In other words, in a gravitational context, perhaps entropy may simply present itself as a mathematical aspect, as it arguably does for the explanation of heat engines. 

\section{Thermodynamic aspects of gravitation}

The issue of thermodynamic properties of the gravitational interaction is also of interest, in which black hole thermodynamics presents itself simply as the most notable front \cite{Wald2001, Mathur2023}. Indeed, the idea that one may interpret the area of a black hole as its physical entropy sparked the discussion about thermodynamic properties of any null hypersurface, in the sense that one may attribute entropy to any region of a spacetime \cite{Srednicki1993}. That this entropy can be interpreted as a result of entanglement entropy (as is the case for black holes) may follow for any region of spacetime \cite{Raamsdonk2010}. Most intriguingly, one can consider a point of view in which such entropy could be interpreted as a measurement of the degrees of freedom of a more fundamental theory of gravitation, rather than just those of a region enclosed by an event horizon. An extensive review of the many thermodynamic aspects of gravity can be found in \cite{Chirco2011}.

Additionally, it is also attractive to consider a perspective in which geometrical and physical arguments that come into play in the framework of general relativity are related to the justification of the laws of thermodynamics. That is, if the principles that give rise to a geometrical interpretation of gravity are simply different manifestations of those that give rise to the laws of thermodynamics. For example, how one should interpret assumptions such as the cosmic censor conjecture or the stationary state conjecture in light of their importance for the derivation of the classical properties of black holes in ``equilibrium''.

\section{Quantum gravity and black hole information problem}

With regard to black holes and the information problem in general, the search for clues of how physics at the Planck scale could interfere with the conclusion of information loss is ongoing and extensive. Proposals for black hole quantization \cite{Stephens1994} preserving unitarity, presence of ``soft hair'' due to the \textit{gravitational wave memory effect} (see, e.g. \cite{Hawking2016,Hawking2017, Favata2010}), and \textit{emergent gravity} \cite{Verlinde2011,Verlinde2017,Visser2019} have surfaced, but arguably the most notable developments are those stated in the AdS/CFT (Anti-de Sitter/conformal field theory) conjecture \cite{Maldacena1999,Hubeny2015}. In essence, this conjecture can be understood as the assertion that a complete theory of gravity is \textit{dual}\footnote{Namely, \textit{duality} is a relation between different theories that can describe the same physical phenomena \cite{Ammon2015}. In the context of this particular duality, the word \textit{holography} is used to refer to the fact that one of the theories is defined in a higher dimension.} to a \textit{conformal field theory}\footnote{A conformal field theory is a quantum field theory that is invariant under conformal transformations \cite{Francesco1996}.} defined on the boundary of \textit{anti-de Sitter} spacetime \cite{Hawking1973}. Namely, this conjecture is an example of the \textit{holographic principle} \cite{Bousso2002}, and stands as one of the intriguing developments in the context of quantum gravity and the black hole information problem. Indeed, the main argument against information loss in this proposal lies in the fact that the conformal field theory for which a complete description of gravity is dual would not admit non-unitary evolution. Thus, since it would be possible to describe the process of black hole formation and evaporation through the dual conformal field theory, the evolution of a pure state to a mixed one would not occur, even in a complete theory of quantum gravity. Although promising, these proposals for developments in black hole physics that may shed light on the evolution of the black hole region are still lacking a precise formulation and, arguably, do not yet constitute a satisfying solution to the question of the final state of a black hole.

In fact, the effect of black holes on the quantum superposition of states has been argued to be much more radical, in the sense that decoherence (i.e., loss of quantum coherence) of quantum superposition is expected to occur for states simply in the presence of a black hole \cite{Danielson2022}. Additionally, in a more general fashion, one can also show that the same conclusion holds for any Killing horizon \cite{Danielson2023}. The full extent of these properties, which at first seem to be restricted to black holes but then are generalized to other gravitational systems, as well as the cosmic censor conjecture, some ``generalized uniqueness theorems'' for black holes, and the nature of Hadamard states, again falls on the developments of a quantum theory to describe the gravitational interaction.

\section{Concluding words}

As discussed, most of these questions and possible directions of study are related to the development of a complete theory of quantum gravity. In a more pragmatic approach, perhaps if one were able to describe the black hole region in a local manner, i.e., without the need to rely on a potential asymptotic structure of a spacetime, a more adequate description of the pertinent phenomena would surface. 
In fact, one can only guess as to what would be the full picture of a more general theory of fundamental interactions, which may even enforce a different philosophical perspective at low energy and low curvature regimes, in a very similar way as general relativity did to Newtonian gravity. For instance, one could be tempted to believe that developments in a quantum theory of gravity could shed light on the nature of the concept of entropy, whose nature even in flat spacetime is far from having the status of universal agreement. In other words, although the black hole information problem may simply be stated as a question concerning the final state of a black hole in the semiclassical framework of gravity, its resolution may impact other theories. It remains for future research to delve deeper into these questions and possible connections.

\postextual

\bibliography{References/References.bib}

\begin{apendicesenv}
\partapendices

\chapter{Manifolds}\label{A}

The purpose of this appendix is to provide an objective review of the basic mathematical framework necessary for the development of this work. We start by presenting the definition of a topology and several properties of interest in topological spaces, as well as proving two results to illustrate their application. From such formalism, an $n$-dimensional manifold is defined as a topological space such that every point in it belongs to a set that can be identified with $\mathbb{R}^n$. Consequently, the notion of tensors on manifolds arises naturally from that of a tangent vector space at each point, and a spacetime is defined as an $n$-dimensional manifold that has a Lorentzian metric tensor field defined on it. This definition will, evidently, carry properties one expects a physical manifold should have. We then discuss concepts of interest and also additional structures that follow directly from the metric tensor, such as the Levi-Civita connection, the notion of curvature, and the integration of functions on arbitrary manifolds. 

It should be noted that this is not meant to serve as a pedagogical introduction, but rather, as a reference to the development of relations and the arguments presented in this work. The interested reader can find a more pedagogical, detailed presentation of such subjects in the references on which the construction of this appendix was based on \cite{Wald1984, Misner1973, Poisson2004, Hawking1973, Darling1994, ONeill1983, Evans, Lee2003,Lee2011}.

\section{Topological spaces}\label{A1}

A topology is a structure one imposes on a set, in the same manner as one imposes an algebraic structure on numbers. Such a structure may seem rather arbitrary, but a topology on a set can be naturally induced by an additional structure, e.g., an inner product over a vector space. In particular, a topology on a set is useful because it tells one how elements in a set are ``connected'', even though the elements of the pertinent set can be any mathematical object. 

More precisely, a \textit{topological space}, $(X,\mathscr{T})$, consists of a set, $X$, and a collection, $\mathscr{T}$, of subsets of $X$, called the \textit{topology}, satisfying the following properties.
\begin{enumerate}
    \item[(1)] The entire set, $X$, and the empty set, $\emptyset$, are in $\mathscr{T}$.
    \item[(2)] The union of an arbitrary collection of subsets is in $\mathscr{T}$, i.e., if $S_{\mu}\in\mathscr{T}$, then $$\bigcup_{\mu}S_{\mu}\in\mathscr{T}.$$
    \item[(3)] The intersection of a finite number of subsets is in $\mathscr{T}$, i.e., if $S_1,...,S_{\mu}\in\mathscr{T}$, then $$\bigcap^{\mu}_{\nu=1}S_{\nu}\in\mathscr{T}.$$
\end{enumerate}

Subsets of $X$ which are in the collection $\mathscr{T}$ are called \textit{open sets}. A \textit{neighborhood} of $S\subset X$ consists of an open subset of $X$ that contains $S$. A subset, $S$, is said to be \textit{closed} if its complement, $X\backslash S=\{x\in X|\;x\notin S\}$, is open. It follows from the topological space axioms that the intersection of an arbitrary collection of closed sets is closed and the union of a finite collection of closed sets is also closed. A subset of $X$ can be both open and closed, or neither open nor closed. If the only subsets which are both open and closed are $X$ and $\emptyset$, the topological space is said to be \textit{connected}. Equivalently, connectedness of a topological space can be defined as the impossibility of the set $X$ being represented as the union of disjoint (i.e., $S\bigcap S'=\emptyset$) non-empty open sets. 

Let $(X,\mathscr{T})$ and $(X',\mathscr{T}')$ be topological spaces. Consider the \textit{Cartesian product}, $X\times X'=\{(a,a')|\;a\in X,\;a'\in X'\}$ and the collection, $\mathcal{T}$, of all subsets of $X\times X'$ that can be expressed as unions of sets of the form $S\times S'$, with $S\in \mathscr{T}$ and $S'\in \mathscr{T}'$. It can be verified that $(X\times X',\mathcal{T})$ is a topological space, in which $\mathcal{T}$ is called the \textit{product topology}. For example, if one considers the standard topology on $\mathbb{R}$, i.e., open sets are open intervals, one can construct a topology on $\mathbb{R}^n$ simply by taking the product topology $n$ times. By doing so, the topology induced on $\mathbb{R}^n$ is known as the \textit{standard topology}.

For an arbitrary subset $S\subset X$, the \textit{closure}, $\overline{S}$, of $S$ is the intersection of all closed sets that contain $S$. Evidently, $\overline{S}$ is closed, contains $S$, and equals $S$ if and only if $S$ is closed. The \textit{interior} of $S$, $\langle S\rangle$, is the union of all open sets contained in $S$. Clearly, $\langle S\rangle$ is open, is contained in $S$, and equals $S$ if and only if $S$ is open. The \textit{boundary} of $S$, $\partial S$, is $\overline{S}\backslash\langle S\rangle$. Equivalently, the boundary of $S$ can be expressed as the intersection between its closure and the closure of its complement, $\partial S=\overline{S}\cap\overline{(X\backslash S)}$, from which it is clear that the boundary of a set is always a closed set.

Let $(X,\mathscr{T})$ and $(X',\mathscr{T}')$ be topological spaces. A map $\psi:X\to X'$ is said to be \textit{continuous} if the inverse image of every open set $S'\subset X'$, $\psi^{-1}[S']=\{x\in X|\;\psi(x)\in S'\}$, is an open set in $X$. In particular, this notion of continuity is equivalent to the $(\epsilon,\delta)$ definition \cite{Leithold1976} if one takes the standard topology on $\mathbb{R}$ and analyses maps $\psi:\mathbb{R}\to\mathbb{R}$. The importance of this more general definition of continuity is exemplified in the proof of the following theorem.

\begin{theorem}\label{01}
    The closed interval $[0,1]$ in $\mathbb{R}$ with the standard topology is connected.
\end{theorem}
\begin{proof}
    First, note that although  $[0,1]$ is a closed interval, it is an open set inside itself, i.e., in the topology $\mathscr{T}$, where $X=[0,1]$ and $(X,\mathscr{T})$ is a topological space. Now, if the set $[0,1]$ were not connected, it would be possible to find non-empty disjoint open sets, $S$ and $S'$, contained in $[0,1]$ such that $S\cup S'=[0,1]$. Suppose this is the case, and define a map $\psi:[0,1]\to\mathbb{R}$ such that
    \begin{equation}
        \psi(a)=
        \begin{cases}
      0, & a\in S, \\
      1, & a\in S'.
    \end{cases}
    \end{equation}
    Evidently, this map is continuous because the inverse image of an open set $S''\in\mathbb{R}$, $\psi^{-1}[S'']$, is either $S$, $S'$, $[0,1]$ or $\emptyset$, since any open set in $\mathbb{R}$ either contains $0$, $1$, both or neither. However, if the image of $\psi$ consists of both $0$ and $1$, then by the intermediate value theorem (see, e.g., \cite{Stewart2007} for a convenient statement and \cite{Ghorpade2006} for a proof), $\psi$ must take on every value in the interval $(0,1)$. Since $\psi$ clearly does not, then $S$ or $S'$ must be the empty set, contradicting the supposition that neither is. Thus, $[0,1]$ is connected in $\mathbb{R}$ with the standard topology.
\end{proof}

A topological space, $(X,\mathscr{T})$, is said to be \textit{path-connected} if for any two points $a,a'\in X$ there exists a continuous map $\psi:[0,1]\to X$ such that $\psi(0)=a$ and $\psi(1)=a'$. The following theorem states an important property of path-connected topological spaces.

\begin{theorem}\label{02}
    A path-connected topological space is connected.
\end{theorem}
\begin{proof}
    Let $(X,\mathscr{T})$ be a path-connected topological space and suppose that it is not connected. Then it would be possible to find non-empty disjoint open sets, $S$ and $S'$ contained in $X$ such that $S\cup S'=X$. By path-connectedness, it is possible to find a continuous map, $\psi$, such that $\psi(0)=a$ and $\psi(1)=a'$ with $a\in S$ and $a'\in S'$. However, this implies that $\psi^{-1}[S]$ and $\psi^{-1}[S']$ are two disjoint non-empty open sets whose union is $[0,1]$, contradicting theorem \ref{01}. Hence, $S$ or $S'$ must be the empty set, contradicting the supposition that neither is. Thus, every path-connected topological space is connected.
\end{proof}

Let $(X,\mathscr{T})$ and $(X',\mathscr{T}')$ be topological spaces. If a map $\psi:X\to X'$ is continuous, a bijection, and its inverse is continuous, then $\psi$ is said to be a \textit{homeomorphism} and $(X,\mathscr{T})$ and $(X',\mathscr{T}')$ are said to be \textit{homeomorphic}. Homeomorphic topological spaces are identical from the perspective of topological properties.
 
A topological space, $(X,\mathscr{T})$, is said to be \textit{Hausdorff} if for each pair of distinct points, $a,\;a'\in X,$ one can find a neighborhood of $a$, $S$, and of $a'$, $S'$, such that $S\bigcap S'=\emptyset$. The usefulness of this property lies in the fact that one can use it to identify topological spaces that possess points that cannot be ``separated''.

Let $(X,\mathscr{T})$ be a topological space and $S\subset X$. A collection of open sets, $\{S_{\mu}\}$, is said to be an \textit{open cover} of $S$ if the union of these sets contains $S$. A subcollection of the open cover which also covers $S$ is called a \textit{subcover}. A set $S$ is said to be \textit{compact} if every open cover of $S$ has a finite subcover. This notion is useful because, under certain circumstances, it allows one to identify when a set is ``finite''. Similarly, an open cover $\{S'_{\mu}\}$ is said to be a \textit{refinement} of $\{S_{\mu}\}$ if for each $S'_{\mu}$ there exists a $S_{\mu}$ such that $S'_{\mu}\subset S_{\mu}$. Moreover, the open cover $\{S'_{\mu}\}$ is said to be \textit{locally finite} if each $a\in X$ has an neighborhood $S''$ such that only a finite amount of $S'_{\mu}$ respects $S''\cap S'\neq\emptyset$. Finally, $(X,\mathscr{T})$ is \textit{paracompact} if every open cover has a locally finite refinement. Paracompactness can then be interpreted as stating that the topological space can be divided into pieces, which can then be used, for example, to extend local properties to global ones.

An $n$-dimensional \textit{manifold}, $M$, is a topological space such that every point has a neighborhood homeomorphic to a neighborhood of a point in $\mathbb{R}^n$ with the standard topology. Let $\psi_{\mu}$ denote a homeomorphism from a neighborhood of $a\in M$, $S_{\mu}$, into $\mathbb{R}^n$, known as a \textit{coordinate system}. For any two non-disjoint open sets, $S_{\mu}$ and $S_{\nu}$, contained in $M$, the manifold is said to $C^r$ if the function $\psi_{\nu}\circ\psi_{\mu}^{-1}:[\psi_{\mu}[S_{\mu}\cap S_{\nu}]]\to[\psi_{\nu}[S_{\mu}\cap S_{\nu}]]$ is $C^{r}$. Fig. \ref{fig:manifold} illustrates this map, with its action being the white region, corresponding to the intersection of the open sets. One also imposes the condition that the collection $\{S_{\mu}\}$ and the family $\{\psi_{\mu}\}$ are \textit{maximal}, i.e., all $S_{\mu},\;\psi_{\mu}$ which are compatible with the condition above are included in $\{S_{\mu}\}$ and $\{\psi_{\mu}\}$. In particular, our discussion will be restricted to $C^{\infty}$ manifolds, which will be referred to simply as manifolds.
\begin{figure}[h]
\centering
\includegraphics[scale=1.35]{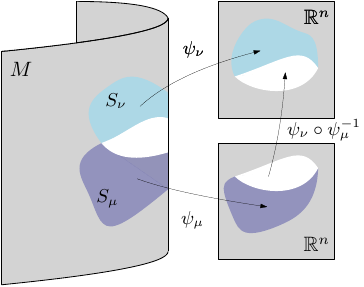} 
\caption{Manifold, $M$, and the smooth map $\psi_{\nu}\circ\psi_{\mu}^{-1}$.}
\caption*{Source: By the author.}
\label{fig:manifold}
\end{figure}

Let $M$ be an $n$-dimensional manifold and $M'$ be an $n'$-dimensional manifold, such that $n'\leq n$. A map $\psi:M'\to M$ is said to be $C^r$ if the function $\psi_{\nu}\circ\psi\circ\psi^{-1}_{\mu}:[\psi_{\mu}[S_{\mu}]]\to[\psi_{\nu}[S_{\nu}]]$ is $C^r$ for all homeomorphisms to $\mathbb{R}^n$ for each respective manifold, $\psi_{\mu}$ and $\psi_{\nu}$. Additionally, if $\psi$ is one-to-one, $\psi[M']$ is said to be a $C^r$ $n'$-dimensional \textit{embedded submanifold} of $M$. A two-dimensional $C^r$ embedded submanifold is called a $C^r$ \textit{surface}, while an $(n-1)$-dimensional (with $n\geq4$) $C^r$ embedded submanifold is called a $C^r$ \textit{hypersurface}. 

\section{Tensor fields}\label{A2}

It is straightforward to see that $\mathbb{R}^n$ with the standard topology has the natural structure of a $n$-dimensional vector space, which arises by taking a vector to be the $n$-tuple associated with each $(x^1,\ldots,x^n)\in\mathbb{R}^n$, together with the field of real numbers and ordinary sum and multiplication by scalar operations. However, in arbitrary manifolds, this global notion is lost due to the nontrivial topology. Since a vector space is necessary to define tensors, which are geometrical quantities associated with the properties of nature from a physical perspective, it is of interest to retrieve this notion for arbitrary manifolds. Still, one wishes to do so in a way that is intrinsic to the manifold structure, that is, without resorting to a higher dimensional manifold in which it might be embedded. We will now see that this can be done precisely by the concept of a vector space that is tangent to each point in the manifold.

Let $M$ be an $n$-dimensional manifold. The \textit{tangent vector space} at a point $a\in M$, $V_a$, is the collection of vectors defined at $a$, together with the field of real numbers and ordinary sum and multiplication by scalar operations. The notion of a vector at a point of an arbitrary manifold arises if one considers that directional derivatives can be interpreted as vectors. More precisely, let $\mathcal{F}$ denote the collection of smooth functions $f:M\to\mathbb{R}$. A tangent vector at a point $a\in M$ is a map, $s:\mathcal{F}\to\mathbb{R}$, which is linear and respects the Leibniz rule. Namely, given a coordinate system, one can associate the directional derivative of each coordinate as a basis vector, and thus, any $s(f)\in V_a$ can be written as 
\begin{equation}\label{eq65}
    s(f)=\sum_{\mu=0}^{n-1}s^{\mu}\partial_{\mu}(f),
\end{equation}
where
\begin{equation}\label{aeq1}
    \partial_{\mu}(f)=\partial_{\mu}(f\circ\psi^{-1})\big|_{\psi(a)},
\end{equation}
with $\partial_{\mu}=\partial/\partial x^{\mu}$, $\{x^{\mu}\}$ the Cartesian coordinates of $\mathbb{R}^n$ and $f\in\mathcal{F}$. Evidently, from eq. \ref{aeq1}, one can see that $\partial_{\mu}(f)$ is the directional derivative of the function $f\circ\psi^{-1}:[\psi[S]]\to \mathbb{R}$ at the point $\psi(a)\in\mathbb{R}^n$. Hence, one may picture $\partial_{\mu}(f)$ as an arrow at $\psi(a)$ pointing in the direction of increasing $x^{\mu}$. It is important to note that a unique function $f$ is not relevant to the definition of a unique vector, as with the information of the directional derivative of any function one can uniquely define a vector. A proof of the dimensionality of $V_a$, which must be equal to the dimensionality of the manifold, can be found in \cite{Wald1984}. Also, note that if one chooses a different coordinate system, $\psi'$, eq. \ref{aeq1} would have defined a different coordinate basis, which would, clearly, span the same tangent vector space. In particular, given a coordinate system, $\psi$, and a vector, $s$, the pertinent association is given by
\begin{equation}
    (s^1,\ldots,s^n)\longleftrightarrow s^1\frac{\partial}{\partial x^1}+\ldots+s^n\frac{\partial}{\partial x^n},
\end{equation}
where the scalars $(s^1,\ldots,s^n)$ are the \textit{coordinate components} of the vector $s$ in the coordinate system $\psi$.

A \textit{vector field} is an assignment of a tangent vector at each point $a\in M$. A \textit{$C^r$ vector field}, $s$, on a manifold is a vector field such that its coordinate components are $C^r$ functions. For any two vector fields, $s$ and $w$, their \textit{commutator}, $[s,w]$, is 
\begin{equation}\label{aeq11}
    [s,w](f)=s[w(f)]-w[s(f)].
\end{equation}
Note that the commutator of two vector fields in a coordinate basis vanishes, as a consequence of the equality of mixed partial derivatives in $\mathbb{R}^n$. In the following, the symbol $(f)$ in vectors will be dropped, as the action of a vector on an arbitrary smooth function is implied by its nature.

Let $V$ be an $n$-dimensional vector space. The elements of $V$ will be denoted with \textit{contravariant} indices, i.e., $s^{\mu}\in V$. Consider the collection of linear maps $f:V\to\mathbb{R}$. This collection, together with the field of real numbers and ordinary sum and multiplication by scalar operations, has the structure of a vector space and is known as the \textit{dual vector space} to $V$, $V^*$. Elements of $V^*$ are called \textit{dual vectors}, and are denoted by \textit{covariant} indices, i.e., $\omega_{\mu}\in V^*$. A \textit{dual vector field} is an assignment of a tangent dual vector at each $a \in M$. A dual vector field is said to be $C^r$ if for each $C^r$ vector field, $s^{\mu}$, the function $\omega_{\mu}(s^{\mu})$ is $C^r$. Given a basis of $V$, $\{s_1^{\mu},\ldots,s_n^{\mu}\}$, a \textit{dual basis}, $\{w^1_{\mu},\ldots,w^n_{\mu}\}$, to $\{s_1^{\mu},\ldots,s_n^{\mu}\}$ is defined as vectors of $V^*$ such that $w^{a}_{\mu}(s_{b}^{\mu})=\delta^{a}\mathstrut_{b}$, where $\delta^{a}\mathstrut_{b}$ is $1$ if $a=b$ and $0$ otherwise. Using this same line of reasoning, the dual to $V^*$, $V^{**}$, can be seen to be \textit{canonically} isomorphic to $V$. In essence, this follows from the fact that the linear maps $f':V^*\to\mathbb{R}$ can be defined ``naturally'' as 
\begin{equation}
    f'(s^{\mu},w_{\mu})=w_{\mu}(s^{\mu}),
\end{equation}
and thus, the only spaces of significance are $V$ and $V^*$. Consequently, one may view vectors as linear maps $f:V^*\to\mathbb{R}$, in which case one may write $w_{\mu}(s^{\mu})$ simply as $w_{\mu}s^{\mu}$ or $s^{\mu}w_{\mu}$. Lastly, due to the one-to-one correspondence between vectors and directional derivatives, it is natural to denote the elements of a basis of the tangent vector space induced by a coordinate system simply as $(\partial_{x^{n}})^{\mu}$, where $x^{n}$ is representative of the coordinates induced in $\mathbb{R}^n$ by the coordinate system. Similarly, constructing a dual basis by requiring that $s_{\mu}^a(\partial_{x^{b}})^{\mu}=\delta^{a}\mathstrut_{b}$ leads one to the convenient representation of elements of the dual basis by $(d{x^{a}})_{\mu}$.

A \textit{tensor}, $T$, of rank $(\ell,\ell')$ over a vector space, $V$, is a multilinear map 
\begin{equation}
    T:\underbrace{V^*\times\cdots\times V^*}_{\ell}\times \underbrace{V\times\cdots\times V}_{\ell'}\to\mathbb{R},
\end{equation}
i.e., given $\ell$ dual vectors and $\ell'$ vectors, $T$ produces a real number, and its linearity does not depend on the number of vectors or dual vectors it operates on. Namely, a tensor of rank $(\ell,\ell')$ is denoted by $T^{\mu_1\cdots\mu_{\ell}}\mathstrut_{\nu_1\cdots\nu_{\ell'}}$, i.e., $\ell$ contravariant indices and $\ell'$ covariant indices. Such a tensor is an element of the vector space of tensors of rank $(\ell,\ell')$, denoted by $\mathcal{T}(\ell,\ell')$. A \textit{tensor field} is an assignment of a tangent tensor at each $a\in M$. A tensor field of rank $(\ell,\ell')$ is said to be $C^r$ if, by operating on $\ell$ $C^r$ dual vector fields and $\ell'$ $C^r$ vector fields, it yields a $C^r$ function. Evidently, a vector is a tensor of rank $(1,0)$, a dual vector is a tensor of rank $(0,1)$ and a scalar is a tensor of rank $(0,0)$. In the following, we will refer to smooth tensor fields simply as tensors, unless stated otherwise. Since tensors are geometrical objects, they are invariant, but their components in a coordinate system are not. The transformation law for the components of a tensor can be readily derived from eq. \ref{eq65} and the relation between a basis and its dual.

The notation adopted for the representation of tensors is called \textit{abstract index notation}. In particular, “abstract” indices, represented by Greek letters, will be used to indicate the rank of the tensor, that is, the objects it acts on. Thus, $T^{\mu_1\cdots\mu_{\ell}}\mathstrut_{\nu_1\cdots\nu_{\ell'}}\in\mathcal{T}(\ell,\ell')$. Additionally, “concrete” indices, represented by Latin letters, will be used to indicate the components of a tensor in a coordinate system. Hence, $T^{a_1\cdots a_{\ell}}\mathstrut_{b_1\cdots b_{\ell'}}\in\mathbb{R}$. Consequently, Latin letters take on values from $0$ to $n-1$, where $n$ is the dimension of the pertinent vector space. Furthermore, following \textit{Einstein's summation convention}, repeated indices on an equation, regardless of being Greek or Latin letters, imply a sum. Such indices are referred to as “dummy” indices, as they carry no information regarding the rank of the tensor and are merely an implication of a “hidden” sum of such indices from $0$ to $n-1$. Finally, equations such as
\begin{equation}
    T^{\mu_1\cdots\mu_{\ell}}\mathstrut_{\nu_1\cdots\nu_{\ell'}}=0
\end{equation}
should be interpreted as stating that the tensor $T^{\mu_1\cdots\mu_{\ell}}\mathstrut_{\nu_1\cdots\nu_{\ell'}}$ is the zero element of the vector space of $\mathcal{T}(\ell,\ell')$. As is known, under the ordinary sum operation, this is just the object that has all of its components in any coordinate system equal to zero.

We now define two operations on tensors. The first one is called the \textit{outer product}, which is a map $\otimes:\mathcal{T}(\ell,\ell')\times\mathcal{T}(\eta,\eta')\to\mathcal{T}(\ell+\eta,\ell'+\eta')$. The outer product of $T^{\mu_1\cdots\mu_{\ell}}\mathstrut_{\nu_1\cdots\nu_{\ell'}}$ and $S^{\mu_1\cdots\mu_{\eta}}\mathstrut_{\nu_1\cdots\nu_{\eta'}}$ is denoted by
\begin{equation}
    T^{\mu_1\cdots\mu_{\ell}}\mathstrut_{\nu_1\cdots\nu_{\ell'}}\otimes S^{\alpha_1\cdots\alpha_{\eta}}\mathstrut_{\beta_1\cdots\beta_{\eta'}}=T^{\mu_1\cdots\mu_{\ell}}\mathstrut_{\nu_1\cdots\nu_{\ell'}}S^{\alpha_1\cdots\alpha_{\eta}}\mathstrut_{\beta_1\cdots\beta_{\eta'}}.
\end{equation}
The second operation is called \textit{contraction} with respect to the $\gamma$th contravariant slot and $\lambda$th covariant slot, and is a map $C_{\gamma\lambda}:\mathcal{T}(\ell,\ell')\to\mathcal{T}(\ell-1,\ell'-1)$ defined as
\begin{equation}
    C_{\gamma,\lambda}T^{\mu_1\cdots \mu_{\gamma}\cdots\mu_{\ell}}\mathstrut_{\nu_1\cdots \nu_{\lambda}\cdots\nu_{\ell'}}=T^{\mu_1\cdots a\cdots\mu_{\ell}}\mathstrut_{\nu_1\cdots a\cdots\nu_{\ell'}},
\end{equation}
where the $\gamma$th contravariant slot, represented by $a$, operates on each element of the basis $\{s_1^{\mu},\ldots,s_n^{\mu}\}$, while the $\lambda$th covariant slot, also represented by $a$, operates on the corresponding dual vector of the dual basis $\{w^1_{\mu},\ldots,w^n_{\mu}\}$. Note that since $a$ is a repeated index, a summation is implied. Also, even though the definition of the contraction requires a basis of $V$ and its dual, the operation is well defined in the sense that $T^{\mu_1\cdots a\cdots\mu_{\ell}}\mathstrut_{\nu_1\cdots a\cdots\nu_{\ell'}}$ is a geometrical object, i.e., independent of the coordinate system. As such, due to the action of the contraction on a $(1,1)$ tensor, one can interpret it as the trace of a tensor over the $\gamma$th contravariant slot and $\lambda$th slot. Finally, note that, since the repeated index in the contraction is merely a representation of a summation, it does not “count” for the rank of the tensor, so that $T^{\mu_1\cdots a\cdots\mu_{\ell}}\mathstrut_{\nu_1\cdots a\cdots\nu_{\ell'}}\in\mathcal{T}(\ell-1,\ell'-1)$. 

It is also useful to study the symmetric and antisymmetric parts of tensors. In order to isolate them for a given tensor, one constructs a new tensor by taking the sum over all permutations of its indices. Evidently, for the antisymmetric case, odd permutations must have a minus sign. For instance, one can write the symmetric part of a tensor $T_{\mu\nu\alpha}$, denoted by $T_{(\mu\nu\alpha)}$, as
\begin{equation}
    T_{(\mu\nu\alpha)}=\frac{1}{3!}(T_{\mu\nu\alpha}+T_{\nu\alpha\mu}+T_{\alpha\mu\nu}+T_{\nu\mu\alpha}+T_{\alpha\nu\mu}+T_{\mu\alpha\nu}),
\end{equation}
and the antisymmetric part, denoted by $T_{[\mu\nu\alpha]}$, as
\begin{equation}
    T_{[\mu\nu\alpha]}=\frac{1}{3!}(T_{\mu\nu\alpha}+T_{\nu\alpha\mu}+T_{\alpha\mu\nu}-T_{\nu\mu\alpha}-T_{\alpha\nu\mu}-T_{\mu\alpha\nu}).
\end{equation}
In general, one has
\begin{equation}\label{symmetric}
    T_{(\mu_{1}...\mu_{\ell})}=\frac{1}{\ell!}T_{\mu_{\alpha(1)}\cdots\mu_{\alpha(\ell)}},
\end{equation}
\begin{equation}\label{antisymmetric}
    T_{[\mu_{1}...\mu_{\ell}]}=\frac{1}{\ell!}\delta_{\alpha}T_{\mu_{\alpha(1)}\cdots\mu_{\alpha(\ell)}},
\end{equation}
where a sum is taken over all permutations of $1,...,\ell$, and $\delta_{\alpha}$ (not a tensor) is $+1$ for even permutations and $-1$ for odd permutations. These prescriptions are also valid for contravariant indices. Furthermore, the contraction of indices of symmetric slots with antisymmetric slots will always vanish, as the terms of the contraction will cancel out exactly.

An $\ell$-form over $V$ is an antisymmetric tensor of rank $(0,\ell)$,
\begin{equation}
    s_{\mu_1\cdots\mu_{\ell}}=s_{[\mu_1\cdots\mu_{\ell}]}.
\end{equation}
The vector space of $\ell$-forms over $V$ is denoted by $\Lambda^{\ell}$. Even though it is not natural to say that a tensor of rank $(0,1)$ is symmetric or antisymmetric, it is customary to refer to dual vectors as $1$-forms. Similarly, a tensor of rank $(0,0)$ is referred to as a $0$-form. Let $s_{\mu_1\cdots\mu_{\ell}}$ be an $\ell$-form and $w_{\nu_1\cdots\nu_{\ell'}}$ be an $\ell'$-form. One can construct an $(\ell+\ell')$-form by antisymmetrizing the outer product of $s_{\mu_1\cdots\mu_{\ell}}$ and $w_{\nu_1\cdots\nu_{\ell'}}$, thus defining a map $\wedge:\Lambda^{\ell}\times\Lambda^{\ell'}\to\Lambda^{\ell+\ell'}$, known as the \textit{wedge product}, by
\begin{equation}\label{ori}
    (s\wedge w)_{\mu_1\cdots\mu_{\ell}\nu_1\cdots\nu_{\ell'}}=\frac{(\ell+\ell')!}{\ell!\ell'!}s_{[\mu_1\cdots\mu_{\ell}}w_{\nu_1\cdots\nu_{\ell'}]}.
\end{equation}
\begin{theorem}\label{propform}
    Let $V$ be a $n$-dimensional vector space and $\Lambda^{\ell}$ denote the vector space of $\ell$-forms over $V$. Then \normalfont{dim}$\;\Lambda^{\ell}=n!/\ell!(n-\ell)!$ \textit{if} $\ell\leq n$ \textit{and} \normalfont{dim}$\;\Lambda^{\ell}=0$ \textit{if} $\ell>n$.
\end{theorem}
\begin{proof}
   Consider a basis of the vector space of two-forms over a three-dimensional vector space, which can be written as $\{dx^0\wedge dx^1,dx^1\wedge dx^2,dx^0\wedge dx^2 \}$. Note that this basis is constructed by taking the linearly independent wedge products of all possible combinations of two elements of $\{dx^0,dx^1,dx^2\}$, which is the basis of one-forms over the three-dimensional vector space. Furthermore, since the wedge product produces a $\delta_{\alpha}$ (see eq. \ref{antisymmetric}) upon exchange of indices, combinations with the same indices in different positions will yield a linearly dependent element. Additionally, due to the antisymmetry, repeated indices will make the wedge product vanish. Thus, a basis of $\ell$-forms over an $n$-dimensional vector space can be constructed by taking $\ell$ distinct elements of the set $\{dx^0,\ldots,dx^{n-1}\}$ such that two arrangements with the same elements are considered to be equivalent. The number of elements that can be constructed in this manner is given by the \textit{binomial coefficient} \cite{Bryant1993}, therefore
    \begin{equation}
        \text{dim}\;\Lambda^{\ell}=\binom{n}{\ell}=\frac{n!}{\ell!(n-\ell)!}.
    \end{equation}
    Moreover, if $\ell>n$, then any combination of $\{dx^0,\ldots,dx^{n-1}\}$ will result in elements with repeated $dx$'s, which, by the rules of the wedge product, will vanish.
\end{proof}

We now define the metric tensor, which is the operator that measures the infinitesimal squared distance given by an infinitesimal displacement. Seeing that the notion of an infinitesimal displacement in a given direction is precisely captured by vectors as derivative operators, the metric can be interpreted as the inner product of the tangent space at each $a\in M$. Hence, a metric, $g_{\mu\nu}$, is a symmetric, nondegenerate map $V\times V\to \mathbb{R}$. A \textit{symmetric} metric respects the property that 
\begin{equation}
    g_{\mu\nu}s^{\mu}w^{\nu}=g_{\mu\nu}s^{\nu}w^{\mu},\;\forall\;s^{\mu},\;w^{\mu}\in V,
\end{equation}
and by \textit{nondegenerate}, it is meant
\begin{equation}\label{deg}
    g_{\mu\nu}s^{\mu}w^{\nu}=0,\;\forall\;w^{\mu}\in V_a\text{ if and only if }s^{\mu}=0.
\end{equation}

Evidently, since any tensor can be reduced to a symmetric and antisymmetric part, one must have that $g_{(\mu\nu)}=g_{\mu\nu}$ and $g_{[\mu\nu]}=0$. The symmetric property can be interpreted when one considers the geometrical notion of an inner product, i.e., the ``measurement'' it makes should not depend on the order of vectors it acts on. Additionally, the nondegenerate property implies that the $n\times n$ matrix form of $g_{\mu\nu}$ has a nonvanishing determinant, i.e., $g_{\mu\nu}$ is invertible. Moreover, note that the metric gives rise to a ``natural'' isomorphism between $V$ and $V^*$, as it can be seen, equivalently, as a map $V\to V^*$. Similarly, the inverse metric, which acts as the inner product of $V^*$, can be seen as a map $V^*\times V^*\to\mathbb{R}$, or equivalently, as a map $V^*\to V$. Since the metric and its inverse can be used to identify vectors with dual vectors, and vice versa, they can be used to “raise” or “lower” indices in tensors. In order for this process to be consistent, the metric and its inverse must obey $g^{\mu\alpha}g_{\alpha\nu}=\delta^{\mu}\mathstrut_{\nu}$, where $\delta^{\mu}\mathstrut_{\nu}$ is the identity map over $V$. For example, the inner product of $s^{\mu}$ and $w^{\mu}$ can be written as
\begin{equation}
    g_{\mu\nu}s^{\mu}w^{\nu}=s^{\mu}w_{\mu}=s_{\mu}w^{\mu}.
\end{equation}

The \textit{signature} of a metric is the number of $+$ and $-$ that one gets by applying it to all vectors of an orthonormal basis of $V$. \textit{Riemannian} metrics, i.e., metrics with signatures $+\ldots+$, are positive definite. However, metrics do not need to be positive definite. For instance, a \textit{Lorentzian} metric has the signature $-+\ldots+$ (with only one minus), which evidently gives rise to three classes of vectors, depending on the sign of their norm. In particular, if $s^{\mu}s_{\mu}<0$, the vector is said to be \textit{timelike}. Similarly, if $s^{\mu}s_{\mu}=0$, the vector is said to be \textit{null}, and if $s^{\mu}s_{\mu}>0$, it is said to be \textit{spacelike}.

A \textit{spacetime}, $(M,g_{\mu\nu})$, is a connected Hausdorff manifold that has a Lorentzian metric defined on it. From a physical perspective, the points of a spacetime are then referred to as \textit{events}. In particular, the restrictions on the manifold structure of a spacetime are necessary in order to exclude unphysical ones. Namely, there are physical reasons to believe that a spacetime is a manifold with such properties. First, a spacetime is expected to be Hausdorff, so that it is always possible to physically distinguish different events. Second, connectedness implies that spacetime cannot be divided into regions that cannot communicate. If this were not the case, one could only consider a connected component, since detection of other connected components would never be possible. Lastly, a Lorentzian metric is necessary for a consistent differentiation between a time dimension and space dimensions, which is also in accord with a limited spatial velocity as measured by observers. It should be noted that the Hausdorff property together with the existence of a Lorentzian metric implies that a spacetime must be paracompact. For an extensive discussion on the restrictions on the topology of the universe and these properties, see, e.g., \cite{Geroch1979}. For simplicity, it will be assumed that the metric defined on a spacetime is a smooth tensor field, but we stress that this is not a necessary property for some of the developments presented in this work. For discussions on the order of differentiability necessary for the metric tensor, see \cite{Hawking1973}. 

A \textit{$C^r$ curve}, $\lambda$, on a manifold, $M$, is a $C^{r}$ map $\lambda:\mathbb{R}\to M$. One can associate a tangent vector, $s^{\mu}$, to $\lambda$ at each point $a\in \lambda$ by setting it to be equal to the derivative of the function $f\circ\lambda$, where $f\in\mathcal{F}$, evaluated at $a$ with respect to the parameter of the curve, $t$. More precisely, by choosing a coordinate system, the curve $\lambda(t)$ will get mapped into a curve $x^{a}(t)$ in $\mathbb{R}^n$. Thus, for any $f\in\mathcal{F}$, one has
\begin{equation}\label{aeq100}
    s^{\mu}=\frac{dx^{a}}{dt}(\partial_{x^a})^{\mu}.
\end{equation}
Conversely, it is possible to find the \textit{integral curves} of $s^{\mu}$, that is, the family of curves in $M$ such that for each point $a\in M$ only one curve of this family passes, and its tangent vector is $s^{\mu}$. By choosing a coordinate system in a neighborhood of $a$, finding these curves reduces to solving the system 
\begin{equation}
    \frac{dx^{a}}{dt}=s^a(x^0,...,x^{n-1}).
\end{equation}
Finally, a curve is said to be timelike, null or spacelike if for each $a\in \lambda$, its tangent vector is timelike, null or spacelike, respectively.

A similar characterization of embedded submanifolds can be made. Let $(M,g_{\mu\nu})$ be an $n$-dimensional spacetime, $S\subset M$ be a hypersurface and $a\in S$. The tangent space, $\overline{V}_a$, of the manifold $S$ can be viewed as a $(n-1)$-dimensional subspace of $V_a$, the tangent space of $a\in M$. Such a subspace can be identified with the one orthogonal to a vector $\ell^{\mu}\in V_a$, i.e., $s^{\mu}\in\overline{V}_a$ if $s^{\mu}\ell_{\mu}=0$, and, consequently, the vector $\ell^{\mu}$ is said to be \textit{normal} to $S$. If the normal to a hypersurface is timelike, the hypersurface is said to be spacelike, and if its normal is null, the hypersurface is said to be null. Similar associations hold for surfaces, where their characterization is made by the nature of the vectors that span them. In particular, a surface is said to be spacelike if the vectors that span it are spacelike. Note that a spacelike surface can be generated by taking the orthogonal space to a timelike and spacelike vector or two non-proportional null vectors.

\section{Derivative operators}\label{derivativeoperators}

A \textit{connection}, $\nabla$, on a manifold, $M$, is a map that takes tensors of rank $(\ell,\ell')$ to tensors of rank $(\ell,\ell'+1)$. Because of this, it is convenient to denote the connection as $\nabla_{\mu}$, even though it is not a dual vector. A connection can be interpreted as the operator that ``takes'' tensors from a point in a manifold to another, so that one can analyze how the tensor varies, since comparison of tensors over different vector spaces is not particularly useful. Indeed, the notion of the ``velocity'' of a curve is already captured by tangent vectors, but a connection will allow one to study ``acceleration'' of curves by providing a means to analyze changes of tensors.

In order to interpret a connection as a derivative operator, one must require it to be linear, respect the Leibniz rule and commute with contraction of indices of a tensor. Additionally, it should be consistent with the notion of a vector as a directional derivative of a function, $f$, i.e., 
\begin{equation}\label{aeq2}
    s(f)=s^{\mu}\nabla_{\mu}f,
\end{equation}
and one also requires that it be torsion free\footnote{Not requiring that the torsion vanishes would lead one to theories of gravitation such as \textit{teleparallel gravity}.}, i.e., 
\begin{equation}\label{aeq3}
    \nabla_{[\mu}\nabla_{\nu]}f=0.
\end{equation}

From the requirement that $\nabla_{\mu}$ obey the Leibniz rule and eqs. \ref{aeq2} and \ref{aeq3}, it is possible to write the commutator of two vectors (see eq. \ref{aeq11}), $s^{\mu}$ and $w^{\mu}$, as
\begin{equation}
    \begin{aligned}[b]
    [s,w](f)& =s[w(f)]-w[s(f)]\\
    &=s^{\mu}\nabla_{\mu}(w^{\nu}\nabla_{\nu}f)-w^{\mu}\nabla_{\mu}(s^{\nu}\nabla_{\nu}f)\\
    &=(s^{\mu}\nabla_{\mu}w^{\nu}-w^{\mu}\nabla_{\mu}s^{\nu})\nabla_{\nu}f,   
    \end{aligned}
\end{equation}
and by comparison with eq. \ref{aeq2} yields
\begin{equation}\label{commu}
    [s,w]^{\mu}=s^{\nu}\nabla_{\nu}w^{\mu}-w^{\nu}\nabla_{\nu}s^{\mu}.
\end{equation}

One reliable connection can be easily defined by a coordinate system, simply by taking it to be the ordinary derivative operator associated with it. Namely, the \textit{ordinary derivative} can be defined as the action of $\partial_{\mu}=\partial/\partial x^{\mu}$ on all components of the tensor with respect to the coordinate basis induced by $\psi$. Consequently, the properties mentioned in the definition of $\nabla_{\mu}$ follow directly from the properties of partial derivatives. However, if one were to choose a different coordinate system, $\psi'$, the same process would yield a different connection, as the tensors resulting from the actions of $\partial_{\mu}$ and $\partial_{\mu}'$ would be different. Thus, simply defining a connection as such is not a unique prescription.

To progress towards a unique definition of a connection, it is useful to study how different choices of connections differ in their application on tensors. As given by eq. \ref{aeq2}, any two connections must agree on their action on scalars. In order to investigate the disagreement of the action of any two connections on vectors, it is convenient to calculate the difference between their action on $fs^{\mu}$ for an arbitrary smooth function $f$ and vector $s^{\mu}$,
\begin{equation}\label{aeq30}
    \nabla_{\mu}(fs^{\nu})-\nabla'_{\mu}(fs^{\nu})=f(\nabla_{\mu}-\nabla'_{\mu})s^{\nu}.
\end{equation}
It is clear that the action of the connection on a vector does not depend only on its value at a given point, but also on a neighborhood of that point. However, the right hand side of eq. \ref{aeq30} can be shown to depend only on the value of the vector at the point where the derivative is being evaluated. 

For example, consider the vector $fs^{\mu}+f'w^{\mu}$, with $f(a)=1$ and $f'(a)=0$. Both functions are smooth, and one has information about their value only at $a$. They are otherwise arbitrary on the rest of $M$. Due to the fact that the test function is factored and the overall result of the right hand side is calculated at $a$, one obtains
\begin{equation}
    [(\nabla_{\mu}-\nabla'_{\mu})(fs^{\nu}+f'w^{\nu})]\big|_a=[(\nabla_{\mu}-\nabla'_{\mu})s^{\nu}]\big|_a.
\end{equation}
From this, one concludes that $[(\nabla_{\mu}-\nabla'_{\mu})(fs^{\mu}+f'w^{\mu})]|_a$ depends only on the value of $fs^{\mu}+f'w^{\mu}$ at $a$. Thus, $(\nabla_{\mu}-\nabla'_{\mu})$ defines a linear map from $(1,0)$ tensors to $(1,1)$ tensors at $a\in M$. More precisely, it defines a symmetric (due to eq. \ref{aeq3}) $(1,2)$ tensor,
\begin{equation}\label{eq443}
(\nabla_{\mu}-\nabla'_{\mu})s^{\nu}=C^{\nu}\mathstrut_{\mu\alpha}s^{\alpha}.
\end{equation}

In the case of $\nabla'_{\mu}=\partial_{\mu}$, the tensor $C^{\nu}\mathstrut_{\mu\alpha}$ is called a \textit{Christoffel symbol} and denoted by $\Gamma^{\nu}\mathstrut_{\mu\alpha}$. Given any coordinate system, $\Gamma^{\nu}\mathstrut_{\mu\alpha}$ relates the action of the ordinary derivative operator to that of an arbitrary connection. In particular, given an arbitrary connection, the component $\Gamma^{a}\mathstrut_{bc}$ is the $a$th component of the variation of the vector $(\partial_c)^{\mu}$ along the direction of $(\partial_b)^{\mu}$, with such variation being the one associated with the arbitrary connection. Lastly, generalizations of eq. \ref{eq443} to tensors of arbitrary rank follows by induction. In other words, one first considers the action of $(\nabla_{\mu}-\nabla'_{\mu})$ on a scalar such as $s^{\mu}w_{\mu}$ to find the relation for dual vectors, and then for tensors of arbitrary rank one considers the pertinent contraction of $C^{\nu}\mathstrut_{\mu\alpha}$ with each contravariant and covariant index.

Although eq \ref{eq443} gives a prescription to relate the action of any two connections, it still does not give rise to a unique connection. In essence, to define a connection associated with the structure of the manifold, one must consider the parallel transport of vectors. A vector, $s^{\mu}$, is said to be \textit{parallel transported} along a curve with tangent $t^{\mu}$ if the equation
\begin{equation}\label{aeq33}
t^{\mu}\nabla_{\mu}s^{\nu}=0
\end{equation}
is satisfied along the curve. Consequently, parallel transport can be interpreted as the lack of variation of a vector along a curve. The desired connection can be derived by requiring that any two vectors $s^{\mu}$ and $w^{\mu}$ which are parallel transported along a curve with tangent $t^{\mu}$ have a constant inner product,
\begin{equation}\label{aeq4}
t^{\mu}\nabla_{\mu}(g_{\nu\alpha}s^{\nu}w^{\alpha})=0.
\end{equation}
Eq. \ref{aeq4} will be valid for all curves and parallel transported vectors if and only if
\begin{equation}\label{aeq5}
\nabla_{\mu}g_{\nu\alpha}=0.
\end{equation}

The connection that respects eq. \ref{aeq5} is said to be \textit{compatible} with the metric, and it is referred to as the \textit{Levi-Civita connection}. In the following, $\nabla_{\mu}$ will denote the Levi-Civita connection. In particular, the action of the ordinary derivative for a given coordinate system is related to that of $\nabla_{\mu}$ by the Christoffel symbol
\begin{equation}
\Gamma^{\alpha}\mathstrut_{\mu\nu}=\frac{1}{2}g^{\alpha\beta}(\partial_{\mu}g_{\nu\beta}+\partial_{\nu}g_{\mu\beta}-\partial_{\beta}g_{\mu\nu}).
\end{equation}
It is useful to evaluate the contracted Christoffel symbol, which reads
\begin{equation}\label{4567}
    \Gamma^{a}\mathstrut_{a\mu}=\frac{g^{ac}}{2}\partial_{\mu} g_{ac},
\end{equation}
but since it is possible to relate the variation of the logarithm of the determinant of a nonsingular matrix to the trace of its inverse \cite{Poisson2004}, one can rewrite eq. \ref{4567} as
\begin{equation}
    \Gamma^{a}\mathstrut_{a\mu}=\partial_{\mu}\ln{\sqrt{-g}},
\end{equation}
where $g=\text{det}(g_{\mu\nu})$. Thus, choosing a coordinate system, considering the properties of $\nabla_{\mu}$ and $\Gamma^{\mu}\mathstrut_{\nu\alpha}$, one finds
\begin{equation}
    \nabla_{a}f=\partial_af, \;\forall\;f\in\mathcal{F},
\end{equation}
\begin{equation}\label{Christoffel vector}
    \nabla_{a}s^{b}=\partial_{a}s^{b}+\Gamma^{b}\mathstrut_{ac}s^{c},
\end{equation}
\begin{equation}\label{Christoffel dual}
    \nabla_{a}s_{b}=\partial_{a}s_{b}-\Gamma^{c}\mathstrut_{ab}s_{c},
\end{equation}
\begin{equation}\label{Christoffel divergence}
    \nabla_{a}s^{a}=\frac{1}{\sqrt{-g}}\partial_a(\sqrt{-g}s^{a}).
\end{equation}
and generalizations for higher rank tensors follow by induction.

It is also useful to define the derivative operator that takes $\ell$-forms to $(\ell+1)$-forms. The \textit{exterior derivative} is a map, $d:\Lambda^{\ell}\to\Lambda^{\ell+1}$, defined by
\begin{equation}
    (ds)_{\nu\mu_1\cdots\mu_{\ell}}=(\ell+1)\nabla_{[\nu}s_{\mu_1\cdots\mu_{\ell}]}.
\end{equation}
Since the action of any two connections is related by the symmetric tensor, $C^{\nu}\mathstrut_{\mu\alpha}$, the definition of $d$ is independent of the choice of connection. 

Lastly, we define the notion of Lie differentiation. Let $M$ and $N$ be $n$-dimensional manifolds, and a map $\psi:M\to N$ be $C^{\infty}$, a bijection, and its inverse be $C^{\infty}$, i.e., a \textit{diffeomorphism}. Now, the case of particular interest is when $M=N$, and the pertinent analysis is that of how tensor fields in $M$ are affected by the action of $\psi$. Evidently, the action of $\psi$ on tensors of rank $(0,0)$ is such that the function smooth $f:\psi[M]\to\mathbb{R}$ is “pulled back” to the function $f\circ\psi:M\to\mathbb{R}$. Similarly, one can identify that the action of $\psi$ on tensors of rank $(1,0)$ is to ``pushforward'' tangent vectors at $a\in M$ to tangent vectors at $\psi(a)\in M$, thus defining a map $\psi^*:V_a\to V_{\psi(a)}$. For all smooth $f:\psi[M]\to \mathbb{R}$, the vector $(\psi^*s)^{\mu}\in V_{\psi(a)}$ is defined by 
\begin{equation}
    (\psi^*s)^{\mu}(f)=s^{\mu}(f\circ\psi),
\end{equation}
where $s^{\mu}\in V_a$. In essence, by choosing a coordinate system in a neighborhood of $a\in M$ and one in a neighborhood of $\psi(a)\in M$, the action of $\psi^*$ can then be verified to be the coordinate transformation associated with the change in a coordinate system. In other words, $\psi$ may be viewed as leaving $a$ and vectors at $a$ unchanged, and effectively being a coordinate transformation.

Analogously, one may view the action of $\psi$ on tensors of rank $(0,1)$ as a ``pullback'' map $\psi_*:V^*_{\psi(a)}\to V^*_a$, so that $(\psi_*s)^{\mu}\in V^*_{\psi(a)}$ respects
\begin{equation}
    (\psi_*w)_{\mu}s^{\mu}=w_{\mu}(\psi^*s)^{\mu},\;\forall\;s^{\mu}\in V_a.
\end{equation}
Since $\psi$ is a diffeomorphism, one can then verify that the action of $\psi_*$ has to be given by $(\psi^{-1})^*$ for scalars to obey corresponding transformation laws, i.e., they must be invariant under coordinate transformations. From the fact that $\psi_*=(\psi^{-1})^*$, one can define the action of $\psi$ on tensors of rank $(\ell,\ell')$, which is denoted by $(\psi^*T)^{\mu_1\cdots\mu_{\ell}}\mathstrut_{\nu_1\cdots\nu_{\ell'}}$. Such a definition is made by requiring that 
\begin{multline}
    (\psi^*T)^{\mu_1\cdots\mu_{\ell}}\mathstrut_{\nu_1\cdots\nu_{\ell'}}(w_1)_{\mu_1}\cdots(w_{\ell})_{\mu_{\ell}}(s_1)^{\nu_1}\cdots(s_{\ell'})^{\nu_{\ell'}}\\=T^{\mu_1\cdots\mu_{\ell}}\mathstrut_{\nu_1\cdots\nu_{\ell'}}(\psi_*w_1)_{\mu_1}\cdots(\psi_*w_{\ell})_{\mu_{\ell}}(\psi_*s_1)^{\nu_1}\cdots(\psi_*s_{\ell'})^{\nu_{\ell'}}.
    \end{multline}
Consequently, from the interpretation of $\psi^*$ and $\psi_*$, one may view the action of the diffeomorphism $\psi$ as inducing a coordinate transformation at each $a\in M$.

From these definitions and interpretations, one can study the variation of tensor fields over the action of a family of diffeomorphisms as follows. Consider a one-parameter group of diffeomorphisms $\psi_s:\mathbb{R}\times M\to M$, that is, for each $s\in\mathbb{R}$, $\psi_s$ is a diffeomorphism. In other words, one has 
\begin{equation}
    \psi_{s}\circ\psi_{s'}=\psi_{s+s'},
\end{equation}
which implies that $\psi_{-s}=(\psi_{s})^{-1}$ and that $\psi_0$ is the identity, mapping every point of $M$ into itself. Now, for a fixed $a\in M$, $\psi_s(a):\mathbb{R}\to M$ defines a curve, called the \textit{orbit} of $\psi_s$. The tangent vector to the orbit of $\psi_s$, $\chi^{\mu}$, can be interpreted as the infinitesimal generator of this one-parameter group. Thus, the \textit{Lie derivative} of a tensor over $\psi_s$ can be defined as
\begin{equation}\label{lie}
   \mathcal{L}_{\chi}T^{\mu_1\cdots\mu_{\ell}}\mathstrut_{\nu_1\cdots\nu_{\ell'}}=\lim_{s\to0}\left(\frac{(\psi_{-s}^*T)^{\mu_1\cdots\mu_{\ell}}\mathstrut_{\nu_1\cdots\nu_{\ell'}}-T^{\mu_1\cdots\mu_{\ell}}\mathstrut_{\nu_1\cdots\nu_{\ell'}}}{s}\right).
\end{equation}
In particular, the abstract index in the representation of $\chi$ has been removed for clearer notation, and the action of the diffeomorphism is represented by $\psi_{-s}^*$ so that the tensors are evaluated at the same point on the manifold. Thus, the Lie derivative is a map of tensors of rank $(\ell,\ell')$ to tensors of rank $(\ell,\ell')$, which clearly obeys the Leibniz rule. In essence, it can be interpreted as the ``infinitesimal'' change that a diffeomorphism has on a tensor, or, more precisely, how a tensor changes as one moves along an orbit of the one-parameter group of diffeomorphisms. 

From eq. \ref{lie}, one can find a prescription for evaluating the Lie derivative of tensors of arbitrary rank. To do so, first note that eq. \ref{lie} implies that
\begin{equation}\label{Liescalar}
    \mathcal{L}_{\chi}f=\chi^{\mu}(f).
\end{equation}
Moreover, using the coordinate transformation perspective of diffeomorphisms, one can consider a coordinate system such that the parameter $s$ along the orbit of $\psi_s$ is chosen as one of the coordinates, $x^0$, so that $\chi^{\mu}=(\partial_{s})^{\mu}$. Then, the action of $\psi_{-s}$ corresponds to the coordinate transformation $x'^0=x^0+s$, which means that
\begin{equation}
    (\psi_{-s}^*T)^{a_1\cdots a_{\ell}}\mathstrut_{b_1\cdots b_{\ell'}}(x^0,\ldots,x^{n-1})=T^{a_1\cdots a_{\ell}}\mathstrut_{b_1\cdots b_{\ell'}}(x^0+s,\ldots,x^{n-1}),
\end{equation}
which yields a simple result for the Lie derivative of said tensor in this convenient coordinate system,
\begin{equation}
    \mathcal{L}_{\chi}T^{a_1\cdots a_{\ell}}\mathstrut_{b_1\cdot sb_{\ell'}}=\frac{\partial T^{a_1\cdots a_{\ell}}\mathstrut_{b_1\cdots b_{\ell'}}}{\partial x^{0}}.
\end{equation}
Hence, if the coordinate components of a tensor do not depend on the coordinate $x^0$, then $\mathcal{L}_{\chi}T^{a_1\cdots a_{\ell}}\mathstrut_{b_1\cdots b_{\ell'}}=0$. 

For tensors of rank $(1,0)$, i.e., vectors, it is not difficult to deduce that in such a coordinate system the components of the $\mathcal{L}_{\chi}w^{\mu}$ and the commutator $[\chi,w]^{\mu}$ are the same. Since both of these quantities are independent of the coordinate system, one can conclude that 
\begin{equation}\label{Lievector}
    \mathcal{L}_{\chi}w^{\mu}=[\chi,w]^{\mu}.
\end{equation}
To evaluate $\mathcal{L}_{\chi}w_{\mu}$, one uses the Leibniz rule and eqs. \ref{Lievector}, \ref{Liescalar} and the commutator of two vectors in terms of the Levi-Civita connection  (i.e., using writing eq. \ref{aeq11} using eq. \ref{aeq2}), which yields
\begin{equation}\label{Lieform}
    \mathcal{L}_{\chi}w_{\mu}=\chi^{\nu}\nabla_{\nu}w_{\mu}+w_{\nu}\nabla_{\mu}\chi^{\nu}.
\end{equation}
The Lie derivative of a tensor of arbitrary rank can then be derived by induction from these results. In particular, its action of tensors of rank $(0,2)$ is given by
\begin{equation}\label{Liemetric}
    \mathcal{L}_{\chi}A_{\mu\nu}=\chi^{\alpha}\nabla_{\alpha}A_{\mu\nu}+A_{\nu\alpha}\nabla_{\mu}\chi^{\alpha}+A_{\mu\alpha}\nabla_{\nu}\chi^{\alpha}.
\end{equation}

Note that the notion of Lie derivative does not depend on an additional structure on the manifold. In contrast, the notion of a unique connection only exists when one is given an additional structure over the manifold, signifying that Lie differentiation is a more fundamental form of differentiation than the one given by the action of a connection. Of course, this translates to the fact that the connection appearing in eqs. \ref{Lieform} and \ref{Liemetric} is arbitrary, i.e., it need not be the Levi-Civita connection. This can be traced back to the fact that all connections must agree on their action on scalars. Similarly, one can also see that the notion of exterior differentiation is also intrinsic to the manifold structure, as it is independent of choice of connection.

\section{Curvature}\label{A3}

The notion of curvature is directly related to the failure of multiple applications of the Levi-Civita connection on a tensor to commute \cite{Wald1984}. Similar to $(\nabla_{\mu}-\nabla'_{\mu})$, one can verify that $(\nabla_{\mu}\nabla_{\nu}-\nabla_{\nu}\nabla_{\mu})$ defines a linear map from $(1,0)$ tensors to $(1,2)$ tensors at $a\in M$. More precisely, it defines a $(1,3)$ tensor,
\begin{equation}\label{eqa20}
(\nabla_{\mu}\nabla_{\nu}-\nabla_{\nu}\nabla_{\mu})s^{\alpha}=-R_{\mu\nu\beta}\mathstrut^{\alpha}s^{\beta},
\end{equation}
known as the \textit{Riemann tensor}. 

From eq. \ref{eqa20} and the same line of reasoning as for $C^{\nu}\mathstrut_{\mu\alpha}$, it is possible to derive 
\begin{equation}\label{Riedual}
    (\nabla_{\mu}\nabla_{\nu}-\nabla_{\nu}\nabla_{\mu})s_{\alpha}=R_{\mu\nu\alpha}\mathstrut^{\beta}s_{\beta},
\end{equation}
\begin{equation}\label{Rie2}
    (\nabla_{\mu}\nabla_{\nu}-\nabla_{\nu}\nabla_{\mu})A^{\alpha\beta}=-R_{\mu\nu\lambda}\mathstrut^{\alpha}A^{\lambda\beta}-R_{\mu\nu\lambda}\mathstrut^{\beta}A^{\alpha\lambda}.
\end{equation}
From these relations, one concludes that $R_{\mu\nu\alpha}\mathstrut^{\beta}$ obeys,
\begin{equation}\label{Rthird}
    R_{\mu\nu\alpha\beta}=-R_{\mu\nu\beta\alpha},
\end{equation}
\begin{equation}\label{Rfirst}
    R_{\mu\nu\beta}\mathstrut^{\alpha}=-R_{\nu\mu\beta}\mathstrut^{\alpha},
\end{equation}
\begin{equation}\label{bianchi}
    \nabla_{[\mu}R_{\nu\alpha]\beta}\mathstrut^{\sigma}=0,
\end{equation}
\begin{equation}\label{Rsecond}
    R_{[\mu\nu\beta]}\mathstrut^{\alpha}=0.
\end{equation}
Eq. \ref{bianchi} is known as the \textit{Bianchi identity}. Eqs. \ref{Rfirst}, \ref{Rsecond} and \ref{Rthird} also imply that
\begin{equation}\label{Rfourth}
    R_{\mu\nu\alpha\beta}=R_{\alpha\beta\mu\nu}.
\end{equation}

The trace of the Riemann tensor over the second and fourth indices (or equivalently, over the first and third) defines the \textit{Ricci tensor}
\begin{equation}
    R_{\mu\nu}=R_{\mu\alpha\nu\beta}g^{\alpha\beta}=R_{\mu\alpha\nu}\mathstrut^{\alpha},
\end{equation}
which is symmetric, as a consequence of eq. \ref{Rfourth}. Finally, the \textit{Ricci scalar} is defined as the trace of the Ricci tensor,
\begin{equation}
    R=R_{\mu\nu}g^{\mu\nu}=R_{\mu}\mathstrut^{\mu}.
\end{equation}
Given a coordinate system, the components of these tensors can be derived from \cite{Wald1984} 
\begin{equation}
    R_{\mu\nu\beta}\mathstrut^{\alpha}=\partial_{\nu}\Gamma^{\alpha}\mathstrut_{\mu\beta}-\partial_{\mu}\Gamma^{\alpha}\mathstrut_{\nu\beta}+\Gamma^{\lambda}\mathstrut_{\mu\beta}\Gamma^{\alpha}\mathstrut_{\lambda\nu}-\Gamma^{\lambda}\mathstrut_{\nu\beta}\Gamma^{\alpha}\mathstrut_{\lambda\mu}.
\end{equation}

By contracting the Bianchi identity, over the second and fourth indices of the Riemann tensor, one finds 
\begin{equation}
    \nabla^{\mu}\left(R_{\mu\nu}-\frac{1}{2}Rg_{\mu\nu}\right)=0.
\end{equation}
This result is of significant importance, as by the field equations postulated by general relativity (see eq. \ref{eq1}), it implies that $\nabla_{\mu}T^{\mu\nu}=0$, i.e., the energy-momentum tensor is locally conserved. Under further assumptions, it is also possible to show \cite{Geroch1975} that the geodesic hypothesis is also a consequence of the contracted Bianchi identity. That is, free \textit{massive} point-like objects follow timelike geodesics. A \textit{geodesic}, $\gamma$, is a curve whose tangent vector is parallel propagated along itself, i.e., a curve whose tangent, $s^{\mu}$, satisfies the \textit{geodesic equation},
\begin{equation}\label{aeq7}
    s^{\mu}\nabla_{\mu}s^{\nu}=f s^{\nu},
\end{equation}
where $f$ is an arbitrary function on the curve. In particular, the geodesic equation can be interpreted as stating that the variation of a curve's tangent is given only in its direction, that is, the vector is varying ``as little'' as possible. In this sense, geodesics can be interpreted as curves that are as ``straight'' as possible in a manifold. Indeed, this is precisely the content of the hypothesis that freely falling bodies follow ``locally straight'' curves.

Furthermore, a geodesic is said to be \textit{affine parametrized} if it satisfies the condition 
\begin{equation}\label{aeq8}
    s^{\mu}\nabla_{\mu}s^{\nu}=0.
\end{equation}
If the tangent vector to $\gamma(t)$ respects eq. \ref{aeq7}, then the tangent vector to $\gamma(t')$ obeys eq. \ref{aeq8}, where $dt'=e^{\int f(t)dt}dt$. In other words, any geodesic can be reparametrized to respect eq. \ref{aeq8}.

Let $(M,g_{\mu\nu})$ be an $n$-dimensional manifold and $\gamma_s(\lambda)$ denote a one-parameter system of affinely parametrized geodesics. This system can be identified in $M$ as the image of the map $\psi:M'\to M$, where $M'$ is a two-dimensional strip of the plane $(\lambda,s)$, with $0<\lambda<\lambda'$ and $-\epsilon<s<\epsilon$ (see fig. \ref{fig:sub1}). The image of the map $\psi$ for lines ${s=\text{constant}}$ is affinely parametrized geodesics, as illustrated in \ref{fig:sub2}.\begin{figure}[h]
\begin{subfigure}[t]{0.5\textwidth}
  \centering
    \includegraphics[scale=1.2]{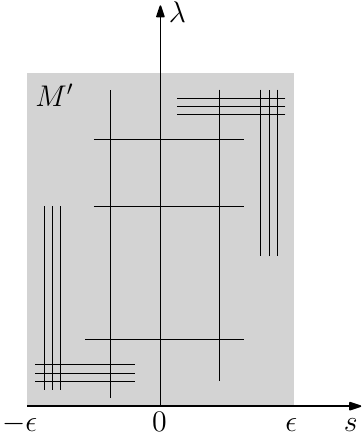}
    \caption{Manifold, $M'$, which is a strip of the plane $(\lambda,s)$ and whose image under the map $\psi$ is a family of null geodesics.}
    \label{fig:sub1}
  \end{subfigure}
  \hfill
  \begin{subfigure}[t]{0.5\textwidth}
  \centering
    \includegraphics[scale=1.2]{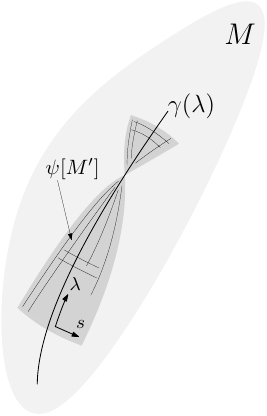}
    \caption{Embedding of $M'$ in $M$ through the map $\psi$. The geodesic $\gamma$ is parametrized by $\lambda$, and the deviation vectors are $(\partial_s)^{\mu}$.}
    \label{fig:sub2}
  \end{subfigure}
  \caption{One-parameter system of geodesics in $M'$ and $M$.}
  \caption*{Source: Adapted from PENROSE \cite{Penrose1972}.}
  \label{fig:sub3}
\end{figure} One may choose $s$ and $\lambda$ as coordinates of $M'$, that is, $(\partial_{\lambda})^{\mu}$ and $(\partial_s)^{\mu}$ comprise a coordinate basis. More precisely, $\ell^{\mu}=(\partial_{\lambda})^{\mu}$ is the tangent vector to the family of geodesics, and $s^{\mu}=(\partial_s)^{\mu}$ is known as the \textit{deviation vector}, which can be interpreted as the displacement between two ``infinitesimally nearby'' geodesics.
By being coordinate vectors, their commutator vanishes, and from eq. \ref{commu} one then has
\begin{equation}\label{aeq10}
    \ell^{\nu}\nabla_{\nu}s^{\mu}=s^{\nu}\nabla_{\nu}\ell^{\mu},
\end{equation}
Namely, the vector $\ell^{\nu}\nabla_{\nu}s^{\mu}$ can be interpreted as the relative velocity of infinitesimally nearby geodesics, as it measures the rate of change of $s^{\mu}$ along geodesics. Correspondingly, the vector $\ell^{\nu}\nabla_{\nu}(\ell^{\alpha}\nabla_{\alpha}s^{\mu})$ can be interpreted as the relative acceleration. 

In this sense, one can readily conclude that the relative acceleration between nearby geodesics is directly related to the Riemann tensor, as
\begin{equation}\label{aeq12}
\begin{aligned}[b]
\ell^{\nu}\nabla_{\nu}(\ell^{\alpha}\nabla_{\alpha}s^{\mu}) 
        & =\ell^{\nu}\nabla_{\nu}(s^{\alpha}\nabla_{\alpha}\ell^{\mu})\\
        & =\ell^{\nu}\nabla_{\nu}s^{\alpha}(\nabla_{\alpha}\ell^{\mu})+\ell^{\nu}s^{\alpha}\nabla_{\nu}\nabla_{\alpha}\ell^{\mu}\\
        & =s^{\nu}\nabla_{\nu}\ell^{\alpha}(\nabla_{\alpha}\ell^{\mu})+\ell^{\nu}s^{\alpha}(\nabla_{\alpha}\nabla_{\nu}\ell^{\mu}-R_{\nu\alpha\beta}\mathstrut^{\mu}\ell^{\beta})\\
        & =-R_{\nu\alpha\beta}\mathstrut^{\mu}\ell^{\nu}s^{\alpha}\ell^{\beta}.
\end{aligned}
\end{equation}
Eq. \ref{aeq12} is known as the \textit{Jacobi equation}, which tells one that geodesics will accelerate in comparison with one another if and only if $R_{\mu\nu\beta}\mathstrut^{\alpha}\neq 0$.

\section{Integration}\label{integration}

Let $M$ be an $n$-dimensional manifold. If there exists an everywhere nonvanishing $n$-form, $s_{\mu_1\cdots\mu_{n}}$, on $M$, then $M$ is said to be \textit{orientable} and $s_{\mu_1\cdots\mu_{n}}$ is said to provide an \textit{orientation}. Due to theorem \ref{propform}, an orientation must be related to any other $n$-form by scalar multiplication. Hence, given a coordinate system, $\psi$, an orientation can be written as (see eq. \ref{ori})
\begin{equation}
    s_{\mu_1\cdots\mu_{n}}=f(x^0,\ldots,x^{n-1})n!(dx^0)_{[\mu_1}\cdots(dx^{n-1})_{\mu_{n}]}=f(x^0,\ldots,x^{n-1})(dx^{0}\wedge\ldots\wedge dx^{n-1})_{\mu_1\cdots\mu_{n}}.
\end{equation}
If $f>0$, the integral of an $n$-form,
\begin{equation}
    w_{\mu_1\cdots\mu_{n}}=f'(x^0,\ldots,x^{n-1})(dx^{0}\wedge\ldots\wedge dx^{n-1})_{\mu_1\cdots\mu_{n}},
\end{equation}
in the open region $S\subset M$ covered by the coordinate system $\psi$ is defined as 
\begin{equation}\label{aeq102}
    \int_{S}w_{\mu_1\cdots\mu_{n}}=\int_{\psi[S]}f'dx^0\ldots dx^{n-1},
\end{equation}
i.e., one uses a coordinate system to perform the integration process on $\mathbb{R}^n$. 

Consequently, the sign of the function $f$ determines the sign of the integration of $n$-forms that arises due to the coordinate system $\psi$. More precisely, eq. \ref{aeq102} would have been defined with a minus sign if $f<0$. In particular, a coordinate system such that the scalar $f$ in the orientation is positive is said to be \textit{right handed}, and if $f<0$, it is said to be \textit{left handed}. This local definition of integration of an $n$-form can be expanded to the entire manifold if it is paracompact \cite{Wald1984}. Such an expansion can be used to state the following result, known as \textit{Stokes' theorem} \cite{Darling1994}.
\begin{theorem}
    Let $M$ be an $n$-dimensional oriented manifold, $S$ be an $n$-dimensional compact oriented submanifold with boundary and let $s_{\mu_1\cdots\mu_{n-1}}$ be an $(n-1)$-form on $M$ which is $C^1$. Then
    \begin{equation}
        \int_{S}(ds)_{\nu\mu_1\cdots\mu_{n-1}}=\int_{\partial S}s_{\mu_1\cdots\mu_{n-1}}.
    \end{equation}
\end{theorem}

On the other hand, integration of functions over $M$ can be defined up to a choice of sign (as a consequence of the arbitrariness in the choice of orientation) when one has a metric defined on it. Such integration is defined by
\begin{equation}\label{aeq1010}
    \int_Mf=\int_Mf\epsilon_{\mu_1\cdots\mu_{n}},
\end{equation}
where $\epsilon_{\mu_1\cdots\mu_{n}}$ is an nonvanishing $n$-form, referred to as a \textit{volume element}, which obeys
\begin{equation}\label{form1}
    \epsilon^{\mu_1\cdots\mu_{n}}\epsilon_{\mu_1\cdots\mu_{n}}=(-1)^s n!,
\end{equation}
where $s$ is the number of minuses in the signature of the metric. Note that eq. \ref{form1} depends not only on the signature of the metric, but also on its explicit form, as the “raised” indices in $\epsilon^{\mu_1\cdots\mu_{n}}$ are a result of multiple contractions with the $g^{\mu\nu}$. Moreover, differentiating eq. \ref{form1} yields
\begin{equation}\label{54}
    \nabla_{\nu}\epsilon_{\mu_1\cdots\mu_{n}}=0.
\end{equation}

Considering that the vector space of antisymmetric tensors of rank $(n,n)$ in an $n$-dimensional manifold is one-dimensional, one can also deduce that
\begin{equation}\label{form2}
    \epsilon^{\mu_1\cdots\mu_{\alpha}\mu_{\alpha+1}\cdots\mu_{\ell}}\epsilon_ {\mu_1\cdots\mu_{\alpha}\nu_{\alpha+1}\cdots\nu_{\ell}}=(-1)^s(\ell-\alpha)!\alpha!\delta^{[\mu_{\alpha+1}}\mathstrut_{\nu_{\alpha+1}}\cdots\delta^{\mu_{\ell}]}\mathstrut_{\nu_{\ell}}.
\end{equation}
Eq. \ref{form2} is merely a consequence of the fact that the outer product $\epsilon^{\mu_1\cdots\mu_{n}}\epsilon_{\nu_1\cdots\nu_{n}}$ must be proportional to the antisymmetrized outer product of multiple identity maps, $\delta^{\mu}\mathstrut_{\nu}$. Furthermore, from eq. \ref{form1}, the coordinate components of $\epsilon_{\mu_1\cdots\mu_{n}}$ satisfy
\begin{equation}
    g^{\mu_1\nu_1}\cdots g^{\mu_n\nu_n}\epsilon_{\mu_1\cdots\mu_{n}}\epsilon_{\nu_1\cdots\nu_{n}}=(-1)^sn!,
\end{equation}
which yields
\begin{equation}
    \epsilon_{12\cdots n}=[(-1)^s\det(g_{\mu\nu})]^{1/2}=\sqrt{|g|}.
\end{equation}
More precisely, in any right handed coordinate system, the volume element is
\begin{equation}\label{volumeexplicit}
    \epsilon_{\mu_1\cdots\mu_{n}}=\sqrt{|g|}(dx^{0}\wedge\ldots\wedge dx^{n-1})_{\mu_1\cdots\mu_{n}}.
\end{equation}
Thus, one can interpret the integration of any $n$-form in an $n$-dimensional manifold, eq. \ref{aeq102}, simply as the integration of a function, eq. \ref{aeq1010}. This is merely a consequence of the fact that any $n$-form is proportional to the volume element. 

It is also of interest to discuss the form of the volume element of a lower dimensional embedded submanifold of $M$. As discussed at the end of appendix \ref{A2}, one can identify the subspace of vectors that are orthogonal to a vector, $\chi^{\mu}$, with the tangent vector space of a hypersurface, $S$, for which $\chi^{\mu}$ is said to be normal to. Without loss of generality, consider that $\chi^{\mu}$ has norm $\pm1$ and note that the volume element of $S$, $\epsilon_{\mu_1\cdots\mu_{n-1}}$, can be identified with the metric that arises as one restricts the action of $g_{\mu\nu}$ to vectors that are orthogonal to $\chi^{\mu}$. This metric, denoted by $h_{\mu\nu}$, has to obey $h^{\mu\alpha}h_{\alpha\nu}=\delta^{\mu}\mathstrut_{\nu}$, as well as
\begin{equation}\label{eq345}
    h^{\mu\nu}\chi_{\nu}=h^{\mu\nu}\chi_{\mu}=0,
\end{equation}
which means that it is ``orthogonal'' to $\chi^{\mu}$. Additionally, in order to identify $h^{\mu}\mathstrut_{\nu}$ as an \textit{orthogonal projection operator} \cite{Horn2013}, one also requires it to obey
\begin{equation}\label{eq3456}
    h^{\mu}\mathstrut_{\alpha}h^{\alpha}\mathstrut_{\nu}=h^{\mu}\mathstrut_{\nu},
\end{equation}
and
\begin{equation}\label{eq34567}
    g_{\mu\alpha}h^{\alpha}\mathstrut_{\nu}=g_{\alpha\mu}h^{\alpha}\mathstrut_{\nu}.
\end{equation}

Respecting these properties, $h_{\mu\nu}$ can take two forms, depending on the nature of the normal vector\footnote{The case $\chi_{\mu}\chi^{\mu}=0$ will be analyzed in detail in \S\;\ref{null}.},
\begin{equation}\label{metric1}
   h_{\mu\nu}=
     \begin{cases}
      g_{\mu\nu}+\chi_{\mu}\chi_{\nu} \text{ for } \chi_{\mu}\chi^{\mu}<0, \\
      g_{\mu\nu}-\chi_{\mu}\chi_{\nu} \text{ for } \chi_{\mu}\chi^{\mu}>0.
    \end{cases}
\end{equation}
Now, consider the $(n-1)$-form $\epsilon_{\mu_1\cdots\mu_{n-1}\nu}\chi^{\nu}$. This tensor is ``orthogonal'' to any vector proportional to $\chi^{\mu}$, in the sense that contraction of any of its indices with it will vanish. Thus, $\epsilon_{\mu_1\cdots\mu_{n-1}\nu}\chi^{\nu}$ is a tensor over the tangent vector space of any $a\in S$, so it must be related by a scalar to the volume element of $S$. Note that this is the case because the subspace spanned by vectors orthogonal to $\chi^{\mu}$ is $(n-1)$-dimensional, so that $(n-1)$-forms over $S$ span a one-dimensional vector space. However, since the volume element of a manifold must obey eq. \ref{form1}, from eq. \ref{form2} one can conclude that $\epsilon_{\mu_1\cdots\mu_{n-1}\nu}\chi^{\nu}$ itself is the volume element of $S$, since $\chi^{\mu}$ is normalized. Hence, in a right-handed coordinate system of $S$, one obtains
\begin{equation}
    \epsilon_{12\cdots n-1\;\nu}\chi^{\nu}=\sqrt{|h|},
\end{equation}
where $h=\text{det}(h_{\mu\nu})$. In a more general way, one can conclude the following. Given an orthonormal set of vectors, $\{\chi^{\mu},\ldots,\eta^{\mu}\}$, the volume element of an embedded submanifold, $S$, whose tangent space is spanned by vectors orthogonal to the set $\{\chi^{\mu},\ldots,\eta^{\mu}\}$ is given by $\epsilon_{\mu_1\cdots\mu_{n-\ell}\nu_1\cdots\nu_{\ell}}\chi^{\nu_1}\ldots\eta^{\nu_{\ell}}$. The coordinate components of this volume element in a right-handed coordinate system of $S$ are given by 
\begin{equation}\label{volumeexplicit2}
    \epsilon_{12\cdots{n-\ell}\;\nu_1\cdots\nu_{\ell}}\chi^{\nu_1}\ldots\eta^{\nu_{\ell}}=\sqrt{|h|},
\end{equation}
where $h_{\mu\nu}$ is the metric that rises when one restricts $g_{\mu\nu}$ to act on the subspace of vectors orthogonal to $\{\chi^{\mu},\ldots,\eta^{\mu}\}$.

Lastly, it is also useful to use the one-dimensional property of $n$-forms in an $n$-dimensional manifold to state two convenient forms of Stokes' theorem. Let $w_{\mu_1\cdots\mu_{n-1}}$ be an $(n-1)$-form such that
\begin{equation}
    w_{\mu_1\cdots\mu_{n-1}}=\epsilon_{\alpha\mu_1\cdots\mu_{n-1}}v^{\alpha},
\end{equation}
where $v^{\mu}$ is an arbitrary vector. Taking the exterior derivative of $w_{\mu_1\cdots\mu_{n-1}}$ yields
\begin{equation}
\begin{aligned}
    (dw)_{\beta\mu_1\cdots\mu_{n-1}}&=n\nabla_{[\beta}(\epsilon_{|\alpha|\mu_1\cdots\mu_{n-1}]}v^{\alpha})\\
    &=n\epsilon_{\alpha[\mu_1\cdots\mu_{n-1}}\nabla_{\beta]}v^{\alpha},
    \end{aligned}
\end{equation}
where eq. \ref{54} was used, and it was indicated in the first line that the index $\alpha$ is not to be antisymmetrized. Of course, $(dw)_{\beta\mu_1\cdots\mu_{n-1}}$ must be related to the volume element by a scalar,
\begin{equation}\label{aeq1021}
    f\epsilon_{\beta\mu_1\cdots\mu_{n-1}}=n\epsilon_{\alpha[\mu_1\cdots\mu_{n-1}}\nabla_{\beta]}v^{\alpha},
\end{equation}
which can be evaluated by applying $\epsilon^{\beta\mu_1\cdots\mu_{n-1}}$ to both sides of eq. \ref{aeq1021}. Using eq. \ref{form2}, this application yields
\begin{equation}
\begin{aligned}[b]
    fn!&=n(n-1)!\delta^{\beta}\mathstrut_{\alpha}\nabla_{\beta}v^{\alpha},\\
    f&=\nabla_{\alpha}v^{\alpha},
    \end{aligned}
\end{equation}
and thus,
\begin{equation}
    (dw)_{\mu_1\cdots\mu_{n}}=\epsilon_{\mu_1\cdots\mu_{n}}\nabla_{\alpha}v^{\alpha}.
\end{equation}
Hence, for a compact and bounded $S\subset M$ that is $n$-dimensional, Stokes' theorem can be expressed as
\begin{equation}\label{stokes2}
    \int_{S}\epsilon_{\mu_1\cdots\mu_{n}}\nabla_{\alpha}v^{\alpha}=\int_{\partial S}\epsilon_{\alpha\mu_1\cdots\mu_{n-1}}v^{\alpha}.
\end{equation}

Similarly, Stokes' theorem also allows one to find a relation for the integral of $w_{\mu_1\cdots\mu_{n-2}}=\epsilon_{\alpha\beta\mu_1\cdots\mu_{n-2}}A^{\alpha\beta}$ over a $(n-2)$-dimensional embedded submanifold of $M$ and the integral of an $(n-1)$-form over a $(n-1)$-dimensional embedded submanifold of $M$. Note that $A^{\alpha\beta}$ can be considered a two-form without loss of generality, since the symmetric part of the tensor would get ``filtered out'' by contraction with the volume element of $M$. Now, from the exterior derivative of the $(n-2)$-form,
\begin{equation}
    (dw)_{\lambda\mu_1\cdots\mu_{n-2}}=(n-1)\epsilon_{\alpha\beta[\mu_1\cdots\mu_{n-2}}\nabla_{\lambda]}A^{\alpha\beta},
\end{equation}
one can find the scalar that relates it to the $(n-1)$-form $\epsilon_{\beta\lambda\mu_1\cdots\mu_{n-2}}\nabla_{\alpha}A^{\alpha\beta}$. Again, note that this is the case because these tensors are defined to be restricted to act on the $(n-1)$-dimensional embedded submanifold of $M$ (i.e., they are ``tangent'' to it), so that $(n-1)$-forms span a one-dimensional vector space. This scalar can be found by using eq. \ref{form2},
\begin{equation}
\begin{aligned}[b]
    f\epsilon_{\beta\lambda\mu_1\cdots\mu_{n-2}}\nabla_{\alpha}A^{\alpha\beta}&=(n-1)\epsilon_{\alpha\beta[\mu_1\cdots\mu_{n-2}}\nabla_{\lambda]}A^{\alpha\beta},\\
    f\epsilon_{\beta\lambda\mu_1\cdots\mu_{n-2}}\epsilon^{\omega\lambda\mu_1\cdots\mu_{n-2}}\nabla_{\alpha}A^{\alpha\beta}(\nabla^{\rho}A_{\rho\omega})&=(n-1)\epsilon^{\omega\lambda\mu_1\cdots\mu_{n-2}}\epsilon_{\alpha\beta[\mu_1\cdots\mu_{n-2}}\nabla_{\lambda]}A^{\alpha\beta}(\nabla^{\rho}A_{\rho\omega}),\\
    f(n-1)!\delta^{\omega}\mathstrut_{\beta}\nabla_{\alpha}A^{\alpha\beta}(\nabla^{\rho}A_{\rho\omega})&=2(n-1)!\delta^{[\omega}\mathstrut_{\alpha}\delta^{\lambda]}\mathstrut_{\beta}\nabla_{\lambda}A^{\alpha\beta}(\nabla^{\rho}A_{\rho\omega}),\\
    f(n-1)!\nabla_{\alpha}A^{\alpha\omega}(\nabla^{\rho}A_{\rho\omega})&=2(n-1)!\nabla_{\alpha}A^{\alpha\beta}(\nabla^{\rho}A_{\rho\beta}),\\
    f&=2.
    \end{aligned}
\end{equation}
Thus, for a compact and bounded $S\subset M$ that is $(n-1)$-dimensional, Stokes' theorem can also be expressed as
\begin{equation}\label{stokes3}
    \int_{S}\epsilon_{\beta\mu_1\cdots\mu_{n}}\nabla_{\alpha}A^{\alpha\beta}=\frac{1}{2}\int_{\partial S}\epsilon_{\alpha\beta\mu_1\cdots\mu_{n-2}}A^{\alpha\beta}.
\end{equation}

\chapter{Kerr spacetime}\label{apb}

The purpose of this appendix is to provide details on several results concerning the Kerr metric, which will mainly be of use in chs. \ref{chapter2} and \ref{chapter3}.

\section{Inverse metric and Christoffel symbols}\label{B1}
The contravariant form of the Kerr metric can be derived from the Kerr-Newman metric with $e=0$, eq. \ref{kerrnewman}. First, its determinant reads
\begin{equation}
    g=-\Sigma^2\sin^2{\theta}.
\end{equation}
Second, the covariant form of the metric can be evaluated to be
\begin{equation}
    \begin{aligned}[b]
    \partial^2_s &=g^{ab}\partial_a\partial_b\\
         &=-\left(\frac{\Sigma'}{\Sigma\Delta}\right)\frac{1}{c^2}\partial^2_t-\left(\frac{2r_s a r }{\Sigma\Delta}\right)\frac{1}{c}\partial_t\partial_\phi+\frac{\Delta}{\Sigma}\partial^2_r+\frac{1}{\Sigma}\partial^2_\theta+\left(\frac{\Delta-a^2\sin^2{\theta}}{\Sigma\Delta\sin^2{\theta}}\right)\partial^2_\phi,
    \end{aligned}
\end{equation}
where
\begin{equation}
    \Sigma'=(r^2+a^2)\Sigma+r_sa^2r\sin^2{\theta},
\end{equation}
and the other variables are given by eq. \ref{kerrnewman}. Lastly, the nonvanishing Christoffel symbols of the Kerr metric are \cite{Muller2010}
\begin{align*}
\Gamma^{r}\mathstrut_{t\phi}&=-\frac{c\Delta r_sa\sin^2{\theta}(r^2-a^2\cos^2{\theta})}{2\Sigma^3}, & \Gamma^{r}\mathstrut_{tt}&=\frac{c^2r_s\Delta(r^2-a^2\cos^2{\theta})}{2\Sigma^3},\\
\Gamma^{r}\mathstrut_{rr}&=\frac{2ra^2\sin^2{\theta}-r_s(r^2-a^2\cos^2{\theta})}{2\Sigma\Delta},  &  \Gamma^{r}\mathstrut_{r\theta}&=-\frac{a^2\sin{\theta}\cos{\theta}}{\Sigma},\\
\Gamma^{r}\mathstrut_{\phi\phi}&=\frac{\Delta\sin^2{\theta}(-2r\Sigma^2+r_s a^2\sin^2{\theta}(r^2-a^2\cos^2{\theta}))}{2\Sigma^3}, &  \Gamma^{r}\mathstrut_{\theta\theta}&=-\frac{r\Delta}{\Sigma},\\
\Gamma^{t}\mathstrut_{r\phi}&=\frac{r_sa\sin^2{\theta}(a^2\cos^2{\theta}(a^2-r^2)-r^2(a^2+3r^2))}{2c\Sigma^2\Delta}, & \Gamma^{t}\mathstrut_{t\theta}&=-\frac{r_sa^2r\sin{\theta}\cos{\theta}}{\Sigma^2},\\
\Gamma^{t}\mathstrut_{tr}&=\frac{r_s(r^2+a^2)(r^2-a^2\cos^2{\theta})}{2\Sigma^2\Delta}, & \Gamma^{t}\mathstrut_{\theta\phi}&=\frac{r_sa^3r\sin^3{\theta}\cos{\theta}}{c\Sigma^2},\\
\Gamma^{\theta}\mathstrut_{\phi\phi}& =-\frac{\sin{\theta}\cos{\theta}(\Sigma'\Sigma+(r^2+a^2)r_sa^2r\sin^2{\theta})}{\Sigma^3},  &  \Gamma^{\theta}\mathstrut_{\theta\theta}&=-\frac{a^2\sin{\theta\cos{\theta}}}{\Sigma},\\
\Gamma^{\theta}\mathstrut_{tt}&=-\frac{c^2r_sa^2r\sin{\theta}\cos{\theta}}{\Sigma^3}, &  \Gamma^{\theta}\mathstrut_{rr}&=\frac{a^2\sin{\theta}\cos{\theta}}{\Sigma\Delta},\\
 \Gamma^{\theta}\mathstrut_{t\phi}&=\frac{cr_sar(r^2+a^2)\sin{\theta}\cos{\theta}}{\Sigma^3}, & \Gamma^{\theta}\mathstrut_{r\theta}&=\frac{r}{\Sigma},\\
 \Gamma^{\phi}\mathstrut_{r\phi}&=\frac{2r\Sigma^2+r_s(a^4\sin^2{\theta}\cos^2{\theta}-r^2(\Sigma+r^2+a^2))}{2\Sigma^2\Delta}, &  \Gamma^{\phi}\mathstrut_{t\theta}&=-\frac{cr_sar\cot{\theta}}{\Sigma^2},\\
 \Gamma^{\phi}\mathstrut_{tr}&=\frac{cr_sa(r^2-a^2\cos^2{\theta})}{2\Sigma^2\Delta}, &  \Gamma^{\phi}\mathstrut_{\theta\phi}&=\frac{\cot{\theta}(\Sigma^2+r_sa^2r\sin^2{\theta})}{\Sigma^2}.
\end{align*}

\section{Neighboring Kerr black holes}\label{B2}

In order to evaluate the infinitesimal variations to first order of the parameters and area of a Kerr black hole discussed in \S\;\ref{kerr}, it is useful to do so by analyzing the difference of such parameters between two neighboring (in the sense that the difference in parameter value is infinitesimal) states of a Kerr black hole. The framework presented is based on \cite{Bardeen1973}. By considering that the original black hole is described by the metric $g_{\mu\nu}$ and the perturbed one is described by $g'_{\mu\nu}$, one can introduce a covariant tensor to represent the perturbation by
\begin{equation}\label{B000}
    \gamma_{\mu\nu}=g'_{\mu\nu}-g_{\mu\nu}.
\end{equation}
Using the symbol $\delta$ to represent the variation of quantities over the perturbation, and requiring that the relation $g^{\mu\alpha}g_{\alpha\nu}=\delta^{\mu}\mathstrut_{\nu}$ holds in first order of $\gamma_{\mu\nu}$, one must have that
\begin{equation}
    \delta g_{\mu\nu}=\gamma_{\mu\nu},
       \end{equation}
    \begin{equation}
    \delta g^{\mu\nu}=-\gamma^{\mu\nu}.
\end{equation}

To evaluate the quantities of interest, it will be necessary to study the variation of tensors that are relevant to the description of the event horizon. Since there is a gauge freedom associated with the mapping of the two spacetimes (i.e., how one compares events of one with the other), it is possible to use this freedom to preserve the actions of the isometries represented by $\xi^{\mu}$ and $\psi^{\mu}$, as well as the position of the event horizon of the two solutions, which correspond to requiring that
\begin{equation}\label{Beq100}
    \delta\xi^{\mu}=0,
\end{equation}
\begin{equation}\label{Beq11}
    \delta\psi^{\mu}=0,
\end{equation}
which in combination with the definition of $\gamma_{\mu\nu}$, yield
\begin{equation}\label{Beq12}
    \delta\xi_{\mu}=\gamma_{\mu\nu}\xi^{\nu},
\end{equation}
\begin{equation}\label{Beq13}
    \delta\psi_{\mu}=\gamma_{\mu\nu}\psi^{\nu}.
\end{equation}
Using the null vector tangent to the horizon, eq. \ref{tangentH}, one also has
\begin{equation}\label{Beq2}
    \delta\chi^{\mu}=\delta\xi^{\mu}+\Omega_H\delta\psi^{\mu}+(\delta\Omega)\psi^{\mu}=(\delta\Omega)\psi^{\mu},
\end{equation}
\begin{equation}\label{Beq3}
    \delta\chi_{\mu}=\gamma_{\mu\nu}\chi^{\nu}+(\delta\Omega)\psi_{\mu},
\end{equation}
which shows that even though the event horizon is the same position, its angular velocity might still change.

Moving forward, by applying the Lie derivative along the Killing vectors, $\xi^{\mu}$ and $\psi^{\mu}$, on eq. \ref{B000} and taking the trace, yields the following relations for $\gamma=\gamma^{\mu\nu}g_{\mu\nu}$,
\begin{equation}\label{Beq5}
     \xi^{\mu}\nabla_{\mu}\gamma=\psi^{\mu}\nabla_{\mu}\gamma=\chi^{\mu}\nabla_{\mu}\gamma=0.
\end{equation}
The invariance of the position of the event horizon also ensures that the scalars $\chi_{\mu}\xi^{\mu}$, $\chi_{\mu}\psi^{\mu}$ and $\chi_{\mu}\eta^{\mu}$ are invariant, where $\eta^{\mu}$ is the auxiliary null vector which can be defined in both solution to be given by \ref{choice}. For such contractions to be invariant, it is necessary that
\begin{equation}\label{Beq6}
    \xi^{\mu}(\delta\chi_{\mu})=0,
\end{equation}
\begin{equation}\label{Beq7}
    \psi^{\mu}(\delta\chi_{\mu})=0,
\end{equation}
\begin{equation}\label{Beq8}
    \eta^{\mu}(\delta\chi_{\mu})=-\chi_{\mu}(\delta\eta^{\mu}).
\end{equation}
Finally, the invariance of the event horizon also implies that $\delta\chi_{\mu}$ and $\chi_{\mu}$ are proportional, so the perturbation is such that 
\begin{equation}\label{Beq9}
    \delta\chi_{\mu}=\alpha\chi_{\mu},
\end{equation}
which coupled with the fact that $\mathcal{L}_{\chi}(\delta\chi_{\mu})$ must vanish, yields
\begin{equation}\label{Beq10}
    \chi^{\mu}\nabla_{\mu}\alpha=0,
\end{equation}
\begin{equation}\label{Beq1}
    \chi^{\mu}\nabla_{\mu}(\delta\chi_{\nu})=(\delta\chi_{\mu})\nabla^{\mu}\chi_{\nu}.
\end{equation}
These are almost all the relations necessary for the analysis of the variation of the parameters and the area of a Kerr black hole. The last one, which will be derived below, is a general result \cite{Wald1984} that can be used to relate the variation of the mass to the variation of asymptotic properties of each solution.

Consider two arbitrary vector fields, $s^{\mu}$ and $w^{\mu}$, such that $\nabla_{\mu}s^{\mu}=\nabla_{\mu}w^{\mu}=0$ and $[s,w]^{\mu}=0$. It can be readily verified that $\nabla_{\mu}(s^{[\mu}w^{\nu]})=0$, which implies that \cite{Wald1984}
\begin{equation} 
    \nabla_{[\mu}(\epsilon_{\nu\alpha]\beta\delta}s^{\beta}w^{\delta})=0.
\end{equation}
By integrating this three-form over a volume bounded by two-spheres, $S$ and $S'$, Stokes' theorem allows one to write
\begin{equation}
    \int_S\epsilon_{\nu\alpha\beta\delta}s^{\beta}w^{\delta}=\int_{S'}\epsilon_{\nu\alpha\beta\delta}s^{\beta}w^{\delta}.
\end{equation}
The significance of this general result to the developments in ch. \ref{chapter2} lies in the case where $s^{\mu}=\xi^{\mu}$, the timelike Killing vector, and $w^{\mu}=\nabla_{\nu}(\gamma^{\mu\nu}-g^{\mu\nu}\gamma)$. Clearly, $\nabla_{\mu}\xi^{\mu}=0$, since it is a Killing vector, and $\nabla_{\mu}\w^{\mu}=0$, as it is the trace of the perturbed Einstein's equation \cite{Wald1984}. Furthermore, by arguments similar to those that led to eq. \ref{Beq5}, $\mathcal{L}_{\xi}w^{\mu}$ must vanish, and thus,  $\xi^{\mu}$ and $w^{\mu}$ commute. By taking $S$ to be a two-sphere at the asymptotic region and $S'=\mathscr{H}$, one then finds
\begin{equation}\label{generalresult}
    \int_S\epsilon_{\nu\alpha\beta\delta}\xi^{\delta}\nabla_{\rho}(\gamma^{\beta\rho}-g^{\beta\rho}\gamma)=\int_{\mathscr{H}}\epsilon_{\nu\alpha\beta\delta}\xi^{\delta}\nabla_{\rho}(\gamma^{\beta\rho}-g^{\beta\rho}\gamma).
\end{equation}
In order to develop the integral on the left hand side, it is convenient to evaluate the function $f$ that relates its two-form to the volume element of the two-sphere, $\epsilon_{\mu\nu\alpha\beta}\zeta^{\alpha}\tau^{\beta}$ (see eqs. \ref{normalt} and \ref{normalr}). In the asymptotic region, one has
\begin{equation}
\begin{aligned}[b]
f\epsilon_{\mu\nu\alpha\beta}\zeta^{\alpha}\tau^{\beta} &=\epsilon_{\mu\nu\alpha\beta}\xi^{\beta}\nabla_{\rho}(\gamma^{\alpha\rho}-g^{\alpha\rho}\gamma), \\
 2f&=\epsilon^{\mu\nu\lambda\omega}\zeta_{\lambda}\tau_{\omega}\epsilon_{\mu\nu\alpha\beta}\xi^{\beta}\nabla_{\rho}(\gamma^{\alpha\rho}-g^{\alpha\rho}\gamma),\\
 f&=-2(-g^{tt}g^{rr})^{-1/2}\delta^{[\lambda}\mathstrut_{\alpha}\delta^{\omega]}\mathstrut_{\beta}\delta^{r}\mathstrut_{\lambda}\delta^{t}\mathstrut_{\omega}\xi^{\beta}\nabla_{\rho}(\gamma^{\alpha\rho}-g^{\alpha\rho}\gamma)\\
 f&=-2c\delta^{[r}\mathstrut_{\alpha}\delta^{t]}\mathstrut_{\beta}\xi^{\beta}\nabla_{\rho}(\gamma^{\alpha\rho}-g^{\alpha\rho}\gamma)\\
 f&=-c\nabla_{\rho}(\gamma^{r\rho}-\eta^{r\rho}\gamma).
\end{aligned}    
\end{equation}

To evaluate $f$, one should first consider the asymptotic form of the Kerr metric, which, due to its asymptotic flat nature, is
\begin{equation}\label{1789}
    g_{\mu\nu}=\eta_{\mu\nu}+\gamma_{\mu\nu},
\end{equation}
with 
\begin{equation}
    \gamma_{ab}dx^adx^b=\frac{r_s}{r}c^2dt^2+\frac{r_s}{r}dr^2+\mathcal{O}(r^{-2}),
\end{equation}
where $(1-x)^{-1}\approx (1+x)$ was used. These relations also lead to
\begin{equation}
    \gamma^{ab}\partial_a\partial_b=\frac{r_s}{r}\frac{1}{c^2}\partial_t^2+\frac{r_s}{r}\partial_r^2+\mathcal{O}(r^{-2}),
\end{equation}
Hence, using eq. \ref{Christoffel divergence}, one finds
\begin{equation}
    \nabla_{\rho}\gamma^{r\rho}=\partial_r(\gamma^{rr})=-\frac{r_s}{r^2},
\end{equation}
\begin{equation}
    \nabla_{\rho}(\eta^{r\rho}\gamma)=0,
\end{equation}
which leads to,
\begin{equation}
    f=\frac{cr_s}{r^2}.
\end{equation}

Returning to the integral on the left hand side of eq. \ref{generalresult}, one obtains
\begin{equation}
\begin{aligned}[b]
\int_S\epsilon_{\nu\alpha\beta\delta}\xi^{\delta}\nabla_{\rho}(\gamma^{\beta\rho}-\eta^{\beta\rho}\gamma) &= \int_Sf\epsilon_{\mu\nu\alpha\beta}\zeta^{\alpha}\tau^{\beta}\\
 &=cr_s\int_Sd\Omega\\
 &=\frac{8\pi G M}{c}.
\end{aligned}    
\end{equation}
This relation can be interpreted as using the asymptotic properties of spacetime to compute its total mass. If one considers this relationship for each individual neighboring Kerr black hole, one can write the difference of their masses as
\begin{equation}\label{admmass}
    \int_S\epsilon_{\nu\alpha\beta\delta}\xi^{\delta}\nabla_{\rho}(\gamma^{\beta\rho}-g^{\beta\rho}\gamma) =\frac{8\pi G \delta  M}{c},
\end{equation}
where the perturbation of the metric is now given by eq. \ref{B000}, which is simply the difference between each individual one, given in eq. \ref{1789}, and $g^{\mu\nu}$ is now given by the metric of the original stationary configuration.

\section{Principal null congruences}\label{congrukerr}
In order to find the equations that describe the behavior of null geodesics, it is useful to consider the conserved quantities associated with them. We will do so here, adopting geometrized units.

In Kerr spacetime, these constants follow from the norm of the null vector and the existence of Killing vectors. More precisely, there are the constants given by eqs. \ref{energy} and \ref{angularmomentum},
\begin{equation}
    E=-g_{tt}\left(\frac{dt}{d\lambda}\right)-g_{t\phi}\left(\frac{d\phi}{d\lambda}\right),
\end{equation}
\begin{equation}
    L=-g_{t\phi}\left(\frac{dt}{d\lambda}\right)-g_{\phi\phi}\left(\frac{d\phi}{d\lambda}\right),
\end{equation}
and the one from the norm of the null vector,
\begin{equation}
    g_{tt}\left(\frac{dt}{d\lambda}\right)^2+2g_{t\phi}\left(\frac{dt}{d\lambda}\right)\left(\frac{d\phi}{d\lambda}\right)+g_{rr}\left(\frac{dr}{d\lambda}\right)^2+g_{\phi\phi}\left(\frac{d\phi}{d\lambda}\right)^2+g_{\theta\theta}\left(\frac{d\theta}{d\lambda}\right)^2=0,
\end{equation}

Similarly, there is also another conserved quantity which arises due to the fact that the Kerr metric possesses a \textit{Killing tensor of valence $2$} \cite{Walker1970}, $K_{\mu\nu}$, which is a symmetric tensor that obeys \cite{Wald1984}
\begin{equation}
    \nabla_{(\mu}K_{\nu\alpha)}=0.
\end{equation}
Although Killing tensors do not rise as generators of a group of diffeomorphisms, they do give rise to conserved quantities along geodesics by the exact same line of reasoning as the one presented in \S~\ref{symmetry}. In essence, the conserved quantity in null geodesics can be expressed by
\begin{equation}
    \frac{1}{\Delta}\left[(r^2+a^2)E-aL\right]^2-\frac{\Sigma^2}{\Delta}\left(\frac{dr}{d\lambda}\right)^2=C,
\end{equation}
\begin{equation}
    (aE\sin{\theta}-L\csc{\theta})^2+\Sigma^2\left(\frac{d\theta}{d\lambda}\right)^2=C.
\end{equation}
Consequently, these results yield the following equations of motion
\begin{equation}
    \Sigma\left(\frac{dt}{d\lambda}\right)^2=\frac{1}{\Delta}\left\{\left[(r^2+a^2)^2-\Delta a\sin^2{\theta}\right]E-ar_srL\right\},
\end{equation}
\begin{equation}
    \Sigma^2\left(\frac{dr}{d\lambda}\right)^2=\left[(r^2+a^2)E-aL\right]-\Delta C,
\end{equation}
\begin{equation}
    \Sigma^2\left(\frac{d\theta}{d\lambda}\right)^2=-\left(aE\sin{\theta}-L\csc{\theta}\right)^2+C,
\end{equation}
\begin{equation}
    \Sigma\left(\frac{d\phi}{d\lambda}\right)^2=\frac{1}{\Delta}\left[ar_srE+(\Sigma-r_sr)L\csc^2{\theta}\right].
\end{equation}

The principal null congruences are those with $\theta(\lambda)=\theta_0$, which will be the case if $C=0$ and $L=aE\sin^2{\theta}$, resulting in
\begin{equation}
    \frac{dt}{d\lambda}=\frac{(r^2+a^2)E}{\Delta},
\end{equation}
\begin{equation}\label{rkerr}
    \frac{dr}{d\lambda}=\pm E,
\end{equation}
\begin{equation}
    \frac{d\theta}{d\lambda}=0,
\end{equation}
\begin{equation}
    \frac{d\phi}{d\lambda}=\frac{aE}{\Delta},
\end{equation}
in which one identifies those with $dr/d\lambda>0$ as outgoing, and those with $dr/d\lambda<0$ as incoming. Thus, the vector that generates the principal outgoing null congruence is
\begin{equation}
    \ell^{\mu}=\left(\frac{r^2+a^2}{\Delta},1,0,\frac{a}{\Delta}\right),
     \end{equation}
     while the one which generates the principal incoming null congruence is
    \begin{equation}\eta^{\mu}=\left(\frac{r^2+a^2}{2\Sigma},-\frac{\Delta}{2\Sigma},0,\frac{a}{2\Sigma}\right),
\end{equation}
being normalized by $\ell^{\mu}\eta_{\mu}=-1$. 

Before proceeding, it is convenient to define a new radial coordinate, $r'$, such that
\begin{equation}\label{tortoiseKerr}
    \frac{dr'}{dr}=\frac{r^2+a^2}{\Delta},
\end{equation}
and an angular coordinate, $\phi'$, 
\begin{equation}\label{phiKerr}
    \phi'=\phi-\Omega t,
\end{equation}
which is well behaved at $r=r_+$. In the same manner, it is also useful to define the incoming and outgoing null coordinates
\begin{equation}\label{null1kerr}
    u=t-r',
\end{equation}
\begin{equation}\label{null2kerr}
    w=t+r'.
\end{equation}

Now, to develop the calculations for the particle creation effect of a Kerr black hole, it will be necessary to evaluate the explicit form of the function $u(\lambda)$, where $\lambda$ is an affine parameter of an incoming null geodesic of the principal null congruence and $u$ is the incoming null coordinate. From eqs. \ref{null1kerr} and \ref{tortoiseKerr}, one finds
\begin{equation}\label{u}
    \begin{aligned}[b]
    \frac{du}{d\lambda}  & = \frac{dt}{d\lambda}-\frac{dr'}{dr}\frac{dr}{d\lambda}\\
    & = 2E\frac{r^2+a^2}{\Delta}.
    \end{aligned}
\end{equation}
Integrating eq. \ref{rkerr}, one obtains
\begin{equation}
    r-r_+=-E\lambda,
\end{equation}
where it was considered that $r(0)=r_+$. Using this result in eq. \ref{u} yields
\begin{equation}
    \begin{aligned}[b]
    \frac{du}{d\lambda}  & = 2E\frac{(r_+-E\lambda)^2+a^2}{(r_+-E\lambda)^2-(r_++r_-)(r_+-E\lambda)+r_+r_-}\\
    & = \frac{2[(r_+-E\lambda)^2+a^2]}{\lambda[E\lambda-(r_+-r_-)]},
    \end{aligned}
\end{equation}
in which eq. \ref{radius12} was used. Integration yields
\begin{equation}\label{dev}
    u(\lambda)=2E\lambda-\frac{1}{\kappa}\ln\left(\frac{\lambda}{C'}\right)+2\left(\frac{r_+-r_-}{r_-^2+a^2}\right)\ln\left[\frac{E\lambda-(r_+-r_-)}{C''}\right],
\end{equation}
where $\kappa$ is the surface gravity of a Kerr black hole, and $C'$ and $C''$ are negative constants. The limit of interest is as $r\to r_+$, in which one has
\begin{equation}\label{incoming}
    u(\lambda)\approx-\frac{1}{\kappa}\ln\left(\frac{\lambda}{C'}\right).
\end{equation}

\chapter{Information}\label{information} 

The purpose of this appendix is to provide an objective review of the quantification of information in quantum mechanics. Namely, the formalism of density operators, the von Neumann entropy, and how they lead to the concept of entanglement of systems. First, we note that even though there is an intrinsic probabilistic notion to most measurements in quantum mechanics, a lack of details about the precise state of a system can introduce another layer of uncertainty, which is precisely the scenario in which one would justifiably state that there is a ``lack of information'' in a quantum system. With this in mind, we review how the formalism of density operators naturally takes into account ``ensembles'' of pure states, which one can argue to be a situation in which one does not have all the possible information about the precise state of the system. Indeed, as measured by the von Neumann entropy, the states known as mixed ones can be interpreted as having an additional ``degree of uncertainty''. Consequently, when dealing with composite systems, the concept of entanglement identifies when the total state of a system is such that measurements in one subsystem will affect measurements in the other, even if they are spacelike separated. 

Again, this is not meant to serve as a pedagogical introduction, but rather, as a reference for the development of relations and the arguments presented in this work, mainly in chs. \ref{chapter3} and \ref{chapter4}. The interested reader can find a more pedagogical, detailed presentation of such subjects in the references on which the construction of this appendix was based on \cite{Sakurai1994, Blum2012, Nielsen2010, Horodecki2009}.

\section{Density operators}\label{informationtheory}

In a classical sense, information about a system is completely characterized by the position and momentum of all the constituents of the energy distribution present in it. From the formalisms of classical mechanics, the complete evolution of the system can then be predicted with certainty. However, realistically, one can only measure averages of these quantities, so that statistics come into play as a consequence of the lack of information about the precise conditions of the system. The purpose of this section is to make these statements precise in a setting in which quantum concepts should be considered, i.e., how ``lack of information'' of a state in a quantum system can be precisely quantified. As it is known, complete information about a quantum system can be expressed as complete knowledge of the state that describes it, rather than the exact values of the physical observations one can make. In other words, although there is an intrinsic uncertainty present in quantum mechanics, which can be exemplified by the probabilistic nature of measurements of observables for which the system is not in any of its eigenstates, one can still identify how much of this uncertainty is a consequence of the theory itself, or ``lack of information'' about the precise state of the system. The first step to analyze these concepts in quantum mechanics is to identify scenarios in which one has complete information about the state of a system. Recall that the state of a system is described by a vector in a Hilbert space, $\mathscr{H}$. Note that since these statements regard vectors that are elements of a complex vector space, it is more convenient to denote them using the \textit{bra-ket} notation \cite{Sakurai1994}. 

Of course, given a basis, $\{|\mu\rangle\}$, an arbitrary state, $|\psi\rangle$, can be written as\footnote{The lower Greek index is merely a label of the complex coefficient, i.e., it is \textit{not} an abstract index.} 
\begin{equation}\label{pure}
    |\psi\rangle=\sum_{\mu}a_{\mu}|{\mu}\rangle,
\end{equation}
where $a_{\mu}\in\mathbb{C}$ and the sum is over $d$, with $d=\text{dim}(\mathscr{H})$. In the following, summations should always be regarded as being over $d$, unless denoted otherwise. If the basis is orthonormal, i.e., $\langle\mu|{\nu}\rangle=\delta_{\mu\nu}$, then requiring that $|\psi\rangle$ be unit length translates to 
\begin{equation}
    \sum_{\mu}|a_{\mu}|^2=1.
\end{equation}
Furthermore, given a basis, the \textit{identity operator}, $\hat{I}$, can be expressed as \cite{Sakurai1994}
\begin{equation}\label{identity}
    \hat{I}=\sum_{\mu}|{\mu}\rangle\langle{\mu}|,
\end{equation}
while the \textit{trace operator}, $\text{Tr}(\cdot)$, is given by
\begin{equation}\label{trace}
    \text{Tr}(\hat{A})=\sum_{\mu}\langle{\mu}|\hat{A}|{\mu}\rangle,
\end{equation}
where $\hat{A}$ is an arbitrary operator. 

A state that can be written as in eq. \ref{pure} is known as a \textit{pure} state. All physically significant properties (such as expectation values and probabilities) of a pure state can be easily evaluated from the form given by eq. \ref{pure} simply by acting operators and states on it. Still, all of this information can also be analyzed through the lens of its \textit{density operator}, given by
\begin{equation}\label{density}
    \hat{\rho}=|\psi\rangle\langle\psi|.
\end{equation}
Evidently, the density operator of a pure state is precisely the \textit{projector operator} associated with the state, which is clearly Hermitian, i.e., 
\begin{equation}\label{hermitian}
    \hat{\rho}^{\dag}=\hat{\rho},
\end{equation}
 idempotent, 
 \begin{equation}
      \hat{\rho}^2=|\psi\rangle\langle\psi|\psi\rangle\langle\psi|=|\psi\rangle\langle\psi|=\hat{\rho},
 \end{equation}
 as per normalization of $|\psi\rangle$, and also respects
 \begin{equation}
      \text{Tr}(\hat{\rho})=\sum_{\mu}\langle{\mu}|\psi\rangle\langle\psi|{\mu}\rangle=\sum_{\mu}|a_{\mu}|^2=1.
  \end{equation}
  
In particular, expectation values for an arbitrary operator, $\hat{A}$, can be evaluated by
\begin{equation}\label{expectation}
\begin{aligned}[b]
\langle\psi|\hat{A}|\psi\rangle &= \sum_{\mu,\nu}\langle\psi|\mu\rangle\langle\mu|\hat{A}|\nu\rangle\langle\nu|\psi\rangle\\
&=\sum_{\mu,\nu}\langle\nu|\psi\rangle\langle\psi|\mu\rangle\langle\mu|\hat{A}|\nu\rangle\\
&=\sum_{\mu,\nu}\langle\nu|\hat{\rho}|\mu\rangle\langle\mu|\hat{A}|\nu\rangle\\
&=\sum_{\nu}\langle\nu|\hat{\rho}\hat{A}|\nu\rangle\\
&=\text{Tr}\left(\hat{\rho}\hat{A}\right),    
\end{aligned}
\end{equation}
where the identity operator, eq. \ref{identity} was applied twice in the first line. Similarly, probabilities for an arbitrary state, $|\psi'\rangle$, read
\begin{equation}\label{probability}
\begin{aligned}[b]
     |\langle\psi'|\psi\rangle|^2 &= \langle\psi'|\psi\rangle\langle\psi|\psi'\rangle\\
     &=\sum_{\mu}\langle\psi'|{\mu}\rangle\langle{\mu}|\psi\rangle\langle\psi|\psi'\rangle\\
     &=\sum_{\mu}\langle\psi'|{\mu}\rangle\langle{\mu}|\hat{\rho}|\psi'\rangle\\
     &=\sum_{\mu}\langle{\mu}|\hat{\rho}|\psi'\rangle\langle\psi'|{\mu}\rangle\\
     &= \text{Tr}\left(\hat{\rho}|\psi'\rangle\langle\psi'|\right).
     \end{aligned}
 \end{equation}
Lastly, \textit{unitary} dynamical evolution of the state, $|\psi\rangle$, of a \textit{closed} system reads
\begin{equation}\label{eqr}
    |\psi(t)\rangle=\hat{U}(t,t_0)|\psi(t_0)\rangle,
\end{equation}
where $\hat{U}(t_0,t_0)=\hat{I}$ and $\hat{U}^{\dagger}(t,t_0)\hat{U}(t,t_0)=\hat{I}$, which translates to 
\begin{equation}\label{eqs}
    \hat{\rho}(t)=\hat{U}(t,t_0)\hat{\rho}(t_0)\hat{U}^{\dagger}(t,t_0).
\end{equation}
Evidently, eqs. \ref{eqr} and \ref{eqs} do not correspond to the \textit{Heisenberg picture} \cite{Sakurai1994}, in which one would have dynamical evolution of observables rather than states. Nonetheless, one can use eqs. \ref{expectation} and \ref{probability} to effectively act $\hat{U}(t,t_0)$ on observables rather than $|\psi\rangle$, providing a completely equivalent formulation of the dynamics of the system. Because of this, one can refer to the dynamical evolution of a system by the evolution of the states or by the evolution of observables.

A simple reason to understand why density operators are a more interesting representation of states is that they are a \textit{unique} representation, as opposed to a ket, which can have an arbitrary phase, $e^{i\theta}$. Furthermore, density operators can be used to study more general states that cannot be simply represented by a ket. Such states are known as \textit{mixed} states, which can be interpreted as representing an \textit{ensemble} of pure states. In particular, the density operator of a more general state can be written as
\begin{equation}\label{densitygen}
    \hat{\rho}=\sum_{\mu}^Np_{\mu}|\psi_{\mu}\rangle\langle\psi_{\mu}|,
\end{equation}
where $\{|\psi_{\mu}\rangle\}$ is a set of normalized pure states, not necessarily orthogonal, so that the number $N$ is not limited by $d$. Namely, given a basis, for each element of $\{|\psi_{\mu}\rangle\}$, i.e., for a fixed $\mu$, one has
\begin{equation}
    |\psi_{\mu}\rangle=\sum_{\nu}a_{\mu\nu}|{\nu}\rangle,
\end{equation}
and normalization of each of them yields
\begin{equation}
    \sum_{\nu}|a_{\mu\nu}|^2=1.
\end{equation}
Additionally, the $N$ \textit{real} numbers, $p_{\mu}$, obey
\begin{equation}\label{mixedrelation}
    0<p_{\mu}\leq1,\quad \sum_{\mu}^Np_{\mu}=1.
\end{equation}
Evidently, this more general density operator is Hermitian (see eq. \ref{hermitian}) and respects
\begin{equation}
    \text{Tr}(\hat{\rho})=\sum_{\nu}\sum_{\mu}^Np_{\mu}\langle{\nu}|\psi_{\mu}\rangle\langle\psi_{\mu}|{\nu}\rangle=\sum_{\mu}^Np_{\mu}\sum_{\nu}|a_{\mu\nu}|^2=1.
\end{equation}
In this manner, the physically significant properties of a more general state whose density operator is given by eq. \ref{densitygen} follow from the same developments of eqs. \ref{expectation} and \ref{probability}, from which one can see that the numbers $p_{\mu}$ act as ``weights'' for each of the individual properties of the pure states in the set $\{|\psi_{\mu}\rangle\}$. Finally, comparison with eq. \ref{density} indicates that a  density operator will describe a pure state if and only if $N=1$, and thus, $p_1=1$.

Given a density operator, one can readily verify if it describes a pure or a mixed state. More specifically, since a density operator is always Hermitian, it can be diagonalized \cite{Sakurai1994}, i.e., written in the form 
\begin{equation}\label{eq30}
    \hat{\rho}=\sum_{\mu}\lambda_{\mu}|{\mu}\rangle\langle{\mu}|,
\end{equation}
where $\{|{\mu}\rangle\}$ is a basis and
\begin{equation}\label{diagrelation}
    0\leq\lambda_{\mu}\leq1,\quad \sum_{\mu}\lambda_{\mu}=1,
\end{equation}
are simply its eigenvalues. Now, from the fact that the density operator of a pure state is idempotent and has a unit trace, one must have 
\begin{equation}\label{tracepure}
    \text{Tr}(\hat{\rho}^2)=1.
\end{equation}
However, since the trace is independent of the choice of basis, this equality will only be respected if one of the eigenvalues of the density operator equals $1$ while all the others vanish. In contrast, if the density operator is of a mixed state, then none of its eigenvalues can be $1$, since this would mean that it could be written in the form of eq. \ref{density}. Hence, the density operator of a mixed state must obey 
\begin{equation}
    \text{Tr}(\hat{\rho}^2)<1,
\end{equation}
which means that 
\begin{equation}
     \hat{\rho}^2\neq\hat{\rho}.
\end{equation}

Consequently, one can use the trace of a squared density operator (or simply the squared density operator) to evaluate if it describes a pure or a mixed state. In particular, $\text{Tr}(\hat{\rho}^2)$ is known as the \textit{purity} of a state, and may take any value in the interval $[1/d,1]$, with the density operator describing a pure state if and only if it equals the upper limit. Evidently, for the trace of a squared density operator to obey the lower limit, the density operator must be proportional to the identity operator, i.e.,
\begin{equation}
    \hat{\rho}=\frac{\hat{I}}{d},
\end{equation}
which is said to describe a \textit{maximally mixed} state. Also, note that unitary evolution does not affect the eigenvalues of an operator, as
\begin{equation}\label{45}
    \begin{aligned}[b]
\text{det}(\hat{\rho}(t)-\hat{I}\lambda) & = \text{det}(\hat{U}(t,t_0)\hat{\rho}(t_0)\hat{U}^{\dagger}(t,t_0)-\hat{U}(t,t_0)\hat{U}^{\dagger}(t,t_0)\hat{I}\lambda)\\
& = \text{det}(\hat{U}(t,t_0)[\hat{\rho}(t_0)-\hat{I}\lambda]\hat{U}^{\dagger}(t,t_0))\\
& = \text{det}([\hat{\rho}(t_0)-\hat{I}\lambda]\hat{U}^{\dagger}(t,t_0)\hat{U}(t,t_0))\\
& = \text{det}(\hat{\rho}(t_0)-\hat{I}\lambda)\text{det}(\hat{U}^{\dagger}(t,t_0)\hat{U}(t,t_0))\\
& = \text{det}(\hat{\rho}(t_0)-\hat{I}\lambda).
    \end{aligned}
\end{equation}
Therefore, unitary evolution of a closed system does not change the nature of the state. In other words, the purity of a state is unaffected by such dynamical evolution.

Perhaps the most notable example of a mixed state is the one describing a system in thermal equilibrium \cite{Blum2012}, for instance, thermal radiation, whose density operator reads
\begin{equation}
    \hat{\rho}=\frac{e^{-\beta \hat{H}}}{Z},
\end{equation}
where $Z=\text{Tr}(e^{-\beta \hat{H}})$, $\beta=k_BT$, $T$ is the temperature of the system and $\hat{H}$ its Hamiltonian. One can verify that this density operator describes a mixed state by writing it in terms of a basis of eigenvalues of $\hat{H}$, $\{|\mu\rangle\}$, i.e.,
\begin{equation}
    \hat{H}|\mu\rangle=E_{\mu}|\mu\rangle,
\end{equation}
so that it takes the form
\begin{equation}
    \hat{\rho}=\frac{1}{Z}\sum_{\mu}e^{-\beta E_{\mu}}|\mu\rangle\langle\mu|.
\end{equation}
Comparison with eq. \ref{eq30} then shows that it describes a mixed state. 

Physically, a mixed state can be completely described by the way one ``prepares'' the ensemble it represents. For example, consider a system which has a probability $1/2$ to be in the state $|\psi\rangle$ and probability $1/2$ to be in the state $|\psi'\rangle$. The mixture described by this system is completely characterized by this statement, which, by eq. \ref{densitygen}, translates to
\begin{equation}\label{0001}
    \hat{\rho}=\frac{1}{2}(|\psi\rangle\langle\psi|+|\psi'\rangle\langle\psi'|).
\end{equation}
It is not difficult to see that this density operator does not equal the one describing a superposition of the states $|\psi\rangle$ and $|\psi'\rangle$ with equal amplitude. Indeed, such a density operator would be associated with the pure state\footnote{Up to a phase, $e^{i\theta}$.}
\begin{equation}
    |\psi\rangle=\frac{1}{\sqrt{2}}\left(|\psi\rangle+|\psi'\rangle\right),
\end{equation}
which from eq. \ref{density}, would read
\begin{equation}
     \hat{\rho}=\frac{1}{2}(|\psi\rangle\langle\psi|+|\psi\rangle\langle\psi'|+|\psi'\rangle\langle\psi|+|\psi'\rangle\langle\psi'|).
\end{equation}

At this point, it is clear that there is a fundamental difference between a superposition of states and a mixture. In a superposition, there is a definite phase relation between the possible states, which allows one to write the joint state of the system in the form of eq. \ref{pure}. However, for a mixture, one only has access to the probabilities to prepare a given state. The lack of a definite phase relation between the states, i.e., information only about fractional populations and not the complex numbers themselves, is precisely what makes it so that it is not possible to write a mixture as a single ket. Indeed, when dealing with a system in which one does not have complete information about the relative phase between the possible states, one refers to it as an \textit{incoherent mixture}\footnote{This terminology is justified when one considers a polarized and an unpolarized beam of light (see \cite{Sakurai1994} for a detailed discussion).}. 

\section{von Neumann entropy}\label{von}

The argument presented above points to the important idea that mixed states are states in which one does not have all the possible information regarding the state of a system. This concept can be quantified similarly to how one quantifies the uncertainty of a random variable through the \textit{Shannon entropy} \cite{Nielsen2010}. The Shannon entropy can be generalized to density operators by the \textit{von Neumann entropy}, which is given by
\begin{equation}
    S(\hat{\rho})=-\text{Tr}(\hat{\rho}\log_2\hat{\rho}),
\end{equation}
where $\log_2$ means that the von Neumann entropy is measured in \textit{bits}. To evaluate the von Neumann entropy associated with a density operator, one writes it in a diagonal form, which allows one to deduce that 
\begin{equation}
    S(\hat{\rho})=-\sum_{\mu}\lambda_{\mu}\log_2\lambda_{\mu}.
\end{equation}
From the properties of a density operator of pure and mixed states, and adopting
\begin{equation}
    \lim_{x\to 0^+}x\log_2{x}=0,
\end{equation}
one then finds
\begin{equation}
    S(\hat{\rho})=0\quad \text{for a pure state},
\end{equation}
\begin{equation}
    \log_2{d}\geq S(\hat{\rho})>0\quad \text{for a mixed state}.
\end{equation}

From this, it is evident that the interpretation of the von Neumann entropy is that it measures the ``missing'' information of a state described by $\hat{\rho}$. For instance, a pure state has no missing information, as it corresponds to the maximum possible information one can have regarding a quantum system. Of course, in such a case one still does not have the maximum degree of certainty of the values of all observables, but a pure state can be interpreted as the maximum amount of knowledge quantum mechanics allows one to have. In contrast, in a mixed state there is not only this fundamental lack of certainty as a consequence of quantum mechanics, but also there is a lack of knowledge regarding the state of the system.

\section{Entanglement}\label{entanglement}

Given these definitions and properties, it is natural to analyze how they generalize to composite systems. Let $|\psi\rangle\in\mathscr{H}_A\otimes\mathscr{H}_B$ denote a pure state of a composite system, comprised of subsystems $A$ and $B$. Let $\{|{\mu}\rangle_A\}$ and $\{|{\mu}\rangle_B\}$ denote orthonormal bases for the subsystems $A$ and $B$, respectively, so that $|\psi\rangle$ can be written as
\begin{equation}\label{ent12345}
    |\psi\rangle=\sum_{\mu,\nu}a_{\mu\nu}|{\mu}\rangle_A|{\nu}\rangle_B,
\end{equation}
where the sum is over $d\times d'$, with $d=\text{dim}(\mathscr{H}_A)$ and $d'=\text{dim}(\mathscr{H}_B)$, and $|{\mu}_A\rangle|{\nu}_B\rangle$ denotes $|{\mu}_A\rangle|\otimes|{\nu}_B\rangle$. Normalization of $|\psi\rangle$ translates to 
\begin{equation}
    \sum_{\mu,\nu}|a_{\mu\nu}|^2=1.
\end{equation}
The density matrix for $|\psi\rangle$ reads
\begin{equation}
    \hat{\rho}_{AB}=\sum_{\mu,\nu,\sigma,\epsilon}a_{\mu\nu}a^*_{\sigma\epsilon}|\mu\rangle_A\langle\nu|\otimes|\sigma\rangle_B\langle\epsilon|.
\end{equation}

Now, one can find an adequate operator to evaluate physically significant quantities in one of the subsystems by taking the trace of the density operator of the composite system over the other, which yields the \textit{reduced density operator}. For example, the reduced density operator for the system $A$ reads 
\begin{equation}
    \hat{\rho}_A=\sum_{\mu}\hat{I}_A\otimes{}_B\langle\mu|\hat{\rho}_{AB}|\mu\rangle_B\otimes\hat{I}_A=\text{Tr}_B(\hat{\rho}_{AB}),
    \end{equation}
where $\hat{I}_A$ is the identity operator of the system $A$. The density operator $\hat{\rho}_A$ then allows one to evaluate physically significant quantities of the subsystem $A$ using eqs. \ref{expectation} and \ref{probability}. Mathematically, the procedure to obtain $\hat{\rho}_A$ from $\hat{\rho}_{AB}$ can be thought of as ``tracing-out'' the subsystem $B$ (see \cite{Nielsen2010} for details). Physically, this would correspond to ``ignoring'' information regarding the subsystem $B$, which is useful when such a subsystem is irrelevant or inaccessible.

Consider a state, $|\psi\rangle\in\mathscr{H}_A\otimes\mathscr{H}_B$, that can be written in the form
\begin{equation}\label{ent1}
    |\psi\rangle=|\psi\rangle_A|\psi\rangle_B,
\end{equation}
where ${}_A\langle\psi|\psi\rangle_A={}_B\langle\psi|\psi\rangle_B=1$. The reduced density operator for the system $A$, $\hat{\rho}_A$, reads
\begin{equation}
     \hat{\rho}_A=\sum_{\mu}\hat{I}_A\otimes{}_B\langle\mu|\hat{\rho}_{AB}|\mu\rangle_B\otimes\hat{I}_A=|\psi\rangle_A\langle\psi|\sum_{\mu}{}_B\langle\mu|\psi\rangle_B\langle\psi|\mu\rangle_B=|\psi\rangle_A\langle\psi|,
\end{equation}
in which the normalization of $|\psi\rangle_B$ was used. Evidently, $\hat{\rho}_A$ describes a pure state, and an analogous development for system $B$ leads one to the conclusion that $\hat{\rho}_B$ also describes a pure state. Indeed, this is a consequence of the fact that since $|\psi\rangle$ can be written in the form of eq. \ref{ent1}, measurements made on one system do not affect measurements on the other. In other words, the action of an operator on system $A$ does not affect expectation values of the action of an operator on system $B$. 

In contrast, if the state $|\psi\rangle$ cannot be written in the form of eq. \ref{ent1}, then the systems $A$ and $B$ are said to be \textit{entangled}, and the state of the systems $A$ and $B$ each are individually mixed. More precisely, the reduced density matrices, $\hat{\rho}_A$ and $\hat{\rho}_B$, will individually describe a mixed state with the same purity, as a consequence of the \textit{Schmidt decomposition} \cite{Nielsen2010}. That is, although the total state of the composite system is completely known (i.e., it is described by a pure state), each individual subsystem has the same degree of ``lack of information''. This is a consequence of the fact that entangled systems are not irrelevant to each other in the process of measurement. For example, if two systems are entangled, it is possible to find an observable for system $A$, $\hat{A}$, and one for system $B$, $\hat{B}$, such that \cite{Unruh2017}
\begin{equation}\label{ent}
    \langle\psi|\hat{A}\otimes\hat{B}|\psi\rangle\neq\langle\psi|\hat{A}\otimes\hat{I}_B|\psi\rangle\langle\psi|\hat{I}_A\otimes\hat{B}|\psi\rangle.
\end{equation}
Physically, this means that there are nontrivial correlations between measurements made on both systems. 

Although entanglement may seem like an exotic feature arising from concepts of quantum mechanics (see \cite{Horodecki2009} for an extensive review), it is an intrinsic feature that was derived from the theory alone. In fact, in Schrödinger's words \cite{Schrödinger1935}:

\noindent\textit{``I would not call [entanglement] one but rather the characteristic trait of quantum mechanics, the one that enforces its entire departure from classical lines of thought.''}

Nowadays, experimental evidence of entanglement is undeniable, with many applications to quantum information theory and quantum communication technology, perhaps the most notable application being quantum teleportation (see, e.g., \cite{Bouwmeester1997, Ren2017}). The fundamental relevance of entanglement in the context of quantum field theory will be exemplified by the line of reasoning presented in \S~\ref{hadmard}.
 


\end{apendicesenv}

\phantompart
\printindex

\end{document}